\documentclass[12pt,reqno]{amsart}
\pdfoutput=1 
\usepackage{latexsym}
\usepackage{amsfonts}
\usepackage{amsthm,amssymb,mathtools,amsfonts,amsmath}
\usepackage{color}
\usepackage{graphicx}
\usepackage{url}
\usepackage{enumerate}
\usepackage{tikz-cd}
\usepackage{hyperref}
\usepackage[foot]{amsaddr}
\usepackage{stackrel}
\usepackage{cancel}
\usepackage{braket}
\usepackage[a4paper, margin=2cm]{geometry}
\usepackage{enumitem}
\usepackage{booktabs}
\usepackage{caption} 
\usepackage{arydshln}
\DeclareMathOperator{\tr}{tr}
\usepackage{comment}

\usepackage{nicematrix}
\usepackage{mathrsfs}  

\usepackage{ytableau}

\ytableausetup{
  mathmode,
  boxsize=0.9em,
  centertableaux
}

\usepackage{tikz}
\usetikzlibrary{arrows.meta}

\usepackage{thmtools}
\declaretheorem[parent=section,  Refname={Thm.}, refname={Thm.}]{theorem}
\declaretheorem[parent=section,numbered=no,name=Theorem]{theorem*}
\declaretheorem[numberlike=theorem, name=Proposition,  Refname={Prop.}, refname={Prop.}]{proposition}
\declaretheorem[name=Proposition, numbered=no]{proposition*}
\declaretheorem[numberlike=theorem, name=Remark,  Refname={Rem.}, refname={Rem.}]{remark}
\declaretheorem[numberlike=theorem, name=Example, Refname={Ex.}, refname={Ex.}]{example}
\declaretheorem[numberlike=theorem, name=Observation, Refname={Obs.}, refname={Obs.}]{observation}
\declaretheorem[numberlike=theorem, name=Lemma, Refname={Lem.}, refname={Lem.}]{lemma}
\declaretheorem[numberlike=theorem, name=Definition, Refname={Def.}, refname={Def.}]{definition}
\declaretheorem[numberlike=theorem, name=Corollary, Refname={Cor.}, refname={Cor.}]{corollary}

\newcommand*{\ketbra}[2]{\lvert #1 \rangle \langle #2 \rvert}

\DeclareMathOperator{\id}{id}

\newcommand*{\lrbracket}[1]{\left(#1\right)}

\newcommand*{\HS}{\mathcal{H}}
\newcommand*{\ZZ}{\mathcal{Z}}

\newcommand*{\Al}{\mathcal{A}}
\newcommand*{\Bl}{\mathcal{B}}
\newcommand*{\Xl}{\mathcal{X}}

\newcommand*{\revisioncolor}{blue}

\title{\large Fixed points in de Finetti hierarchies}

\author{Gereon Koßmann$^{1, *}$}
\email{kossmann@physik.rwth-aachen.de}
\address{$^1$ Institute for Quantum Information, RWTH Aachen University, Aachen, Germany}
\author{Julius A. Zeiss$^{1, *}$}
\email{jzeiss@physik.rwth-aachen.de}
\thanks{$^*$ The authors contributed equally to this work.}

\begin{document}

\begin{abstract}
    De Finetti theorems convert permutation symmetry into approximate mixtures of product states and thereby justify a wide range of reductions in classical and quantum statistics. In this work we study de Finetti hierarchies in which the feasible states are additionally constrained to be fixed points of quantum channels, a condition that subsumes invariance under arbitrary compact symmetry groups. Combining the mean-ergodic theorem with the structure theory of conditional expectations, we prove a tight bound on the entanglement-assisted classical capacity of the dual of a conditional expectation, block-wise distortion bounds for informationally complete measurements adapted to fixed point algebras, and an exact type-based refinement of the chain rule for permutation-invariant states. From these tools we derive several de Finetti theorems: a double-sided extension theorem with $O\left(\sqrt{\log n}/n\right)$ convergence, an interpolation theorem whose dimension dependence is governed solely by the block structure of the fixed point algebras and which recovers the dimension-independent classical behavior for maximal tori, and a Bose-symmetric variant. Exploiting Schur-Weyl duality and Gelfand-Tsetlin bases, we further show that the rounding scheme producing certifiably good separable inner approximations for (constrained) separability problems can be implemented in time polynomial in $1/\epsilon$ for fixed local dimensions, complementing the known efficient outer hierarchies. Applications to bilinear optimization under symmetries and to approximate quantum error correction are discussed.
\end{abstract}

\maketitle
\tableofcontents
\section{Introduction}

Permutation invariant statistical processes tend to be mixtures of independent and identically distributed ones.
By translating permutation symmetries into mixtures of products, de Finetti theorems justify a wide range of reductions in statistical theories. The core idea of a de Finetti theorem, named after its very first investigator \cite{deFinetti1931, alvarezmelis2015translationthecharacteristicfunction}, for a permutation invariant random variable $X_1 \times X_2 \times \ldots \times X_n$ is to show that local statistical dependencies break down if the length $n$ tends to infinity, so that the joint distribution can be replaced by a mixture of distributions of independent and identically distributed random variables. Subsequent refinements for finite-level approximations in $n$ have been shown in \cite{Diaconis1977,Diaconis1980} and a first version with non-commutative statistics in \cite{Raggio89}. 

In quantum information theory, de Finetti theorems and the closely related structure of the symmetric subspace \cite{harrow2013churchsymmetricsubspace} have become a standard tool with applications across the field: they underpin the approximation hierarchies for separability problems of Doherty, Parrilo and Spedalieri \cite{Doherty_2004}, the security analysis of quantum cryptographic protocols via exponential and one-and-a-half quantum de Finetti theorems \cite{Knig2005, Christandl2007}, tomography and channel estimation, and, more recently, converging hierarchies for bilinear optimization problems such as approximate quantum error correction \cite{Berta2021, kossmann2025_aqec}. Beyond quantum theory, de Finetti-type statements have been established for general probabilistic theories and abstract convex bodies \cite{Plvala2025, Aubrun2024, zeiss2026finitefinetticonvexbodies}, which delineates how much of the phenomenon is due to convex geometry alone.

Structurally, one distinguishes infinite de Finetti theorems, which characterize the states admitting extensions of every length $n$ as exact mixtures of product states \cite{deFinetti1931, Raggio89}, from finite (approximate) de Finetti theorems, which quantify the distance of an $n$-extendible state to the set of such mixtures. More precisely, in the finite setting one considers a state on $k$ subsystems that extends to a permutation invariant state on $n \geq k$ subsystems, and bounds its trace distance to mixtures of $k$-fold product states; throughout, the parameter $k$ thus counts the retained subsystems on which the approximation is stated, while $n$ counts the length of the assumed extension, and the error vanishes as $n\to\infty$ for fixed $k$. On the classical side, the finite theory goes back to Diaconis and Freedman \cite{Diaconis1977, Diaconis1980}, whose bounds are notably independent of the alphabet size. In the quantum setting, three families of proof techniques have emerged. Group-theoretic proofs based on the representation theory of the symmetric group and the symmetric subspace \cite{Knig2005, Christandl2007, harrow2013churchsymmetricsubspace} yield the best known dependence $O\left(d^2 k /n\right)$ for $k$-copy statements in local dimension $d$, but the approximating states are obtained non-constructively and interact poorly with additional constraints. Information-theoretic proofs based on the self-decoupling of the mutual information under measurements \cite{Brando2016, Brandao_2017, jee2020quasi, Berta2021} typically give an $O\left(\sqrt{\log\left(d\right)/n}\right)$ rate, but produce explicit, certifiably good measurement-based candidate states that remain feasible for linear and semidefinite constraints; this compatibility is what makes them the method of choice for converging semidefinite programming hierarchies for constrained separability problems \cite{Berta2021, kossmann2025_aqec, zeiss2025approximatingfixedsizequantum}. Finally, convex-geometric arguments \cite{Aubrun2024, zeiss2026finitefinetticonvexbodies} extend the finite theory to cones and convex bodies.  

The first two quantum proof routes therefore realize complementary virtues at quantifiable costs. The group-theoretic route attains the linear-in-$1/n$ rate $O\left(d^2 k/n\right)$, but it is non-constructive and interacts poorly with additional constraints; the information-theoretic route is constructive and constraint-compatible, but it converges only as $O\left(\sqrt{\log\left(d\right)/n}\right)$ and additionally incurs dimension-dependent prefactors through the distortion constants of the informationally complete measurements involved. A central aim of the present work is to close this gap between the constrained (information-theoretic) and the unconstrained (group-theoretic) de Finetti theorems. We ask: \emph{are there regimes in which one can benefit from the information-theoretic proof route --- namely, retain its compatibility with constraints and its certifiably good rounding schemes --- while recovering the convergence rates provided by the unconstrained group-theoretic de Finetti theorems?} Our answer is affirmative up to logarithmic factors in $n$: for double-sided and Bose-symmetric extensions, the fixed point techniques developed in this work upgrade the information-theoretic rate to $O\left(\sqrt{\log n}/n\right)$ (\autoref{thm:double_extended_de_Finetti}, \autoref{thm:bose_de_Finetti}; see \autoref{example_double_extension} for a direct comparison with \cite{Christandl2007}), while the approximating states remain explicit, measurement-based, feasible under the imposed constraints, and efficiently computable (\autoref{thm:efficient_rounding}).

The starting point of this article is the observation that in the information-theoretic proof route the permutation symmetry itself is used only in a single, final step, while all remaining steps hold for arbitrary states. This suggests treating symmetries not as a global assumption but as one instance of a much more flexible class of constraints: fixed point equations $\Phi\left(\rho\right) = \rho$ for quantum channels $\Phi$, which include invariance under arbitrary compact groups via the Haar twirl, and whose fixed point sets carry the structure of finite-dimensional $C^*$-algebras. Systematically replacing symmetry arguments by fixed point arguments, we obtain sharpened mutual information bounds, refined chain rules, and adapted measurement constructions, which in turn yield new de Finetti theorems and, on the algorithmic side, an efficient implementation of the associated rounding schemes. A detailed account of the challenges and of our results is given in the next section.

\section{Challenges and results}\label{sec:challenges_and_results}

Constrained de Finetti theorems, as pioneered in \cite{Berta2021}, approximate optimization problems over constrained separable states by converging hierarchies of semidefinite programs. Two inefficiencies of the current state of the art motivate this work. First, on the analytic side, the known convergence guarantees of the measurement-based information-theoretic proof technique scale as $O\left(1/\sqrt{n}\right)$ and carry a dimension dependence through the distortion constants of informationally complete measurements, even in situations where additional structure --- for instance classicality of one subsystem, Bose symmetry, or invariance under a local symmetry group --- is known to allow for much better behavior \cite{Diaconis1977, Christandl2007}. Second, on the computational side, the outer relaxations can be evaluated in time polynomial in the hierarchy level $n$ for fixed local dimensions \cite{Chee2025, kossmann2025_aqec, zeiss2025approximatingfixedsizequantum}, whereas the certifying inner candidates have so far only been represented in an exponentially large computational basis.

Our guiding principle for both problems is to identify symmetries as fixed points of quantum channels: a state is invariant under a compact group $\mathcal{V}$ of unitaries if and only if it is a fixed point of the corresponding Haar twirl, and, more generally, the fixed point set of any channel with a full-rank invariant state is the range of a conditional expectation onto a finite-dimensional $C^*$-algebra (cf.\ \autoref{thm:ergodic_projection_condexp}). This algebraic normal form makes symmetries, classicality constraints and genuinely quantum constraints amenable to one uniform treatment; conditional expectations are, in particular, idempotent channels, whose structural and approximation theory has recently attracted independent interest \cite{delsol2025emulationcapacityidempotentchannels, kitaev2025almostidempotentquantumchannelsapproximate}.

\subsection{Results}

Our contributions can be grouped as follows.
\begin{enumerate}
     \item \emph{Capacity bounds for conditional expectations.} In \autoref{thm:upper_bound_entanglement_assisted_classical_capacity} we show that the entanglement-assisted classical capacity of the dual $E^*$ of a conditional expectation onto a fixed point algebra with block dimensions $\left(d_1,\ldots,d_k\right)$ is tightly bounded by $\log\left(\sum_{i=1}^k d_i^2\right)$; the capacity formula for this channel class is essentially known from the direct sum constructions of \cite[Prop.~1]{Fukuda2007} (cf.\ also \cite[Prop.~5]{Gao2018}), and the novelty lies in its identification, via the mean-ergodic theorem (\autoref{lem:fixed points_equal_ergodic_projection}), with the maximal mutual information of arbitrary bipartite states subject to a fixed point constraint, together with a classical-side-information variant (\autoref{prop:classical_capacity_conditional_expectation}) adapted to the measured settings of our de Finetti arguments. Applied to Haar twirls (\autoref{cor:entanglement_assisted_classcial_capacity_group_averages}), this bounds the mutual information of any invariant state by the multiplicities of the group action and, for the symmetric group, replaces the trivial bound $2n\log d$ by the polylogarithmic bound $\left(d^2-1\right)\log\left(n+d\right)$.
    \item \emph{Fixed-point-adapted distortion bounds and chain rules.} \autoref{prop:distortion_conditional_expectation} constructs, for every channel with a full-rank fixed point, an informationally complete measurement whose distortion constant is governed by the largest block of the fixed point algebra rather than by the ambient dimension. \autoref{prop:classification_of_types} and \autoref{prop:chain_rule} show that for permutation-invariant states the conditioning registers in the self-decoupling chain rule can be compressed exactly to their types, reducing the number of distinct summands from exponential to polynomial in the number of measured systems.
    \item \emph{De Finetti theorems under fixed point constraints.} Combining these tools we prove a double-sided de Finetti theorem with $O\left(\sqrt{\log n}/n\right)$ convergence (\autoref{thm:double_extended_de_Finetti}), an interpolation de Finetti theorem in which the dimension dependence is carried entirely by the block data of the local fixed point algebras and which recovers dimension-independent behavior in the classical (maximal torus) case (\autoref{thm:interpolation_deFinetti}), a $k$-copy version with side information (\autoref{thm:proper_deFinetti}), and a Bose-symmetric double-sided version (\autoref{thm:bose_de_Finetti}).
    \item \emph{An efficient inner sequence for $\operatorname{SEP}$ and $\operatorname{cSEP}$.} In \autoref{sec:efficient_inner_sequence} we show that the de Finetti-based rounding scheme can be executed in time polynomial in $n$ for fixed local dimensions: the type-coarse-grained measurements, all required partial traces and the resulting separable candidates are computed directly in the Schur-Weyl block picture (\autoref{thm:efficient_rounding}), including Bose-symmetric and constrained variants (\autoref{cor:bose_efficient_rounding}, \autoref{cor:constrained_efficient_rounding}). Together with the efficient outer relaxations of \cite{kossmann2025_aqec, zeiss2025approximatingfixedsizequantum}, this yields additive $\epsilon$-error algorithms with matching inner certificates in time $\operatorname{poly}\left(1/\epsilon\right)$.
    \item \emph{Applications.} In \autoref{sec:applications} we derive convergence rates for symmetry-reduced hierarchies for bilinear optimization problems and give a Bose-symmetric hierarchy for approximate quantum error correction that avoids the Kronecker-coefficient couplings reported as the main implementation obstacle in \cite{kossmann2025_aqec}.
\end{enumerate}

\subsection{Previous work}

The literature on de Finetti theorems can be organized along a few axes. Concerning the presence of \emph{side information}, the classical statements \cite{Diaconis1977, Diaconis1980} and the symmetric-subspace-based quantum statements \cite{Knig2005, Christandl2007} concern the extended systems alone, whereas the information-theoretic route \cite{Brando2016, Brandao_2017, Berta2021, jee2020quasi} naturally accommodates an arbitrary reference system $A$, which is essential for optimization applications; our results belong to the latter family and quantify how constraints on $A$ improve the bounds. Concerning \emph{infinite versus finite} statements, the infinite theorems \cite{deFinetti1931, Raggio89} give exact representations, while the finite theorems trade exactness for explicit rates; all results of this paper are finite and quantitative. Concerning \emph{proof methods}, group-theoretic arguments \cite{Knig2005, Christandl2007, harrow2013churchsymmetricsubspace} achieve the best known rates but non-constructively, information-theoretic arguments \cite{Brando2016, Brandao_2017, Berta2021} are constructive and constraint-compatible, and convex-geometric arguments \cite{Aubrun2024, zeiss2026finitefinetticonvexbodies} isolate the role of the underlying state spaces; we combine the second route with representation theory to import some of the quantitative advantages of the first. Concerning \emph{constraints}, converging hierarchies for constrained separability problems were established in \cite{Berta2021} and developed into a symmetry-reduction framework in \cite{Chee2025, kossmann2025_aqec, zeiss2025approximatingfixedsizequantum}; the present work extends the constraint class to arbitrary fixed point conditions and supplies the missing efficient inner sequence. Finally, concerning \emph{quantum versus general probabilistic theories}, finite de Finetti theorems for cones and convex bodies \cite{Plvala2025, Aubrun2024, zeiss2026finitefinetticonvexbodies} show that a dimension dependence is unavoidable in general; our interpolation theorem locates the quantum problem on the spectrum between the dimension-free classical case and the fully quantum case in terms of the block structure of the constraint algebras.

\section{Notation and preliminaries}\label{sec:notation_preliminaries}

Throughout this article we use the following notation and conventions. Operator-algebraic preliminaries on finite-dimensional $C^*$-algebras and conditional expectations are collected in \autoref{subsec:definitions_notations}, a recap on fixed points of quantum channels is given in the appendix (cf.\ \autoref{thm:ergodic_projection_condexp}), and the representation-theoretic background is deferred to \autoref{sec:Young_diagrams_tableaux}, \autoref{sec:weyls_dimension_formula} and \autoref{sec:subgroup_adapted_basis}.

\paragraph{\textbf{General conventions.}}
We write $\mathbb{N} \coloneqq \lbrace 1,2,\ldots\rbrace$ and $\mathbb{N}_0 \coloneqq \mathbb{N}\cup \lbrace 0\rbrace$, and for $n\in\mathbb{N}$ we abbreviate $\left[n\right]\coloneqq\lbrace 1,\ldots, n\rbrace$. The cardinality of a finite set $\mathcal{Z}$ is denoted by $\lvert \mathcal{Z}\rvert$ and $\operatorname{conv}$ denotes the convex hull. Throughout, $\log$ denotes the logarithm to base $2$ and $\ln$ the natural logarithm. We use the Landau symbol $O\left(\cdot\right)$ and write $\operatorname{poly}\left(n\right)$ for quantities bounded by a polynomial in $n$. The Kronecker delta is denoted by $\delta_{ij}$, an (isometric) embedding by $\hookrightarrow$, and $\cong$ denotes an isomorphism, where the category (Hilbert spaces, groups, algebras or representations) is clear from the context; a group $G$ decorating the symbol, as in $\cong_G$ or $\stackrel{G}{\cong}$, indicates an isomorphism of $G$-representations. For a linear map $F$ we write $\operatorname{ran}\left(F\right)$ for its range and $\ker\left(F\right)$ for its kernel. Finite strings are abbreviated by $z_k^m \coloneqq \left(z_k,z_{k+1},\ldots,z_m\right)$, and the same convention applies to collections of systems, e.g.\ $B_1^n = B_1B_2\cdots B_n$; the objects $z_1^0$ and $B_1^0$ are understood as the empty string and the trivial system, respectively.

\paragraph{\textbf{Hilbert spaces and operators.}}
All Hilbert spaces in this work are finite-dimensional complex Hilbert spaces. Quantum systems are labeled by capital Latin letters, and we write $\mathcal{H}_A$ for the Hilbert space of system $A$ and $d_A \coloneqq \dim \mathcal{H}_A$ for its dimension; composite systems correspond to tensor products, $\mathcal{H}_{AB} \coloneqq \mathcal{H}_A\otimes \mathcal{H}_B$ and $\mathcal{H}_{B_1^n} \coloneqq \mathcal{H}_B^{\otimes n}$. The letter $R$ is reserved for purifying reference systems and $E$ for environments. For vector spaces $V,W$ over a field $\mathbb{K}$ we denote by $\operatorname{Hom}_{\mathbb{K}}\left(V,W\right)$ the space of $\mathbb{K}$-linear maps from $V$ to $W$ and set $\operatorname{End}_{\mathbb{K}}\left(V\right)\coloneqq \operatorname{Hom}_{\mathbb{K}}\left(V,V\right)$. In this article we restrict to $\mathbb{K}=\mathbb{C}$ and simply write $\operatorname{Hom}\left(V,W\right)$ and $\operatorname{End}\left(V\right)$; in particular, $\operatorname{End}\left(\mathbb{C}^d\right)$ is identified with the algebra $M_d\left(\mathbb{C}\right)$ of complex $d\times d$ matrices, and for a Hilbert space $\mathcal{H}$ we write $\mathcal{B}\left(\mathcal{H}\right)=\operatorname{End}\left(\mathcal{H}\right)$ for the algebra of bounded linear operators on $\mathcal{H}$. We use Dirac notation and fix a computational (orthonormal) basis $\lbrace \ket{i}\rbrace_{i\in\left[d\right]}$ of $\mathbb{C}^d$, with product basis $\ket{i_1}\otimes\ket{i_2}\otimes\cdots\otimes\ket{i_n} = \ket{i_1,i_2,\ldots,i_n}$ of $\left(\mathbb{C}^d\right)^{\otimes n}$; rank-one operators are written as $\ketbra{\psi}{\varphi}$. The adjoint of an operator $X$ is denoted by $X^\dagger$, whereas $a^*$ denotes the involution of an abstract $*$-algebra; $X^T$ and $\bar{X}$ denote the transpose and the entry-wise complex conjugate with respect to the computational basis. The identity operator on $\mathcal{H}_A$ is denoted by $1_A$ (also $1_d$ on $\mathbb{C}^d$, or simply $1$), and the same symbol is used for the unit of an abstract algebra; it is to be distinguished from the identity map $\operatorname{id}_A:\mathcal{B}\left(\mathcal{H}_A\right)\to \mathcal{B}\left(\mathcal{H}_A\right)$. The trace is denoted by $\tr$ and the partial trace over a subsystem $A$ by $\tr_A$; we endow $\mathcal{B}\left(\mathcal{H}\right)$ with the Hilbert--Schmidt inner product $\langle X,Y\rangle \coloneqq \tr\left[X^\dagger Y\right]$. We write $\lVert X\rVert_1 \coloneqq \tr\left[\sqrt{X^\dagger X}\right]$ for the Schatten-$1$-norm (trace norm) and $\lVert X\rVert_\infty$ (or simply $\lVert X\rVert$) for the operator norm. For Hermitian operators $X,Y$ the partial order $X\succeq Y$ means that $X-Y$ is positive semidefinite, and $X\succ0$ means that $X$ is positive definite.

\paragraph{\textbf{Quantum states.}}
The set of quantum states (density operators) on $\mathcal{H}$ is
\begin{align}
    \mathcal{S}\left(\mathcal{H}\right) \coloneqq \lbrace \rho \in \mathcal{B}\left(\mathcal{H}\right) \ \vert \ \rho \succeq 0, \ \tr\left[\rho\right] = 1\rbrace.
\end{align}
Subscripts indicate the systems a state acts on, e.g.\ $\rho_{AB_1^n}\in \mathcal{S}\left(\mathcal{H}_A\otimes \mathcal{H}_B^{\otimes n}\right)$, and reduced states are obtained by tracing out the omitted systems, e.g.\ $\rho_A = \tr_{B_1^n}\left[\rho_{AB_1^n}\right]$. Pure states are identified with their rank-one density operators, $\psi = \ketbra{\psi}{\psi}$, and a purification of $\rho_A \in \mathcal{S}\left(\mathcal{H}_A\right)$ is a pure state $\psi_{RA}$ with $\tr_R\left[\psi_{RA}\right] = \rho_A$. The maximally mixed state on $\mathcal{H}$ is $\tau_{\mathcal{H}} \coloneqq 1_{\mathcal{H}}/\dim\left(\mathcal{H}\right)$, abbreviated $\tau$ whenever no confusion can arise (the symbols $\tau_A$ and $\tau_B$ are reserved for the type maps introduced below). The maximally entangled state on $\mathcal{H}_L\otimes \mathcal{H}_L$ is $\Phi_{LL} = \ketbra{\Phi_{LL}}{\Phi_{LL}} \in \mathcal{S}\left(\mathcal{H}_L\otimes \mathcal{H}_L\right)$ with $\ket{\Phi_{LL}}\coloneqq d_L^{-1/2}\sum_{i=1}^{d_L}\ket{i}\otimes \ket{i}$. The fidelity between quantum states $\rho,\sigma \in \mathcal{S}\left(\mathcal{H}\right)$ is defined as $\mathcal{F}\left(\rho,\sigma\right)\coloneqq \lVert \sqrt{\rho}\sqrt{\sigma}\rVert_1^2$ (see e.g.\ \cite{Wilde2016}). For a bipartite system with cut $A : B$, the set of separable states is
\begin{align}
    \operatorname{SEP}\left(A:B\right)\coloneqq \operatorname{conv}\lbrace \rho_A\otimes \sigma_B \ \vert \ \rho_A \in \mathcal{S}\left(\mathcal{H}_A\right), \ \sigma_B\in \mathcal{S}\left(\mathcal{H}_B\right)\rbrace.
\end{align}
More generally, we consider \emph{constrained separability} ($\operatorname{cSEP}$) problems \cite{Berta2021, zeiss2025approximatingfixedsizequantum, kossmann2025_aqec, jee2020quasi}, i.e.\ optimizations of linear functionals over sets of the form\begin{align}\label{eq:def_cSEP_set}    \operatorname{cSEP}_{\mathcal{C}_A,\,\mathcal{C}_B}\left(A:B\right)\coloneqq \operatorname{conv}\lbrace \rho_A\otimes \sigma_B \ \vert \ \rho_A \in \mathcal{C}_A, \ \sigma_B\in \mathcal{C}_B\rbrace \subseteq \operatorname{SEP}\left(A:B\right),\end{align}where $\mathcal{C}_A\subseteq \mathcal{S}\left(\mathcal{H}_A\right)$ and $\mathcal{C}_B\subseteq \mathcal{S}\left(\mathcal{H}_B\right)$ are constraint sets specified by finitely many linear conditions on the product factors. Typical instances are fixed point constraints $\Phi\left(\rho_A\right) = \rho_A$ for given channels --- the case studied in this work, cf.\ the set $\Sigma\left(A:B\right)$ introduced in \autoref{sec:deFinetti_theorems_constraints} --- and constraints with fixed marginals $\Theta_{A_L \rightarrow C_{A_L}}\left(\rho_A\right) = W_{C_{A_L}}\otimes \rho_{A_R}$ with respect to a bipartition $\mathcal{H}_A = \mathcal{H}_{A_L}\otimes \mathcal{H}_{A_R}$, as they arise for non-local games with fixed-dimensional entanglement assistance \cite[App.\ B]{zeiss2025approximatingfixedsizequantum} or in approximate quantum error correction \cite{kossmann2025_aqec}. 

\paragraph{\textbf{Channels, duals and fixed points.}}
A linear map $\Phi:\mathcal{B}\left(\mathcal{H}_A\right)\to \mathcal{B}\left(\mathcal{H}_{A^\prime}\right)$ is written $\Phi_{A\to A^\prime}$ whenever input and output systems are to be emphasized. A quantum channel in the Schr\"odinger picture is a completely positive and trace-preserving (CPTP) map, whereas in the Heisenberg picture channels are unital completely positive maps. The Hilbert--Schmidt dual (adjoint) $\Phi^*$ of a linear map $\Phi$ is determined by $\tr\left[\Phi^*\left(\rho\right)a\right] = \tr\left[\rho\, \Phi\left(a\right)\right]$ for all $\rho,a\in \mathcal{B}\left(\mathcal{H}\right)$; it interchanges the two pictures. Composition of maps is denoted by $\circ$, with iterates $\Phi^{\circ t}$, $t \in \mathbb{N}_0$. The set of fixed points of $\Phi$ is
\begin{align}
    \operatorname{Fix}\left(\Phi\right)\coloneqq \ker\left(\Phi - \operatorname{id}\right) = \lbrace X \in \mathcal{B}\left(\mathcal{H}\right) \ \vert \ \Phi\left(X\right) = X\rbrace.
\end{align}
We abbreviate the data-processing inequality by DPI. Conditional expectations onto unital $*$-subalgebras, which arise as ergodic projections of channels admitting a full-rank fixed point, are introduced in \autoref{subsec:definitions_notations} (see also \autoref{thm:ergodic_projection_condexp}).

\paragraph{\textbf{Classical systems, measurements and post-measurement conventions.}}
Classical systems (registers) are denoted by the capital letters $X,Y,Z$ (and $T$ for the type registers introduced below), with associated finite alphabets written in calligraphic letters, e.g.\ $\mathcal{Y}$ and $\mathcal{Z}$; elements of alphabets, i.e.\ measurement outcomes, are denoted by the corresponding lowercase letters, e.g.\ $y\in\mathcal{Y}$ and $z \in \mathcal{Z}$. The set of probability distributions on a finite set $\mathcal{Z}$, identified with the standard simplex in $\mathbb{R}^{\mathcal{Z}}$, is denoted by $\mathcal{M}_1\left(\mathcal{Z}\right)$. A classical register $Z$ with alphabet $\mathcal{Z}$ is represented by the Hilbert space $\mathcal{H}_Z \coloneqq \mathbb{C}^{\lvert \mathcal{Z}\rvert}$ with a fixed orthonormal basis $\lbrace \ket{z}\rbrace_{z\in\mathcal{Z}}$, and a classical-quantum (cq) state on $AZ$ is of the form
\begin{align}
    \rho_{AZ} = \sum_{z \in \mathcal{Z}} p\left(z\right)\, \rho_{A\vert z}\otimes \ketbra{z}{z},
\end{align}
where $p \in \mathcal{M}_1\left(\mathcal{Z}\right)$ is the outcome distribution and $\rho_{A \vert z} \in \mathcal{S}\left(\mathcal{H}_A\right)$ denotes the conditional state given the outcome $z$; analogously we write $p\left(z_1^m\right)$ and $\rho_{A\vert z_1^m}$ for outcome strings $z_1^m \in \mathcal{Z}^m$. A positive operator-valued measure (POVM) on $\mathcal{H}_B$ with outcomes in $\mathcal{Z}$ is a family $\lbrace M_{B\vert z}\rbrace_{z\in \mathcal{Z}}\subseteq \mathcal{B}\left(\mathcal{H}_B\right)$ of positive semidefinite operators with $\sum_{z\in\mathcal{Z}} M_{B\vert z} = 1_B$. The associated measurement (channel) is the map
\begin{align}
    \mathcal{M}_B : \mathcal{S}\left(\mathcal{H}_B\right)\to \mathcal{M}_1\left(\mathcal{Z}\right), \quad \rho \mapsto \left(\tr\left[M_{B\vert z}\, \rho\right]\right)_{z \in \mathcal{Z}},
\end{align}
which we freely identify with the quantum-classical channel $\rho \mapsto \sum_{z\in \mathcal{Z}} \tr\left[M_{B\vert z}\,\rho\right] \ketbra{z}{z}$. A measurement is called informationally complete (IC) if the map $\mathcal{M}_B$, extended linearly to $\mathcal{B}\left(\mathcal{H}_B\right)$, is injective, and minimal informationally complete (MIC) if in addition $\lvert \mathcal{Z}\rvert = d_B^2$; MIC measurements exist in every finite dimension \cite{DAriano2004}. Following \cite{jee2020quasi,Brandao_2017,Lami_2018}, a constant $c\left(\mathcal{M}_B\right)>0$ is called a distortion bound for $\mathcal{M}_B$ if
\begin{align}
    \lVert X_{AB}\rVert_1 \leq c\left(\mathcal{M}_B\right)\, \lVert \left(\operatorname{id}_A\otimes \mathcal{M}_B\right)\left(X_{AB}\right)\rVert_1
\end{align}
for every auxiliary system $A$ and every traceless Hermitian operator $X_{AB}\in \mathcal{B}\left(\mathcal{H}_A\otimes \mathcal{H}_B\right)$. Throughout, Alice holds the $A$-systems and Bob the $B$-systems; the post-measurement registers on Alice's side are denoted by $Y$ and those on Bob's side by $Z$. In a hierarchy of $n$ subsystems, the integers $m_A$ and $m_B$ with $0\leq m_A,m_B\leq n-1$ denote the number of subsystems measured on Alice's, respectively Bob's, side.

\paragraph{\textbf{Types.}}
For $m \in \mathbb{N}_0$ and $r\in \mathbb{N}$ the type set is
\begin{align}
    \mathcal{T}_{m,r}\coloneqq \Big\lbrace t = \left(t_1,\ldots,t_r\right)\in \mathbb{N}_0^{r} \ \Big\vert \ \sum_{i=1}^{r} t_i = m\Big\rbrace, \qquad \lvert \mathcal{T}_{m,r}\rvert = \binom{m+r-1}{r-1} \leq \left(m+1\right)^{r-1},
\end{align}
i.e.\ the set of weak compositions of $m$ into $r$ parts (see e.g.\ \cite[Sec.~11.1]{Cover2006}). For a finite alphabet $\mathcal{Z}$ we identify $\mathcal{T}_{m,\lvert\mathcal{Z}\rvert}$ with $\lbrace t \in \mathbb{N}_0^{\mathcal{Z}} \ \vert \ \sum_{z\in\mathcal{Z}}t_z = m\rbrace$ and define the type map
\begin{align}
    \tau: \mathcal{Z}^m \to \mathcal{T}_{m,\lvert \mathcal{Z}\rvert}, \quad \tau\left(z_1^m\right)_z \coloneqq \lvert \lbrace k \in \left[m\right] \ \vert \ z_k = z\rbrace\rvert,
\end{align}
which counts how often each outcome appears in a string. The type class of $t\in \mathcal{T}_{m,\lvert \mathcal{Z}\rvert}$ is $\mathcal{C}\left(t\right)\coloneqq \tau^{-1}\left(t\right)$ with $\lvert \mathcal{C}\left(t\right)\rvert = m!/\prod_{z\in\mathcal{Z}}t_z!$; two strings have the same type if and only if they are related by a permutation, such that types enumerate the $S_m$-orbits of $\mathcal{Z}^m$. We write $\tau_A$ and $\tau_B$ for the type maps applied to Alice's and Bob's outcome strings, and $T_{m_A}$ (also $T_{m_A,\lvert\mathcal{Y}\rvert}$, if the alphabet is to be emphasized) and $T_{m_B}$ for the classical registers carrying the types $\tau_A\left(Y_2^{m_A+1}\right)$ and $\tau_B\left(Z_2^{m_B+1}\right)$. If the alphabet is clear from the context we abbreviate $\mathcal{T}_{m_A} \equiv \mathcal{T}_{m_A,\lvert \mathcal{Y}\rvert}$ and $\mathcal{T}_{m_B}\equiv \mathcal{T}_{m_B, \lvert \mathcal{Z}\rvert}$.

\paragraph{\textbf{Entropic quantities.}}
The Shannon entropy of $p\in\mathcal{M}_1\left(\mathcal{Z}\right)$ is $H\left(p\right)\coloneqq -\sum_{z\in\mathcal{Z}}p\left(z\right)\log p\left(z\right)$ and the von Neumann entropy of $\rho \in \mathcal{S}\left(\mathcal{H}\right)$ is $H\left(\rho\right)\coloneqq -\tr\left[\rho \log \rho\right]$; for multipartite states we write $H\left(A\right)_\rho \coloneqq H\left(\rho_A\right)$ and $H\left(A\vert B\right)_\rho \coloneqq H\left(AB\right)_\rho - H\left(B\right)_\rho$. The Umegaki relative entropy of $\rho\in\mathcal{S}\left(\mathcal{H}\right)$ with respect to a positive semidefinite operator $\sigma$ is
\begin{align}
    D\left(\rho \Vert \sigma\right)\coloneqq \begin{cases} \tr\left[\rho\left(\log \rho - \log\sigma\right)\right] & \text{if } \operatorname{supp}\left(\rho\right)\subseteq \operatorname{supp}\left(\sigma\right),\\ +\infty & \text{otherwise},\end{cases}
\end{align}
where $\operatorname{supp}$ denotes the support. The mutual information and the conditional mutual information are defined as
\begin{align}
    I\left(A:B\right)_\rho \coloneqq D\left(\rho_{AB}\Vert \rho_A\otimes \rho_B\right), \qquad I\left(A : B \vert C\right)_\rho \coloneqq H\left(A\vert C\right)_\rho + H\left(B \vert C\right)_\rho - H\left(AB\vert C\right)_\rho,
\end{align}
and satisfy $I\left(A:B\right)_\rho = H\left(A\right)_\rho + H\left(B\right)_\rho - H\left(AB\right)_\rho$ (see e.g.\ \cite{Wilde2016}). A subscript indicates the state in which an entropic quantity is evaluated and is omitted if it is clear from the context.

\paragraph{\textbf{Groups, algebras and symmetries.}}
We denote by $S_n$ the symmetric group on $\left[n\right]$, by $\operatorname{GL}\left(V\right)$ the general linear group of a vector space $V$, by $\mathcal{U}\left(\mathcal{H}\right)$ the group of unitary operators on $\mathcal{H}$, and abbreviate $\mathcal{U}\left(d\right)\coloneqq \mathcal{U}\left(\mathbb{C}^d\right)$. Subgroups $\mathcal{V}\subseteq \mathcal{U}\left(\mathcal{H}\right)$ are assumed to be closed and hence compact, and $\mu_{\mathcal{V}}$ denotes the unique normalized Haar measure on $\mathcal{V}$ \cite[Chap.~9]{Cohn2013}. The group algebra of a finite group $G$ is denoted by $\mathbb{C}\left[G\right]$. On $\mathcal{H}^{\otimes n}$ with $d = \dim \mathcal{H}$ we consider throughout the tensor (permutation) representation of $S_n$,
\begin{align}\label{eq:def_permutation_action_Sn_H_otimes_n}
     \psi_n^d\,:\,S_n \to \operatorname{GL}\left( \mathcal{H}^{\otimes n}\right), \quad \sigma \mapsto \left[v_1\otimes v_2\otimes \cdots \otimes v_n \mapsto v_{\sigma^{-1}\left(1\right)}\otimes v_{\sigma^{-1}\left(2\right)}\otimes \cdots \otimes v_{\sigma^{-1}\left(n\right)}\right],
\end{align}
whose image consists of unitary operators; we write $U_{B_1^n}\left(\sigma\right)\coloneqq \psi_n^{d_B}\left(\sigma\right)$ whenever the systems on which a permutation acts are to be emphasized. The permutation matrix algebra $\mathcal{A}_n^d \coloneqq \operatorname{span}\lbrace \psi_n^d\left(\sigma\right) \ \vert \ \sigma \in S_n\rbrace$ is the image of $\mathbb{C}\left[S_n\right]$ under the linear extension of $\psi_n^d$.\footnote{The extension has a non-trivial kernel for $d$ smaller than $n$.} A state $\rho_{AB_1^n}\in \mathcal{S}\left(\mathcal{H}_A \otimes \mathcal{H}_B^{\otimes n}\right)$ is called permutation invariant over the $B$-systems (with respect to $A$) if
\begin{align}
    \left(1_A \otimes U_{B_1^n}\left(\sigma\right)\right)\rho_{AB_1^n}\left(1_A \otimes U_{B_1^n}\left(\sigma\right)\right)^\dagger = \rho_{AB_1^n} \quad \text{for all} \ \sigma \in S_n,
\end{align}
and Bose-symmetric over the $B$-systems (with respect to $A$) if $\left(1_A\otimes U_{B_1^n}\left(\sigma\right)\right)\rho_{AB_1^n} = \rho_{AB_1^n}$ for all $\sigma\in S_n$, i.e.\ if $\rho_{AB_1^n}$ is supported on the symmetric subspace of $\mathcal{H}_B^{\otimes n}$ (see e.g.\ \cite{harrow2013churchsymmetricsubspace, zeiss2025approximatingfixedsizequantum}). A state $\rho_{AB}$ admits a (one-sided, permutation invariant) $n$-extension if there exists a state $\rho_{AB_1^n}$, permutation invariant over the $B$-systems with respect to $A$, such that $\rho_{AB_1} = \rho_{AB}$; it admits a double-sided $n$-extension if there exists a state $\rho_{A_1^nB_1^n}$, permutation invariant over the $A$-systems with respect to $B_1^n$ and over the $B$-systems with respect to $A_1^n$, such that $\rho_{A_1B_1}=\rho_{AB}$. For a subset $\mathcal{S}\subseteq \mathcal{B}\left(\mathcal{H}\right)$ the commutant is
\begin{align}
    \mathcal{S}^\prime \coloneqq \lbrace X \in \mathcal{B}\left(\mathcal{H}\right) \ \vert \ XS = SX \quad \text{for all} \ S \in \mathcal{S}\rbrace.
\end{align}
If an algebra $\mathcal{A}\subseteq \operatorname{End}\left(V\right)$ acts on $V$, its commutant (centralizer) is also denoted by $\operatorname{End}_{\mathcal{A}}\left(V\right)$; for a representation of a group $G$ on $V$ we write $\operatorname{End}_G\left(V\right)\equiv \operatorname{End}_{\mathbb{C}\left[G\right]}\left(V\right)$ for the algebra of $G$-equivariant maps and, analogously, $\operatorname{Hom}_G\left(V,W\right)$ for the space of intertwiners. For a compact subgroup $\mathcal{W}\subseteq \mathcal{U}\left(\mathcal{H}\right)$ we denote by $\mathcal{A}_{\mathcal{W}}\subseteq \mathcal{B}\left(\mathcal{H}\right)$ the $*$-algebra generated by $\mathcal{W}$, such that $\mathcal{A}_{\mathcal{W}}^\prime = \mathcal{W}^\prime$; if $\mathcal{W}$ is given as a product of groups with specified unitary actions on $\mathcal{H}$, e.g.\ $\mathcal{W} = \mathcal{V}_A\times \mathcal{V}_B\times S_n$, then $\mathcal{A}_{\mathcal{W}}$ denotes the $*$-algebra generated by the images of all factors.

\paragraph{\textbf{Partitions and weak compositions.}}
A partition of $n\in\mathbb{N}$ is a non-increasing sequence $\lambda = \left(\lambda_1,\lambda_2,\ldots\right)$ of non-negative integers with $\sum_{i}\lambda_i = n$, denoted $\lambda \vdash n$; its length $\ell\left(\lambda\right)$ is the number of non-zero parts, and we write $\lambda \vdash_d n$ if in addition $\ell\left(\lambda\right)\leq d$. A weak composition of $n$ into $r$ parts is a tuple $\mu = \left(\mu_1,\ldots,\mu_r\right)\in \mathbb{N}_0^r$ with $\sum_{i=1}^r \mu_i = n$ and no ordering imposed, denoted $\mu \vDash n$; in particular, type sets consist of weak compositions, cf.\ above. The irreducible representations of $S_n$ are labeled by partitions $\lambda \vdash n$: the corresponding Specht module is denoted by $V_\lambda$ with dimension $d_\lambda \coloneqq \dim V_\lambda$, whereas $U_\lambda$ denotes the irreducible polynomial representation (Weyl module) of $\mathcal{U}\left(d\right)$, equivalently of $\operatorname{GL}\left(d\right)$, of highest weight $\lambda \vdash_d n$, with dimension $m_\lambda \coloneqq \dim U_\lambda$; the latter equals the multiplicity of $V_\lambda$ in $\mathcal{H}^{\otimes n}$ under \autoref{eq:def_permutation_action_Sn_H_otimes_n} (see \autoref{sec:weyls_dimension_formula}). The combinatorial background on Young diagrams and tableaux is summarized in \autoref{sec:Young_diagrams_tableaux} (cf.\ also \cite{sagan2013symmetric,Fulton2004}).

\subsection{A de Finetti argument revisited}\label{subsec:deFinetti_revisited}

Our interest in the following proof technique for de Finetti theorems, as outlined in \autoref{sec:challenges_and_results}, is twofold: first, it is compatible with additional constraints on the involved states and hence applies to constrained separability problems; second, it gives rise to certifiably good rounding schemes, i.e.\ explicit separable candidate states with an a priori approximation guarantee (see \autoref{sec:efficient_inner_sequence}). In order to motivate the aim of this work we consolidate the following de Finetti argument for a separability proof \cite{Berta2021}. Assume that a state $\rho_{AB}$ admits, for every $n \in \mathbb{N}$, a permutation-invariant $n$-extension $\rho_{AB_1^n}$ over the $B$-systems with respect to $A$. We want to show that $\rho_{AB}$ is then well approximated by separable states on the $A:B$-cut, i.e., there exists a sequence of separable states $\sigma_{AB}^{\left(n\right)}\in \operatorname{SEP}\left(A:B\right)$ such that 
\begin{align}\label{eq:de_Finetti_argument_for_separability}
    \lVert \sigma_{AB}^{\left(n\right)} - \rho_{AB}\rVert_1 \to 0 \quad \text{for} \quad n\to \infty.
\end{align}
For this we fix an $n\in \mathbb{N}$ with given $\rho_{AB_1^n}$ and a measurement $\mathcal{M}_{B\to Z}$ as introduced in \autoref{sec:notation_preliminaries}, applied to each $B$-system, i.e.\ $\rho_{AZ_1^n} \coloneqq \left(\operatorname{id}_A \otimes \mathcal{M}_{B\to Z}^{\otimes n}\right)\left(\rho_{AB_1^n}\right)$. Then by \cite[Ex.\ 11.6.3]{Wilde2016} and the data-processing inequality for the channel $\mathcal{M}_{B\to Z}$ we have
\begin{align}\label{eq:mutual_information_bounds}
    I\left(A:Z_1^n\right)_{\rho_{AZ_1^n}} \leq I\left(A:B_1^n\right)_{\rho_{AB_1^n}} \leq 2\log d_A.
\end{align}
Since the $Z_1^n$-systems are classical, the left-hand side of \autoref{eq:mutual_information_bounds} even satisfies the sharper bound
\begin{align}\label{eq:cq_mutual_information_bound}
    I\left(A:Z_1^n\right)_{\rho_{AZ_1^n}} = H\left(A\right)_{\rho_{AZ_1^n}} - H\left(A\vert Z_1^n\right)_{\rho_{AZ_1^n}} \leq \log d_A,
\end{align}
because $H\left(A\vert Z_1^n\right)_{\rho_{AZ_1^n}} = \sum_{z_1^n} p\left(z_1^n\right) H\left(\rho_{A\vert z_1^n}\right)\geq 0$ for classical conditioning systems; this is the bound used in \cite[Lem.~2.1]{Berta2021}. The detour via $I\left(A:B_1^n\right)_{\rho_{AB_1^n}}$ in \autoref{eq:mutual_information_bounds} is nevertheless instructive: the fixed point techniques of \autoref{sec:techniques} bound precisely this fully quantum mutual information, for which the dimension bound $2 \log d_A$ is tight, and the factor of two reappears there as the square in \autoref{thm:upper_bound_entanglement_assisted_classical_capacity}. For the present sketch we continue with \autoref{eq:cq_mutual_information_bound}. Using furthermore the chain rule for the mutual information together with \autoref{lemma:conditioning_mutual_information} (which is adapted from \cite{Brandao_2017} and \cite{Berta2021}), we can decompose
\begin{align}\label{eq:self_decoupling_lemma_intro}
    I\left(A:Z_1^n\right)_{\rho_{AZ_1^n}} = \sum_{m=0}^{n-1}\sum_{z_1^m} p\left(z_1^m\right)  D \left(\rho_{AZ_{m+1\vert  z_1^m}}\Vert \rho_{A\vert  z_1^m} \otimes \rho_{Z_{m+1\vert  z_1^m}}\right).
\end{align}
In \autoref{eq:self_decoupling_lemma_intro} there exists a summand $0\leq m \leq n-1$ whose value is at most the average $\frac{\log d_A}{n}$ (by \autoref{eq:cq_mutual_information_bound}), such that we have with Pinsker's inequality and the convexity of $x\mapsto x^2$ and of $\lVert\cdot\rVert_1$
\begin{equation}
    \begin{aligned}
      \frac{\log d_A}{n} &\geq\sum_{z_1^m} p\left(z_1^m\right)
      D  \left(\rho_{AZ_{m+1\vert  z_1^m}}\Vert \rho_{A\vert  z_1^m}\otimes \rho_{Z_{m+1\vert  z_1^m}} \right) \\
      &\geq\frac{1}{2\ln 2} \lVert \rho_{AZ_{m+1}}-\sum_{z_1^m} p\left(z_1^m\right) \rho_{A\vert  z_1^m}\otimes \rho_{Z_{m+1\vert  z_1^m}} \rVert_1^2.
    \end{aligned}
\end{equation}
Assume now that the measurement $\mathcal{M}$ is \emph{informationally complete} with distortion bound $c\left(\mathcal{M}\right)$ (cf.\ \autoref{sec:notation_preliminaries}), i.e.\ such that the following inequality is true
\begin{align}\label{eq:informationally_complete_measurements}
     \lVert \rho_{AZ_{m+1}}-\sum_{z_1^m}p\left(z_1^m\right)\rho_{A\vert  z_1^m}\otimes \rho_{Z_{m+1\vert  z_1^m}} \rVert_1
    \geq\frac{1}{c\left(\mathcal{M}\right)}  \lVert \rho_{AB_{m+1}}-\sum_{z_1^m}p\left(z_1^m\right)\rho_{A\vert  z_1^m}\otimes \rho_{B_{m+1\vert  z_1^m}} \rVert_1.
\end{align}
We conclude to have
\begin{align}\label{eq:de_Finetti_without_symmetry}
    c\left(\mathcal{M}\right)   \sqrt{\frac{2 \ln 2 \, \log d_A}{n}} \geq  \lVert \rho_{AB_{m+1}}-\sum_{z_1^m}p\left(z_1^m\right)\rho_{A\vert  z_1^m}\otimes \rho_{B_{m+1\vert  z_1^m}} \rVert_1.
\end{align}
Now we apply the permutation invariance, which implies $\rho_{AB_{m+1}} = \rho_{AB_1}$ for the joint marginal, and conclude that \autoref{eq:de_Finetti_without_symmetry} can be transferred to
\begin{align}\label{eq:de_Finetti_with_symmetry}
    c\left(\mathcal{M}\right)   \sqrt{\frac{2 \ln 2 \, \log d_A}{n}} \geq  \lVert \rho_{AB_{1}}-\sum_{z_1^m}p\left(z_1^m\right)\rho_{A\vert  z_1^m}\otimes \rho_{B_{m+1\vert  z_1^m}} \rVert_1.
\end{align}
Note that the conditional states on the right-hand side of \autoref{eq:de_Finetti_with_symmetry} remain those of the $\left(m+1\right)$-st $B$-system --- conditioning $B_1$ on outcomes obtained by measuring $B_1,\ldots,B_m$ themselves would be ill-defined --- and only the joint marginal $\rho_{AB_{m+1}}$ is identified with $\rho_{AB_1}$, cf.\ \cite{Berta2021}. \autoref{eq:de_Finetti_with_symmetry} identifies a separable state $\sigma_{AB}^{\left(n\right)}$ at trace distance $O\left(1/\sqrt{n}\right)$ from $\rho_{AB}$, which is exactly a de Finetti argument as in \autoref{eq:de_Finetti_argument_for_separability} and yields convergence, as $n \to \infty$, to a separable state. 

Two further structural features of this proof route are worth recording. First, it is compatible with linear constraints in the sense of \autoref{eq:def_cSEP_set}: constraint maps acting on $A$ or on the retained system $B_{m+1}$ commute with the measurements performed on $B_1^m$, so every per-factor linear constraint satisfied by the extension $\rho_{AB_1^n}$ --- fixed point constraints as well as constraints with fixed marginals --- is inherited by the conditional states $\rho_{A\vert z_1^m}$ and $\rho_{B_{m+1}\vert z_1^m}$, and hence by the separable candidate in \autoref{eq:de_Finetti_with_symmetry}, cf.\ \cite{Berta2021} and \cite[Lem.~B.5]{zeiss2025approximatingfixedsizequantum}. Second, while precisely this compatibility makes the information-theoretic proof route the method of choice for constrained hierarchies, the constraints themselves have so far only ensured feasibility: prior to this work, they have not been used to improve the convergence rate of de Finetti theorems, with the exception of the special case of equality (decoupling) constraints in \cite[Lem.~6]{jee2020quasi}, which is of a different nature than the fixed point constraints studied here and which we recover and generalize in \autoref{prop:equality_constraint_upper_bound}. This paper investigates several steps in the argument above. We summarize a few key observations in the following box. 

\begin{observation}\label{observation}
Considering the de Finetti argument above once again, a crucial observation is that permutation invariance is used only in the very last step, from \autoref{eq:de_Finetti_without_symmetry} to \autoref{eq:de_Finetti_with_symmetry}. However, this property of the states $\rho_{AB_1^n}$ has not been used in \autoref{eq:mutual_information_bounds} and has also not been used in \autoref{eq:self_decoupling_lemma_intro}. In particular, although the self-decoupling lemma \cite[Lem.~2.1]{Berta2021} is \emph{stated} for classical-quantum states that are permutation invariant over the $Z_1^n$-systems, an inspection of its proof shows that only the classicality of the $Z_1^n$-systems is used --- through the bound \autoref{eq:cq_mutual_information_bound} and the conditioning identity of \autoref{lemma:conditioning_mutual_information} --- so that its conclusion holds for arbitrary classical-quantum states. Permutation invariance proper enters only when the lemma is applied, namely through the identification $\rho_{AB_{m+1}} = \rho_{AB_1}$ in the step from \autoref{eq:de_Finetti_without_symmetry} to \autoref{eq:de_Finetti_with_symmetry}. We thus ask whether we can substantially make use of
\begin{enumerate}
    \item the bound on $I\left(A : B_1^n\right)_{\rho_{AB_1^n}}$ if we assume more structure on $\rho_{AB_1^n}$, such as permutation invariance of the $B$-systems with respect to $A$.
    \item the chain rule in \autoref{eq:self_decoupling_lemma_intro} in order to get rid of the exponentially many (in $m$) measurement outcomes $z_1^m \in \mathcal{Z}^m$.
    \item the choice of the informationally complete measurement in \autoref{eq:informationally_complete_measurements} in order to reduce the distortion factor $c\left(\mathcal{M}\right)$, provided the post-measurement states inherit additional structure, again in the form of fixed point constraints.
\end{enumerate}
In particular, the improvements in~(1) and~(3) are obtained from constraints on the $A$- and $B$-systems, such as symmetries, which we treat as fixed points of general quantum channels.     
\end{observation}

\section{Bounds and techniques}\label{sec:techniques}
Symmetries such as the permutation symmetry used in the step from \autoref{eq:de_Finetti_without_symmetry} to \autoref{eq:de_Finetti_with_symmetry} can be seen as a fixed point constraint on the set of feasible states for the de Finetti argument. To be concrete consider a compact subgroup $\mathcal{V}$ of the unitary group $\mathcal{U}\left(\mathcal{H}\right)$ on a finite dimensional Hilbert space $\mathcal{H}$ and its unique normalized Haar measure \cite[Chap.~9]{Cohn2013} $\mu_\mathcal{V}$. Then the map 
\begin{align}
    \rho \mapsto \int_{\mathcal{V}} V\rho V^\dagger d\mu_{\mathcal{V}}\left(V\right), \quad \rho \in \mathcal{S}\left(\mathcal{H}\right)
\end{align}
is a completely positive and trace-preserving map. Moreover, it has the maximally mixed state as a full-rank fixed point. By the mean-ergodic theorem and its consequences (revisited in \autoref{thm:ergodic_projection_condexp}), fixed points induce an algebraic decomposition of the underlying Hilbert space and thus a structural decomposition of the states acting on it. In this section we provide several very general tools and techniques in order to handle such constraints in a de Finetti argument such as in \autoref{subsec:deFinetti_revisited} in a systematic way. 
\subsection{Definitions and notation}\label{subsec:definitions_notations}
In this work we consider finite-dimensional $C^*$-algebras $\mathcal{A}$, i.e., operator-norm-closed $*$-algebras, which can be identified with direct sums of full matrix algebras (see \cite[Thm.~11.9]{Takesaki2001OperatorAlgebrasI}). Concretely, a complex $*$-algebra is a complex algebra equipped with an involution $a \mapsto a^*$ for all $a \in \mathcal{A}$ such that the $*$-map is a conjugate-linear anti-automorphism satisfying $\left(a^*\right)^* = a$. The $*$-algebra $\mathcal{A}$ is a $C^*$-algebra if it is endowed with a norm $\lVert \cdot\rVert $ satisfying the $C^*$-identity
\begin{align}
    \lVert a^* a\rVert  = \lVert a\rVert ^2 \qquad \text{for all } a \in \mathcal{A}.
\end{align}
For our purposes it will be sufficient to work throughout with the full matrix algebra $\mathcal{A} \cong M_n\left(\mathbb{C}\right)$ of complex $n \times n$ matrices, equipped with the usual operator norm induced by the Euclidean norm on vectors. A unital $*$-subalgebra $\mathcal{B}$ is a subset $\mathcal{B} \subseteq \mathcal{A}$ containing the identity of $\mathcal{A}$ that is itself a $C^*$-algebra. By \cite[Thm.~11.9]{Takesaki2001OperatorAlgebrasI} we can identify $\mathcal{B}$ with
\begin{align}\label{eq:decomposition_finite_star_algebras}
    \mathcal{B} \cong  \bigoplus_{i=1}^k M_{d_i}\left(\mathbb{C}\right)
    \qquad \text{and} \qquad
    \mathcal{B}^\prime \cong  \bigoplus_{i=1}^k M_{m_i}\left(\mathbb{C}\right),
\end{align}
where $\mathcal{B}^\prime$ denotes the commutant of $\mathcal{B}$ in $\mathcal{A} \cong M_n\left(\mathbb{C}\right)$. The tuples $\left(d_1,\ldots,d_k\right)$ and $\left(m_1,\ldots,m_k\right)$ are unique up to permutation and quantify the dimension and the multiplicity, respectively, of the irreducible representations of $\mathcal{B}$. 

As we would like to investigate fixed points of channels later in this work, we introduce conditional expectations from $\mathcal{A} \cong M_n\left(\mathbb{C}\right)$ to a unital $*$-subalgebra $\mathcal B \subseteq \mathcal{A}$. A positive and unital map
\begin{align}
    E:\mathcal{A} \to \mathcal{B}
\end{align}
is a \emph{conditional expectation} if it satisfies $E\left(b\right) = b$ for all $b \in \mathcal{B}$ and the bimodule property
\begin{align}
    E\left(bac\right) = bE\left(a\right)c \quad \text{for all} \ b,c \in \mathcal{B},\ a \in \mathcal{A}.
\end{align}
By Tomiyama's theorem, the bimodule property is automatic: every norm-one projection of $\mathcal{A}$ onto $\mathcal{B}$ is already a conditional expectation (see e.g. \cite[Sec.~9.1]{Petz2008QITQS}). Conditional expectations are automatically completely positive \cite[Eq.~(9.6)]{Petz2008QITQS} and thus channels in the Heisenberg picture. As we aim to prove a bound on the entanglement assisted classical capacity \cite{Bennett1999} for conditional expectations as a first main goal, we will make use of the following structure result about conditional expectations. The result is from \cite[Prop.~1.5]{Wolf2012QuantumChannelsGuidedTour}. The decomposition in \autoref{eq:decomposition_finite_star_algebras} induces a decomposition of the underlying Hilbert space $\mathcal{H} \cong \mathbb{C}^n$ given by 
\begin{align}\label{eq:decomposition_level_Hilbert_spaces}
    \mathcal H \cong  \bigoplus_{i=1}^k \mathcal{H}_{i,1}\otimes \mathcal{H}_{i,2}.
\end{align}
\begin{proposition}[\cite{Wolf2012QuantumChannelsGuidedTour}]\label{prop:structure_prop_cond_expectation_mmwolf}
    Let $\mathcal{B} \subseteq \mathcal{A} \cong M_n\left(\mathbb{C}\right)$ be a unital $*$-subalgebra and $E:\mathcal{A}\to \mathcal{B}$ a conditional expectation. Then there exists a family of mutually orthogonal isometries $\{V_i:\mathcal{H}\to \mathcal{H}_{i,1}\otimes\mathcal{H}_{i,2}\}_{i=1}^k$ and density operators $\rho_i \in \mathcal{B}\left(\mathcal{H}_{i,2}\right)$\footnote{see \autoref{eq:decomposition_level_Hilbert_spaces}.} such that 
    \begin{align}
        E\left(a\right) = \sum_{i=1}^k V_i^\dagger  \left( \tr_{i,2}  \left[V_i aV_i^\dagger \ 1_{d_i}\otimes \rho_i \right]\otimes 1_{m_i} \right)V_i, \quad a \in \mathcal{A}.
    \end{align}
\end{proposition}

\subsection{A mutual information bound for conditional expectations}

Given the prerequisites from \autoref{subsec:definitions_notations}, we can investigate our first result. Let $\Phi:\mathcal{B}\left(\mathcal{H}_A\right)\to \mathcal{B}\left(\mathcal{H}_A\right)$ be completely positive and unital, and write $\Phi^*$ for its Schr\"odinger-picture dual (trace-preserving and completely positive). Fix some side information system $\mathcal{H}_B$. Motivated from the \autoref{observation}~(1) in \autoref{subsec:deFinetti_revisited}, we consider the following optimization problem
\begin{equation}\label{eq:mutual_information_upper_bound_equation}
    \begin{aligned}
        \sup \quad & I\left(A : B\right)_{\rho_{AB}} \\
        \operatorname{s.t.}\quad & \left(\Phi^*_{A\to A}\otimes \id_B\right)\left(\rho_{AB}\right) = \rho_{AB}, \\
        & \rho_{AB}\in \mathcal{S}\left(\mathcal{H}_A\otimes \mathcal{H}_B\right).
    \end{aligned}
\end{equation}
Using the fixed point constraint in \autoref{eq:mutual_information_upper_bound_equation}, we may equivalently write
\begin{equation}\label{eq:mutual_information_upper_bound_equation_fixed points}
    \begin{aligned}
        \sup \quad & I\left(A : B\right)_{\rho_{AB}} \\
        \operatorname{s.t.}\quad & \rho_{AB}\in \operatorname{Fix}\left(\Phi^*_{A\to A}\otimes \id_B\right), \\
        & \rho_{AB}\in \mathcal{S}\left(\mathcal{H}_A\otimes \mathcal{H}_B\right).
    \end{aligned}
\end{equation}

In order to handle the fixed point constraint in \autoref{eq:mutual_information_upper_bound_equation_fixed points} in a more sophisticated manner, we use the following lemma. The lemma is well known: that ergodic projections of channels admitting a full-rank invariant state are conditional expectations onto the fixed point algebra is classical \cite{Frigerio1978, Arias2002, BlumeKohout2010} (cf.\ also \cite{Wolf2012QuantumChannelsGuidedTour} and \autoref{thm:ergodic_projection_condexp}), and the extension to side information is routine; we include the short proof for completeness.

\begin{lemma}\label{lem:fixed points_equal_ergodic_projection}
    Let $\Phi\otimes\operatorname{id}_B:\mathcal{B}\left(\mathcal{H}_A\otimes\mathcal{H}_B\right)\to\mathcal{B}\left(\mathcal{H}_A\otimes\mathcal{H}_B\right)$ be completely positive and unital, and assume that $\Phi^*$ admits a full-rank fixed point. Let $E:\mathcal{B}\left(\mathcal{H}_A\right)\to\operatorname{Fix}\left(\Phi\right)$ denote the (Heisenberg-picture) ergodic projection from \autoref{thm:ergodic_projection_condexp},\footnote{The extension to side information is routine, since $\left(\Phi\otimes\operatorname{id}_B\right)^n = \Phi^n \otimes \operatorname{id}_B$ for all $n\in\mathbb{N}$.} i.e.
    \begin{align}\label{eq:cesaro_E_def}
        E   =   \lim_{N\to\infty}\frac{1}{N}\sum_{n=0}^{N-1}\Phi^n,
    \end{align}
    which is a conditional expectation onto the fixed point subalgebra $\operatorname{Fix}\left(\Phi\right)$ by \autoref{thm:ergodic_projection_condexp}. Then
    \begin{align}
        \operatorname{Fix}\left(\Phi^*\otimes\operatorname{id}_B\right)   =   \operatorname{ran}\left(E^*\otimes\operatorname{id}_B\right).
    \end{align}
\end{lemma}

\begin{proof}
    By the mean ergodic theorem \autoref{thm:ergodic_projection_condexp}~(1), we have that $E$ is a projection, $E$ satisfies $E\circ \Phi= \Phi \circ E = E$ and $\operatorname{ran}\left(E\right)=\operatorname{Fix}\left(\Phi\right)$. Under the assumption that $\Phi^*$ has a full-rank fixed point, \autoref{thm:ergodic_projection_condexp}~(2) moreover identifies the projection $E$ as a \emph{conditional expectation}
    onto the fixed point subalgebra $\operatorname{Fix}\left(\Phi\right)$. Taking Hilbert-Schmidt adjoints and using that adjoints commute with finite sums and norm-limits, we obtain 
    \begin{equation}\label{eq:cesaro_Estar_def}
    \begin{aligned}
        \left(\lim_{N\to\infty}\frac{1}{N}\sum_{n=0}^{N-1}\left(\Phi \otimes \operatorname{id}_B\right)^n\right)^* 
        &  =  \lim_{N\to\infty}\frac{1}{N}\sum_{n=0}^{N-1}\left(\left(\Phi\otimes \operatorname{id}_B\right)^*\right)^n \\
        &  =   \left( \lim_{N\to\infty}\frac{1}{N}\sum_{n=0}^{N-1}\left(\left(\Phi\right)^*\right)^{n} \right)\otimes \operatorname{id}_B \\
        &  =   E^* \otimes \id_B.
    \end{aligned}
    \end{equation}
    Hence, we conclude
    \begin{align}\label{eq:range_Estar_fix}
        \operatorname{ran}\left(E^* \otimes \operatorname{id}_B\right)   =   \operatorname{Fix}\left(\Phi^* \otimes \operatorname{id}_B\right).
    \end{align}
\end{proof}

By \autoref{lem:fixed points_equal_ergodic_projection}, the constraint $\rho_{AB}\in \operatorname{Fix}\left(\Phi^*_{A\to A}\otimes \id_B\right)$ in \autoref{eq:mutual_information_upper_bound_equation_fixed points} is equivalent to $\rho_{AB}\in \operatorname{ran}\left(E^*_{A\to A}\otimes \id_B\right)$, and since $E^*$ is idempotent, every feasible
$\rho_{AB}$ can be written as $\rho_{AB}=\left(E^*_{A\to A}\otimes \id_B\right)\left(\sigma_{AB}\right)$ for some state $\sigma_{AB}$.
Therefore \autoref{eq:mutual_information_upper_bound_equation_fixed points} is equivalent to the following optimization problem
\begin{equation}\label{eq:mutual_information_upper_bound_with_cond_exp}
    \begin{aligned}
        \sup \quad & I\left(A : B\right)_{\left(E^*_{A\to A}\otimes \id_B\right)\left(\rho_{AB}\right)} \\
       \operatorname{s.t.} \quad & \rho_{AB}\in \mathcal{S}\left(\mathcal{H}_A\otimes \mathcal{H}_B\right).
    \end{aligned}
\end{equation}

Interestingly, in case of arbitrarily large side-information, this is a well-known quantity in quantum information theory, called the \emph{entanglement assisted classical capacity} \cite{Bennett1999} (see also the Bennett–Shor–Smolin–Thapliyal Theorem \cite[Thm.\ 21.3.1]{Wilde2016}). Thus, in order to solve \autoref{eq:mutual_information_upper_bound_equation_fixed points}, we need to calculate the entanglement assisted classical capacity of the dual of a conditional expectation, which is a trace-preserving and completely positive map. 
\begin{lemma}[\cite{Wolf2012QuantumChannelsGuidedTour}]\label{lem:calculation_adjoint}
      Let $\mathcal{B} \subseteq \mathcal{A} \cong M_n\left(\mathbb{C}\right)$ be a unital $*$-subalgebra, $\{V_i:\mathcal{H}\to \mathcal{H}_{i,1}\otimes\mathcal{H}_{i,2}\}_{i=1}^k$ a family of mutually orthogonal isometries and $\rho_i \in \mathcal{B}\left(\mathcal{H}_{i,2}\right)$\footnote{see \autoref{eq:decomposition_level_Hilbert_spaces}.} density operators. Then the Hilbert-Schmidt dual of 
    \begin{align}
        E\left(a\right) = \sum_{i=1}^k V_i^\dagger  \left( \tr_{i,2}  \left[V_i aV_i^\dagger \ 1_{d_i}\otimes \rho_i \right]\otimes 1_{m_i} \right)V_i, \quad a \in \mathcal{A},
    \end{align}
    is given by
    \begin{align}\label{eq:lem_adjoints}
        E^*\left(\omega\right) = \sum_{i=1}^k V_i^\dagger \left( \tr_{i,2}\left[V_i \omega V_i^\dagger\right]\otimes \rho_i \right)V_i, \quad \omega \in \mathcal{S}\left(\mathcal{H}\right).
    \end{align}
\end{lemma}
\begin{proof}
    See \autoref{appendix:adjoints_conditional_expectations}.
\end{proof}

Given all the preliminary comments, we can now summarize the observed ingredients in order to prove the following proposition, which yields a tight upper bound on the entanglement assisted classical capacity of the Hilbert-Schmidt duals of conditional expectations. We emphasize that the capacity statement itself is essentially known: up to the blockwise re-preparation of the states $\rho_i$ from \autoref{lem:calculation_adjoint}, which leaves all capacities unchanged, the channel $E^*$ is a direct sum of partial traces $\bigoplus_{i} \operatorname{id}_{d_i}\otimes \tr_{m_i}$, for which exact, strongly additive single-letter capacity formulas follow from the direct sum calculus of \cite[Prop.\ 1]{Fukuda2007} and are stated explicitly in \cite[Prop.\ 5]{Gao2018} in the context of TRO (ternary ring of operators) channels. The contribution here lies in the combination with \autoref{lem:fixed points_equal_ergodic_projection}: for an \emph{arbitrary} unital completely positive map $\Phi$ admitting a full-rank invariant state, the fixed point constraint in \autoref{eq:mutual_information_upper_bound_equation_fixed points} is identified with the range of $E^*\otimes \id_B$, so that the known capacity formula becomes a mutual information bound for all constrained bipartite states; moreover, we give a self-contained and elementary proof adapted to this setting.

\begin{proposition}\label{thm:upper_bound_entanglement_assisted_classical_capacity}
    Let $\Phi:\mathcal{B}\left(\mathcal{H}_A\right) \to \mathcal{B}\left(\mathcal{H}_A\right)$ be a completely positive and unital channel such that $\Phi^*$ has a full-rank fixed point, and let $\left(d_1,\ldots,d_k\right)$ be the tuple, unique up to permutation, defined via $d_i \coloneqq \dim \mathcal{H}_{i,1}$ in \autoref{eq:decomposition_level_Hilbert_spaces}. Then the entanglement assisted classical capacity of $E^*$, the dual of the conditional expectation $E$ onto the fixed point algebra of $\Phi$ from \autoref{thm:ergodic_projection_condexp}, defined by
    \begin{align}\label{eq:def_entanglement_asssisted_classical_capacity}
        C_{\operatorname{EA}}\left(E^*\right) = \max_{\rho_{A} \in \mathcal{S}\left(\mathcal{H}_A\right)} I\left(R:A\right)_{\sigma_{RA}}, \quad \sigma_{RA} \coloneqq \left(\operatorname{id}_{R}\otimes E^*\right)\left(\psi_{RA}\right),
    \end{align}
    where $\psi_{RA}$ is any purification of $\rho_A\in \mathcal{S}\left(\mathcal{H}_A\right)$, is tightly bounded by
    \begin{align}
        C_{\operatorname{EA}}\left(E^*\right)\leq \log \left(\sum_{i=1}^k d_i^2\right).
    \end{align}
\end{proposition}
\begin{proof}
    We use \autoref{eq:def_entanglement_asssisted_classical_capacity} and define
    \begin{align}
        U:\mathcal{H}\to  \bigoplus_{i=1}^k  \left(\mathcal{H}_{i,1} \otimes \mathcal{H}_{i,2} \right), \quad U\phi =  \bigoplus_{i=1}^k V_i\phi.
    \end{align}
    Using the properties of the isometries $\{V_i:\mathcal{H}\to \mathcal{H}_{i,1}\otimes\mathcal{H}_{i,2}\}_{i=1}^k$ 
     \begin{align}
        V_i V_j^\dagger = \delta_{ij} 1_{\mathcal{H}_{i,1} \otimes \mathcal{H}_{i,2}}, \quad \sum_{i=1}^k V_i^\dagger V_i  = 1_{\mathcal{H}}. 
    \end{align}
    yields 
    \begin{align}
        UU^\dagger =  \bigoplus_{i=1}^k V_iV_i^\dagger = 1_{\oplus_{i} \left(\mathcal{H}_{i,1}\otimes\mathcal{H}_{i,2}\right)}, \quad U^\dagger U = \sum_{i=1}^k V_i^\dagger V_i = 1_{\mathcal{H}}.
    \end{align}
    Now we compute with the help of \autoref{lem:calculation_adjoint}
    \begin{equation}\label{eq:thm_entanglement_assited_1}
        \begin{aligned}
            UE^*\left(\omega\right) U^\dagger &=   \bigoplus_{i=1}^kV_i  \left(\sum_{j=1}^k V_j^\dagger \left( \tr_{j,2}\left[V_j \omega V_j^\dagger\right]\otimes \rho_j \right)V_j \right)V_i^\dagger \\
            &=  \bigoplus_{i=1}^k  \left( \tr_{i,2}\left[V_i \omega V_i^\dagger\right]\otimes \rho_i \right).
        \end{aligned}
    \end{equation}
    Now let us define for a purification $\psi_{RA}$ of $\rho_A$ the states 
    \begin{equation}
        \begin{aligned}
            \sigma_{RA} &\coloneqq \left(\operatorname{id}_R\otimes E^*\right)\left(\psi_{RA}\right) \\
            \tilde{\sigma}_{R\oplus_i \left(\mathcal{H}_{i,1}\otimes\mathcal{H}_{i,2}\right)} &\coloneqq \left(\operatorname{id}_R\otimes U\right)\sigma_{RA}\left(\operatorname{id}_R\otimes U\right)^\dagger.
        \end{aligned}
    \end{equation}
    Because local isometries do not change the mutual information\footnote{Recall that $I\left(A : B\right)_{\rho_{AB}} = D\left(\rho_{AB}\Vert \rho_A \otimes \rho_B\right)$ and the relative entropy is invariant under isometries (see e.g. \cite{Tomamichel2016book}).}, we have
    \begin{align}
        I\left(R:A\right)_{\sigma_{RA}} = I\left(R: \bigoplus_{i=1}^k \left(\mathcal{H}_{i,1}\otimes\mathcal{H}_{i,2}\right)\right)_{\tilde{\sigma}_{R:\oplus_i \left(\mathcal{H}_{i,1}\otimes\mathcal{H}_{i,2}\right)}}.
    \end{align}
    Using the concrete representation from \autoref{eq:lem_adjoints} in \autoref{eq:thm_entanglement_assited_1}, we can write  
    \begin{align}
        \tilde{\sigma}_{R\oplus_i \left(\mathcal{H}_{i,1}\otimes\mathcal{H}_{i,2}\right)} =  \bigoplus_{i=1}^k \left(X_{RA_{i,1}}^i \otimes \rho_i\right), \quad X_{RA_{i,1}}^i \coloneqq \tr_{\mathcal{H}_{i,2}}\left[\left(\operatorname{id}_R\otimes V_i\right) \psi_{RA} \left(\operatorname{id}_R\otimes V_i\right)^\dagger\right],
    \end{align}
    whereby $X_{RA_{i,1}}^i$ are subnormalized states. We normalize them with $p_i \coloneqq \tr\left[X_{RA_{i,1}}^i\right]$, i.e.\ we set $\hat{X}^i_{RA_1}\coloneqq X^i_{RA_{i,1}}/p_i$, and introduce a classical register $K$ for the block-diagonal term
    \begin{align}\label{eq:proof_entanglement_assisted_concrete_form_state}
    \tilde{\sigma}_{R K A_1A_2} = \sum_{i=1}^k p_i \ketbra{i}{i}_K \otimes \hat{X}_{RA_1}^i \otimes \rho^i_{A_2}.
    \end{align}
    As we are interested on the mutual information $I\left(R:KA_1A_2\right)_{\tilde{\sigma}_{R K A_1A_2}}$, we can apply the chain rule
    \begin{align}\label{eq:proof_entanglement_assisted_chain_rule}
        I\left(R:K A_1 A_2\right)_{\tilde{\sigma}_{R K A_1A_2}} = I\left(R:K\right)_{\tilde{\sigma}_{R K}} + I\left(R:A_1 A_2  \vert K\right)_{\tilde{\sigma}_{RKA_1A_2}}.
    \end{align}
    Now we apply the \autoref{lemma:conditioning_mutual_information} in a simple form in order to get
    \begin{align}\label{eq:proof_thm_entanglement_classical_rewriting}
    I\left(R:A_1 A_2  \vert K\right)_{\tilde{\sigma}_{RKA_1A_2}} = \sum_{i}p_iI\left(R:A_1 A_2\right)_{\tilde{\sigma}_{RA_1A_2\vert i}}.
    \end{align}
    Writing this mutual information in terms of relative entropy for each $1\leq i \leq k$ yields
    \begin{equation}
        \begin{aligned}
            I\left(R:A_1 A_2\right)_{\tilde{\sigma}_{RA_1A_2\vert i}} &= D\left(\tilde{\sigma}_{RA_1A_2\vert i} \Vert \tilde{\sigma}_{R\vert i} \otimes \tilde{\sigma}_{A_1A_2\vert i}\right) \\
            &=D\left(\hat{X}_{RA_1}^i \otimes \rho^i_{A_2} \Vert \hat{X}^i_{R}\otimes  \hat{X}_{A_1}^i \otimes \rho^i_{A_2}\right) \\
            &=D\left(\hat{X}_{RA_1}^i  \Vert \hat{X}^i_{R}\otimes  \hat{X}_{A_1}^i\right) + D\left(\rho^i_{A_2} \Vert \rho^i_{A_2}\right) \\
            &=D\left(\hat{X}_{RA_1}^i \Vert \hat{X}^i_{R}\otimes  \hat{X}_{A_1}^i\right),
        \end{aligned}
    \end{equation}
    whereby we have used \autoref{eq:proof_entanglement_assisted_concrete_form_state} and the facts that relative entropy is additive and satisfies $D\left(\rho^i_{A_2} \Vert \rho^i_{A_2}\right) =0$. Thus, concluding \autoref{eq:proof_thm_entanglement_classical_rewriting}, we have
    \begin{align}\label{eq:proof_thm_geting_rid_of_A2}
        I\left(R:A_1 A_2  \vert K\right)_{\tilde{\sigma}_{RKA_1A_2}} = \sum_{i}p_iI\left(R:A_1\right)_{\tilde{\sigma}_{RA_1\vert i}}.
    \end{align}
    Thus, we combine all the steps \autoref{eq:proof_entanglement_assisted_chain_rule}, \autoref{eq:proof_thm_entanglement_classical_rewriting} and \autoref{eq:proof_entanglement_assisted_concrete_form_state}  in order to get 
    \begin{align}
        I\left(R:K A_1 A_2\right)_{\tilde{\sigma}_{R K A_1A_2}} = I\left(R:K\right)_{\tilde{\sigma}_{R K}} + \sum_{i}p_iI\left(R:A_1\right)_{\tilde{\sigma}_{RA_1\vert i}}.
    \end{align}
    The first term has as an upper bound $H\left(p\right)$, because $I\left(R:K\right) = H\left(K\right) - H\left(K \vert R \right) \leq H\left(K\right) = H\left(p\right)$, because $K$ is classical and thus $H\left(K\vert R \right)\geq 0$, and the second term has the usual upper bound $\log d_i^2$ if $\dim \mathcal{H}_{i,1} = d_i$, such that we conclude 
    \begin{align}\label{eq:proof_entanglement_assisted_close_to_final}
        I\left(R:K A_1 A_2\right)_{\tilde{\sigma}_{R K A_1A_2}} \leq H\left(p\right) + \sum_{i=1}^k p_i \log d_i^2.
    \end{align}
    On the right hand side we now apply the standard optimization over $\left(p_i\right)$. For that purpose we define $W \coloneqq \sum_{i=1}^k d_i^2$ and define $q_i \coloneqq d_i^2/W$. Then we have
    \begin{equation}
        \begin{aligned}
            0\leq D\left(p\Vert q\right) &= \sum_{i=1}^k p_i \log \frac{p_i}{q_i} \\
            &= \sum_{i=1}^k  \left(p_i \log p_i - p_i\log q_i \right) \\
            &= -H\left(p\right)- \sum_{i=1}^k  \left(p_i \log d_i^2 \right) + \log W.
        \end{aligned}
    \end{equation}
    Rearranging leads to 
    \begin{align}
        H\left(p\right) + \sum_{i=1}^k  \left(p_i \log d_i^2 \right) \leq  \log W,
    \end{align}
    which, inserted into \autoref{eq:proof_entanglement_assisted_close_to_final}, leads to the desired result
    \begin{align}
        C_{\operatorname{EA}}\left(E^*\right) \leq H\left( p\right) + \sum_{i=1}^k p_i \log d_i^2 \leq \log\left(\sum_{i=1}^k d_i^2\right).
    \end{align}

    Given the proof technique it is easy to see that the state
    \begin{align}
        \ket{\psi_{RA}} = \sum_{i=1}^k  \sqrt{p_i} \ket{i}_{R_0}\otimes\ket{i}_{A_0} \otimes \ket{\Phi^i}_{R_{1,i}A_{1,i}} \otimes \ket{\varphi^i}_{R_{2,i}A_{2,i}}
    \end{align}
    with $p_i\coloneqq \frac{d_i^2}{W}$ actually achieves the bound if the $\Phi^i$ are maximally entangled and the $\varphi^i$ correspond to purifications of the $\rho_i$ from \autoref{lem:calculation_adjoint}, which do not contribute to the bound.   
\end{proof}

We conclude this section with some comments on the proof strategy of \autoref{thm:upper_bound_entanglement_assisted_classical_capacity}. The main ingredient in the proof is that conditional expectations yield a specific decomposition, given by a finite dimensional $C^*$-algebra. Moreover we note that we can improve the bound by a factor of $2$ in case that we consider block-wise measurements in \autoref{eq:decomposition_level_Hilbert_spaces}, because in \autoref{eq:proof_entanglement_assisted_close_to_final} the subsystems would be classical and we could use the stronger uniform classical bound on mutual information there. Without the assumption of classicality, this result can even be proven with the reduction criterion and a PPT (positive partial transpose) assumption on the underlying states \cite{Horodecki1999}.

The following proposition makes the role of classicality precise for the case relevant to the de Finetti arguments below, namely when the \emph{side information} system is classical, as is the case after the measurements in \autoref{eq:mutual_information_bounds}. It is the constrained analogue of the sharper classical-quantum bound \autoref{eq:cq_mutual_information_bound}, which it recovers for the trivial constraint $\Phi = \operatorname{id}_A$. Its formula mirrors the formula $\chi\left(\bigoplus_i \operatorname{id}_{n_i}\otimes \tr_{m_i}\right) = \log\left(\sum_i n_i\right)$ of \cite[Prop.~5]{Gao2018} (cf.\ also \cite[Prop.~1]{Fukuda2007}) for the one-shot expression $\chi\left(\mathcal{N}\right) = \max_{\rho_{XA}} I\left(X:B\right)$ of the classical capacity, i.e.\ the maximized Holevo information, in the same way as \autoref{thm:upper_bound_entanglement_assisted_classical_capacity} mirrors the entanglement-assisted formula $C_{\operatorname{EA}} = \log\left(\sum_i n_i^2\right)$. Indeed, since $E^*$ is idempotent, every feasible classical-quantum state can be written as $\rho_{AZ} = \sum_z p\left(z\right) E^*\left(\rho_{A\vert z}\right)\otimes \ketbra{z}{z}$, so the supremum of $I\left(A:Z\right)$ in \autoref{prop:classical_capacity_conditional_expectation} is precisely this one-shot capacity of $E^*$ in ensemble form.

\begin{proposition}[Classical side information]\label{prop:classical_capacity_conditional_expectation}
    Let $\Phi:\mathcal{B}\left(\mathcal{H}_A\right) \to \mathcal{B}\left(\mathcal{H}_A\right)$ be a completely positive and unital channel such that $\Phi^*$ has a full-rank fixed point, and let $\left(d_1,\ldots,d_k\right)$ be as in \autoref{thm:upper_bound_entanglement_assisted_classical_capacity}. Then every classical-quantum state $\rho_{AZ} \in \mathcal{S}\left(\mathcal{H}_A \otimes \mathcal{H}_Z\right)$ with classical $Z$-system satisfying $\left(\Phi^*_{A\to A}\otimes \id_Z\right)\left(\rho_{AZ}\right) = \rho_{AZ}$ obeys
    \begin{align}
        I\left(A:Z\right)_{\rho_{AZ}} \leq \log\left(\sum_{i=1}^k d_i\right),
    \end{align}
    and the bound is attained.
\end{proposition}
\begin{proof}
    By \autoref{lem:fixed points_equal_ergodic_projection} (with side information system $\mathcal{H}_Z$) we can write $\rho_{AZ} = \left(E^*_{A\to A}\otimes \id_Z\right)\left(\sigma_{AZ}\right)$ for some state $\sigma_{AZ}$. Repeating the steps of the proof of \autoref{thm:upper_bound_entanglement_assisted_classical_capacity} with the system $Z$ in place of the reference $R$, i.e.\ conjugating with the isometry $U$ and introducing the classical register $K$ for the block index, we obtain the analogue of \autoref{eq:proof_entanglement_assisted_concrete_form_state},
    \begin{align}
        \tilde{\sigma}_{K A_1 A_2 Z} = \sum_{i=1}^k p_i \ketbra{i}{i}_K \otimes \hat{X}^i_{A_1 Z} \otimes \rho^i_{A_2},
    \end{align}
    with normalized states $\hat{X}^i_{A_1Z}$ that are classical on $Z$, and $I\left(A:Z\right)_{\rho_{AZ}} = I\left(K A_1 A_2:Z\right)_{\tilde{\sigma}}$ by isometric invariance of the mutual information. The chain rule together with \autoref{lemma:conditioning_mutual_information}, applied to the classical register $K$, yields
    \begin{align}
        I\left(K A_1 A_2:Z\right)_{\tilde{\sigma}} = I\left(K:Z\right)_{\tilde{\sigma}} + \sum_{i=1}^k p_i \, I\left(A_1 A_2:Z\right)_{\hat{X}^i_{A_1Z}\otimes \rho^i_{A_2}} = I\left(K:Z\right)_{\tilde{\sigma}} + \sum_{i=1}^k p_i \, I\left(A_1:Z\right)_{\hat{X}^i_{A_1Z}},
    \end{align}
    where the second equality uses that $\rho^i_{A_2}$ is in a product with the remaining systems. Since the $Z$-system is classical, we have, exactly as in \autoref{eq:cq_mutual_information_bound}, $I\left(A_1:Z\right)_{\hat{X}^i_{A_1Z}} \leq H\left(A_1\right)_{\hat{X}^i_{A_1Z}} \leq \log d_i$, and moreover $I\left(K:Z\right)_{\tilde{\sigma}} \leq H\left(K\right)_{\tilde{\sigma}} = H\left(p\right)$ because $K$ is classical. Hence
    \begin{align}
        I\left(A:Z\right)_{\rho_{AZ}} \leq H\left(p\right) + \sum_{i=1}^k p_i \log d_i \leq \log\left(\sum_{i=1}^k d_i\right),
    \end{align}
    where the last step is the optimization from the proof of \autoref{thm:upper_bound_entanglement_assisted_classical_capacity} with $W \coloneqq \sum_{i=1}^k d_i$ and $q_i \coloneqq d_i/W$. The bound is attained by the classically maximally correlated state
    \begin{align}
        \tilde{\sigma}_{K A_1 A_2 Z} = \sum_{i=1}^k \frac{d_i}{W} \, \ketbra{i}{i}_K \otimes \left(\frac{1}{d_i}\sum_{x=1}^{d_i} \ketbra{x}{x}_{A_1}\otimes \ketbra{i,x}{i,x}_Z \right)\otimes \rho^i_{A_2},
    \end{align}
    which is of the form \autoref{eq:lem_adjoints} and hence feasible, and for which $I\left(A:Z\right) = H\left(p\right) + \sum_{i=1}^k p_i \log d_i = \log W$.
\end{proof}

\begin{remark}\label{rem:classical_side_information_deFinetti}
    In the de Finetti arguments of this work the mutual information is ultimately evaluated on post-measurement states with classical side systems, cf.\ \autoref{eq:mutual_information_bounds}. Whenever the fixed point constraint acts on the \emph{unmeasured} side --- so that it commutes with the measurement channel and is inherited by the post-measurement state --- \autoref{prop:classical_capacity_conditional_expectation} may therefore be used in place of \autoref{thm:upper_bound_entanglement_assisted_classical_capacity}, replacing $\log\left(\sum_i d_i^2\right)$ by $\log\left(\sum_i d_i\right)$. In the unconstrained case this is exactly the improvement from $2\log d_A$ to $\log d_A$ in \autoref{eq:cq_mutual_information_bound}, and in the Schur--Weyl settings below it roughly halves the polylogarithmic exponents; e.g., $\left(d_A^2-1\right)\log\left(n+d_A\right)$ improves to $\tfrac{1}{2}\left(d_A-1\right)\left(d_A+2\right)\log\left(n+d_A\right)$ in \autoref{thm:double_extended_de_Finetti}, $\log m_{\left(n\right)}^2$ improves to $\log m_{\left(n\right)}$ in \autoref{thm:bose_de_Finetti}, and $\log\left(\sum_i m_{i,A}^2\right)$ improves to $\log\left(\sum_i m_{i,A}\right)$ in \autoref{thm:interpolation_deFinetti}. For uniformity of exposition we nevertheless state the results below via the entanglement-assisted route.
\end{remark}

\subsection{\texorpdfstring{Applications of \autoref{thm:upper_bound_entanglement_assisted_classical_capacity}}{Applications of the capacity bound}}
Arguably one of the most interesting applications of \autoref{thm:upper_bound_entanglement_assisted_classical_capacity} is in case of a compact group action, where fixed points are identified with the commutant of the group. We state it as the following corollary.
\begin{corollary}\label{cor:entanglement_assisted_classcial_capacity_group_averages}
    Let $\mathcal{V}$ be a compact subgroup of the unitary group $\mathcal{U}\left(\mathcal{H}\right)$ on $\mathcal{H}$ and consider the channel 
    \begin{align}
        \Phi_{\mathcal{V}}^*:\mathcal{S}\left(\mathcal{H}\right) \to \mathcal{S}\left(\mathcal{H}\right), \quad \rho \mapsto \int_{\mathcal{V}} V\rho V^\dagger d\mu_{\mathcal{V}}\left(V\right), 
    \end{align}
    whereby $\mu_{\mathcal{V}}$ denotes the unique normalized Haar measure corresponding to $\mathcal{V}$ (see e.g. \cite[chap.~9]{Cohn2013}) and $\left(m_1,\ldots,m_k\right)$, $\left(d_1,\ldots,d_k\right)$ the multiplicities and dimensions of the irreducible representations in the decomposition of $\mathcal{H}$ under $\mathcal{V}$ (see e.g. \cite[Thm.~4.3]{Bump2013}). Then we have for the dual $E_{\mathcal{V}}^*$ of the conditional expectation $E_{\mathcal{V}}$ onto the commutant $\mathcal{V}^\prime$ 
    \begin{align}
        C_{\operatorname{EA}}\left(E_{\mathcal{V}}^*\right) \leq \log  \left(\sum_{i=1}^k m_i^2 \right).
    \end{align}
\end{corollary}
\begin{proof}
    As group averages with respect to a unitary representation share the property that they equal their dual up to isomorphism of spaces, i.e. $\Phi_{\mathcal{V}} = \Phi^*_{\mathcal{V}}$, and leave the maximally mixed state $\tau \in \mathcal{S}\left(\mathcal{H}\right)$ invariant, the prerequisites of \autoref{thm:upper_bound_entanglement_assisted_classical_capacity} are satisfied. Furthermore, by translation invariance of the Haar measure (e.g. \cite[sec.~9.2]{Cohn2013} by definition) we have for $X \in\operatorname{Fix}\left(\Phi_{\mathcal{V}}\right)$ and $W \in \mathcal{V}$ 
    \begin{equation}
        \begin{aligned}
            X &= \Phi_{\mathcal{V}}^*\left(X\right) \\
            &= \int VXV^\dagger d\mu_{\mathcal{V}}\left(V\right) \\
            &=\int \left(WV\right)X\left(WV\right)^\dagger d\mu_{\mathcal{V}}\left(V\right) \\
            &= W\int VXV^\dagger d\mu_{\mathcal{V}}\left(V\right)W^\dagger \\
            &= WXW^\dagger.
        \end{aligned}
    \end{equation}
    Conversely, every $X$ commuting with all $W \in \mathcal{V}$ is clearly invariant under the average. Thus we have
    \begin{align}
        \operatorname{Fix}\left(\Phi_{\mathcal{V}}^*\right) = \{X \in \mathcal{B}\left(\mathcal{H}\right) \ \vert \ WX = XW \quad \text{for all} \ W \in \mathcal{V}\} \eqqcolon \mathcal{V}^\prime.
    \end{align}
    It is easy to check that the commutant of any subset of $\mathcal{B}\left(\mathcal{H}\right)$ is an algebra; since $\mathcal{V}$ consists of unitaries and is closed under adjoints, $\mathcal{V}^\prime$ is moreover a $*$-algebra and hence a finite dimensional $C^*$-algebra (see \cite[chap.~3.1]{Stratila2019}). We conclude by \cite[Thm.~11.9]{Takesaki2001OperatorAlgebrasI}, as written in \autoref{eq:decomposition_finite_star_algebras}, that 
    \begin{align}
       \operatorname{Fix}\left(\Phi^*_{\mathcal{V}}\right) = \mathcal{V}^\prime \cong  \bigoplus_{i=1}^k M_{m_i}\left(\mathbb{C}\right). 
    \end{align}
    Comparing this with \autoref{thm:upper_bound_entanglement_assisted_classical_capacity} immediately yields the bound
    \begin{align}
        C_{\operatorname{EA}}\left(E_{\mathcal{V}}^*\right) \leq \log  \left(\sum_{i=1}^k m_i^2 \right).
    \end{align}
\end{proof}

For finite groups the quantum-capacity analogue appears in \cite[Sec.\ 4.1]{Gao2018}; the entanglement-assisted formula for conditional expectations onto general fixed point algebras --- of which the Haar twirl, being idempotent, is a special case --- is also explicit in the Markovian-noise setting \cite[Thm.\ 3.1]{SinghDatta2024}, with zero-error and unassisted counterparts in \cite{SinghRahamanDatta2025, FawziRahamanTaheri2024}.

\begin{remark}[Equality in the capacity bounds]\label{rem:equality_capacity_bounds}
\autoref{thm:upper_bound_entanglement_assisted_classical_capacity} and
\autoref{cor:entanglement_assisted_classcial_capacity_group_averages} are stated as upper
bounds, but both hold with equality. The state constructed at the end of the proof of
\autoref{thm:upper_bound_entanglement_assisted_classical_capacity} lies by construction in
the range of $\operatorname{id}_R\otimes E^*$ and is hence feasible by
\autoref{lem:fixed points_equal_ergodic_projection}; it attains the bound, as the block
distribution $p_i = d_i^2/\sum_{j=1}^k d_j^2$ contributes $I\left(R:K\right) = H\left(p\right)$, the maximally
entangled states $\Phi^i$ contribute $2\log d_i$ each, and
$H\left(p\right) + \sum_{i=1}^k p_i \log d_i^2 = \log\left(\sum_{i=1}^k d_i^2\right)$ for this choice of $p$.
Consequently
\begin{align}
    C_{\operatorname{EA}}\left(E^*\right)
    = \sup\left\{ I\left(A:B\right)_{\rho_{AB}} \,\middle\vert\, \left(\Phi^*_{A\to A}\otimes \operatorname{id}_B\right)\left(\rho_{AB}\right) = \rho_{AB}\right\}
    = \log\left(\sum_{i=1}^k d_i^2\right),
\end{align}
where the two inequalities follow from the feasibility of the states in
\autoref{eq:def_entanglement_asssisted_classical_capacity} and from purifying the $B$-system
of a feasible state together with data processing, respectively; this is in accordance with
the exact, strongly additive formulas for direct sums of partial traces
\cite[Prop.~1]{Fukuda2007}, \cite[Prop.~5]{Gao2018}. In
\autoref{cor:entanglement_assisted_classcial_capacity_group_averages} the optimizer commutes
with $1\otimes V$ for all $V \in \mathcal{V}$, so the multiplicity bound
$\log\left(\sum_{i=1}^k m_i^2\right)$ is exhausted already within the set of $\mathcal{V}$-invariant
states. We nevertheless phrase both statements as inequalities, since only this direction
enters the de Finetti arguments below, and since the value is not stable under further
restriction of the feasible set: for classical side information the attainable optimum drops
from $\log\left(\sum_i d_i^2\right)$ to $\log\left(\sum_i d_i\right)$, cf.\
\autoref{prop:classical_capacity_conditional_expectation}.
\end{remark}

\subsection{Distortion bounds given fixed points}
In a recent work \cite{kossmann2025_aqec}, more symmetries than just permutation symmetries in de Finetti hierarchies and thus separability problems have been introduced in order to study outer bounds in the finite block-length coding regime for the achievable regions in approximate quantum error correction (see also \cite{Tomamichel2016}). This motivates us to study fixed point constraints on one of the systems of the separable cut $A : B$. Particularly, we aim to approximate separable states $\sigma_{AB}$ such that there exists a decomposition in product states 
\begin{align}
 \sigma_{AB} = \sum_{i\in I} p\left(i\right) \sigma_{A\vert i}   \otimes   \sigma_{B\vert i}, 
\end{align}
with the property that each of the $\sigma_{B\vert i}$ satisfies the following fixed point constraint
\begin{align}
    \Phi_{B\to B}\left(\sigma_{B\vert i}\right) = \sigma_{B\vert i}, \quad i \in I.
\end{align}
Considering the proof-sketch in \autoref{subsec:deFinetti_revisited}, a hidden dimension dependence of the convergence behavior occurs in \autoref{eq:informationally_complete_measurements}, where the bound on the informationally complete measurement is introduced. However, it is known from \cite{Diaconis1977,Diaconis1980} that this bound vanishes in case of classical systems $A$ or $B$. The aim of this subsection is to investigate the new point of view of fixed points as a type of interpolation result. Using the formula for conditional expectations \autoref{prop:structure_prop_cond_expectation_mmwolf} and its dual \autoref{lem:calculation_adjoint}, we show that including fixed points into the distortion bounds $c\left(\mathcal{M}\right)$ indeed is possible and yields, given the near-optimal distortion bound from \cite[Lem.~8]{jee2020quasi} block-wise bounds. Next to our immediate applications for the de Finetti proof strategy from \autoref{subsec:deFinetti_revisited} we note, that distortion bounds can be seen as a general tool in quantum information theory (see e.g.\ \cite{Lami_2018,Brandao_2017}).

\begin{proposition}\label{prop:distortion_conditional_expectation}
Assume a channel $\Phi:\mathcal{B}\left(\mathcal{H}_B\right)\to\mathcal{B}\left(\mathcal{H}_B\right)$ in the Heisenberg picture and assume that its Schr\"odinger
dual $\Phi^{*}$ admits a full-rank fixed point.
Let $E:\mathcal{B}\left(\mathcal{H}_B\right)\to \operatorname{Fix}\left(\Phi\right)$ be the conditional expectation onto the fixed point algebra from \autoref{thm:ergodic_projection_condexp}, and let $E^{*}$ denote its dual.
Then there exists a measurement $\mathcal{M}_B$ such that, for every Hermitian operator
$X_{AB}\in\mathcal{B}\left(\mathcal{H}_A\otimes\mathcal{H}_B\right)$ satisfying
\begin{align}
\left(\operatorname{id}_A\otimes E^{*}\right)\left(X_{AB}\right)=X_{AB},
\end{align}
and which is \emph{block-wise traceless} with respect to the block projections $\{P_i\}_{i=1}^k$\footnote{Given the formula from \autoref{lem:calculation_adjoint}, we have $P_i\coloneqq V_i^\dagger V_i$, which are projections by construction.} of $\operatorname{Fix}\left(\Phi\right)$, i.e.
\begin{align}\label{eq:block-wise_traceless_assumption}
\tr  \left[\left(1_A\otimes P_i\right) X_{AB} \right]=0, \qquad 1\leq i \leq k,
\end{align}
the following distortion bound holds:
\begin{align}
\lVert X_{AB}\rVert_1 \leq 2  d_{\max}   \lVert \left(\operatorname{id}_A\otimes\mathcal{M}_B\right)\left(X_{AB}\right) \rVert_1,
\end{align}
with $d_{\max}\coloneqq \max_{1\leq i\leq k} d_i$, where $\left(d_1,\ldots,d_k\right)$ are the block dimensions of the conditional expectation from \autoref{thm:ergodic_projection_condexp}.
\end{proposition}

\begin{proof}
By \autoref{thm:ergodic_projection_condexp} (2), the existence of a full-rank invariant state for $\Phi^{*}$ ensures that the ergodic projection yields a conditional expectation $E$ onto the fixed point algebra of $\Phi$.
By \autoref{lem:calculation_adjoint}, the dual $E^{*}$ has the form
\begin{align}\label{eq:def_dual_conditional_exp_in_distortion_proof}
E^*\left(\omega_B\right)
= \sum_{i=1}^k V_i^\dagger \left( \tr_{i,2}\left[V_i \omega_B V_i^\dagger\right]\otimes \rho_i \right)V_i,
\qquad
 \omega_B \in \mathcal{S}\left(\mathcal{H}_B\right),
\end{align}
with mutually orthogonal isometries
\begin{align}
V_i :\mathcal{H}_B \to \mathcal{H}_{i,1} \otimes \mathcal{H}_{i,2},
\qquad
P_i\coloneqq V_i^\dagger V_i,
\end{align}
and density operators $\rho_i\in\mathcal{S}\left(\mathcal{H}_{i,2}\right)$ in the decomposition \autoref{eq:decomposition_level_Hilbert_spaces} used in \autoref{lem:calculation_adjoint}.

Let $X_{AB}\in\mathcal{B}\left(\mathcal{H}_A\otimes\mathcal{H}_B\right)$ be Hermitian and define the block-compressions
\begin{equation}
\begin{aligned}
X_{A}^{\left(i\right)}
&\coloneqq \left(1_A \otimes V_i\right)  X_{AB}  \left(1_A \otimes V_i\right)^\dagger \in \mathcal{B}\left(\mathcal{H}_A \otimes \mathcal{H}_{i,1} \otimes \mathcal{H}_{i,2}\right),
\\
Y_A^{\left(i\right)}
&\coloneqq \tr_{i,2}\left[X_{A}^{\left(i\right)}\right]\in \mathcal{B}\left(\mathcal{H}_A \otimes \mathcal{H}_{i,1} \right).
\end{aligned}
\end{equation}
Then by \autoref{eq:def_dual_conditional_exp_in_distortion_proof}
\begin{align}
\left(\operatorname{id}_A\otimes E^*\right)\left(X_{AB}\right)
= \sum_{i=1}^k\left(1_A \otimes V_i\right)^\dagger \left( Y_A^{\left(i\right)}\otimes \rho_i \right)\left(1_A \otimes V_i\right).
\end{align}
Since the summands have orthogonal support, using invariance of the Schatten-$1$-norm under isometries and $\tr\left[\rho_i\right]=1$ gives
\begin{align}\label{eq:proof_proposition_distortion}
 \lVert \left(\operatorname{id}_A\otimes E^*\right)\left(X_{AB}\right) \rVert_1
= \sum_{i=1}^k \lVert   Y_A^{\left(i\right)}\otimes \rho_i \rVert_1
= \sum_{i=1}^k \lVert   Y_A^{\left(i\right)} \rVert_1 .
\end{align}
Moreover, the assumption \autoref{eq:block-wise_traceless_assumption} implies $\tr\left[Y_A^{\left(i\right)}\right]=0$ for all $i$, since
\begin{align}
\tr\left[Y_A^{\left(i\right)}\right]
= \tr  \left[X_A^{\left(i\right)} \right]
= \tr  \left[\left(1_A\otimes V_i^\dagger V_i\right) X_{AB} \right]
= \tr  \left[\left(1_A\otimes P_i\right) X_{AB} \right]
=0.
\end{align}

Now, for each block $i$ we use \cite[Lem.~8]{jee2020quasi} to obtain a measurement $\mathcal{M}_{i,1}$ on $\mathcal{H}_{i,1}$ (identified with a POVM $\{M_{i,1\vert r,s}\}_{r,s}$) such that for every Hermitian traceless operator $Z_A^{\left(i\right)}\in \mathcal{B}\left(\mathcal{H}_A \otimes \mathcal{H}_{i,1} \right)$,
\begin{align}\label{eq:proof_proposition_distortion2}
\Vert Z_A^{\left(i\right)} \rVert_1 \leq 2\, d_i \,  \lVert \left(\operatorname{id}_A \otimes \mathcal{M}_{i,1}\right)\left(Z_A^{\left(i\right)}\right) \rVert_1.
\end{align}

Define a single POVM on $\mathcal{H}_B$ by gluing the block POVMs:
\begin{align}\label{eq:proof_distortion_measurements}
M_{B\vert i,r,s} \coloneqq V_i^\dagger  \left(M_{i,1\vert r,s}\otimes 1_{i,2}\right)V_i.
\end{align}
Then $\sum_{r,s}M_{B\vert i,r,s}=P_i$ and $\sum_{i=1}^k P_i=1_B$, hence $\{M_{B\vert i,r,s}\}_{i,r,s}$ is a POVM. For the associated measurement map $\mathcal{M}_B$, the orthogonality of the blocks together with \autoref{eq:proof_distortion_measurements} yields
\begin{align}\label{eq:proof_proposition_distortion3}
\lVert \left(\operatorname{id}_A \otimes \mathcal{M}_B\right)\left(X_{AB}\right)\rVert_1
= \sum_{i=1}^k  \lVert \left(\operatorname{id}_A \otimes \mathcal{M}_{i,1}\right)\left(Y_A^{\left(i\right)}\right) \rVert_1.
\end{align}

Finally, using the invariance assumption $\left(\operatorname{id}_A\otimes E^{*}\right)\left(X_{AB}\right)=X_{AB}$, we have
\begin{align}
\lVert X_{AB}\rVert_1
&= \lVert \left(\operatorname{id}_A\otimes E^*\right)\left(X_{AB}\right) \rVert_1
= \sum_{i=1}^k\lVert  Y_A^{\left(i\right)}\rVert_1
\\ \label{eq:proof_distortion_final_calc_1}
&\leq \sum_{i=1}^k 2d_i   \lVert \left(\operatorname{id}_A \otimes \mathcal{M}_{i,1}\right)\left(Y_A^{\left(i\right)}\right) \rVert_1
\\ \label{eq:proof_distortion_final_calc_2}
&\leq 2 \max_{1\leq i \leq k} d_i \sum_{i=1}^k  \lVert \left(\operatorname{id}_A \otimes \mathcal{M}_{i,1}\right)\left(Y_A^{\left(i\right)}\right) \rVert_1
\\ \label{eq:proof_distortion_final_calc_3}
&= 2 d_{\max}  \lVert \left(\operatorname{id}_A \otimes \mathcal{M}_B\right)\left(X_{AB}\right)\rVert_1.
\end{align}
Here we used in \autoref{eq:proof_distortion_final_calc_1} the isometric invariance of the Schatten-1-norm, in \autoref{eq:proof_distortion_final_calc_2} the fact from \autoref{eq:proof_proposition_distortion2} and in \autoref{eq:proof_distortion_final_calc_3} the \autoref{eq:proof_proposition_distortion3}. 
\end{proof}

For later purposes we state a concrete corollary adapted to side-information and a concrete structural assumption on $X_{AB}$ in \autoref{prop:distortion_conditional_expectation}.

\begin{corollary}\label{cor:distortion_correlation_postmeasurement}
Under the assumptions of \autoref{prop:distortion_conditional_expectation}, let $\mathcal{M}_{B}$ denote the measurement from the proposition with classical outcome register $Z$.
Let $J$ be a finite classical register and consider an ensemble of states $\{\rho_{AB\vert j}\}_{j\in J}$, each of which is invariant under $\operatorname{id}_A\otimes E^{*}$, i.e.
\begin{align}
\left(\operatorname{id}_A\otimes E^{*}\right)\left(\rho_{AB\vert j}\right)=\rho_{AB\vert j},
\qquad j \in J.
\end{align}
Then
\begin{align}\label{eq:distortion_history_lowerbound}
\lVert \rho_{AZ\vert j}
-\rho_{A\vert j}\otimes \rho_{Z\vert j} \rVert_1
  \geq  
\frac{1}{2  d_{\max}} 
\lVert \rho_{AB \vert j}
- \rho_{A\vert j}\otimes \rho_{B\vert j}\rVert_1, \quad j \in J.
\end{align}
\end{corollary}

\begin{proof}
The technical assumption of \autoref{prop:distortion_conditional_expectation} is the block-wise traceless assumption \autoref{eq:block-wise_traceless_assumption}. For this purpose we define
\begin{align}
    X_{AB\vert j}\coloneqq \rho_{AB\vert j}
-\rho_{A\vert j}\otimes \rho_{B\vert j},
\end{align}
which is Hermitian, traceless and, by assumption, satisfies $\left(\operatorname{id}_A\otimes E^{*}\right)\left(X_{AB\vert j}\right) = X_{AB\vert j}$, since $E^*$ is trace preserving and $E^*\left(\rho_{B\vert j}\right) = \rho_{B\vert j}$.
For each block projection $P_i$, $1\leq i \leq k$, from the conditional expectation onto $\operatorname{Fix}\left(\Phi\right)$ (cf. \autoref{thm:ergodic_projection_condexp}),
\begin{align}
\tr  \left[\left(1_{A}\otimes P_i\right)X_{AB\vert j} \right]
&=
\tr\left[P_i \rho_{B\vert j}\right]
-\tr\left[\rho_{A\vert j}\right] \tr\left[P_i \rho_{B\vert j}\right]
=0,
\end{align}
using $\tr\left[\rho_{A\vert j}\right] = 1$. Thus, we conclude from \autoref{prop:distortion_conditional_expectation} that 
\begin{align}
\lVert \rho_{AZ\vert j}
- \rho_{A\vert j}\otimes \rho_{Z\vert j} \rVert_1
  \geq  
\frac{1}{2  d_{\max}} 
\lVert \rho_{AB\vert j}
- \rho_{A\vert j}\otimes \rho_{B\vert j}\rVert_1, \quad j \in J,
\end{align}
which is \autoref{eq:distortion_history_lowerbound}.
\end{proof}
\subsection{Symmetries in post-measurement states coming from a permutation invariant state}\label{subsec:symmetries_post_measurement}
    Assume a state $\rho_{AB_1^n} \in \mathcal{S}\left(\mathcal{H}_A \otimes \mathcal{H}_B^{\otimes n}\right)$ that is permutation invariant over the $B$-systems with respect to $A$, i.e.\ invariant under the tensor representation \autoref{eq:def_permutation_action_Sn_H_otimes_n} of $S_n$ on the $B$-systems. Applying a measurement $\mathcal{M}_B:\mathcal{S}\left(\mathcal{H}_B\right) \to \mathcal{M}_1\left(\mathcal{Z}\right)$ to the subsystems $B_2,\ldots, B_{m+1}$ leaves us with a convex mixture over conditional states,
    \begin{align}
        \rho_{AB_1} = \sum_{z_2^{m+1}} p\left(z_2^{m+1}\right)\rho_{AB_1\vert z_2^{m+1}},
    \end{align}
    and, as in \autoref{subsec:deFinetti_revisited}, with the associated separable state 
    \begin{align}\label{eqn:separable_state_conditional_state}
        \tilde{\rho}_{AB_1} = \sum_{z_2^{m+1}} p\left(z_2^{m+1}\right)\rho_{A\vert z_2^{m+1}}\otimes \rho_{B_1\vert z_2^{m+1}} \in \operatorname{SEP}\left(A:B\right).
    \end{align}
    In \autoref{observation}~(2) it was already shown that the permutation invariance of $\rho_{AB_1^n}$ is neither used in any of the arguments nor really necessary until the very last step from \autoref{eq:de_Finetti_without_symmetry} to \autoref{eq:de_Finetti_with_symmetry}. We subsequently show that the $S_n$-invariance allows for a significantly more economical representation of the state $\tilde{\rho}_{AB_1}$ than the decomposition in \autoref{eqn:separable_state_conditional_state}: all conditional data depend on the outcome string $z_2^{m+1}$ only through its type $\tau_B\left(z_2^{m+1}\right) \in \mathcal{T}_{m,\lvert\mathcal{Z}\rvert}$, i.e.\ through its $S_m$-orbit, where the type sets, type maps, type classes and type registers are those introduced in \autoref{sec:notation_preliminaries}. Since $\lvert \mathcal{T}_{m,\lvert\mathcal{Z}\rvert}\rvert$ grows only polynomially in $m$ for a fixed alphabet, this \emph{classification of types} replaces the exponentially many summands in \autoref{eqn:separable_state_conditional_state} by polynomially many, which is also one of the combinatorial backbones of the efficiency results in \autoref{sec:efficient_inner_sequence}. We state the result directly in the double-sided setting.
\begin{proposition}[Classification of Types]\label{prop:classification_of_types}
Let $\rho_{A_1^n B_1^n}\in \mathcal{S}\left(\mathcal{H}_{A}^{\otimes n} \otimes \mathcal{H}_B^{\otimes n}\right)$ be permutation-invariant over the $A$-systems with respect to $B_1^n$ and over the $B$-systems with respect to $A_1^n$. Fix $m_A,m_B\leq n-1$ and consider the post-measurement state obtained by measuring
$\mathcal{M}_A:\mathcal{S}\left(\mathcal{H}_A\right)\to \mathcal{M}_1\left(\mathcal{Y}\right)$ on $A_2,\dots,A_{m_A+1}$ and
$\mathcal{M}_B:\mathcal{S}\left(\mathcal{H}_B\right)\to \mathcal{M}_1\left(\mathcal{Z}\right)$ on $B_2,\dots,B_{m_B+1}$ (i.e.\ the product measurement $\mathcal{M}_A^{\otimes m_A} \otimes \mathcal{M}_B^{\otimes m_B}$).
Writing $y_2^{m_A+1}\in\mathcal Y^{m_A}$ and $z_2^{m_B+1}\in\mathcal Z^{m_B}$ for the classical outcome strings, we obtain
\begin{align}
\rho_{A_1B_1} &= \sum_{y_2^{m_A+1},  z_2^{m_B+1}} p\left(y_2^{m_A+1}, z_2^{m_B+1}\right) \rho_{A_1B_1\vert y_2^{m_A+1}\ z_2^{m_B+1}},  \\
\tilde{\rho}_{A_1B_1} &= \sum_{y_2^{m_A+1},  z_2^{m_B+1}} p\left(y_2^{m_A+1}, z_2^{m_A+1}\right)  \rho_{A_1\vert y_2^{m_A+1}\ z_2^{m_B+1}} \otimes \rho_{B_1\vert y_2^{m_A+1}\ z_2^{m_B+1}} .
\end{align}
Let $\mathcal T_{m_A}\coloneqq \{s\in\mathbb N_0^{\mathcal Y}:\sum_{y\in\mathcal Y}s_y=m_A\}$ and $\mathcal T_{m_B} \coloneqq \{t\in\mathbb N_0^{\mathcal Z}:\sum_{z\in\mathcal Z}t_z=m_B\}$ be the type sets, and let
$\tau_A:\mathcal Y^{m_A}\to\mathcal T_{m_A}$ and $\tau_B:\mathcal Z^{m_B}\to\mathcal T_{m_B}$ be the corresponding type maps, as introduced in \autoref{sec:notation_preliminaries}.

Then:
\begin{enumerate}
\item There exists a probability distribution $p\in \mathcal{M}_1\left(\mathcal{T}_{m_A} \times \mathcal{T}_{m_B}\right)$ and states $\rho_{A_1\vert s,t}\in \mathcal{S}\left(\mathcal{H}_A\right)$ and $\rho_{B_1\vert s,t} \in \mathcal{S}\left(\mathcal{H}_B\right)$ such that
\begin{align}
\tilde{\rho}_{A_1B_1} = \sum_{s \in \mathcal{T}_{m_A}} \sum_{t\in \mathcal{T}_{m_B}} p\left(s,t\right)  \rho_{A_1\vert s,t}\otimes \rho_{B_1\vert s,t},
\end{align}
and if $\tau_A\left(y_2^{m_A+1}\right) = \tau_A\left({y^\prime}_2^{m_A+1}\right)$, then for all $z_2^{m_B+1}\in\mathcal Z^{m_B}$,
\begin{align}
p\left(y_2^{m_A+1}, z_2^{m_B+1}\right) &= p\left({y^\prime}_2^{m_A+1}, z_2^{m_B+1}\right), \\
\rho_{A_1\vert y_2^{m_A+1}\ z_2^{m_B+1}} &= \rho_{A_1\vert {y^\prime}_2^{m_A+1}\ z_2^{m_B+1}}, \\
\rho_{B_1\vert y_2^{m_A+1}\ z_2^{m_B+1}} &= \rho_{B_1\vert {y^\prime}_2^{m_A+1}\ z_2^{m_B+1}},
\end{align}
and similarly with the roles of $A$ and $B$ exchanged (i.e.\ for fixed $y_2^{m_A+1}$ and $\tau_B\left(z_2^{m_B+1}\right)=\tau_B\left({z^\prime}_2^{m_B+1}\right)$).

\item If we define the coarse-grained, symmetric POVM elements
\begin{equation}
\begin{aligned}
F_s^A &\coloneqq \sum_{y_2^{m_A+1}:\ \tau_A\left(y_2^{m_A+1}\right) = s}\  \bigotimes_{i=2}^{m_A+1}M^A_{y_i}, \quad s \in \mathcal{T}_{m_A}, \\
G_t^B &\coloneqq \sum_{z_2^{m_B+1}:\ \tau_B\left(z_2^{m_B+1}\right) = t}\  \bigotimes_{i=2}^{m_B+1}M^B_{z_i}, \quad t \in \mathcal{T}_{m_B},
\end{aligned}
\end{equation}
such that
\begin{align}
    \sum_{s\in \mathcal{T}_{m_A}} F_s^A=1_{A_2^{m_A+1}},\quad  \sum_{t\in \mathcal{T}_{m_B}} G_t^B=1_{B_2^{m_B+1}},
\end{align}
then
\begin{align}
I\left(A_1:B_1 \vert  Y_{2}^{m_A+1}Z_2^{m_B+1}\right)_{\rho_{A_1B_1Y_2^{m_A+1}Z_{2}^{m_B+1}}}
=
I\left(A_1:B_1\vert  T_{m_A} \ T_{m_B}\right)_{\rho_{A_1B_1T_{m_A} T_{m_B}}},
\end{align}
where $T_{m_A}$ and $T_{m_B}$ are the classical outcome registers of the coarse-grained POVMs $\{F_s^A\}_{s\in \mathcal{T}_{m_A}}$ and $\{G_t^B\}_{t\in \mathcal{T}_{m_B}}$, respectively, i.e.\ the type registers of \autoref{sec:notation_preliminaries}. Since the conditional states depend on the outcome strings only through their types by $\left(1\right)$, the same identity holds after any further local processing of $A_1$ and $B_1$; in particular, measuring $A_1$ and $B_1$ with $\mathcal{M}_A$ and $\mathcal{M}_B$ yields $I\left(Y_1:Z_1\vert Y_2^{m_A+1}Z_2^{m_B+1}\right)_{\omega} = I\left(Y_1:Z_1\vert T_{m_A}T_{m_B}\right)_{\omega}$ for the fully measured state $\omega$.

\item  If $|\mathcal{Y}|$ and $|\mathcal{Z}|$ are the numbers of outcomes, then
\begin{align}
|\mathcal{T}_{m_A}| = \binom{m_A+|\mathcal{Y}| -1}{|\mathcal{Y}| -1} \leq \left(m_A+1 \right)^{|\mathcal{Y}| -1},
\end{align}
and similarly for $|\mathcal T_{m_B}|$. Moreover, for each $\left(s,t\right)\in\mathcal T_{m_A}\times\mathcal T_{m_B}$,
\begin{align}
p\left(s,t\right)
=
\sum_{\substack{y_2^{m_A+1}:\ \tau_A\left(y_2^{m_A+1}\right)=s}}\ \sum_{\substack{z_2^{m_B+1}:\ \tau_B\left(z_2^{m_B+1}\right)=t}} p\left(y_2^{m_A+1},z_2^{m_B+1}\right)
=
\frac{m_A!}{\prod_{y\in\mathcal Y} s_y!}\cdot \frac{m_B!}{\prod_{z\in\mathcal Z} t_z!}\cdot p\left(\bar y,\bar z\right),
\end{align}
for any fixed representatives $\bar y\in\mathcal Y^{m_A}$, $\bar z\in\mathcal Z^{m_B}$ with $\tau_A\left(\bar y\right)=s$ and $\tau_B\left(\bar z\right)=t$.
\end{enumerate}
\end{proposition}
\begin{proof}
We start with the proof of $\left(1\right)$. Let the product measurement outcomes on $A_2,\dots,A_{m_A+1}$ and $B_2,\dots,B_{m_B+1}$ be
$y_2^{m_A+1}\in\mathcal Y^{m_A}$ and $z_2^{m_B+1}\in\mathcal Z^{m_B}$, and define the subnormalized post-measurement operator\footnote{The operators $\Omega^{y,z}_{A_1B_1}$ are the conditional states weighted by their outcome probabilities; the separate notation avoids dividing by possibly vanishing probabilities.} on
$A_1B_1$ by
\begin{align}\label{eq:Omega_yz_def}
\Omega^{y_2^{m_A+1},z_2^{m_B+1}}_{A_1B_1}
\coloneqq 
\tr_{A_2^{m_A+1}B_2^{m_B+1}}
 \left[
 \left(
 1_{A_1B_1}\otimes  \bigotimes_{i=2}^{m_A+1} M^A_{y_i}\otimes  \bigotimes_{i=2}^{m_B+1} M^B_{z_i}
 \right) 
\rho_{A_1^{m_A+1}B_1^{m_B+1}}
 \right],
\end{align}
so that 
\begin{align}
p\left(y_2^{m_A+1},z_2^{m_B+1}\right)=\tr \left[\Omega^{y_2^{m_A+1},z_2^{m_B+1}}_{A_1B_1} \right] \quad \text{and} \quad     p\left(y_2^{m_A+1},z_2^{m_B+1}\right)\cdot \rho_{A_1B_1\vert  y_2^{m_A+1}z_2^{m_B+1}}=\Omega^{y_2^{m_A+1},z_2^{m_B+1}}_{A_1B_1}.
\end{align}
Moreover,
\begin{align}
\rho_{A_1\vert y_2^{m_A+1}z_2^{m_B+1}}=\tr_{B_1} \left[\rho_{A_1B_1\vert y_2^{m_A+1}z_2^{m_B+1}} \right],
\qquad
\rho_{B_1\vert y_2^{m_A+1}z_2^{m_B+1}}=\tr_{A_1} \left[\rho_{A_1B_1\vert y_2^{m_A+1}z_2^{m_B+1}} \right].
\end{align}

Let $\sigma \in S_{m_A}$ act on the index set $\{2,\dots,m_A+1\}$ and extend it to $\widehat\sigma$ on $\{1,\dots,m_A+1\}$ by
$\widehat\sigma\left(1\right)=1$. Let $U^A_{\widehat\sigma}$ be the unitary permuting $A_1,\dots,A_{m_A+1}$ accordingly.
Define $\left(\sigma \cdot y\right)_i\coloneqq y_{\sigma^{-1}\left(i\right)}$ for $i=2,\dots,m_A+1$.
Permutation invariance of $\rho$ over the $A$-systems w.r.t.\ $B_1^{m_B+1}$ implies
\begin{align}
\left(U^A_{\widehat\sigma}\otimes  1_{B_1^{m_B+1}}\right) \rho_{A_1^{m_A+1}B_1^{m_B+1}} 
\left(U^A_{\widehat\sigma}\otimes  1_{B_1^{m_B+1}}\right)^\dagger
=
\rho_{A_1^{m_A+1}B_1^{m_B+1}}.
\end{align}
Additionally, since $\widehat\sigma$ fixes the index $1$,
\begin{align}
\left(U^A_{\widehat\sigma}\right)^\dagger \left(\bigotimes_{i=2}^{m_A+1} M^A_{\left(\sigma\cdot y\right)_i} \right)U^A_{\widehat\sigma}
=
\bigotimes_{i=2}^{m_A+1} M^A_{y_i}.
\end{align}
Inserting these identities into \autoref{eq:Omega_yz_def} and using the cyclicity of the partial trace with respect to operators supported on the traced-out systems yields
\begin{align}
\Omega^{\left(\sigma \cdot y\right)_2^{m_A+1},z_2^{m_B+1}}_{A_1B_1}
=
\Omega^{y_2^{m_A+1},z_2^{m_B+1}}_{A_1B_1}.
\end{align}
Taking traces and normalizing gives, for all $z_2^{m_B+1}$ and every $y_2^{m_A+1}$,
\begin{equation}
\begin{aligned}
p\left(\left(\sigma \cdot y\right)_2^{m_A+1},z_2^{m_B+1}\right)&=p\left(y_2^{m_A+1},z_2^{m_B+1}\right), \\
\rho_{A_1\vert \left(\sigma \cdot y\right)_2^{m_A+1}\ z_2^{m_B+1}}&=\rho_{A_1\vert y_2^{m_A+1}\ z_2^{m_B+1}},
\\
\rho_{B_1\vert \left(\sigma \cdot y\right)_2^{m_A+1}\ z_2^{m_B+1}}&=\rho_{B_1\vert y_2^{m_A+1}\ z_2^{m_B+1}}.
\end{aligned}
\end{equation}
Since two strings $y_2^{m_A+1},{y_2^\prime}^{m_A+1}\in\mathcal Y^{m_A}$ have the same type
$\tau_A\left(y_2^{m_A+1}\right)=\tau_A\left({y_2^\prime }^{m_A+1}\right)$ iff they are related by such a permutation,
this proves the claimed invariances in $y_2^{m_A+1}$ (for fixed $z_2^{m_B+1}$). The analogous argument using
permutation invariance over the $B$-systems w.r.t.\ $A_1^{m_A+1}$ gives the corresponding invariances in
$z_2^{m_B+1}$ (for fixed $y_2^{m_A+1}$).

Let $s\coloneqq \tau_A\left(y_2^{m_A+1}\right)\in\mathcal T_{m_A}$ and $t\coloneqq \tau_B\left(z_2^{m_B+1}\right)\in\mathcal T_{m_B}$.
From the two invariances above, the triple
\begin{align}
 \left(p\left(y_2^{m_A+1},z_2^{m_B+1}\right),
\rho_{A_1\vert y_2^{m_A+1}z_2^{m_B+1}},
\rho_{B_1\vert y_2^{m_A+1}z_2^{m_B+1}} \right)
\end{align}
depends only on $\left(s,t\right)$. Hence we may define $\rho_{A_1\vert s,t}$ and $\rho_{B_1\vert s,t}$ by choosing any
representatives $\left(\bar y,\bar z\right)$ with $\tau_A\left(\bar y\right)=s$ and $\tau_B\left(\bar z\right)=t$ and setting
\begin{align}
\rho_{A_1\vert s,t} \coloneqq \rho_{A_1\vert \bar y \bar z},
\qquad
\rho_{B_1\vert s,t} \coloneqq \rho_{B_1\vert \bar y \bar z}.
\end{align}
Define also the induced probability distribution on types by
\begin{align}
p\left(s,t\right)\coloneqq
\sum_{\substack{y_2^{m_A+1}:\ \tau_A\left(y_2^{m_A+1}\right)=s}}
\ \sum_{\substack{z_2^{m_B+1}:\ \tau_B\left(z_2^{m_B+1}\right)=t}}
p\left(y_2^{m_A+1},z_2^{m_B+1}\right).
\end{align}
Grouping the sum for $\tilde\rho_{A_1B_1}$ by type classes then yields
\begin{align}
\tilde\rho_{A_1B_1}
=
\sum_{s\in\mathcal T_{m_A}}\sum_{t\in\mathcal T_{m_B}}
p\left(s,t\right)\ \rho_{A_1\vert s,t}\otimes \rho_{B_1\vert s,t},
\end{align}
as claimed. 

The proof for $\left(2\right)$: Let $\rho_{A_1B_1Y_2^{m_A+1}Z_2^{m_B+1}}$ denote the post-measurement state, where $Y_2^{m_A+1},Z_2^{m_B+1}$ are classical, i.e.
\begin{align}
    \rho_{A_1B_1Y_2^{m_A+1}Z_{2}^{m_B+1}}= \sum_{y_2^{m_A+1},\, z_2^{m_B+1}} p\left(y_2^{m_A+1}, z_2^{m_B+1}\right)\, \rho_{A_1B_1\vert y_2^{m_A+1} z_2^{m_B+1}}\otimes \ketbra{y_2^{m_A+1}}{y_2^{m_A+1}}\otimes\ketbra{z_2^{m_B+1}}{z_2^{m_B+1}}.
\end{align}
Applying \autoref{lemma:conditioning_mutual_information} gives
\begin{align}
I\left(A_1:B_1\vert Y_2^{m_A+1}Z_2^{m_B+1}\right)_{\rho_{A_1B_1Y_2^{m_A+1}Z_{2}^{m_B+1}}}
=\sum_{y_2^{m_A+1}, z_2^{m_B+1}} p\left(y_2^{m_A+1},z_2^{m_B+1}\right)  I\left(A_1:B_1\right)_{\rho_{A_1 B_1\vert y_2^{m_A+1},z_2^{m_B+1}}}.
\end{align}
By $\left(1\right)$, the conditional state $\rho_{A_1 B_1\vert y_2^{m_A+1},\, z_2^{m_B+1}}$ depends only on the pair of types
$\left(s,t\right)=\left(\tau_A\left(y_2^{m_A+1}\right),\tau_B\left(z_2^{m_B+1}\right)\right)$.
Hence, there exist states $\rho_{A_1B_1 \vert s,t}$, defined via arbitrary orbit representatives as in the proof of $\left(1\right)$, such that
\begin{align}
\rho_{A_1 B_1\vert y_2^{m_A+1},z_2^{m_B+1}}=\rho_{A_1 B_1\vert \tau_A\left(y_2^{m_A+1}\right) \ \tau_B\left(z_2^{m_B+1}\right)}.
\end{align}
Grouping the above sum by type classes yields
\begin{align}
I\left(A_1:B_1\vert Y_2^{m_A+1}Z_2^{m_B+1}\right)
=\sum_{s,t} p\left(s,t\right)  I\left(A_1:B_1\right)_{\rho_{A_1 B_1 \vert s,t}}.
\end{align}
Now define the classical type registers $T_{m_A}\coloneqq \tau_A\left(Y_2^{m_A+1}\right)$ and $T_{m_B}\coloneqq \tau_B\left(Z_2^{m_B+1}\right)$ of \autoref{sec:notation_preliminaries}.
Since $\left(T_{m_A},T_{m_B}\right)$ is a function of $\left(Y_2^{m_A+1},Z_2^{m_B+1}\right)$, the right-hand side is exactly
the expansion of \autoref{lemma:conditioning_mutual_information}:
\begin{align}
\sum_{s,t} p\left(s,t\right)  I\left(A_1:B_1\right)_{\rho_{A_1B_1\vert s,t}}
=
I\left(A_1:B_1\vert T_{m_A} T_{m_B}\right)_{\rho_{A_1B_1T_{m_A}T_{m_B}}}.
\end{align}
Therefore,
\begin{align}
I\left(A_1:B_1\vert Y_2^{m_A+1}Z_2^{m_B+1}\right)
=
I\left(A_1:B_1\vert T_{m_A} T_{m_B}\right).
\end{align}
Finally, $T_{m_A}$ (resp.\ $T_{m_B}$) can equivalently be obtained by the single coarse-grained POVM
$\{F_s^A\}_s$ (resp.\ $\{G_t^B\}_t$) because classical postprocessing of $\mathcal M_A^{\otimes m_A}$
(resp. $\mathcal M_B^{\otimes m_B}$) by $\tau_A$ (resp.\ $\tau_B$) is equivalent to summing POVM
elements over the corresponding type classes, i.e.\ over the $S_{m_A}$- respectively $S_{m_B}$-orbits.
The claimed identity for the fully measured state follows in the same way, since by $\left(1\right)$ the conditional states $\rho_{A_1B_1\vert y_2^{m_A+1} z_2^{m_B+1}}$, and hence also their images under $\mathcal{M}_A\otimes\mathcal{M}_B$, depend only on the types.

For $\left(3\right)$:
A type $s\in\mathcal T_{m_A}$ is a vector $s=\left(s_y\right)_{y\in\mathcal Y}\in\mathbb N_0^{\mathcal Y}$ with $\sum_{y\in\mathcal Y}s_y=m_A$.
Hence $|\mathcal T_{m_A}|$ equals the number of weak compositions (cf.\ \autoref{sec:notation_preliminaries}) of $m_A$ into $\lvert\mathcal{Y}\rvert$ parts, i.e.\
$|\mathcal T_{m_A}|=\binom{m_A+\vert \mathcal{Y}\vert -1}{\vert \mathcal{Y} \vert-1}$. The polynomial bound is a direct consequence. The cardinality results are standard and can be deduced from e.g. \cite[Sec.~11.1.1]{Cover2006}.

Fix $s\in\mathcal T_{m_A}$ and define the type class
$\mathcal C_A\left(s\right)\coloneqq \{y_2^{m_A+1}\in\mathcal Y^{m_A}:\tau_A\left(y_2^{m_A+1}\right)=s\}$.
Then
\begin{align}
|\mathcal C_A\left(s\right)|=\frac{m_A!}{\prod_{y\in\mathcal Y} s_y!},
\end{align}
because $\mathcal C_A\left(s\right)$ consists of all length-$m_A$ strings with exactly $s_y$ occurrences of each symbol $y$,
counted by the multinomial coefficient. Similarly,
$|\mathcal C_B\left(t\right)|=\frac{m_B!}{\prod_{z\in\mathcal Z} t_z!}$ for $t\in\mathcal T_{m_B}$.

Now for $\left(s,t\right)\in\mathcal T_{m_A}\times\mathcal T_{m_B}$ we have by definition
\begin{align}
p\left(s,t\right)=\sum_{y\in\mathcal C_A\left(s\right)}\ \sum_{z\in\mathcal C_B\left(t\right)} p\left(y,z\right).
\end{align}
Choose representatives $\bar y\in\mathcal C_A\left(s\right)$ and $\bar z\in\mathcal C_B\left(t\right)$.
By $\left(1\right)$, $p\left(y,z\right)$ depends only on the pair of types $\left(\tau_A\left(y\right),\tau_B\left(z\right)\right)$, such that for all
$y\in\mathcal C_A\left(s\right)$ and $z\in\mathcal C_B\left(t\right)$,
$p\left(y,z\right)=p\left(\bar y,\bar z\right)$. Therefore
\begin{equation}
\begin{aligned}
p\left(s,t\right)&=|\mathcal C_A\left(s\right)| |\mathcal C_B\left(t\right)| p\left(\bar y,\bar z\right) \\
&=\frac{m_A!}{\prod_{y\in\mathcal Y} s_y!}\cdot \frac{m_B!}{\prod_{z\in\mathcal Z} t_z!}\cdot p\left(\bar y,\bar z\right),
\end{aligned}
\end{equation}
which is the desired result.
\end{proof}
An immediate consequence of \autoref{prop:classification_of_types} is the following corollary, where we just extend one side. 
\begin{corollary}[One-sided type compression]\label{cor:one_sided_types}
Let $\rho_{AB_1^n}\in\mathcal S\left(\mathcal H_A\otimes\mathcal H_B^{\otimes n}\right)$ be permutation-invariant over the $B$-systems (with $A$ fixed). Fix $m\leq n-1$ and measure $\mathcal M_B$ on $B_2,\dots,B_{m+1}$, producing an outcome string $Z_2^{m+1}\in\mathcal Z^m$ and the separable state
\begin{align}
\tilde\rho_{AB_1}
=\sum_{z_2^{m+1}} p\left(z_2^{m+1}\right) \rho_{A\vert z_2^{m+1}}\otimes \rho_{B_1\vert z_2^{m+1}}.
\end{align}
Let $\mathcal T_{m,\lvert\mathcal{Z}\rvert} = \{t\in\mathbb N_0^{\mathcal Z}:\sum_{z\in\mathcal Z}t_z=m\}$ and $\tau_B:\mathcal Z^m\to\mathcal T_{m,\lvert\mathcal{Z}\rvert}$ be the type set and the type map of \autoref{sec:notation_preliminaries}, and define the type-sorted (coarse-grained) POVM on $\mathcal{H}_B^{\otimes m}$ by
\begin{align}
G_t^B\coloneqq \sum_{z_2^{m+1}:\ \tau_B\left(z_2^{m+1}\right)=t}\ \bigotimes_{i=2}^{m+1} M^B_{z_i},
\qquad t\in\mathcal T_{m, \lvert\mathcal{Z}\rvert},
\end{align}
with outcome register $T_m\coloneqq \tau_B\left(Z_2^{m+1}\right)$. Then there exist probabilities $p\left(t\right)$ and states $\rho_{A\vert t},\rho_{B_1\vert t}$ such that
\begin{align}
\widetilde\rho_{AB_1}=\sum_{t\in\mathcal T_{m,\lvert\mathcal{Z}\rvert}} p\left(t\right) \rho_{A\vert t}\otimes \rho_{B_1\vert t},
\end{align}
and whenever $\tau_B\left(z\right)=\tau_B\left(z'\right)$ we have $p\left(z\right)=p\left(z'\right)$, $\rho_{A\vert z}=\rho_{A\vert z'}$, and $\rho_{B_1\vert z}=\rho_{B_1\vert z'}$.
Moreover, after also measuring $B_1$ with $\mathcal M_B$ (producing $Z_1$),
\begin{align}
I\left(A:Z_1\vert Z_2^{m+1}\right) = I\left(A:Z_1\vert T_m\right).
\end{align}
\end{corollary}

From \autoref{prop:classification_of_types} and \autoref{cor:one_sided_types} we can now deduce a chain rule tailored for the case of a permutation invariant state. 
\begin{proposition}[Chain rule]\label{prop:chain_rule}
Let $\rho_{A_1^n B_1^n}\in \mathcal{S}\left(\mathcal{H}_{A}^{\otimes n} \otimes \mathcal{H}_B^{\otimes n}\right)$ be permutation-invariant over the $A$-systems with respect to $B_1^n$ and over the $B$-systems with respect to $A_1^n$. Fix $m_A,m_B\leq n-1$ and consider the post-measurement state obtained by iid-measurements
$\mathcal{M}_A:\mathcal{S}\left(\mathcal{H}_A\right)\to \mathcal{M}_1\left(\mathcal{Y}\right)$ and 
$\mathcal{M}_B:\mathcal{S}\left(\mathcal{H}_B\right)\to \mathcal{M}_1\left(\mathcal{Z}\right)$ 
\begin{align}
    \omega_{Y_1^n Z_1^n}
      \coloneqq  
     \left(\mathcal{M}_A^{\otimes n}\otimes \mathcal{M}_B^{\otimes n}\right)  \left(\rho_{A_1^nB_1^n}\right).
\end{align}
Then:
\begin{enumerate}
    \item[\textup{(1)}] (\textup{double extensions}) For every $m_A,m_B\in\{0,\dots,n-1\}$ let
    $T_{m_A}$ and $T_{m_B}$ denote the post-measurement type registers associated with
    $Y_2^{m_A+1}$ and $Z_2^{m_B+1}$ from \autoref{prop:classification_of_types}. Then
    \begin{align}
        I\left(Y_1^n:Z_1^n\right)_\omega
          =  
        \sum_{m_A=0}^{n-1} \sum_{m_B=0}^{n-1}
        I\left(Y_1:Z_1\vert T_{m_A}  T_{m_B}\right)_{\omega_{Y_1Z_1T_{m_A} T_{m_B} }}.
    \end{align}
    \item[\textup{(2)}] (\textup{one-sided extensions}) Let $\rho_{AB_1^n}\in\mathcal{S}\left(\mathcal{H}_A\otimes\mathcal{H}_B^{\otimes n}\right)$
    be invariant under permutations of the $B$-systems.
    Define $\omega_{AZ_1^n}\coloneqq\left(\mathrm{id}_A\otimes \mathcal{M}_B^{\otimes n}\right)\left(\rho_{AB_1^n}\right)$.
    For $m_B\in\{0,\dots,n-1\}$ let $T_{m_B} $ be the type register from \autoref{cor:one_sided_types}. Then
    \begin{align}
        I\left(A:Z_1^n\right)_\omega
          =  
        \sum_{m_B=0}^{n-1} I\left(A:Z_1\vert  T_{m_B} \right)_{\omega_{AZ_1T_{m_B} }}.
    \end{align}
\end{enumerate}
\end{proposition}

\begin{proof}
\textup{(1)} Apply \autoref{lem:double_chain_rule_MI} to $\omega_{Y_1^nZ_1^n}$ to obtain
\begin{align}\label{eq:double_chain_rule_start}
    I\left(Y_1^n:Z_1^n\right)_\omega
      =  
    \sum_{m_A=0}^{n-1}\sum_{m_B=0}^{n-1}
    I  \left(Y_{m_A+1}:Z_{m_B+1}  | Y_1^{m_A}Z_1^{m_B}\right)_\omega.
\end{align}
Since $\rho_{A_1^nB_1^n}$ is invariant under independent permutations and the measurement is i.i.d.,
the classical state $\omega_{Y_1^nZ_1^n}$ is invariant under independent permutations of the $Y$-registers and $Z$-registers, i.e.\ for all $\pi,\sigma\in S_n$,
\begin{align}\label{eq:omega_perm_invar}
    \left(P_\pi^Y\otimes P_\sigma^Z\right) \omega_{Y_1^nZ_1^n} \left(P_\pi^Y\otimes P_\sigma^Z\right)^\dagger
      =  \omega_{Y_1^nZ_1^n},
\end{align}
where $P_\pi^Y$ permutes $Y_1,\dots,Y_n$ and $P_\sigma^Z$ permutes $Z_1,\dots,Z_n$.

Fix $m_A,m_B\in\{0,\dots,n-1\}$ and let $\pi_{m_A+1}\in S_n$ be the cycle $\left(1\ 2\ \cdots\ m_A+1\right)$
(and likewise $\sigma_{m_B+1}$ on the $Z$-indices). Then
\begin{align}
    \pi_{m_A+1}\left(Y_{m_A+1}\right)=Y_1,\quad \pi_{m_A+1}\left(Y_1^{m_A}\right)=Y_2^{m_A+1},\qquad
    \sigma_{m_B+1}\left(Z_{m_B+1}\right)=Z_1,\quad \sigma_{m_B+1}\left(Z_1^{m_B}\right)=Z_2^{m_B+1}.
\end{align}
Because $\omega$ is classical, applying $\left(P_{\pi_{m_A+1}}^Y\otimes P_{\sigma_{m_B+1}}^Z\right)$ is just a relabeling of the random variables,
hence all (conditional) Shannon entropies, and thus conditional mutual informations, are invariant under this relabeling. Using \autoref{eq:omega_perm_invar} we therefore obtain
\begin{align}\label{eq:perm_step}
    I  \left(Y_{m_A+1}:Z_{m_B+1}  | Y_1^{m_A}Z_1^{m_B}\right)_\omega =  I  \left(Y_1:Z_1  | Y_2^{m_A+1}Z_2^{m_B+1}\right)_\omega.
\end{align}
Inserting \autoref{eq:perm_step} into \autoref{eq:double_chain_rule_start} yields
\begin{align}
    I\left(Y_1^n:Z_1^n\right)_\omega =  \sum_{m_A=0}^{n-1}\sum_{m_B=0}^{n-1}
    I  \left(Y_1:Z_1  | Y_2^{m_A+1}Z_2^{m_B+1}\right)_\omega.
\end{align}
Now apply \autoref{prop:classification_of_types}~(2), in the form for the fully measured state, to replace the conditioning registers
$Y_2^{m_A+1}$ and $Z_2^{m_B+1}$ by the corresponding type registers $T_{m_A} $ and $T_{m_B} $, giving
\begin{align}
    I  \left(Y_1:Z_1  | Y_2^{m_A+1}Z_2^{m_B+1}\right)_\omega =   I\left(Y_1:Z_1\vert T_{m_A}  T_{m_B} \right)_{\omega_{Y_1Z_1T_{m_A} T_{m_B} }},
\end{align}
and thus the claim in \textup{(1)}.

\textup{(2)} Apply the usual chain rule $I\left(A:Z_1^n\right)=\sum_{m_B=0}^{n-1} I\left(A:Z_{m_B+1}\vert Z_1^{m_B}\right)$ to $\omega_{AZ_1^n}$.
Using $B$-permutation-invariance of $\rho_{AB_1^n}$ and i.i.d.\ measurement, we have invariance of $\omega_{AZ_1^n}$
under permutations of the $Z$-registers, and with the same cycle argument as above obtain
\begin{align}
    I\left(A:Z_{m_B+1}\vert Z_1^{m_B}\right)_\omega
      =  
    I\left(A:Z_1\vert Z_2^{m_B+1}\right)_\omega.
\end{align}
Finally apply \autoref{cor:one_sided_types} to replace $Z_2^{m_B+1}$ by $T_{m_B} $, which yields \textup{(2)}.
\end{proof}

\subsection{Bounds on the mutual information given constraints}

In this section we investigate in which sense equality constraints on the set of feasible states restrict the bounds on the mutual information. To be concrete, we consider the value of the following optimization program
\begin{equation}\label{eq:program_equality_constraints}
    \begin{aligned}
        \sup \quad & I\left(A : B\right)_{\rho_{AB}} \\ 
        \operatorname{s.t.} \quad &\left(\Phi_{A\to A^\prime}\otimes \operatorname{id}_B\right)\left(\rho_{AB}\right) = \omega_{A^\prime}\otimes \rho_B \\
        &\rho_{AB} \in \mathcal{S}\left(\mathcal{H}_A \otimes \mathcal{H}_B\right).
    \end{aligned}
\end{equation}
Here $\Phi_{A\to A^\prime}$ is assumed to be a quantum channel in the Schr\"odinger picture and $\omega_{A^\prime} \in \mathcal{S}\left(\mathcal{H}_{A^\prime}\right)$ is a fixed state. The discussion is inspired by \cite[Lem.~6]{jee2020quasi} and generalizes it to arbitrary channels.
\begin{proposition}\label{prop:equality_constraint_upper_bound}
    The value of the program \autoref{eq:program_equality_constraints} is bounded by $2\log d_E$, where $d_E$ denotes the Choi rank of the channel $\Phi_{A\to A^\prime}$, i.e.\ the minimal environment dimension of a Stinespring dilation.
\end{proposition}
\begin{proof}
    Let $V_{A\to A^\prime E}:\mathcal{H}_A\to \mathcal{H}_{A^\prime}\otimes \mathcal{H}_E$ be a minimal Stinespring dilation (see e.g.\ \cite{Kretschmann2008} for a rigorous definition of minimality). Due to the isometric invariance of the relative entropy we have
    \begin{equation}
    \begin{aligned}
        I\left(A : B\right)_{\rho_{AB}} &= D\left(\rho_{AB} \Vert \rho_A \otimes \rho_B\right) \\
        &= D\left(\left(V_{A\to A^\prime E}\otimes 1_B\right) \rho_{AB} \left(V_{A\to A^\prime E}\otimes 1_B\right)^\dagger \Vert V_{A\to A^\prime E}\, \rho_A V_{A\to A^\prime E}^\dagger\otimes \rho_B\right) \\
        &= I\left(A^\prime E:B\right)_{\sigma_{A^\prime E B}}
    \end{aligned}
    \end{equation}
    with $\sigma_{A^\prime E B} \coloneqq \left(V_{A\to A^\prime E}\otimes 1_B\right) \rho_{AB} \left(V_{A\to A^\prime E}\otimes 1_B\right)^\dagger$. Now we use the chain rule
    \begin{align}
        I\left(A^\prime E:B\right)_{\sigma_{A^\prime E B}} = I\left(A^\prime : B\right)_{\sigma_{A^\prime B}} + I\left(E:B\vert A^\prime\right)_{\sigma_{A^\prime E B}}.
    \end{align}
    By the constraint we have
    \begin{equation}
    \begin{aligned}
        \sigma_{A^\prime B} &= \tr_{E}\left[\left(V_{A\to A^\prime E}\otimes 1_B\right) \rho_{AB} \left(V_{A\to A^\prime E}\otimes 1_B\right)^\dagger \right] \\
        &= \left(\Phi_{A\to A^\prime}\otimes \operatorname{id}_B\right)\left(\rho_{AB}\right)\\
        &= \omega_{A^\prime} \otimes \rho_B,
    \end{aligned}
    \end{equation}
    which is product. Thus $I\left(A^\prime : B\right)_{\sigma_{A^\prime B}} = 0$, and the standard estimate $I\left(E:B\vert A^\prime\right)_{\sigma_{A^\prime E B}} \leq 2 \log d_E$ yields the desired result.
\end{proof}
Generally, \autoref{prop:equality_constraint_upper_bound} need not improve on the trivial upper bound $2 \log d_A$, as the Choi rank can be as large as $d_A d_{A^\prime}$ (see e.g.\ \cite[Cor.~1.63]{Quantum_Information_Processing}); already for the completely dephasing channel, where $A^\prime = A$ and $d_E = d_A$, the bound reduces to the trivial one.

\begin{example}[Partial trace constraints]\label{ex:partial_trace_constraint}
    Let $\mathcal{H}_A = \mathcal{H}_{A_1}\otimes \mathcal{H}_{A_2}$ and take $\Phi_{A\to A_1} = \tr_{A_2}$ to be the partial trace. A minimal Stinespring dilation is the identity embedding with environment $E = A_2$, i.e.\ $d_E = d_{A_2}$. \autoref{prop:equality_constraint_upper_bound} then states that every state with $\tr_{A_2}\left[\rho_{AB}\right] = \omega_{A_1}\otimes \rho_B$ satisfies $I\left(A:B\right)_{\rho_{AB}}\leq 2\log d_{A_2}$, which recovers the constrained decoupling bound of \cite[Lem.~6]{jee2020quasi} used in the hierarchies of \cite{Berta2021, kossmann2025_aqec}; note that the environment system makes a proof without a third auxiliary system possible. In this case the bound improves on the trivial one whenever $d_{A_2} < d_{A_1}d_{A_2} = d_A$, i.e.\ whenever the constraint decouples (correlation breaking) a proper subsystem.
\end{example}

\section{de Finetti theorems and constraints}\label{sec:deFinetti_theorems_constraints}

In this section we prove a new de Finetti theorem as well as improved versions of several known ones. Our starting point is \autoref{observation}: the tools of \autoref{sec:techniques} answer each of its three questions in the affirmative, and each answer strengthens the de Finetti argument of \autoref{subsec:deFinetti_revisited} at a different point.
\begin{enumerate}
    \item \emph{Symmetries in the mutual information bound (\autoref{observation}~(1)).} The states $\rho_{AB_1^n}\in \mathcal{S}\left(\mathcal{H}_A \otimes \mathcal{H}_B^{\otimes n}\right)$ under consideration carry symmetry --- at the very least permutation invariance over the $B$-systems --- which the dimension bound in \autoref{eq:mutual_information_bounds} does not exploit. \autoref{thm:upper_bound_entanglement_assisted_classical_capacity} incorporates general fixed point constraints, and hence in particular such symmetries, directly into the bound on the mutual information.
    \item \emph{Compression of the chain rule (\autoref{observation}~(2)).} For permutation invariant states the number of distinct summands in \autoref{eq:self_decoupling_lemma_intro} reduces from exponential to polynomial in the number of measured systems, cf.\ \autoref{prop:classification_of_types}~(3), which yields the refined chain rule of \autoref{prop:chain_rule}.
    \item \emph{Smaller distortion factor (\autoref{observation}~(3)).} Whenever the post-measurement states $\rho_{AB\vert z_1^m}$ in \autoref{eq:informationally_complete_measurements} inherit local fixed point constraints of the form in \autoref{prop:distortion_conditional_expectation}, the distortion factor $c\left(\mathcal{M}\right)$ can be reduced by means of \autoref{cor:distortion_correlation_postmeasurement}.
\end{enumerate}
Combining all three ingredients, we first prove a double-sided de Finetti theorem (\autoref{thm:double_extended_de_Finetti}) and an interpolation de Finetti theorem under fixed point constraints (\autoref{thm:interpolation_deFinetti}), in which the dimension dependence is carried entirely by the block data of the local fixed point algebras; \autoref{thm:proper_deFinetti} and \autoref{thm:bose_de_Finetti} then complement these with a $k$-copy version with side information and a Bose-symmetric double-sided version, respectively. All de Finetti theorems of this section are moreover compatible with additional constraints with fixed marginals in the sense of \autoref{eq:def_cSEP_set}: by the inheritance argument
recorded at the end of \autoref{subsec:deFinetti_revisited}, the measurement-based separable candidates satisfy every such constraint imposed on the extension, so that the statements below apply to $\operatorname{SEP}$ and $\operatorname{cSEP}$ problems alike.

\subsection{Double-sided extensions and interpolation de Finetti theorem}
    For de Finetti theorems derived from the self-decoupling lemma of \cite{Brandao_2017}, it has remained an open question whether double-sided extensions can be accommodated and how they affect the convergence behavior. The obstacle is the mutual information bound itself: for general states, \autoref{eq:mutual_information_bounds} only yields
    \begin{align}\label{eq:bound_mutual_information_no_further_information}
        I(A:B_1^n)_{\rho_{AB_1^n}} \leq 2\log \min \{d_A, d_B^n\}.
    \end{align}
    For a double-sided extension, i.e.\ states of the form $\rho_{A_1^n B_1^n}\in \mathcal{S}\left(\mathcal{H}_A^{\otimes n} \otimes \mathcal{H}_B^{\otimes n}\right)$, the bound in \autoref{eq:bound_mutual_information_no_further_information} deteriorates to
    \begin{align}
        I(A_1^n:B_1^n)_{\rho_{A_1^nB_1^n}} \leq 2\log \min \{d_A^n, d_B^n\} = 2\, n\, \log \min \{d_A, d_B\} \ \sim \ n,
    \end{align}
    growing linearly in the number $n$ of copies. Traced through the argument of \autoref{subsec:deFinetti_revisited}, this linear growth exactly cancels the $1/n$ gain from the chain rule: the overall bound in \autoref{eq:de_Finetti_with_symmetry} stagnates at a constant, precluding any convergence argument as $n\to \infty$ along these lines.

    The way out is the point made in \autoref{observation}: since the step from \autoref{eq:de_Finetti_without_symmetry} to \autoref{eq:de_Finetti_with_symmetry} requires permutation invariance anyway, the states under consideration may be assumed permutation invariant from the outset. Consequently, the relevant quantity is no longer the general bound in \autoref{eq:bound_mutual_information_no_further_information}, but its counterpart for permutation invariant states --- that is, for fixed points of the group average in \autoref{cor:entanglement_assisted_classcial_capacity_group_averages}, applied to the permutation action of the symmetric group from \autoref{eq:def_permutation_action_Sn_H_otimes_n}. Together with the well-known bounds of \autoref{lem:bounds_symmetric_group}, the corollary replaces the linear growth by a polylogarithmic one --- additional structure sharpens the bound precisely where the dimension argument fails --- and thereby restores convergence as $n\to \infty$. This is made precise in the following theorem.


\begin{theorem}\label{thm:double_extended_de_Finetti}
        Let $\rho_{AB} \in \mathcal{S}\left(\mathcal{H}_A\otimes \mathcal{H}_B\right)$ be a state such that there exists a permutation invariant double sided $n$-extension. Then there exists a separable state $\tilde{\rho}_{AB} \in \operatorname{SEP}\left(A:B\right)$ such that
        \begin{align}
            \lVert \rho_{AB} - \tilde{\rho}_{AB}\rVert_1 \leq 4 d_A d_B\frac{\sqrt{2 \ln 2 \left(d_A^2 -1\right)\,\log\left(n+d_A \right)}}{n}.
        \end{align}
    \end{theorem}
    \begin{proof}
        We start by considering the permutation invariant extension $\rho_{A_1^n B_1^n} \in \mathcal{S}\left(\mathcal{H}_A^{\otimes n}\otimes \mathcal{H}_B^{\otimes n}\right)$ of $\rho_{AB}$ and $I\left(A_1^n:B_1^n\right)_{\rho_{A_1^nB_1^n}}$. By \autoref{cor:entanglement_assisted_classcial_capacity_group_averages} we find that for permutation invariant states
        \begin{align}
            I\left(A_1^n:B_1^n\right)_{\rho_{A_1^nB_1^n}} \leq \log  \left(\sum_{i=1}^k m_i^2 \right).
        \end{align}
        Here the coefficients $m_i$ are determined by the decomposition of $\mathcal{H}_A^{\otimes n}$ under the action of $S_n$ from \autoref{eq:def_permutation_action_Sn_H_otimes_n} into blocks as discussed in \autoref{sec:weyls_dimension_formula}. Particularly the blocks can be identified with all partitions $\lambda_1\geq \lambda_2 \geq \ldots \geq \lambda_{d_A}\geq \lambda_{d_A+1}\geq 0$ such that $\lambda_{d_A+1}= 0$ and $\sum_{i=1}^{d_A}\lambda_i =  n$. For this specific action the amount of blocks and the size of the blocks can then be estimated with Weyl's dimension formula with the bounds from \autoref{lem:bounds_symmetric_group}
        \begin{equation}
            \begin{aligned}
                \lvert \{\lambda \vdash n \ \vert \ \lambda_{d_A+1}= 0\}\rvert &\leq \left(n+d_A\right)^{d_A-1}\\
                m_\lambda &\leq \left(n+d_A\right)^{d_A\left(d_A-1\right)/2}\quad \text{for all} \quad \lambda \vdash n, \quad \lambda_{d_A+1}=0.
            \end{aligned}
        \end{equation}
        Thus, expressing the sum over blocks as a sum over the admissible partitions and bounding it by the number of partitions times the largest squared multiplicity, we obtain
        \begin{equation}
        \begin{aligned}
            I\left(A_1^n:B_1^n\right)_{\rho_{A_1^nB_1^n}} &\leq\log \left( \sum_{\substack{\lambda \vdash n \\ \lambda_{d_A+1}=0}} m_\lambda^2 \right) \leq \log \left( \lvert \{\lambda \vdash n \ \vert \ \lambda_{d_A+1}= 0\}\rvert \cdot \max_{\substack{\lambda \vdash n \\ \lambda_{d_A+1}=0}} m_\lambda^2\right) \\
            &\leq\log \left( \left(n+d_A\right)^{d_A-1} \left(\left(n+d_A\right)^{d_A\left(d_A-1\right)/2}\right)^2 \right) \\
            &= \left(d_A^2 -1\right)\,\log\left(n+d_A \right).
        \end{aligned}
        \end{equation}
        Now we apply informationally complete measurements $\mathcal{M}_A$ and $\mathcal{M}_B$ from \cite[Lem.~8]{jee2020quasi} such that we obtain by DPI
        \begin{align}\label{eq:proof_double_extended_deFinetti_bound_DPI}
            I\left(Y_1^n:Z_1^n\right)_{\rho_{Y_1^nZ_1^n}}\leq I\left(A_1^n:B_1^n\right)_{\rho_{A_1^nB_1^n}} \leq \left(d_A^2 -1\right)\,\log\left(n+d_A \right).
        \end{align}
        By \autoref{prop:chain_rule} we have 
        \begin{align}
        I\left(Y_1^n:Z_1^n\right)_\rho
          =  
        \sum_{m_A=0}^{n-1} \sum_{m_B=0}^{n-1}
        I\left(Y_1:Z_1\vert T_{m_A }  T_{m_B}\right)_{\rho_{Y_1Z_1T_{m_A }T_{m_B}}}.
    \end{align}
    Note that the double sum runs over the $n^2$ pairs $\left(m_A,m_B\right)$, i.e.\ over the numbers of already measured systems on the two sides; the type registers $T_{m_A}$ and $T_{m_B}$ enter each summand only as conditioning systems and are not summed over, so the number of summands is exactly $n^2$, irrespective of the cardinalities of the type sets $\mathcal{T}_{m_A,\lvert\mathcal{Y}\rvert}$ and $\mathcal{T}_{m_B,\lvert\mathcal{Z}\rvert}$. Since each summand is a conditional mutual information of a classical state and hence nonnegative, at least one of them is bounded by the average, i.e.\ there exists a pair $\left(m_A,m_B\right)$ such that
    \begin{align}\label{eq:proof_two_sided_deFinetti1}
        I\left(Y_1:Z_1\vert T_{m_A }  T_{m_B}\right)_{\rho_{Y_1Z_1T_{m_A }T_{m_B}}} \leq \frac{\left(d_A^2 -1\right)\,\log\left(n+d_A \right)}{n^2}.
    \end{align}
    By \autoref{lemma:conditioning_mutual_information} we have 
    \begin{align}
        I\left(Y_1:Z_1\vert T_{m_A }  T_{m_B} \right)_{\rho_{Y_1Z_1T_{m_A }T_{m_B} }} = \sum_{t_{m_A},t_{m_B}} p\left(t_{m_A},t_{m_B}\right) I\left(Y_1:Z_1\right)_{\rho_{Y_1Z_1\vert t_{m_A},t_{m_B}}}.
    \end{align}
    Thus, with the convexity of $x  \mapsto x^2$ and of the Schatten-$1$-norm we have
    \begin{equation}\label{eq:proof_two_sided_deFinetti2}
    \begin{aligned}
        \sum_{t_{m_A},t_{m_B}}&p\left(t_{m_A},t_{m_B}\right) I\left(Y_1:Z_1\right)_{\rho_{Y_1Z_1\vert t_{m_A},t_{m_B}}} = \sum_{t_{m_A},t_{m_B}} p\left(t_{m_A},t_{m_B}\right)D\left(\rho_{Y_1Z_1\vert t_{m_A},t_{m_B}}\Vert \rho_{Y_1\vert t_{m_A},t_{m_B}} \otimes \rho_{Z_1\vert t_{m_A},t_{m_B}}\right) \\
        &  \geq \frac{1}{2\ln 2}\sum_{t_{m_A},t_{m_B}} p\left(t_{m_A},t_{m_B}\right)\Vert\rho_{Y_1Z_1\vert t_{m_A},t_{m_B}}-\rho_{Y_1\vert t_{m_A},t_{m_B}} \otimes \rho_{Z_1\vert t_{m_A},t_{m_B}}\rVert^2_1 \\
        &\geq \frac{1}{2\ln 2}\ \lVert\sum_{t_{m_A},t_{m_B}}p\left(t_{m_A},t_{m_B}\right) \rho_{Y_1Z_1\vert t_{m_A},t_{m_B}}-\sum_{t_{m_A},t_{m_B}}p\left(t_{m_A},t_{m_B}\right) \rho_{Y_1\vert t_{m_A},t_{m_B}} \otimes \rho_{Z_1\vert t_{m_A},t_{m_B}}\rVert^2_1 \\
        &= \frac{1}{2\ln 2}\ \lVert\rho_{Y_1Z_1}-\sum_{t_{m_A},t_{m_B}}p\left(t_{m_A},t_{m_B}\right) \rho_{Y_1\vert t_{m_A},t_{m_B}} \otimes \rho_{Z_1\vert t_{m_A},t_{m_B}}\rVert^2_1. 
    \end{aligned}
    \end{equation}
    Combining \autoref{eq:proof_two_sided_deFinetti1} and \autoref{eq:proof_two_sided_deFinetti2} yields
    \begin{align}
        \lVert\rho_{Y_1Z_1}-\sum_{t_{m_A},t_{m_B}}p\left(t_{m_A},t_{m_B}\right) \rho_{Y_1\vert t_{m_A},t_{m_B}} \otimes \rho_{Z_1\vert t_{m_A},t_{m_B}}\rVert_1 \leq \frac{\sqrt{2 \ln 2 \left(d_A^2 -1\right)\,\log\left(n+d_A \right)}}{n}.
    \end{align}
    Applying now the linear distortion of informationally complete measurements from \cite[Lem.~8]{jee2020quasi}, once for each side with $c\left(\mathcal{M}_A\right)\leq 2d_A$ and $c\left(\mathcal{M}_B\right)\leq 2d_B$, yields
    \begin{align}\label{eq:final_result_double_extended_deFinetti}
        \lVert\rho_{AB}-\sum_{t_{m_A},t_{m_B}}p\left(t_{m_A},t_{m_B}\right) \rho_{A\vert t_{m_A},t_{m_B}} \otimes \rho_{B\vert t_{m_A},t_{m_B}}\rVert_1 \leq 4 d_A d_B \frac{\sqrt{2 \ln 2 \left(d_A^2 -1\right)\,\log\left(n+d_A \right)}}{n}.
    \end{align}
    We define $\tilde{\rho}_{AB} \coloneqq \sum_{t_{m_A},t_{m_B}}p\left(t_{m_A},t_{m_B}\right) \rho_{A\vert t_{m_A},t_{m_B}} \otimes \rho_{B\vert t_{m_A},t_{m_B}} \in \operatorname{SEP}\left(A:B\right)$ in order to conclude the desired assertion.
    \end{proof}
    
    We close this section by considering the following example.
    \begin{example}\label{example_double_extension}
Let $\rho_{A_1^{2n}} \in \mathcal{S}\left(\mathcal{H}_A^{\otimes 2n}\right)$ be permutation invariant with respect to the natural action of $S_{2n}$ from \autoref{eq:def_permutation_action_Sn_H_otimes_n}. Consider, moreover, the embeddings of two isomorphic copies of $S_n$, acting on the first and on the last $n$ tensor factors respectively, into $S_{2n}$:
\begin{align}
    S_{\{1,2,\ldots,n\}} \hookrightarrow S_{2n}
    \quad \text{and} \quad
    S_{\{n+1,n+2,\ldots,2n\}} \hookrightarrow S_{2n}.
\end{align}
Using these embeddings together with $S_{2n}$-permutation invariance, we can apply the construction of \autoref{thm:double_extended_de_Finetti} in the special case $\mathcal{H}_A=\mathcal{H}_B$. In this setting, the theorem effectively addresses a separability problem between subsystems of the $A$-chain.

Now consider, without loss of generality, the bipartition at the cut $A_1:A_2$. This yields a direct comparison between the bound from \cite[Thm.~II.7]{Christandl2007} and the one from \autoref{thm:double_extended_de_Finetti}. There exist separable states $\sigma_{A_1A_2},\tilde{\sigma}_{A_1A_2}\in \operatorname{SEP}(A_1:A_2)$ such that
\begin{equation}
\begin{aligned}
\lVert \rho_{A_1A_2} - \sigma_{A_1A_2}\rVert_1
&\le \frac{d_A^2}{n}
&& \text{\cite[Thm.~II.7]{Christandl2007}},\\
\lVert \rho_{A_1A_2} - \tilde{\sigma}_{A_1A_2}\rVert_1
&\le \frac{4d_A^3}{n}\,\sqrt{2\ln 2\,\log\left(n+d_A\right)}
&& \autoref{thm:double_extended_de_Finetti}.
\end{aligned}
\end{equation}
Hence, \autoref{thm:double_extended_de_Finetti} achieves a scaling that is close to the best known de Finetti bound for this type of separability task. At the same time, it offers an additional structural advantage: the approximating separable state $\tilde{\sigma}_{A_1A_2}$ is compatible with a rich set of constraints and can be obtained constructively. For $\sigma_{A_1A_2}$, such a constructive constrained characterization is not known (see also the discussion in \cite{Berta2021} and its application in \cite{kossmann2025_aqec,zeiss2025approximatingfixedsizequantum,zeiss2026finitefinetticonvexbodies}).
\end{example}

In the following, we aim to prove an interpolation result for de Finetti theorems, motivated by the following observation. The convex hull of product states on two quantum systems $\mathcal{H}_A$ and $\mathcal{H}_B$ is known to be the set of separable states $\operatorname{SEP}(A:B)$. Assume now that we consider the convex hull of states whereby all states on $\mathcal{H}_B$ satisfy a fixed point constraint. This can be written as follows
\begin{equation}
\begin{aligned}
   \Sigma(A:B) &\coloneqq \operatorname{conv}\{\rho_A \otimes \sigma_B \ \vert \ \rho_A \in \mathcal{S}\left(\mathcal{H}_A\right), \ \sigma_B \in \mathcal{S}\left(\mathcal{H}_B\right), \ \Phi_{A\to A}\left(\rho_A\right)= \rho_A, \ \Phi_{B\to B}\left(\sigma_B\right) =\sigma_B\} \\ &\subseteq \operatorname{SEP}\left(A:B\right).
\end{aligned}
\end{equation}
    The constraint in $\Sigma(A:B)$ can be understood as defining general interpolation sets between the classical case, i.e. invariance under the group average corresponding to the $d_A$-fold maximal torus $\mathbb{T}^{d_A}$, and similarly for $d_B$, and the quantum case, where $\Phi_{A\to A} = \operatorname{id}_A$ and similarly for $\Phi_{B\to B}$. In terms of the complexity of a de Finetti theorem, this should in particular lead to improvements in the dimension dependence of the bounds, which occurs in all finite quantum de Finetti theorems \cite{Knig2005,Christandl2007,Berta2021} and even beyond, in the general finite de Finetti theorem on convex bodies \cite{zeiss2026finitefinetticonvexbodies}, while in the classical setting it is well known that this dependence disappears \cite{Diaconis1977,Diaconis1980}. Combining all of our techniques from \autoref{sec:techniques}, we show the following result.

    \begin{theorem}[de Finetti theorem under fixed-point constraints]\label{thm:interpolation_deFinetti}
        Assume channels $\Phi_{A\to A}:\mathcal{S}\left(\mathcal{H}_A\right) \to \mathcal{S}\left(\mathcal{H}_A\right)$ and $\Phi_{B\to B}:\mathcal{S}\left(\mathcal{H}_B\right) \to \mathcal{S}\left(\mathcal{H}_B\right)$ admitting a full-rank fixed point. Then for every state $\rho_{AB}$ admitting a one-sided permutation-invariant $n$-extension $\rho_{AB_1^n}$ that satisfies $(\Phi_{A\to A}\otimes\, \operatorname{id}_{B_1^n})\left(\rho_{AB_1^n}\right) = \rho_{AB_1^n}$ and $(\operatorname{id}_{AB_1^{n-1}}\otimes \, \Phi_{B\to B})\left(\rho_{AB_1^n}\right) = \rho_{AB_1^n}$, there exists a state $\tilde{\rho}_{AB} \in \Sigma\left(A:B\right)$ satisfying  
        \begin{align}
            \lVert \rho_{AB} - \tilde{\rho}_{AB}\rVert_1 \leq 2 \, m_{B,\operatorname{max}} \sqrt{\frac{2 \ln 2 \, \log  \left(\sum_{i=1}^{k_A} m_{i,A}^2 \right)}{n}},
        \end{align}
        whereby $m_{i,A}$ are the dimensions of the fixed point algebra of $\Phi_{A\to A}$ and $m_{B,\operatorname{max}} \coloneqq \max_{1\leq i \leq k_B} m_{i,B}$ is the largest block-size of the fixed point algebra of $\Phi_{B\to B}$.
    \end{theorem}
    \begin{proof}
        Given the full-rank fixed point of $\Phi_{A\to A}$, we can apply \autoref{thm:upper_bound_entanglement_assisted_classical_capacity} in order to get
        \begin{align}
            I\left(A:B_1^n\right)_{\rho_{AB_1^n}} \leq \log  \left(\sum_{i=1}^{k_A} \left(m_{i,A}\right)^2 \right).
        \end{align}
By \autoref{thm:ergodic_projection_condexp}~(2) the existence of a full-rank fixed point of $\Phi_{B\to B}$ yields the existence of a conditional expectation $E_B$ and with \autoref{lem:calculation_adjoint}, the dual $E_B^{*}$ has the form
\begin{align}
E^*\left(\omega_B\right)
= \sum_{i=1}^{k_B} V_i^\dagger \left( \tr_{i,2}\left[V_i \omega_B V_i^\dagger\right]\otimes \rho_i \right)V_i,
\qquad
 \omega_B \in \mathcal{S}\left(\mathcal{H}_B\right),
\end{align}
with mutually orthogonal isometries
\begin{align}
V_i :\mathcal{H}_B \to \mathcal{H}_{i,1} \otimes \mathcal{H}_{i,2},
\qquad
P_i\coloneqq V_i^\dagger V_i,
\end{align}
and density operators $\rho_i\in\mathcal{S}\left(\mathcal{H}_{i,2}\right)$ in the decomposition \autoref{eq:decomposition_level_Hilbert_spaces} used in \autoref{lem:calculation_adjoint}. Now, for each block $i$ we use \cite[Lem.~8]{jee2020quasi} to obtain a measurement $\mathcal{M}_{i,1}$ on $\mathcal{H}_{i,1}$ (identified with a POVM $\{M_{i,1\vert r,s}\}_{r,s}$). Define a single POVM $\mathcal{M}_B$ on $\mathcal{H}_B$ by gluing the block POVMs 
\begin{align}\label{eq:construction_measurement_interpolation_deFinetti}
    M_{B\vert i,r,s} \coloneqq V_i^\dagger  \left(M_{i,1\vert r,s}\otimes 1_{i,2}\right)V_i.
\end{align}
Then $\sum_{r,s}M_{B\vert i,r,s}=P_i$ and $\sum_{i=1}^{k_B} P_i=1_B$, hence $\{M_{B\vert i,r,s}\}_{i,r,s}$ is a POVM. We apply the $n$-fold product measurement $\mathcal{M}_B^{\otimes n}$ of this measurement channel and apply DPI for the mutual information with this informationally complete measurement on $B$, such that we have
        \begin{align}
            I\left(A:Z_1^n\right)_{\rho_{AZ_1^n}} \leq I\left(A:B_1^n\right)_{\rho_{AB_1^n}} \leq \log  \left(\sum_{i=1}^{k_A} m_{i,A}^2 \right).
        \end{align}
       Applying now  \autoref{prop:chain_rule}~(2) yields
        \begin{align}\label{eq:interpolation_deFinetti_bound}
        I\left(A:Z_1^n\right)_\rho
          =  
        \sum_{m_B=0}^{n-1} I\left(A:Z_1\vert  T_{m_B}\right)_{\rho_{AZ_1T_{m_B}}}.
    \end{align}
    Thus, there exists a $0\leq m_B \leq n-1$ such that
    \begin{align}
        I\left(A:Z_1\vert  T_{m_B}\right)_{\rho_{AZ_1T_{m_B}}} \leq \frac{\log  \left(\sum_{i=1}^{k_A} m_{i,A}^2 \right)}{n}.
    \end{align}
    In the following estimate we apply 
    \begin{enumerate}
        \item \autoref{lemma:conditioning_mutual_information} from \autoref{eq:calculation_1} to \autoref{eq:calculation_2}
        \item Pinsker inequality from \autoref{eq:calculation_3} to \autoref{eq:calculation_4}
        \item \autoref{cor:distortion_correlation_postmeasurement} from \autoref{eq:calculation_4} to \autoref{eq:calculation_5}
        \item the convexity of $x\mapsto x^2$ and of the norm $\lVert \cdot \rVert_1$ in the remaining steps.
    \end{enumerate}
    \begin{align}\label{eq:calculation_1}
        \frac{\log  \left(\sum_{i=1}^{k_A} m_{i,A}^2 \right)}{n}&\geq I\left(A:Z_1\vert  T_{m_B}\right)_{\rho_{AZ_1T_{m_B}}}   \\ \label{eq:calculation_2}
        &=\sum_{t_{m_B}}p\left(t_{m_B}\right) I\left(A:Z_1\right)_{\rho_{AZ_1\vert t_{m_B}}}\\ \label{eq:calculation_3}
        &= \sum_{t_{m_B}} p\left(t_{m_B}\right)D\left(\rho_{AZ_1\vert t_{m_B}}\Vert \rho_{A\vert t_{m_B}} \otimes \rho_{Z_1\vert t_{m_B}}\right) \\ \label{eq:calculation_4}
        &  \geq \frac{1}{2 \, \ln 2}\sum_{t_{m_B}} p\left(t_{m_B}\right)\lVert\rho_{AZ_1\vert t_{m_B}}-\rho_{A\vert t_{m_B}} \otimes \rho_{Z_1\vert t_{m_B}}\rVert^2_1 \\ \label{eq:calculation_5}
        &  \geq \frac{1}{2 \, \ln 2\,  c(\mathcal{M})^2} \sum_{t_{m_B}} p\left(t_{m_B}\right) \,\lVert\rho_{AB_1\vert t_{m_B}}-\rho_{A\vert t_{m_B}} \otimes \rho_{B_1\vert t_{m_B}}\rVert^2_1 \\ \label{eq:calculation_6}
        &\geq \frac{1}{2 \ln 2 \, c(\mathcal{M})^2} \, \lVert\sum_{t_{m_B}}p\left(t_{m_B}\right) \rho_{AB_1\vert t_{m_B}}-\sum_{t_{m_B}}p\left(t_{m_B}\right) \rho_{A\vert t_{m_B}} \otimes \rho_{B_1\vert t_{m_B}}\rVert^2_1 \\ \label{eq:calculation_7}
        &= \frac{1}{2 \ln 2 \, c(\mathcal{M})^2} \, \lVert\rho_{AB_1}-\sum_{t_{m_B}}p\left(t_{m_B}\right) \rho_{A\vert t_{m_B}} \otimes \rho_{B_1\vert t_{m_B}}\rVert^2_1. 
    \end{align}
    As the construction of $\mathcal{M}_B$ on $Z_1$ in \autoref{eq:construction_measurement_interpolation_deFinetti} is identical to the construction in \autoref{prop:distortion_conditional_expectation}, we can verbatim apply the bound on the distortion $c(\mathcal{M}_B) \leq 2\, m_{B,\operatorname{max}}$. Note that this step is admissible: the conditional states inherit the fixed-point constraint, $\left(\operatorname{id}_A\otimes E_B^*\right)\left(\rho_{AB_1\vert t_{m_B}}\right) = \rho_{AB_1\vert t_{m_B}}$, since the measurement acts on $B_2^{m_B+1}$ only and commutes with the channel on $B_1$ (cf.\ the argument in \autoref{cor:constrained_efficient_rounding}), and the differences $\rho_{AB_1\vert t_{m_B}} - \rho_{A\vert t_{m_B}}\otimes\rho_{B_1\vert t_{m_B}}$ are blockwise traceless as verified in \autoref{cor:distortion_correlation_postmeasurement}. We conclude
    \begin{align}
        \lVert\rho_{AB}-\sum_{t_{m_B}}p\left(t_{m_B}\right) \rho_{A\vert t_{m_B}} \otimes \rho_{B\vert t_{m_B}}\rVert_1 \leq 2\,  m_{B,\operatorname{max}} \sqrt{\frac{2 \, \ln 2 \, \log  \left(\sum_{i=1}^{k_A} m_{i,A}^2 \right)}{n}},
    \end{align}
    whereby $m_{B,\operatorname{max}}$ is defined as in \autoref{prop:distortion_conditional_expectation} as the maximal block-dimension of the fixed point algebra of $\Phi_{B\to B}$, and $\tilde{\rho}_{AB}\coloneqq \sum_{t_{m_B}}p\left(t_{m_B}\right)\rho_{A\vert t_{m_B}}\otimes\rho_{B\vert t_{m_B}} \in \Sigma\left(A:B\right)$ by the inherited fixed-point constraints.
    \end{proof}
\subsection{Further de Finetti theorems}
We adapt this theorem to a proper de Finetti argument, i.e.\ to an approximation of $\rho_{AB_1^k}$ by a mixture of $k$-fold product states with side information (cf.\ \cite[Thm.~2.4]{Berta2021}).
\begin{theorem}\label{thm:proper_deFinetti}
Assume a channel $\Phi_{B\to B}:\mathcal{S}\left(\mathcal{H}_B\right) \to \mathcal{S}\left(\mathcal{H}_B\right)$ admitting a full-rank fixed point. Furthermore fix $0<k\leq n-1$ and set $N\coloneqq n-k$. Then for every state $\rho_{AB_1^k}$ admitting a one-sided permutation-invariant $n$-extension $\rho_{AB_1^n}$ that satisfies $(\operatorname{id}_{AB_1^{n-1}}\otimes \, \Phi_{B\to B})\left(\rho_{AB_1^n}\right) = \rho_{AB_1^n}$, there exist $m \in \lbrace k,\ldots,n-1\rbrace$ and a measurement $\mathcal{M}_B$ on the systems $B_{k+1},\ldots,B_{m}$ with outcome strings $z_{k+1}^m$ such that
\begin{align}
\lVert
\rho_{AB_1^k}-
\sum_{z_{k+1}^{m}}p(z_{k+1}^{m})\rho_{A\vert z_{k+1}^{m}}\otimes \left(\rho_{B_1\vert z_{k+1}^{m}}\right)^{\otimes k}
\rVert_1
\leq
2\, k \, m_{B,\operatorname{max}}\sqrt{\frac{2\ln2\,\log  \left( d_A^2 \left(k-1+d_B\right)^{d_B^2-1}\right)}{N}},
\end{align}
where $m_{B,\operatorname{max}} \coloneqq \max_{1\leq i \leq k_B} m_{i,B}$ is the largest block-size of the fixed point algebra of $\Phi_{B\to B}$.
\end{theorem}
\begin{proof}
The state $\rho_{AB_1^n}$ is invariant under permutations of the systems $B_1,\ldots,B_{i-1}$ for every $i \leq n$. Applying \autoref{cor:entanglement_assisted_classcial_capacity_group_averages} to this $S_{i-1}$-action, whose commutant on $\mathcal{H}_A\otimes\mathcal{H}_B^{\otimes\left(i-1\right)}$ has blocks with multiplicities $d_A m_\lambda$, $\lambda \vdash_{d_B}\left(i-1\right)$ (cf.\ \autoref{sec:weyls_dimension_formula}), together with \autoref{lem:bounds_symmetric_group} gives, for every $1\leq i\leq k$,
        \begin{align}\label{eq:proof_proper_deFinetti_mutual_info_bound}
        \begin{split}
            I\left(AB_1^{i-1}:B_{k+1}^n\right)_{\rho_{AB_1^n}} &\leq \log\Big(\sum_{\lambda\vdash_{d_B}\left(i-1\right)}\left(d_A m_\lambda\right)^2\Big)\leq \log  \left( d_A^2 \left(i-1+d_B\right)^{d_B^2-1}\right)\\
            &\leq \log\left(d_A^2\left(k-1+d_B\right)^{d_B^2-1}\right).
        \end{split}
        \end{align}
Moreover, by \autoref{thm:ergodic_projection_condexp}~(2) the existence of a full-rank fixed point of $\Phi_{B\to B}$ yields the existence of a conditional expectation $E_B$ and with \autoref{lem:calculation_adjoint}, the dual $E_B^{*}$ has the form
\begin{align}
E^*\left(\omega_B\right)
= \sum_{i=1}^{k_B} V_i^\dagger \left( \tr_{i,2}\left[V_i \omega_B V_i^\dagger\right]\otimes \rho_i \right)V_i,
\qquad
 \omega_B \in \mathcal{S}\left(\mathcal{H}_B\right),
\end{align}
with mutually orthogonal isometries
\begin{align}
V_i :\mathcal{H}_B \to \mathcal{H}_{i,1} \otimes \mathcal{H}_{i,2},
\qquad
P_i\coloneqq V_i^\dagger V_i,
\end{align}
and density operators $\rho_i\in\mathcal{S}\left(\mathcal{H}_{i,2}\right)$ in the decomposition \autoref{eq:decomposition_level_Hilbert_spaces} used in \autoref{lem:calculation_adjoint}. Now, for each block $i$ we use \cite[Lem.~8]{jee2020quasi} to obtain a measurement $\mathcal{M}_{i,1}$ on $\mathcal{H}_{i,1}$ (identified with a POVM $\{M_{i,1\vert r,s}\}_{r,s}$). Define a single POVM $\mathcal{M}_B$ on $\mathcal{H}_B$ by gluing the block POVMs 
\begin{align}
    M_{B\vert i,r,s} \coloneqq V_i^\dagger  \left(M_{i,1\vert r,s}\otimes 1_{i,2}\right)V_i.
\end{align}
Then $\sum_{r,s}M_{B\vert i,r,s}=P_i$ and $\sum_{i=1}^{k_B} P_i=1_B$, hence $\{M_{B\vert i,r,s}\}_{i,r,s}$ is a POVM. We apply the product measurement $\mathcal{M}_B^{\otimes \left(n-k\right)}$ of this measurement channel to the systems $B_{k+1},\ldots,B_n$ and apply DPI for the mutual information with this informationally complete measurement, such that, for every $1\leq i \leq k$, we have
        \begin{align}\label{eq:proof_deFinetti_under_fixed point1}
    I\left(AB_{1}^{i-1}:Z_{k+1}^n\right)_{\rho_{AB_1^kZ_{k+1}^n}} \leq I\left(AB_1^{i-1}:B_{k+1}^n\right)_{\rho_{AB_1^n}} \leq \log  \left( d_A^2 \left(k-1+d_B\right)^{d_B^2-1}\right).
\end{align}
        Now we apply the chain rule over the measured block, with the convention that $Z_{k+1}^{k}$ is the trivial register,
        \begin{align}
    I\left(AB_1^{i-1}:Z_{k+1}^n\right) = \sum_{m=k}^{n-1}I\left(AB_1^{i-1}:Z_{m+1}\vert Z_{k+1}^m\right).
\end{align}
        From \autoref{eq:proof_deFinetti_under_fixed point1} and $N=n-k $ we get
\begin{align}
\frac{1}{N}\sum_{m=k}^{n-1}
I\left(AB_1^{i-1}:Z_{m+1}\vert Z_{k+1}^m\right)
=
\frac{I(AB_1^{i-1}:Z_{k+1}^n)}{N}
\leq \frac{\log  \left( d_A^2 \left(k-1+d_B\right)^{d_B^2-1}\right)}{N}.
\end{align}
For each $m\in\left\{k ,\dots,n-1\right\}$, write, using \autoref{lemma:conditioning_mutual_information}, Pinsker's inequality and the distortion bound of $\mathcal{M}_B$ with auxiliary system $AB_1^{i-1}$,
\begin{align}
I\left(AB_1^{i-1}:Z_{m+1}\vert Z_{k+1}^m\right)
&=
\sum_{z_{k+1}^m} p(z_{k+1}^m)\,
D\left(
\rho_{AB_1^{i-1}Z_{m+1}\vert z_{k+1}^m}
\Vert
\rho_{AB_1^{i-1}\vert z_{k+1}^m}\otimes \rho_{Z_{m+1}\vert z_{k+1}^m}
\right)\\
&\ge \frac{1}{2\ln2}\sum_{z_{k+1}^m}p(z_{k+1}^m)
\lVert
\rho_{AB_1^{i-1}Z_{m+1}\vert z_{k+1}^m}
-\rho_{AB_1^{i-1}\vert z_{k+1}^m}\otimes \rho_{Z_{m+1}\vert z_{k+1}^m}
\rVert_1^2\\
&\geq \frac{1}{2\ln2\,c\left(\mathcal{M}_B\right)^2}\sum_{z_{k+1}^m}p(z_{k+1}^m)
\lVert
\rho_{AB_1^{i-1}B_{m+1}\vert z_{k+1}^m}
-\rho_{AB_1^{i-1}\vert z_{k+1}^m}\otimes \rho_{B_{m+1}\vert z_{k+1}^m}
\rVert_1^2.
\end{align}
Conditioned on the outcomes $z_{k+1}^m$, the state is still invariant under permutations of the unmeasured systems $B_1,\ldots,B_k, B_{m+1},\ldots,B_n$, such that $\rho_{AB_1^{i-1}B_{m+1}\vert z_{k+1}^m} = \rho_{AB_1^{i}\vert z_{k+1}^m}$ after relabeling $B_{m+1}\leftrightarrow B_i$. Therefore, we have, averaging over $m$ and using the convexity of $x\mapsto x^2$,
\begin{align}\label{eq:general_i_fixedpoint}
\frac{1}{N}\sum_{m=k}^{n-1}
\sum_{z_{k+1}^m}p(z_{k+1}^m)
\lVert
\rho_{AB_1^{i}\vert z_{k+1}^m}
-\rho_{AB_1^{i-1}\vert z_{k+1}^m}\otimes \rho_{B_i\vert z_{k+1}^m}
\rVert_1
\leq
c\left(\mathcal{M}_B\right)\sqrt{\frac{2\ln2\,\log  \left( d_A^2 \left(k-1+d_B\right)^{d_B^2-1}\right)}{N}}.
\end{align}
Now fix $m\in\left\{k,\dots,n-1\right\}$ and $z_{k+1}^m$. Using triangle inequality $k-1$ times yields then
\begin{equation}
\begin{aligned}
&\lVert
\rho_{AB_1^k\vert z_{k+1}^m}
-\rho_{A\vert z_{k+1}^m}\otimes \rho_{B_1\vert z_{k+1}^m}\otimes\cdots\otimes\rho_{B_k\vert z_{k+1}^m}
\rVert_1\\
&\leq
\sum_{i=1}^k
\lVert
\rho_{AB_1^i\vert z_{k+1}^m}\otimes \rho_{B_{i+1}\vert z_{k+1}^m}\otimes\cdots\otimes\rho_{B_k\vert z_{k+1}^m}\\
&\hspace{4cm}-
\rho_{AB_1^{i-1}\vert z_{k+1}^m}\otimes \rho_{B_i\vert z_{k+1}^m}\otimes \rho_{B_{i+1}\vert z_{k+1}^m}\otimes\cdots\otimes\rho_{B_k\vert z_{k+1}^m}
\rVert_1\\
&=
\sum_{i=1}^k
\lVert
\rho_{AB_1^i\vert z_{k+1}^m}
-\rho_{AB_1^{i-1}\vert z_{k+1}^m}\otimes \rho_{B_i\vert z_{k+1}^m}
\rVert_1,
\end{aligned}
\end{equation}
Averaging now over $m$ and $z_{k+1}^m$ and applying \autoref{eq:general_i_fixedpoint} for each $i$ yields
\begin{align}
\frac{1}{N}\sum_{m=k}^{n-1}\sum_{z_{k+1}^m}p(z_{k+1}^m)
\lVert
\rho_{AB_1^k\vert z_{k+1}^m}
-\rho_{A\vert z_{k+1}^m}\otimes& \rho_{B_1\vert z_{k+1}^m}\otimes\cdots\otimes\rho_{B_k\vert z_{k+1}^m}
\rVert_1\\
&\leq 
\sum_{i=1}^k c(\mathcal{M}_B)\sqrt{\frac{2\ln2\, \log  \left( d_A^2 \left(k-1+d_B\right)^{d_B^2-1}\right)}{N}}\\
&=
k\, c(\mathcal{M}_B)\sqrt{\frac{2\ln2\,\log  \left( d_A^2 \left(k-1+d_B\right)^{d_B^2-1}\right)}{N}}.
\end{align}
Hence there exists $m\in\{k,\dots,n-1\}$ such that
\begin{equation}
\begin{aligned}
\sum_{z_{k+1}^{m}}p(z_{k+1}^{m})
\lVert
\rho_{AB_1^k\vert z_{k+1}^{m}}
-\rho_{A\vert z_{k+1}^{m}}\otimes &\rho_{B_1\vert z_{k+1}^{m}}\otimes\cdots\otimes\rho_{B_k\vert z_{k+1}^{m}}
\rVert_1
\leq \\
&k\,c(\mathcal{M}_B)\sqrt{\frac{2\ln2\,\log  \left( d_A^2 \left(k-1+d_B\right)^{d_B^2-1}\right)}{N}}.
\end{aligned}
\end{equation}
By convexity of the Schatten-$1$-norm we conclude
\begin{align}
\lVert
\rho_{AB_1^k}-
\sum_{z_{k+1}^{m}}p(z_{k+1}^{m})\rho_{A\vert z_{k+1}^{m}}\otimes \rho_{B_1\vert z_{k+1}^{m}}\otimes\cdots\otimes\rho_{B_k\vert z_{k+1}^{m}}
\rVert_1
\leq
k\,c(\mathcal{M}_B)\sqrt{\frac{2\ln2\,\log  \left( d_A^2 \left(k-1+d_B\right)^{d_B^2-1}\right)}{N}}.
\end{align}
By symmetry, $\rho_{B_i\vert z_{k+1}^{m}}=\rho_{B_1\vert z_{k+1}^{m}}$ for all $i=1,\dots,k$, such that we conclude
\begin{align}
\lVert
\rho_{AB_1^k}-
\sum_{z_{k+1}^{m}}p(z_{k+1}^{m})\rho_{A\vert z_{k+1}^{m}}\otimes \left(\rho_{B_1\vert z_{k+1}^{m}}\right)^{\otimes k}
\rVert_1
\leq
k\,c(\mathcal{M}_B)\sqrt{\frac{2\ln2\,\log  \left( d_A^2 \left(k-1+d_B\right)^{d_B^2-1}\right)}{N}}.
\end{align}
Now we introduce the bound $c\left(\mathcal{M}_B\right)\leq 2\,m_{B,\operatorname{max}}$ from \autoref{prop:distortion_conditional_expectation} (cf.\ also \autoref{cor:distortion_correlation_postmeasurement}) in order to get
\begin{align}
\lVert
\rho_{AB_1^k}-
\sum_{z_{k+1}^{m}}p(z_{k+1}^{m})\rho_{A\vert z_{k+1}^{m}}\otimes \left(\rho_{B_1\vert z_{k+1}^{m}}\right)^{\otimes k}
\rVert_1
\leq
2\, k \, m_{B,\operatorname{max}}\sqrt{\frac{2\ln2\,\log  \left( d_A^2 \left(k-1+d_B\right)^{d_B^2-1}\right)}{N}},
\end{align}
where $m_{B,\operatorname{max}} \coloneqq \max_{1\leq i \leq k_B} m_{i,B}$ is the largest block-size of the fixed point algebra of $\Phi_{B\to B}$, as in \autoref{prop:distortion_conditional_expectation}; the required blockwise tracelessness and invariance of the conditional states hold as in the proof of \autoref{thm:interpolation_deFinetti}.
\end{proof}
We remark that \autoref{thm:proper_deFinetti} is in some sense not optimal as it uses just the $S_{k-1}$-symmetry in the mutual information bound established in \autoref{thm:upper_bound_entanglement_assisted_classical_capacity} and not the additional decomposition imposed by the fixed point condition of $\Phi_{B\to B}$. However, if we are just interested in a scaling estimate, it is worth emphasizing that, given a local symmetry $\mathcal{V}$ as in \autoref{cor:entanglement_assisted_classcial_capacity_group_averages}, the overall decomposition of $\mathcal{H}_B^{\otimes (k-1)}$ can be written as 
\begin{align}
    \mathcal{H}^{\otimes (k-1)}_B \cong \bigoplus_{\lambda \vdash_d \ k-1} \bigoplus_{i=1}^l \mathbb{C}^{m_{\lambda,i}} \otimes \mathbb{C}^{d_\lambda} \otimes \mathbb{C}^{e_i},
\end{align}
whereby $d_\lambda$ corresponds to the dimension of the corresponding Specht module and $e_i$ to the dimension of the irreducibles of $\mathcal{V}$. Without further concrete structure, this scales at least as $\left(k-1 \right)^{d_B}$, as this is  roughly the number of summands in $\bigoplus_{\lambda \vdash_d \ k-1}$ by \autoref{sec:weyls_dimension_formula}. Thus, we will not get a substantial scaling improvement in comparison to \autoref{eq:proof_proper_deFinetti_mutual_info_bound}. 

\begin{example}[Classical subsystems]
    We consider the example of $\mathcal{H}_A$ to be the trivial system and a maximal torus, which is by definition a maximal compact, connected and abelian subgroup of $\mathcal{U}\left(\mathcal{H}_B\right)$ \cite[Chap.~15]{Bump2013}. For the case of a unitary group $\mathcal{U}\left(\mathcal{H}_B\right)$ on a $d_B$-dimensional Hilbert space, it is given by $\mathbb{T}^{d_B} \subseteq \mathcal{U}\left(\mathcal{H}_B\right)$, whereby $\mathbb{T}\cong \mathcal{U}\left(1\right)$, the unit circle with Lie-group structure. It is readily seen that its commutant is isomorphic to $\mathbb{C}^{d_B}$ with point-wise multiplication as multiplication, such that the blocks are all one-dimensional and $m_{B,\operatorname{max}} = 1$. Applying \autoref{thm:proper_deFinetti} with this constraint on the $B$-systems and $d_A = 1$ leads to a de Finetti bound  given by
    \begin{align}
            \lVert \rho_{B_1^k} - \tilde{\rho}_{B_1^k}\rVert_1 \leq 2\, k \sqrt{\frac{2\ln2\,\log  \left(  \left(k-1+d_B\right)^{d_B^2-1}\right)}{N}},
    \end{align}
    which just has the dimension dependence in $\log\left( \left(k-1+d_B\right)^{d_B^2-1}\right)$ in comparison to \cite{Diaconis1977,Diaconis1980}, which lead to dimension-free bounds. It is worth mentioning that we could also estimate the bound with $2\left(k-1\right)\log d_B$ if we do not apply \autoref{sec:weyls_dimension_formula}. It should be also noted that for polytopes the de Finetti hierarchy stops after finitely many steps \cite{Aubrun2024}. As in \autoref{example_double_extension}, it is worth noting that the candidate in \autoref{thm:interpolation_deFinetti} can be accessed and is compatible with certain constraints. 
\end{example}

The commutant of the $*$-algebra generated by the permutation representation \autoref{eq:def_permutation_action_Sn_H_otimes_n} decomposes into one matrix block for every partition $\lambda \vdash_{d_A} n$, as discussed in detail in \autoref{sec:weyls_dimension_formula} and \autoref{sec:subgroup_adapted_basis}. Working with all of these blocks simultaneously comes at a technical price: the standard bases of the corresponding Weyl modules, such as the semistandard basis, are not orthonormal, so the associated changes of basis fail to be isometries and suitable Gramian matrices have to be determined (cf., e.g., \cite[Sec.\ 2.10]{sagan2013symmetric} and \autoref{rem:young_symmetrizers}). The block of the partition $\left(n\right)$ --- the symmetric subspace --- is the exception: it admits an explicit orthogonal basis indexed by weak compositions, arising as a special case of the general construction in \cite{zeiss2025approximatingfixedsizequantum}. Concretely, for a Hilbert space $\mathcal{H}_A$ and a weak composition $\mu \coloneqq (\mu_1,\ldots,\mu_{d_A})$ of $n$, define the unnormalized vectors
    \begin{align}
        v_{\mu} \coloneqq \sum_{x\in [d_A]^n, \ \vert \{t \ \vert \ x_t = i\}\vert = \mu_i} e_{x_1}\otimes e_{x_2}\otimes \cdots \otimes e_{x_n}.
    \end{align}
    Distinct weak compositions label disjoint sets of strings, so the vectors $\{v_\mu\}_{\mu \vDash n}$ are pairwise orthogonal, and normalizing them yields an orthonormal basis of the symmetric subspace; in particular, their number recovers $\dim \operatorname{Sym}^n\left(\mathcal{H}_A\right) = \binom{n+d_A-1}{d_A-1} = m_{\left(n\right)}$, the multiplicity entering the bound of \autoref{thm:bose_de_Finetti} below. Restricting a de Finetti argument to the symmetric subspace therefore avoids the basis issues above entirely. Operators supported on the symmetric subspace are called \emph{Bose-symmetric} and are fully characterized by the multiplicative invariance (see e.g. \cite{harrow2013churchsymmetricsubspace, zeiss2025approximatingfixedsizequantum})
    \begin{align}
        \left(1_A \otimes U_\sigma\right) \rho_{AB_1^n} = \rho_{AB_1^n} \quad \text{for all} \quad \sigma \in S_n,
    \end{align}
    which is strictly stronger than invariance under conjugation. While \cite{Christandl2007} proved a Bose-symmetric de Finetti theorem whose convergence rate improves on the merely permutation invariant case ($2kd/n$ vs.\ $2kd^2/n$), the underlying proof method is neither compatible with additional constraints nor constructive. This gap was closed in \cite{zeiss2025approximatingfixedsizequantum} with a constraint-compatible Bose-symmetric de Finetti theorem; however, as the construction there proceeds via a purification step, the resulting convergence rate does not improve on the permutation invariant counterpart. The following theorem extends the Bose-symmetric de Finetti theorem to double-sided extensions, allowing for applications to cSEP and for constructive, measurement-based rounding procedures (see \autoref{sec:efficient_inner_sequence}).

    \begin{theorem}[Bose-symmetric de Finetti]\label{thm:bose_de_Finetti}
        Let $\rho_{AB} \in \mathcal{S}\left(\mathcal{H}_A\otimes \mathcal{H}_B\right)$ be a state such that there exists a Bose-symmetric double sided $n$-extension, i.e. there exists a state $\rho_{A_1^n B_1^n}$ with $\rho_{A_1B_1} = \rho_{AB}$ such that $\rho_{A_1^n B_1^n}$ is Bose-symmetric over the $A$-systems with respect to $B_1^n$ and similarly over the $B$-systems with respect to $A_1^n$. Then there exists a separable state $\tilde{\rho}_{AB} \in \operatorname{SEP}\left(A:B\right)$ such that
        \begin{align}
            \lVert \rho_{AB} - \tilde{\rho}_{AB}\rVert_1 \leq 4 d_A d_B \frac{\sqrt{4 \ln 2 \left(d_A -1\right)\,\log\left(n+d_A \right)}}{n}.
        \end{align}
    \end{theorem}
    \begin{proof}
        Bose symmetry over the $A$-systems implies in particular that $\rho_{A_1^nB_1^n}$ is invariant under the $A$-side twirl, so \autoref{cor:entanglement_assisted_classcial_capacity_group_averages}, applied exactly as in the proof of \autoref{thm:double_extended_de_Finetti}, yields
        \begin{align}\label{eq:bound_entanglement_assisted_capacity_bose_symmetric}
            I\left(A_1^n:B_1^n\right)_{\rho_{A_1^nB_1^n}} \leq \log  \left(\sum_{\lambda \vdash n} m_{\lambda}^2 \right).
        \end{align}
        The support condition improves this further: in the intermediate bound \autoref{eq:proof_entanglement_assisted_close_to_final} from the proof of \autoref{thm:upper_bound_entanglement_assisted_classical_capacity}, the block distribution $\left(p_\lambda\right)$ of $\rho_{A_1^nB_1^n}$ is concentrated on the block $\lambda = \left(n\right)$ carrying the symmetric subspace, so that $H\left(p\right) = 0$ and only the summand $\log m_{\left(n\right)}^2$ survives. Note that the full-rank assumption of \autoref{thm:upper_bound_entanglement_assisted_classical_capacity} refers to the invariant state of the twirl on the ambient space $\mathcal{H}_A^{\otimes n}$ --- the maximally mixed state, which is invariant although not Bose-symmetric --- and imposes no rank condition on the feasible states themselves. Hence \autoref{eq:bound_entanglement_assisted_capacity_bose_symmetric} improves to
        \begin{align}\label{eq:bose_de_Finetti_scaling}
             I\left(A_1^n:B_1^n\right)_{\rho_{A_1^nB_1^n}} \leq \log m_{(n)}^2.
        \end{align}
        Alternatively, since $\rho_{A_1^nB_1^n}$ is, up to a local isometry on the $A$-systems, a state on $\operatorname{Sym}^n\left(\mathcal{H}_A\right)\otimes \mathcal{H}_B^{\otimes n}$, the dimension bound underlying \autoref{eq:bound_mutual_information_no_further_information} directly gives $I\left(A_1^n:B_1^n\right) \leq 2 \log \dim \operatorname{Sym}^n\left(\mathcal{H}_A\right) = \log m_{\left(n\right)}^2$.
        By \autoref{lem:bounds_symmetric_group}~(3), we have $m_{(n)} \leq \left(n+d_A\right)^{d_A-1}$. The rest of the proof follows verbatim the proof of \autoref{thm:double_extended_de_Finetti} from \autoref{eq:proof_double_extended_deFinetti_bound_DPI} on. Given the improved scaling from \autoref{eq:bose_de_Finetti_scaling}, we end up as in \autoref{eq:final_result_double_extended_deFinetti} with
        \begin{align}
        \lVert\rho_{AB}-\sum_{t_{m_A},t_{m_B}}p\left(t_{m_A},t_{m_B}\right) \rho_{A\vert t_{m_A},t_{m_B}} \otimes \rho_{B\vert t_{m_A},t_{m_B}}\rVert_1 \leq 4 d_A d_B \frac{\sqrt{4 \ln 2 \left(d_A -1\right)\,\log\left(n+d_A \right)}}{n}.
    \end{align}
    Here $\tilde{\rho}_{AB}\coloneqq \sum_{t_{m_A},t_{m_B}}p\left(t_{m_A},t_{m_B}\right) \rho_{A\vert t_{m_A},t_{m_B}} \otimes \rho_{B\vert t_{m_A},t_{m_B}}$.
    \end{proof}
    The purification scheme of \cite{zeiss2025approximatingfixedsizequantum} extends to our double-sided variant. We elaborate the different applications of the de Finetti theorems and the fixed point constraints of this section in \autoref{sec:applications} in more detail.


\section{Computationally efficient inner sequence for SEP and cSEP}\label{sec:efficient_inner_sequence}

In \cite{kossmann2025_aqec, zeiss2025approximatingfixedsizequantum, zeiss2026finitefinetticonvexbodies}, the authors constructed a rounding scheme yielding certifiably good inner approximations to separability problems from SDP outer approximations formulated in terms of symmetric extensions, and their measurement-based methods extend to the constrained setting. The convergence guarantees result from an information-theoretic argument in relation to the self-decoupling lemma \cite{Brando2016, Brandao_2017} central to de Finetti representation theorems studied in, e.g., \cite{Berta2021, jee2020quasi, kossmann2025_aqec, zeiss2025approximatingfixedsizequantum, zeiss2026finitefinetticonvexbodies} and revisited in \autoref{subsec:deFinetti_revisited}. While \cite{kossmann2025_aqec, zeiss2025approximatingfixedsizequantum, Chee2025} provide a variety of symmetry-reduction techniques computing additive $\epsilon$-error \emph{outer} approximations to cSEP problems in $\operatorname{poly}\left(1/\epsilon\right)$-time whenever the local dimensions are fixed, the proposed convergent \emph{inner} sequence has so far lacked a comparably efficient representation: its implementation in an exponentially large computational basis comes at prohibitive computational cost. Building on the techniques developed in \autoref{sec:techniques}, we close this gap in this section. Exploiting the symmetries underlying the de Finetti-based rounding scheme, we show that certifiably good inner points can be obtained up to an additive $\epsilon$-error in $\operatorname{poly}\left(1/\epsilon\right)$-time whenever the local dimensions are fixed. The representation-theoretic background is collected in \autoref{sec:Young_diagrams_tableaux}, \autoref{sec:weyls_dimension_formula} and \autoref{sec:subgroup_adapted_basis}, whose notation and results we use freely.

Throughout this section, the local dimensions $d_A = \dim \mathcal{H}_A$ and $d_B = \dim \mathcal{H}_B$ are fixed, and all complexity statements are understood as functions of the extension parameter $n$ in the standard arithmetic model\footnote{That is, we count exact arithmetic operations $\left(+,\, -,\, \times,\, \div,\, \sqrt{\,\cdot\,}\,\right)$ at unit cost, cf.\ \cite{Burgisser1997}. Together with the exact, polynomially sized representations provided by \autoref{prop:cascade_data_poly}, all complexity statements of this section therefore also hold in the bit model up to polynomial overhead.}; the numbers involved are rational up to explicitly computable square roots (cf.\ \autoref{prop:cascade_data_poly}). We consider permutation-invariant states, i.e.\ one-sided $n$-extensions $\rho_{AB_1^n}\in \operatorname{End}_{\mathbb{C}\left[S_n\right]}\left(\mathcal{H}_A\otimes \mathcal{H}_B^{\otimes n}\right)$ as introduced in \autoref{sec:notation_preliminaries}, and freely identify the registers $B_2^{k+1}$, respectively $B_1^{k}$, with $\mathcal{H}_B^{\otimes k}$ when applying the permutation action $\psi_{k}^{d_B}$ from \autoref{eq:def_permutation_action_Sn_H_otimes_n}. For any $1\leq k \leq n$, Schur-Weyl duality (\autoref{thm:schur_weyl_duality}) provides the level-$k$ decomposition
\begin{align}\label{eqn:Schur_Weyl_k}
    \mathcal{H}_B^{\otimes k} \cong \bigoplus_{\lambda \vdash_{d_B} k} U_\lambda^{\left(k\right)}\otimes V_\lambda^{\left(k\right)},
\end{align}
where $V_\lambda$ is a Specht module and $U_\lambda$ is a Weyl module and the superscript records the level, and correspondingly, since $S_k$ acts trivially on $\mathcal{H}_A$,\footnote{Cf.\ \cite[Lem.\ 2.5]{Chee2025} and \cite{polak2020new} for an alternative argument in terms of the trivial group action.}
\begin{align}\label{eqn:commutant_with_auxiliary_system}
    \operatorname{End}_{\mathbb{C}\left[S_k\right]}\left(\mathcal{H}_A\otimes \mathcal{H}_B^{\otimes k}\right)\cong \bigoplus_{\lambda\vdash_{d_B} k} \operatorname{End}\left(\mathcal{H}_A\otimes U^{\left(k\right)}_\lambda\right)\otimes 1_{V^{\left(k\right)}_\lambda},
\end{align}
extending \autoref{eq:sw_commutant_blocks}. All Weyl-module blocks appearing below are expressed with respect to the Gelfand-Tsetlin bases of \autoref{subsubsec:rep_theory_Ud}, i.e.\ in the Schur basis \autoref{eq:def_schur_basis}; this convention fixes every matrix representation in this section.

\begin{definition}[Compressed Schur form]\label{def:compressed_schur_form}
    Let $1\leq k\leq n$. A permutation-invariant state $\rho_{AB_1^k}\in \operatorname{End}_{\mathbb{C}\left[S_k\right]}\left(\mathcal{H}_A\otimes \mathcal{H}_B^{\otimes k}\right)$ is said to be given in compressed Schur form if one is provided, for every $\lambda \vdash_{d_B} k$, with a weight $p^{\left(k\right)}_\lambda \geq 0$ and a state $\rho_{AU^{\left(k\right)}_\lambda}\in \mathcal{S}\left(\mathcal{H}_A\otimes U^{\left(k\right)}_\lambda\right)$, stored as a $\left(d_A m_\lambda\right)\times \left(d_A m_\lambda\right)$ matrix in the Schur basis, such that
    \begin{align}\label{eqn:schur_weyl_state_normalised}
        \rho_{AB_1^k}\cong \bigoplus_{\lambda \vdash_{d_B}k} p^{\left(k\right)}_\lambda\, \rho_{AU^{\left(k\right)}_\lambda}\otimes \tau_{V^{\left(k\right)}_\lambda},
    \end{align}
    where $m_{\lambda}=\dim_{\mathbb{C}}\lrbracket{U_{\lambda}^{(k)}}$ and $\tau_{V} = 1_{V}/\dim\left(V\right)$ denotes the maximally mixed state.
\end{definition}

\begin{remark}\label{rem:compressed_schur_form}
    Every permutation-invariant state admits a unique compressed Schur form. Indeed, \autoref{eqn:commutant_with_auxiliary_system} yields a decomposition $\bigoplus_\lambda X_{AU^{\left(k\right)}_\lambda}\otimes 1_{V_\lambda^{\left(k\right)}}$ with positive blocks $X_{AU^{\left(k\right)}_\lambda}$, and setting $p^{\left(k\right)}_\lambda \coloneqq \tr\left[X_{AU^{\left(k\right)}_\lambda}\right]\dim\left(V^{\left(k\right)}_\lambda\right)$ and $\rho_{AU^{\left(k\right)}_\lambda}\coloneqq X_{AU^{\left(k\right)}_\lambda}/\tr\left[X_{AU^{\left(k\right)}_\lambda}\right]$ (dropping blocks with vanishing trace) yields \autoref{eqn:schur_weyl_state_normalised} with $\sum_\lambda p^{\left(k\right)}_\lambda = 1$. By \autoref{lem:bounds_symmetric_group}, the number of blocks and each block dimension $d_A m_\lambda$ are polynomial in $k$ for fixed $d_A, d_B$, such that the description size is $\operatorname{poly}\left(n\right)$. The assumption of access in compressed Schur form is well motivated: the outer approximation schemes of \cite{Doherty_2004, kossmann2025_aqec, zeiss2025approximatingfixedsizequantum} output permutation-invariant optimizers in exactly this representation.
\end{remark}

\begin{remark}[The Specht factor never enters]\label{rem:specht_factor_inert}
    The compressed Schur form stores no data on the Specht modules $V^{\left(k\right)}_\lambda$, and this is the structural reason for its polynomial size. In the dimension count $d_B^{\,k} = \sum_{\lambda\vdash_{d_B}k} m_\lambda\, d_\lambda$ of the ambient tensor space, the polynomially bounded factors $m_\lambda$ are carried by the Weyl modules, whereas the in
    general exponential factors $d_\lambda = \dim V^{\left(k\right)}_\lambda$ are carried by the Specht modules (cf.\ the discussion following \autoref{thm:schur_weyl_duality}); permutation invariance, however, forces the state to be maximally mixed on precisely these registers, so they hold no information about $\rho_{AB_1^k}$. Accordingly, $d_\lambda$
    enters the algorithms below only as a number of polynomially many bits, supplied by the hook length formula \autoref{eq:hook_length_formula}: as the normalization in the weights $p^{\left(k\right)}_\lambda$ of \autoref{rem:compressed_schur_form} and as the branching
    weights $w_{\mu\vert\lambda}$ of \autoref{lem:reblocking}. No operator on a Specht module is ever represented: the measurements of \autoref{prop:coarse_grained_blocks} act on the Weyl
    side only, and the re-blocking of \autoref{lem:reblocking} merely regroups the invisible path registers, converting maximally mixed factors into branching weights.
\end{remark}

\begin{remark}[Quantum and classical Schur transforms]\label{rem:quantum_schur_context}
    The algorithms of this section are entirely classical; they require only the combinatorial and Clebsch-Gordan data of \autoref{lem:schur_transform_efficiency} and \autoref{prop:cascade_data_poly} (cf.\ \cite{Litjens2016, polak2020new} for alternative block-diagonalization algorithms based on non-orthogonal representative sets). For context we note that the Schur transform can also be implemented as a quantum circuit of size polynomial in $n$, $d_B$ and $\log\left(1/\epsilon\right)$ \cite{Bacon2006, bacon2005quantumschurtransformi, harrow2005applicationscoherentclassicalcommunication}, with subsequent improvements and extensions in \cite{Kirby_2018, Krovi2019, grinko2023gelfandtsetlinbasispartiallytransposed, Nguyen_2024, wills2024generalisedcouplingelementaryalgorithm, grinko2025mixed}; the restriction to finite gate sets only realizes a unitary $U_\epsilon$ with $\lVert U_\epsilon - U_{\operatorname{Sch}}\rVert\leq \epsilon$, which is immaterial at the inverse-polynomial accuracies relevant here.\footnote{Weak Schur sampling, which only measures the label $\lambda$, should not be confused with this gate-set approximation, cf.\ \cite{cervero2023weakschursamplinglogarithmic, cerveromartin2024memorygateefficientalgorithm}.} For the complementary regime $n < d_B$ we refer to \cite{Krovi2019, burchardt2025highdimensionalquantumschurtransforms}; since our hierarchies converge in the limit $n\to \infty$ at fixed $d_B$, the regime $n \geq d_B$ is the relevant one, and the truncated Young lattice handles both regimes uniformly (cf.\ \autoref{rem:branching_truncation}).
\end{remark}


The rounding procedure requires two efficient subroutines: the Schur-Weyl block form of the type-coarse-grained measurements of \autoref{cor:one_sided_types} (cf.\ \autoref{observation} (2)), and the block-wise computation of partial traces. Both rest on two elementary facts about the Schur basis \autoref{eq:def_schur_basis}. First, the blocks of \autoref{eq:sw_commutant_blocks} are compressions along path isometries: for $\lambda \vdash_{d_B} j$ and $T \in \operatorname{Paths}_{j}\left(\lambda\right)$ let $W_{\lambda, T}: U^{\left(j\right)}_\lambda \to \mathcal{H}_B^{\otimes j}$, $\ket{N}\mapsto \ket{\lambda, N, T}$, denote the isometric embedding of the Weyl module along the path $T$; for any $X \in \operatorname{End}_{\mathbb{C}\left[S_j\right]}\left(\mathcal{H}_B^{\otimes j}\right)$, invariance forces the Schur-basis matrix elements to be diagonal in the path label and independent of the path, i.e.
\begin{align}\label{eqn:blocks_as_compressions}
    X^{\left(\lambda\right)} = W_{\lambda,T}^{\dagger}\, X\, W_{\lambda,T} \qquad \text{for every} \ T \in \operatorname{Paths}_{j}\left(\lambda\right),
\end{align}
and conversely $X \cong \bigoplus_{\lambda} X^{\left(\lambda\right)}\otimes 1_{V^{\left(j\right)}_\lambda}$ is recovered as $X = \sum_{\lambda\vdash_{d_B}j}\sum_{T \in \operatorname{Paths}_j\left(\lambda\right)} W_{\lambda,T}\, X^{\left(\lambda\right)} W_{\lambda,T}^{\dagger}$; in particular, $X \mapsto X^{\left(\lambda\right)}$ is a unital $*$-homomorphism on the commutant, preserving adjoints, positivity and the identity. Second, since the Schur basis is constructed recursively --- the cascade \autoref{eq:schur_transform_cascade} embeds every Gelfand-Tsetlin vector into $\mathcal{H}_B^{\otimes j}$ by iterating the single-box Clebsch-Gordan step --- the path isometries factor through the Pieri isometries of \autoref{eq:pieri_single_box}: for every $\mu \in \lambda^-$ and every $S \in \operatorname{Paths}_{j-1}\left(\mu\right)$,
\begin{align}\label{eqn:path_isometry_factorization}
    W_{\lambda,\, S\rightarrow\lambda} = \left(W_{\mu, S}\otimes 1_{\mathcal{H}_B}\right) C_{\mu\leftarrow\lambda},
\end{align}
since $U^{\left(j\right)\dagger}_{\operatorname{Sch}}\left(\ket{N}\otimes\ket{S\rightarrow\lambda}\right) = \left(U^{\left(j-1\right)\dagger}_{\operatorname{Sch}}\otimes 1_{\mathcal{H}_B}\right)\left(\left(C_{\mu\leftarrow\lambda}\ket{N}\right)\otimes \ket{S}\right)$ after regrouping the tensor factors, by \autoref{eq:CG_transform_action} read backwards as in \autoref{eq:def_C_isometry_matrix_elements}. Both facts hold verbatim in the presence of an auxiliary register on which the symmetric group acts trivially, with $W_{\lambda,T}$ replaced by $1\otimes W_{\lambda,T}$ (cf.\ \autoref{eqn:commutant_with_auxiliary_system}). Measurement operators are assembled upward along the truncated Young lattice by \autoref{eqn:path_isometry_factorization}, one tensor factor at a time, whereas states are re-blocked downward, one split-off factor at a time. We treat them in turn.

\begin{proposition}[Coarse-grained measurements in block form]\label{prop:coarse_grained_blocks}
    Let $\mathcal{M}_B: \mathcal{S}\left(\mathcal{H}_B\right)\to \mathcal{M}_1\left(\ZZ\right)$ be a MIC measurement with POVM $\lbrace M_{B\vert z}\rbrace_{z \in \ZZ}$, $\lvert \ZZ\rvert = d_B^2$, and let $2\leq m_B \leq n-1$. For every type $t \in \mathcal{T}_{m_B, \lvert \ZZ\rvert}$ consider the type-coarse-grained POVM element of \autoref{cor:one_sided_types},
    \begin{align}\label{eqn:def_coarse_grained_measurement}
        G_{B_2^{m_B+1}\vert\, t}\coloneqq \sum_{z_2^{m_B+1}:\ \tau_B\left(z_2^{m_B+1}\right)=t}\ \bigotimes_{i=2}^{m_B+1} M_{B\vert z_i} \ \in \ \operatorname{End}_{\mathbb{C}\left[S_{m_B}\right]}\left(\mathcal{H}_B^{\otimes m_B}\right),
    \end{align}
    with Schur-Weyl block decomposition $G_{B_2^{m_B+1}\vert\, t}\cong \bigoplus_{\lambda \vdash_{d_B}m_B} G^{\left(\lambda\right)}_t \otimes 1_{V^{\left(m_B\right)}_\lambda}$ according to \autoref{eq:sw_commutant_blocks}. Then all blocks $G^{\left(\lambda\right)}_t \in \operatorname{End}\left(U^{\left(m_B\right)}_\lambda\right)$, expressed in the Schur basis, can be computed in time polynomial in $n$. Moreover, for every $\lambda \vdash_{d_B} m_B$ the family $\lbrace G^{\left(\lambda\right)}_t\rbrace_{t \in \mathcal{T}_{m_B,\lvert\ZZ\rvert}}$ is a POVM on $U^{\left(m_B\right)}_\lambda$.
\end{proposition}

\begin{proof}
    Since type classes are invariant under permutations of the positions, conjugation by $\psi_{m_B}^{d_B}\left(\sigma\right)$ merely permutes the summands in \autoref{eqn:def_coarse_grained_measurement}, such that $G_{B_2^{m_B+1}\vert t}$ indeed lies in the commutant and admits the asserted block decomposition; the number of types is $\lvert \mathcal{T}_{m_B,\lvert\ZZ\rvert}\rvert = \binom{m_B + d_B^2 -1}{d_B^2-1} = \operatorname{poly}\left(m_B\right)$ by \autoref{prop:classification_of_types}.

    To compute the blocks we assemble the coarse-grained operators one tensor factor at a time along the Clebsch-Gordan cascade of \autoref{subsubsec:CG_series_transform}; this guarantees that all blocks are expressed in the same Gelfand-Tsetlin bases, including phase conventions, in which \autoref{prop:cascade_data_poly} supplies the isometries $C_{\mu\leftarrow\lambda}$ (for the orbit-basis alternative in the spirit of the algebra-level block-diagonalizations of \cite{Litjens2016, polak2020new}, cf.\ \autoref{rem:orbit_closed_form} below). For every $1\leq j\leq m_B$ and every partial type $s \in \mathcal{T}_{j,\lvert\ZZ\rvert}$, the defining formula \autoref{eqn:def_coarse_grained_measurement} with $m_B$ replaced by $j$ yields the level-$j$ coarse-grained operators
    \begin{align}\label{eqn:def_partial_coarse_grained}
        G_{B_2^{j+1}\vert\, s} = \sum_{z_2^{j+1}:\ \tau_B\left(z_2^{j+1}\right)=s}\ \bigotimes_{i=2}^{j+1} M_{B\vert z_i} \ \in \ \operatorname{End}_{\mathbb{C}\left[S_{j}\right]}\left(\mathcal{H}_B^{\otimes j}\right),
    \end{align}
    which lie in the respective commutants by the argument above, with blocks $G^{\left(\lambda\right)}_s \in \operatorname{End}\left(U^{\left(j\right)}_\lambda\right)$, $\lambda \vdash_{d_B} j$, expressed in the Schur basis; the case $j = m_B$ recovers the objects of the statement. Splitting off the last outcome decomposes each type class according to the outcome occupying the final position: for $2\leq j\leq m_B$,
    \begin{align}\label{eqn:partial_type_recursion}
        G_{B_2^{j+1}\vert\, s} = \sum_{\substack{z\in\ZZ\\ s_z\geq 1}} G_{B_2^{j}\vert\, s-e_z}\otimes M_{B\vert z},
    \end{align}
    where $e_z \in \mathbb{N}_0^{\ZZ}$ denotes the standard basis vector; the base case is $G_{B_2\vert\, e_z} = M_{B\vert z}$, whose single block at $\lambda = \left(1\right)$ is $M_{B\vert z}$ itself, the Gelfand-Tsetlin basis of $U^{\left(1\right)}_{\left(1\right)}\cong \mathcal{H}_B$ being the computational basis (cf.\ \autoref{subsubsec:rep_theory_Ud}). For the blocks, fix $\lambda\vdash_{d_B} j$, an arbitrary $\mu \in \lambda^-$ and a path $S \in \operatorname{Paths}_{j-1}\left(\mu\right)$. Compressing \autoref{eqn:partial_type_recursion} along the extended path $S\rightarrow\lambda$ via \autoref{eqn:blocks_as_compressions} and inserting the factorization \autoref{eqn:path_isometry_factorization} yields
    \begin{align}\label{eqn:blocks_coarse_grained}
    \begin{split}
        G^{\left(\lambda\right)}_s = W_{\lambda,\, S\rightarrow\lambda}^{\dagger}\, G_{B_2^{j+1}\vert\, s}\, W_{\lambda,\, S\rightarrow\lambda} &= \sum_{\substack{z\in\ZZ\\ s_z\geq 1}} C_{\mu\leftarrow\lambda}^{\dagger}\left(\left(W_{\mu,S}^{\dagger}\, G_{B_2^{j}\vert\, s-e_z}\, W_{\mu,S}\right)\otimes M_{B\vert z}\right)C_{\mu\leftarrow\lambda}\\
        &= \sum_{\substack{z\in\ZZ\\ s_z\geq 1}} C_{\mu\leftarrow\lambda}^{\dagger}\left(G^{\left(\mu\right)}_{s-e_z}\otimes M_{B\vert z}\right)C_{\mu\leftarrow\lambda},
    \end{split}
    \end{align}
    using \autoref{eqn:blocks_as_compressions} again at level $j-1$ in the second step; in particular, the result is independent of the choices of $\mu$ and $S$. At level $j$ there are at most $\binom{j+d_B^2-1}{d_B^2-1}$ partial types (\autoref{prop:classification_of_types}) and at most $\left(j+d_B\right)^{d_B-1}$ shapes with at most $d_B$ removable cells each (\autoref{lem:bounds_symmetric_group} and \autoref{prop:young_branching_rule}), and each application of \autoref{eqn:blocks_coarse_grained} multiplies matrices of dimensions at most $m_\mu d_B \times m_\lambda$ supplied by \autoref{prop:cascade_data_poly}, such that a single upward sweep computes all blocks $G^{\left(\lambda\right)}_t$, $\lambda \vdash_{d_B} m_B$, simultaneously for all $t \in \mathcal{T}_{m_B,\lvert\ZZ\rvert}$, in time polynomial in $n$.

    Finally, since the types partition $\ZZ^{m_B}$, we have $\sum_{t}G_{B_2^{m_B+1}\vert t} = 1_{\mathcal{H}_B^{\otimes m_B}}$ and hence $\sum_t G^{\left(\lambda\right)}_t = 1_{U^{\left(m_B\right)}_\lambda}$ blockwise; positivity of each $G^{\left(\lambda\right)}_t$ follows from \autoref{eqn:blocks_as_compressions}, the block map being a unital $*$-homomorphism on the commutant. This establishes the POVM property and completes the proof.
\end{proof}

\begin{remark}[Orbit-basis closed form]\label{rem:orbit_closed_form}
    Alternatively, the blocks admit an explicit closed form. Identify the computational basis of $\mathcal{H}_B^{\otimes m_B}$ with strings $x \in \left[d_B\right]^{m_B}$, assign to any pair of strings the joint type $\kappa\left(x,y\right)\in\mathbb{N}_0^{d_B\times d_B}$ with $\kappa_{p,q}\left(x,y\right)\coloneqq \big\lvert\lbrace i \in \left[m_B\right] \ \vert \ \left(x_i,y_i\right) = \left(p,q\right)\rbrace\big\rvert$, and let $E_\kappa \coloneqq \sum_{\kappa\left(x,y\right)=\kappa}\ketbra{x}{y}$ denote the associated orbit matrices, which form a basis of $\operatorname{End}_{\mathbb{C}\left[S_{m_B}\right]}\left(\mathcal{H}_B^{\otimes m_B}\right)$ indexed by the $\binom{m_B+d_B^2-1}{d_B^2-1}$ weak compositions of $m_B$ into $d_B^2$ parts. Since the matrix entries of $G_{B_2^{m_B+1}\vert t}$ depend only on the joint type, one has $G_{B_2^{m_B+1}\vert\, t} = \sum_\kappa g_{t,\kappa}\, E_\kappa$ with $g_{t,\kappa} = \bra{x_\kappa}G_{B_2^{m_B+1}\vert t}\ket{y_\kappa}$ for an arbitrary representative pair with $\kappa\left(x_\kappa,y_\kappa\right)=\kappa$, and grouping the positions by basis pairs and the outcome strings by their induced contingency tables yields
    \begin{align}\label{eqn:coefficients_coarse_grained}
        g_{t,\kappa} = \sum_{C \ \text{feasible}} \left(\frac{\prod_{p,q\in\left[d_B\right]}\kappa_{p,q}!}{\prod_{z,p,q}c_{z,p,q}!}\right) \prod_{z,p,q}\left(\left(M_{B\vert z}\right)_{p,q}\right)^{c_{z,p,q}},
    \end{align}
    where a table $C = \left(c_{z,p,q}\right)\in \mathbb{N}_0^{\ZZ\times d_B\times d_B}$ is feasible if $\sum_{p,q}c_{z,p,q} = t_z$ for all $z \in \ZZ$ and $\sum_{z}c_{z,p,q} = \kappa_{p,q}$ for all $p,q\in\left[d_B\right]$. The feasible tables are the lattice points of a transportation polytope of fixed dimension $\lvert\ZZ\rvert d_B^2 = d_B^4$ with marginals bounded by $m_B$; they can be enumerated in time $\operatorname{poly}\left(m_B\right)$ by fixed-dimension lattice-point methods (equivalently, via Ehrhart theory) \cite{ehrhart1962polyedres, beck2007computing, deloera2013combinatoricsgeometrytransportationpolytopes}, cf.\ also \cite{zeiss2025approximatingfixedsizequantum}. The blocks $E^{\left(\lambda\right)}_\kappa$ are independent of the measurement: since $E_\kappa = \sum_{p,q:\ \kappa_{p,q}\geq 1} E_{\kappa - e_{p,q}}\otimes \ketbra{p}{q}$ with the elementary matrices $e_{p,q}\in\mathbb{N}_0^{d_B\times d_B}$ and base case $E_{e_{p,q}} = \ketbra{p}{q}$, they are computed by the recursion \autoref{eqn:blocks_coarse_grained} with the matrix units $\lbrace \ketbra{p}{q}\rbrace_{p,q}$ in place of the POVM elements and joint types in place of partial types --- neither self-adjointness nor positivity enters the derivation. They may therefore be precomputed once and reused across different measurements, each new measurement costing only the evaluation of the coefficients $g_{t,\kappa}$, such that $G^{\left(\lambda\right)}_t = \sum_\kappa g_{t,\kappa}\, E^{\left(\lambda\right)}_\kappa$ by linearity. This orbit-basis route is the algebra-level block-diagonalization in the spirit of \cite{Litjens2016, polak2020new} (cf.\ also \cite{Chee2025, zeiss2025approximatingfixedsizequantum, kossmann2025_aqec}); carried out in the semistandard bases it determines the blocks only up to the Gramian corrections of \autoref{rem:GT_vs_semistandard}, which the Gelfand-Tsetlin recursion of \autoref{prop:coarse_grained_blocks} avoids.
\end{remark}

The second subroutine reduces the extension level. Its core is the following identity, which expresses a state in level-$k$ block form at level $k-1$ with one $B$-system split off explicitly; it combines the Pieri isometries $C_{\mu\leftarrow\lambda}: U^{\left(k\right)}_\lambda \hookrightarrow U^{\left(k-1\right)}_\mu \otimes \mathcal{H}_B$ of \autoref{eq:pieri_single_box} with the Young branching weights
\begin{align}\label{eqn:def_branching_weights}
    w_{\mu\vert\lambda}\coloneqq \frac{\dim V^{\left(k-1\right)}_\mu}{\dim V^{\left(k\right)}_\lambda}, \qquad \mu \in \lambda^-, \qquad \text{with} \quad \sum_{\mu\in\lambda^-}w_{\mu\vert\lambda} = 1
\end{align}
by the dimension recursion in \autoref{prop:young_branching_rule}.

\begin{lemma}[Re-blocking along Young branching]\label{lem:reblocking}
    Let $2\leq k \leq n$ and let $\rho_{AB_1^k}$ be given in compressed Schur form \autoref{eqn:schur_weyl_state_normalised}. Then, splitting off one $B$-factor, which by permutation invariance may be taken to be $B_1$,
    \begin{align}\label{eqn:schur_weyl_block_separating_B}
        \rho_{AB_1^{k}} \cong \bigoplus_{\mu \vdash_{d_B}\left(k-1\right)} Y_{AB_1U^{\left(k-1\right)}_\mu}\otimes \tau_{V^{\left(k-1\right)}_\mu},
    \end{align}
    where, up to a reordering of the tensor factors $U^{\left(k-1\right)}_\mu\otimes \mathcal{H}_B \cong \mathcal{H}_B \otimes U^{\left(k-1\right)}_\mu$,
    \begin{align}\label{eqn:reblocking_identity}
        Y_{AB_1U^{\left(k-1\right)}_\mu} = \sum_{\substack{\lambda \in \mu^+\\ \ell\left(\lambda\right)\leq d_B}} p^{\left(k\right)}_\lambda\, w_{\mu\vert\lambda}\, \left(1_{\mathcal{H}_A}\otimes C_{\mu\leftarrow\lambda}\right) \rho_{AU^{\left(k\right)}_\lambda} \left(1_{\mathcal{H}_A}\otimes C_{\mu\leftarrow\lambda}\right)^{\dagger}.
    \end{align}
\end{lemma}

\begin{proof}
    By \autoref{eqn:blocks_as_compressions}, applied with the auxiliary register $\mathcal{H}_A$, the compressed Schur form \autoref{eqn:schur_weyl_state_normalised} is the resolution
    \begin{align}\label{eqn:compressed_schur_resolution}
        \rho_{AB_1^{k}} = \sum_{\lambda\vdash_{d_B}k} \frac{p^{\left(k\right)}_\lambda}{\dim V^{\left(k\right)}_\lambda}\ \sum_{T \in \operatorname{Paths}_{k}\left(\lambda\right)} \left(1_{\mathcal{H}_A}\otimes W_{\lambda, T}\right) \rho_{AU^{\left(k\right)}_\lambda} \left(1_{\mathcal{H}_A}\otimes W_{\lambda, T}\right)^{\dagger}.
    \end{align}
    Every path decomposes uniquely as $T = S\rightarrow\lambda$ with $\mu \in \lambda^-$ its level-$\left(k-1\right)$ vertex and $S \in \operatorname{Paths}_{k-1}\left(\mu\right)$, where every $\mu \in \lambda^-$ automatically satisfies $\ell\left(\mu\right)\leq d_B$ (cf.\ \autoref{rem:branching_truncation}). Inserting the factorization \autoref{eqn:path_isometry_factorization}, exchanging the order of summation via $\sum_{\lambda}\sum_{\mu\in\lambda^-} = \sum_{\mu}\sum_{\lambda\in\mu^+,\,\ell\left(\lambda\right)\leq d_B}$, and using $w_{\mu\vert\lambda} = \dim V^{\left(k-1\right)}_\mu/\dim V^{\left(k\right)}_\lambda$ from \autoref{eqn:def_branching_weights} turns \autoref{eqn:compressed_schur_resolution} into
    \begin{align}
        \rho_{AB_1^{k}} = \sum_{\mu\vdash_{d_B}\left(k-1\right)} \frac{1}{\dim V^{\left(k-1\right)}_\mu}\ \sum_{S\in\operatorname{Paths}_{k-1}\left(\mu\right)} \left(1_{\mathcal{H}_A}\otimes W_{\mu,S}\otimes 1_{\mathcal{H}_B}\right) Y_{AB_1U^{\left(k-1\right)}_\mu} \left(1_{\mathcal{H}_A}\otimes W_{\mu,S}\otimes 1_{\mathcal{H}_B}\right)^{\dagger}
    \end{align}
    with $Y_{AB_1U^{\left(k-1\right)}_\mu}$ as in \autoref{eqn:reblocking_identity}, up to the reordering in the statement. By \autoref{eqn:blocks_as_compressions}, read in reverse at level $k-1$ with the auxiliary register $\mathcal{H}_A\otimes\mathcal{H}_B$, this is precisely the block decomposition \autoref{eqn:schur_weyl_block_separating_B}. Finally, the split-off factor is the last register $B_{k+1}$ in the labeling $B_2^{k+1}\cong \mathcal{H}_B^{\otimes k}$; conjugating with the cyclic permutation exchanging it with $B_1$ leaves $\rho_{AB_1^k}$ invariant, which justifies the relabeling in the statement.
\end{proof}

\begin{proposition}[Efficient blockwise partial trace]\label{prop:efficient_partial_trace}
    Let $\rho_{AB_1^n}$ be given in compressed Schur form. Then, for every $1\leq m_B\leq n-1$, the level-$m_B$ data
    \begin{align}\label{eqn:level_mB_block_form}
        \rho_{AB_1^{m_B+1}} \cong \bigoplus_{\lambda\vdash_{d_B} m_B} p^{\left(m_B\right)}_\lambda\, \rho_{AB_1U^{\left(m_B\right)}_\lambda}\otimes \tau_{V^{\left(m_B\right)}_\lambda}, \qquad \rho_{AB_1U^{\left(m_B\right)}_\lambda}\in \mathcal{S}\left(\mathcal{H}_A\otimes\mathcal{H}_B\otimes U^{\left(m_B\right)}_\lambda\right),
    \end{align}
    carrying the retained system $B_1$ explicitly, can be computed in time polynomial in $n$. In particular, all reduced states $\rho_{AB_1^k}$, $1\leq k \leq n$, are accessible in compressed Schur form in time polynomial in $n$.
\end{proposition}

\begin{proof}
    We iterate two steps, starting from the level-$n$ data. \emph{Separation:} given the level-$k$ data $\lbrace p^{\left(k\right)}_\lambda, \rho_{AU^{\left(k\right)}_\lambda}\rbrace$, \autoref{lem:reblocking} produces the level-$\left(k-1\right)$ blocks $Y_{AB_1U^{\left(k-1\right)}_\mu}$ of \autoref{eqn:reblocking_identity}. Since the $C_{\mu\leftarrow\lambda}$ are isometries and the $\rho_{AU^{\left(k\right)}_\lambda}$ are normalized,
    \begin{align}\label{eqn:branching_pushforward}
        p^{\left(k-1\right)}_\mu \coloneqq \tr\left[Y_{AB_1U^{\left(k-1\right)}_\mu}\right] = \sum_{\substack{\lambda\in\mu^+\\ \ell\left(\lambda\right)\leq d_B}} p^{\left(k\right)}_\lambda\, w_{\mu\vert\lambda},
    \end{align}
    i.e.\ the level-$\left(k-1\right)$ weights are the pushforward of the level-$k$ weights along Young branching; in particular $p^{\left(k-1\right)}_\mu\geq 0$ and, by \autoref{eqn:def_branching_weights},
    \begin{align}
        \sum_{\mu\vdash_{d_B}\left(k-1\right)} p^{\left(k-1\right)}_\mu = \sum_{\lambda\vdash_{d_B}k} p^{\left(k\right)}_\lambda \sum_{\mu\in\lambda^-}w_{\mu\vert\lambda} = \sum_{\lambda\vdash_{d_B}k} p^{\left(k\right)}_\lambda = 1.
    \end{align}
    Setting $\rho_{AB_1U^{\left(k-1\right)}_\mu}\coloneqq Y_{AB_1U^{\left(k-1\right)}_\mu}/p^{\left(k-1\right)}_\mu$ for $p^{\left(k-1\right)}_\mu>0$ (and dropping vanishing blocks) yields the state $\rho_{AB_1^{k}}$ in the form \autoref{eqn:level_mB_block_form} at level $k-1$. \emph{Trace:} discarding the explicit factor gives $\rho_{AB_1^{k-1}}$ in compressed Schur form with blocks
    \begin{align}\label{eqn:partial_trace_recursion}
        \rho_{AU^{\left(k-1\right)}_\mu} = \frac{1}{p^{\left(k-1\right)}_\mu}\sum_{\substack{\lambda\in\mu^+\\ \ell\left(\lambda\right)\leq d_B}} p^{\left(k\right)}_\lambda\, w_{\mu\vert\lambda}\, \Gamma_{\mu\leftarrow\lambda}\left(\rho_{AU^{\left(k\right)}_\lambda}\right),
    \end{align}
    where the completely positive and trace-preserving maps
    \begin{align}\label{eqn:def_Gamma_maps}
    \begin{split}
        \Gamma_{\mu\leftarrow\lambda}&: \operatorname{End}\left(\mathcal{H}_A\otimes U^{\left(k\right)}_\lambda\right)\to \operatorname{End}\left(\mathcal{H}_A\otimes U^{\left(k-1\right)}_\mu\right),\\
        \Gamma_{\mu\leftarrow\lambda}\left(X\right)&\coloneqq \tr_{\mathcal{H}_B}\left[\left(1_{\mathcal{H}_A}\otimes C_{\mu\leftarrow\lambda}\right)X\left(1_{\mathcal{H}_A}\otimes C_{\mu\leftarrow\lambda}\right)^\dagger\right]
    \end{split}
    \end{align}
    are trace preserving because the $C_{\mu\leftarrow\lambda}$ are isometries. To obtain the level-$m_B$ data \autoref{eqn:level_mB_block_form}, perform separation and trace for $k = n, n-1,\ldots, m_B+2$ and a final separation at $k = m_B+1$, retaining the explicit factor.

    For the complexity, note that at level $k$ there are at most $\left(k+d_B\right)^{d_B-1}$ partitions with at most $d_B$ incoming edges each (\autoref{lem:bounds_symmetric_group} and \autoref{prop:young_branching_rule}), the matrices involved have dimensions bounded by $d_A d_B m_\lambda$ with $m_\lambda \leq \left(n+d_B\right)^{d_B\left(d_B-1\right)/2}$, and the isometries $C_{\mu\leftarrow\lambda}$ in the Schur basis are supplied in time polynomial in $n$ by \autoref{prop:cascade_data_poly}. Each of the at most $n$ iterations therefore costs $\operatorname{poly}\left(n\right)$ arithmetic operations, which proves the claim.
\end{proof}

We are now in the position to state and prove the main theorem of this section. Recall from \autoref{sec:notation_preliminaries} the notion of a distortion bound $c\left(\mathcal{M}_B\right)$ for an informationally complete measurement.

\begin{theorem}[Efficient rounding]\label{thm:efficient_rounding}
    Let $d_A, d_B$ be fixed, let $\mathcal{M}_B$ be a MIC measurement on $\mathcal{H}_B$ with outcome set $\ZZ$ and distortion bound $c\left(\mathcal{M}_B\right)$, and let $\rho_{AB_1^n}$ be a permutation-invariant state given in compressed Schur form. Then one can compute, in time polynomial in $n$, separable states
    \begin{align}\label{eqn:def_rounding_candidates}
        \sigma^{\left(m\right)}_{AB}\coloneqq \sum_{t\in\mathcal{T}_{m,\lvert\ZZ\rvert}} p\left(t\right)\, \rho_{A\vert t}\otimes \rho_{B_1\vert t}\ \in\ \operatorname{SEP}\left(A:B\right), \qquad m = 0,1,\ldots,n-1,
    \end{align}
    each given explicitly by its polynomially many factors, such that
    \begin{align}\label{eqn:rounding_guarantee}
        \min_{0\leq m\leq n-1}\ \big\lVert \rho_{AB_1}-\sigma^{\left(m\right)}_{AB}\big\rVert_1 \ \leq \ c\left(\mathcal{M}_B\right)\sqrt{\frac{2\ln 2 \cdot \log d_A}{n}}\ \eqqcolon\ \epsilon\left(d_A, n\right).
    \end{align}
    Consequently, for any observable $H$ on $\mathcal{H}_A\otimes\mathcal{H}_B$,
    \begin{align}\label{eqn:rounding_value_guarantee}
        \max_{0\leq m \leq n-1} \tr\left[H\sigma^{\left(m\right)}_{AB}\right] \ \geq \ \tr\left[H \rho_{AB_1}\right] - \lVert H\rVert_\infty\, \epsilon\left(d_A,n\right).
    \end{align}
    In particular, if $\rho_{AB_1^n}$ is an optimizer of a level-$n$ outer relaxation of a $\operatorname{SEP}$ problem instance with value $p^\star$, such that $\tr\left[H\rho_{AB_1}\right]\geq p^\star$, then $p^\star - \lVert H\rVert_\infty \epsilon\left(d_A,n\right) \leq \max_m \tr\left[H\sigma^{\left(m\right)}_{AB}\right]\leq p^\star$, and choosing $n = \big\lceil 2\ln 2\, \log\left(d_A\right) c\left(\mathcal{M}_B\right)^2 \lVert H\rVert_\infty^2/\epsilon^2\big\rceil$ yields an additive $\epsilon$-error lower bound on $p^\star$ in time $\operatorname{poly}\left(1/\epsilon\right)$.
\end{theorem}

\begin{proof}
    Fix $m \in \lbrace 0,\ldots,n-1\rbrace$. For $m = 0$ set $\sigma^{\left(0\right)}_{AB}\coloneqq \rho_A\otimes \rho_{B_1}$, computable from \autoref{prop:efficient_partial_trace}; for $m \geq 1$ compute the level-$m$ data \autoref{eqn:level_mB_block_form} via \autoref{prop:efficient_partial_trace} and the measurement blocks $G^{\left(\lambda\right)}_t$ via \autoref{prop:coarse_grained_blocks} (for $m=1$ the types are the single outcomes, $\mathcal{T}_{1,\lvert\ZZ\rvert}\cong\ZZ$, and $G_t = M_{B\vert z}$). Measuring the systems $B_2^{m+1}$ with the POVM $\lbrace G_{B_2^{m+1}\vert t}\rbrace_t$ and discarding them produces, using the block forms and $\tr_{V_\lambda}\left[\tau_{V_\lambda}\right]=1$,
    \begin{align}\label{eqn:def_pt}
        p\left(t\right) &= \tr\left[\left(1_{AB_1}\otimes G_{B_2^{m+1}\vert t}\right)\rho_{AB_1^{m+1}}\right] = \sum_{\lambda\vdash_{d_B}m} p^{\left(m\right)}_\lambda \tr\left[\left(1_{AB_1}\otimes G^{\left(\lambda\right)}_t\right)\rho_{AB_1U^{\left(m\right)}_\lambda}\right],\\
        \label{eqn:def_conditional_states}
        \rho_{AB_1\vert t} &= \frac{1}{p\left(t\right)}\sum_{\lambda\vdash_{d_B}m} p^{\left(m\right)}_\lambda\, \tr_{U^{\left(m\right)}_\lambda}\left[\left(1_{AB_1}\otimes G^{\left(\lambda\right)}_t\right)\rho_{AB_1U^{\left(m\right)}_\lambda}\right]
    \end{align}
    for $p\left(t\right)>0$, with marginals $\rho_{A\vert t}$ and $\rho_{B_1\vert t}$ (outcomes with $p\left(t\right)=0$ are dropped). All quantities are positive semidefinite by the POVM property in \autoref{prop:coarse_grained_blocks}, we have $\sum_t p\left(t\right) = 1$, and they coincide with the type-conditional states of \autoref{cor:one_sided_types}. Since $\lvert \mathcal{T}_{m,\lvert\ZZ\rvert}\rvert = \operatorname{poly}\left(n\right)$ and each factor is a $d_A\times d_A$, respectively $d_B\times d_B$, matrix, every $\sigma^{\left(m\right)}_{AB}$ in \autoref{eqn:def_rounding_candidates} is an explicitly represented separable state, computable in time $\operatorname{poly}\left(n\right)$, and $\sum_t p\left(t\right)\rho_{AB_1\vert t} = \rho_{AB_1}$ by the resolution $\sum_t G_{B_2^{m+1}\vert t} = 1$.

    For the approximation guarantee, let $\omega_{AZ_1^n}\coloneqq \left(\operatorname{id}_A\otimes \mathcal{M}_B^{\otimes n}\right)\left(\rho_{AB_1^n}\right)$ and let $T_m$ denote the type registers of \autoref{cor:one_sided_types}. By \autoref{prop:chain_rule}~(2) and the classicality of the $Z_1^n$-systems as in \autoref{eq:cq_mutual_information_bound},
    \begin{align}\label{eqn:pigeonhole_CMI}
        \sum_{m=0}^{n-1} I\left(A : Z_1 \vert T_m\right)_\omega = I\left(A:Z_1^n\right)_\omega \leq\log d_A,
    \end{align}
    such that there exists an $m^\star \in \lbrace 0,\ldots,n-1\rbrace$ with $I\left(A:Z_1\vert T_{m^\star}\right)_\omega \leq \log\left(d_A\right)/n$. Repeating the estimate \autoref{eq:calculation_2}--\autoref{eq:calculation_7} verbatim for this $m^\star$ --- i.e.\ Pinsker's inequality applied to the type-conditional states, the distortion bound $c\left(\mathcal{M}_B\right)$ lifting the classical closeness on $Z_1$ back to $B_1$, and convexity of the trace norm together with $\sum_t p\left(t\right)\rho_{AB_1\vert t} = \rho_{AB_1}$ --- yields
    \begin{align}
        \big\lVert \rho_{AB_1} - \sigma^{\left(m^\star\right)}_{AB}\big\rVert_1 \leq c\left(\mathcal{M}_B\right)\sqrt{2\ln 2\cdot I\left(A:Z_1\vert T_{m^\star}\right)_\omega} \leq c\left(\mathcal{M}_B\right)\sqrt{\frac{2\ln2 \cdot \log d_A}{n}}.
    \end{align}
    Since $m^\star$ is not identified by the data, the algorithm computes all $n$ candidates and the minimum in \autoref{eqn:rounding_guarantee} is achieved by comparison against $\rho_{AB_1}$, which is available in explicit form by \autoref{prop:efficient_partial_trace}; alternatively, for the optimization statement no identification is needed. Indeed, \autoref{eqn:rounding_value_guarantee} follows from \autoref{eqn:rounding_guarantee} by H\"older's inequality, $\lvert \tr\left[H\left(\rho_{AB_1}-\sigma^{\left(m\right)}_{AB}\right)\right]\rvert \leq \lVert H \rVert_\infty \lVert \rho_{AB_1}-\sigma^{\left(m\right)}_{AB}\rVert_1$, evaluated at any $m$ attaining the minimum. The final claim follows since each $\sigma^{\left(m\right)}_{AB}$ is separable and hence feasible, so that $\max_m \tr\left[H\sigma^{\left(m\right)}_{AB}\right]\leq p^\star$, while the stated choice of $n$ makes $\lVert H\rVert_\infty\epsilon\left(d_A,n\right)\leq \epsilon$; the overall runtime is polynomial in $n = \operatorname{poly}\left(1/\epsilon\right)$.
\end{proof}

As a direct consequence we obtain the following corollaries. The first concerns the Bose-symmetric setting, which arises naturally in the outer hierarchies: by \cite[Lem.~5.7]{zeiss2025approximatingfixedsizequantum}, Bose symmetry can always be arranged by purifying each $B$-system to $B\bar{B}$, at the price of the enlarged local dimension $d_{B\bar B}\coloneqq d_B d_{\bar B}$ and hence a MIC measurement with $d_{B\bar B}^2$ outcomes; in exchange, the algebraic structure collapses to a single block, and in particular the positivity- and algebra-preserving compressions of \autoref{prop:coarse_grained_blocks} act on one Weyl module only.

\begin{corollary}[Bose-symmetric efficient rounding]\label{cor:bose_efficient_rounding}
    In the setting of \autoref{thm:efficient_rounding}, assume in addition that $\rho_{A\left(B\bar B\right)_1^n}$ is Bose-symmetric over the $\left(B\bar B\right)$-systems with respect to $A$, i.e.\ supported on $\mathcal{H}_A\otimes \operatorname{Sym}^n\left(\mathcal{H}_B\otimes\mathcal{H}_{\bar B}\right)$. Then the compressed Schur form consists of the single block $\lambda = \left(n\right)$, of size $d_A\, m_{\left(n\right)}$ with $m_{\left(n\right)}\leq \left(n + d_{B\bar B}\right)^{d_{B\bar B}-1}$, all objects in \autoref{prop:coarse_grained_blocks}, \autoref{lem:reblocking} and \autoref{prop:efficient_partial_trace} restrict to the one-row partitions $\left(k\right)$, $k\leq n$, and the conclusions of \autoref{thm:efficient_rounding} hold verbatim with local dimension $d_{B\bar B}$.
\end{corollary}

\begin{proof}
    Bose symmetry means $p^{\left(n\right)}_\lambda = 0$ for all $\lambda \neq \left(n\right)$, since $\operatorname{Sym}^n\left(\mathcal{H}_{B\bar B}\right)\cong U^{\left(n\right)}_{\left(n\right)}\otimes V^{\left(n\right)}_{\left(n\right)}$ with the trivial one-dimensional Specht module $V_{\left(n\right)}$. A one-row partition has a single removable cell, $\left(k\right)^- = \lbrace\left(k-1\right)\rbrace$, so the pushforward \autoref{eqn:branching_pushforward} gives $p^{\left(k\right)}_{\left(k\right)} = 1$ and $w_{\left(k-1\right)\vert\left(k\right)} = 1$ at every level; consequently every step of \autoref{prop:efficient_partial_trace} involves a single block and a single isometry $C_{\left(k-1\right)\leftarrow\left(k\right)}$, and in \autoref{prop:coarse_grained_blocks} only the block $\lambda = \left(m\right)$ contributes to \autoref{eqn:def_pt} and \autoref{eqn:def_conditional_states}. The runtime and the guarantee \autoref{eqn:rounding_guarantee} are those of \autoref{thm:efficient_rounding} with $d_B$ replaced by $d_{B\bar B}$; note that $\epsilon\left(d_A,n\right)$ depends on the local dimension of $B$ only through the distortion bound of the chosen MIC measurement.
\end{proof}

Moreover, the rounding scheme is compatible with the full constraint class of the cSEP framework \autoref{eq:def_cSEP_set}: per-factor linear constraints --- fixed point as well as with fixed marginals --- imposed on the retained systems are inherited by the candidates, as in \cite{kossmann2025_aqec, zeiss2025approximatingfixedsizequantum, zeiss2026finitefinetticonvexbodies}.

\begin{corollary}[Efficient rounding under cSEP constraints]\label{cor:constrained_efficient_rounding}
    In the setting of \autoref{thm:efficient_rounding}, let $\Psi:\mathcal{S}\left(\mathcal{H}_A\right)\to \mathcal{S}\left(\mathcal{H}_A\right)$ and $\Phi:\mathcal{S}\left(\mathcal{H}_B\right)\to \mathcal{S}\left(\mathcal{H}_B\right)$ be channels, and let, with respect to bipartitions $\mathcal{H}_A = \mathcal{H}_{A_L}\otimes\mathcal{H}_{A_R}$ and $\mathcal{H}_B = \mathcal{H}_{B_L}\otimes\mathcal{H}_{B_R}$, the linear maps $\Theta_{A_L\to C_{A_L}}$ and $\mathrm{Y}_{B_L\to C_{B_L}}$ together with fixed operators $W_{C_{A_L}}$ and $K_{C_{B_L}}$ specify constraints with fixed marginals as in \autoref{eq:def_cSEP_set}. Assume that $\rho_{AB_1^n}$ satisfies the corresponding extension-level constraints: it is invariant under $\Psi$ applied to $A$ and under $\Phi$ applied to any single $B$-factor, and it fulfills
    \begin{align}\label{eqn:extension_level_cSEP_constraints}
        \Theta\left(\rho_{AB_1^n}\right) = W_{C_{A_L}}\otimes \rho_{A_RB_1^n}, \qquad
        \mathrm{Y}\left(\rho_{AB_1^n}\right) = K_{C_{B_L}}\otimes \rho_{AB_1^{n-1}\left(B_R\right)_n},
    \end{align}
    where $\mathrm{Y}$ acts on the $n$-th $B$-factor; by permutation invariance the same then holds for every $B$-factor. Then every conditional state inherits all four constraints, i.e.\ for all $t$,
    \begin{align}\label{eqn:inherited_cSEP_constraints}
        \Psi\left(\rho_{A\vert t}\right) = \rho_{A\vert t}, \qquad
        \Phi\left(\rho_{B_1\vert t}\right) = \rho_{B_1\vert t}, \qquad
        \Theta\left(\rho_{A\vert t}\right) = W_{C_{A_L}}\otimes \rho_{A_R\vert t}, \qquad
        \mathrm{Y}\left(\rho_{B_1\vert t}\right) = K_{C_{B_L}}\otimes \rho_{\left(B_R\right)_1\vert t},
    \end{align}
    such that each candidate $\sigma^{\left(m\right)}_{AB}$ of \autoref{eqn:def_rounding_candidates} lies in the set $\operatorname{cSEP}_{\mathcal{C}_A,\mathcal{C}_B}\left(A:B\right)$ of \autoref{eq:def_cSEP_set} cut out by these constraints, and the guarantees \autoref{eqn:rounding_guarantee} and \autoref{eqn:rounding_value_guarantee} hold for the corresponding $\operatorname{cSEP}$ problem instance.
\end{corollary}
\begin{proof}
    By permutation invariance it suffices to consider maps acting on the retained factor $B_1$, respectively on $A$. Any such map commutes with the measurement channel acting on $B_2^{m+1}$, hence, using the invariance of $\rho_{AB_1^{m+1}}$ under $\Phi_{B_1}$,
    \begin{align}
        \left(\operatorname{id}_A\otimes \Phi_{B_1}\right)\left(\tr_{B_2^{m+1}}\left[\left(1_{AB_1}\otimes G_{B_2^{m+1}\vert t}\right)\rho_{AB_1^{m+1}}\right]\right) = \tr_{B_2^{m+1}}\left[\left(1_{AB_1}\otimes G_{B_2^{m+1}\vert t}\right)\rho_{AB_1^{m+1}}\right],
    \end{align}
    so that $\rho_{AB_1\vert t}$, and thus its marginals, inherit the fixed point constraints; the same argument applies to $\Psi$ on $A$. For the constraints with fixed marginals the identical computation applies with $\Phi_{B_1}$ replaced by $\mathrm{Y}_{\left(B_L\right)_1\to C_{B_L}}$ and the invariance replaced by the extension-level constraint \autoref{eqn:extension_level_cSEP_constraints}, which, after normalization by $p\left(t\right)$, produces the right-hand sides in \autoref{eqn:inherited_cSEP_constraints}; likewise for $\Theta$ on $A$, cf.\ \cite[Eqs.~(B.8)--(B.11)]{zeiss2025approximatingfixedsizequantum}. Feasibility of $\sigma^{\left(m\right)}_{AB}$ for $\operatorname{cSEP}_{\mathcal{C}_A,\mathcal{C}_B}\left(A:B\right)$ follows since it is a convex combination of the constrained product states $\rho_{A\vert t}\otimes \rho_{B_1\vert t}$, and the error guarantees are those of \autoref{thm:efficient_rounding}. We emphasize the scope of the hypothesis: constraint maps coupling several $B$-factors, or acting non-trivially only on the measured factors, are not covered by this argument; for the compatibility of such constraints with symmetry reduction we refer to the discussion in \cite[Sec.~7]{kossmann2025_aqec} and \cite{zeiss2025approximatingfixedsizequantum, zeiss2026finitefinetticonvexbodies}.
\end{proof}

In summary, for a given cSEP problem instance of fixed local dimension, the outer approximation schemes providing $\operatorname{poly}\left(1/\epsilon\right)$-time additive $\epsilon$-error outer approximations \cite{kossmann2025_aqec, zeiss2025approximatingfixedsizequantum, zeiss2026finitefinetticonvexbodies}, together with the inner approximations derived here, constitute an efficient algorithmic procedure determining the optimal value of such problems up to an additive $\epsilon$-interval. Moreover, explicit separable candidates certifying additive $\epsilon$-error lower bounds on any such cSEP problem can be computed in $\operatorname{poly}\left(1/\epsilon\right)$-time.

\section{Applications}\label{sec:applications}

In \cite{kossmann2025_aqec} the framework of combined symmetry reduction for semidefinite programming (SDP) has been introduced, which turned out to be a powerful tool for solving de Finetti hierarchies with symmetries as SDP constraints. To be concrete, the underlying idea can be explained as follows. Let $\mathcal{U}$ and $\mathcal{V}$ be subgroups of the unitary group $\mathcal{U}\left(\mathcal{H}\right)$ and consider an SDP 
\begin{equation}\label{eq:symmtry_theorem_problem}
        \begin{aligned}
            \sup \ &\tr\left[H \rho\right] \\
            \operatorname{s.t.} \  &\tr\left[A_j \rho \right] \leq b_j,  \quad  1\leq j \leq m \\
            &\rho \geq 0 \\
            &UHU^\dagger = H  \quad \text{and} \quad U A_j U^\dagger = A_j \quad \text{for all} \  U \in \mathcal{U}, \ 1\leq j \leq m \\
    &V\rho V^\dagger = \rho \quad \text{for all} \ V \in \mathcal{V}.
        \end{aligned}
    \end{equation}
Can we benefit from both symmetry groups $\mathcal{U}$ and $\mathcal{V}$ for the purpose of symmetry reduction of SDPs? The answer to this is involved as \cite[Thm.~7.6]{kossmann2025_aqec} gives a concrete case where combined symmetry reduction is not possible in the first place. In order to state clearly what it means to have a joint symmetry of the SDP \autoref{eq:symmtry_theorem_problem}, a corresponding well-defined notion has been established as follows. Concretely, the SDP \autoref{eq:symmtry_theorem_problem} exhibits a \emph{joint symmetry} if the set 
        \begin{align}\label{eq:strong_joint_symmetry_def}
            \mathcal{W} \coloneqq \left\{UV \ \vert \ U \in \mathcal{U}, \ V \in \mathcal{V}\right\}
        \end{align}
        is a group. Particularly, it is shown that $\mathcal{W}$ is immediately a compact subgroup of $\mathcal{U}\left(\mathcal{H}\right)$ \cite[Lem.~7.4]{kossmann2025_aqec}. In this case \cite[Prop.~7.5~(d)]{kossmann2025_aqec} shows that  the SDP
\begin{equation}\label{eq:symmetry_reduced_programm_prop}
    \begin{aligned}
        \sup \ &\operatorname{tr}\left[H \rho\right] \\
        \text{s.t.} \quad &\operatorname{tr}\left[A_j \rho \right] \leq b_j, \quad 1 \leq j \leq m, \\
        &\rho \geq 0, \\
        &\rho \in \mathcal{A}_{\mathcal{W}}^\prime
    \end{aligned}
\end{equation}
admits the same optimal value as the SDP in \autoref{eq:symmtry_theorem_problem}. In the particular case of a de Finetti hierarchy, $\mathcal{V} \cong S_n$ acts via the tensor representation \autoref{eq:def_permutation_action_Sn_H_otimes_n}. In this section we aim to underpin the practical point of view of combined symmetry reduction with theoretical results on improved convergence behavior of de Finetti hierarchies under constraints. 

\subsection{A bilinear optimization problem with symmetries and linear constraints}
Specifically the constraint situation of \autoref{thm:interpolation_deFinetti} falls in the category of joint symmetry reduction. Let us define the notion of covariant maps as follows. A linear map $ \Psi: \mathcal{B}\left(\mathcal{H}\right) \to \mathcal{B}\left(\mathcal{H}\right) $ is said to be \emph{covariant} with respect to a subgroup $ \mathcal{V} $ of the unitary group $ \mathcal{U}\left(\mathcal{H}\right) $ if 
    \begin{align}
    \Psi\left(V \cdot V^\dagger\right) = V \Psi\left(\cdot\right) V^\dagger, \quad \text{for all} \ V \in \mathcal{V}. 
    \end{align}
For the purpose of application we consider the following bilinear optimization program. Let $H_{AB} \in \mathcal{B}\left(\mathcal{H}_A \otimes \mathcal{H}_B\right)$ be an operator and $\mathcal{V}_i \subseteq \mathcal{U}\left(\mathcal{H}_i\right)$ for $i=A,B$ closed subgroups and consider
\begin{equation}\label{eq:bilinear_optimization_under_symmetries}
    \begin{aligned}
        c_{\operatorname{opt}}\coloneqq \sup \ &  \tr\left[H_{AB} \,  \rho_A \otimes \rho_B\right] \\
        \operatorname{s.t.} \ & \Psi_{A\to A}\left(\rho_A\right) = X_A \\
        & \Psi_{B\to B}\left(\rho_B\right) = X_B \\
        & \Psi_{A\to A} \ \text{is covariant w.r.t.} \ \mathcal{V}_A \\
        &\Psi_{B\to B} \ \text{is covariant w.r.t.} \ \mathcal{V}_B \\
        &\left(V_A\otimes V_B\right)\, H_{AB}\, \left(V_A\otimes V_B\right)^\dagger = H_{AB}, \quad V_A \in \mathcal{V}_A,\ V_B \in \mathcal{V}_B \\
        &\rho_A\in \mathcal{S}\left(\mathcal{H}_A\right), \quad \rho_B\in \mathcal{S}\left(\mathcal{H}_B\right).
    \end{aligned}
\end{equation}
\begin{proposition}[Convergence under combined symmetry reduction]\label{prop:bilinear_hierarchy}
    Consider \autoref{eq:bilinear_optimization_under_symmetries}. Then the outer hierarchy given by 
    \begin{equation}\label{eq:bilinear_optimization_under_symmetries_hierarchy}
    \begin{aligned}
        c^{\left(n\right)}\coloneqq \sup \ &  \tr\left[H_{AB} \,  \rho_{AB_1^n}\right] \\
        \operatorname{s.t.} \ & (\Psi_{A\to A}\otimes \id_{B_1^n})\left(\rho_{AB_1^n}\right) = X_A \otimes \rho_{B_1^n} \\
        & (\id_{AB_1^{n-1}}\otimes \Psi_{B\to B})\left(\rho_{AB_1^n}\right) = \rho_{AB_1^{n-1}}\otimes  X_B \\
        & \rho_{AB_1^n} \in \mathcal{A}_{\mathcal{V}_A\times \mathcal{V}_B \times S_n}^\prime \subseteq \mathcal{S}\left(\mathcal{H}_A \otimes \mathcal{H}_B^{\otimes n}\right)
    \end{aligned}
\end{equation}
satisfies the following bound
\begin{align}
    c^{\left(n\right)}-c_{\operatorname{opt}}\leq \lVert H_{AB}\rVert_\infty \, 2 \, m_{B,\operatorname{max}} \sqrt{\frac{2 \ln 2 \, \log  \left(\sum_{i=1}^{k_A} m_{i,A}^2 \right)}{n}},
\end{align}
 whereby $m_{i,A}$ are the dimensions of the fixed point algebra of $\mathcal{V}_{A}$ and $m_{B,\operatorname{max}} \coloneqq \max_{1\leq i \leq k_B} m_{i,B}$ is the largest block-size of the fixed point algebra of $\mathcal{V}_{B}$.
\end{proposition}
\begin{proof}
    As the objective function in \autoref{eq:bilinear_optimization_under_symmetries} is linear, we can consider the set of feasible points
    \begin{align}
        \mathcal{F} \coloneqq \operatorname{conv}\left \{ \rho_A\otimes \rho_B \ \vert \ \Psi_{A\to A}\left(\rho_A\right) = X_A,\ \Psi_{B\to B}\left(\rho_B\right) = X_B \right\}.
    \end{align}
    Applying the outer hierarchy from \cite{Berta2021} on \autoref{eq:bilinear_optimization_under_symmetries} yields the following optimization program
    \begin{equation}
    \begin{aligned}
        c^{\left(n\right)}\coloneqq \sup \ &  \tr\left[H_{AB} \,  \rho_{AB_1^n}\right] \\
        \operatorname{s.t.} \ & (\Psi_{A\to A}\otimes \id_{B_1^n})\left(\rho_{AB_1^n}\right) = X_A \otimes \rho_{B_1^n} \\
        & (\id_{AB_1^{n-1}}\otimes \Psi_{B\to B})\left(\rho_{AB_1^n}\right) = \rho_{AB_1^{n-1}}\otimes  X_B \\
        & \rho_{AB_1^n} \in \mathcal{A}_{ S_n}^\prime \subseteq \mathcal{S}\left(\mathcal{H}_A \otimes \mathcal{H}_B^{\otimes n}\right).
    \end{aligned}
\end{equation}
    As $H_{AB}$ is invariant under the action of $\mathcal{V}_A$ and $\mathcal{V}_B$ and the constraints are covariant with respect to $\mathcal{V}_A$ and $\mathcal
    {V}_B$, the objective value does not change and the constraints are satisfied. Moreover, as the $S_n$-symmetry and the $\mathcal{V}_A \times \mathcal{V}_B$ symmetry commute, \cite[Prop.~7.5~(d)]{kossmann2025_aqec} implies that we can combine the symmetries and find a solution in $\mathcal{A}^\prime_{\mathcal{V}_A \times \mathcal{V}_B \times S_n}$. Moreover, given $\rho_{AB_1^n} \in \mathcal{A}^\prime_{\mathcal{V}_A \times \mathcal{V}_B \times S_n}$, we can apply \autoref{thm:interpolation_deFinetti} with the choices $\Phi_{A\to A} \coloneqq \int_{\mathcal{V}_A} V \, \cdot \, V^\dagger \, d\mu_{\mathcal{V}_A}\left(V\right)$ and $\Phi_{B\to B} \coloneqq \int_{\mathcal{V}_B} V \, \cdot \, V^\dagger \, d\mu_{\mathcal{V}_B}\left(V\right)$, i.e.\ the Haar twirls over $\mathcal{V}_A$ and $\mathcal{V}_B$, which admit the maximally mixed state as a full-rank fixed point. By \autoref{cor:entanglement_assisted_classcial_capacity_group_averages} these twirls are duals of conditional expectations onto the commutants $\mathcal{A}^\prime_{\mathcal{V}_A}$ and $\mathcal{A}^\prime_{\mathcal{V}_B}$, and any state in $\mathcal{A}^\prime_{\mathcal{V}_A \times \mathcal{V}_B \times S_n}$ satisfies the corresponding fixed point constraints together with permutation invariance over the $B$-systems. Applying Hölder's inequality finally yields the desired convergence speed. 
\end{proof}

\subsection{Approximate quantum error correction}
We consider the task of approximate quantum error correction (AQEC) with encoding a code space $\mathcal{H}_L$ into a Hilbert space $\mathcal{H}_P$, which can be considered as the task of, for a given noisy channel $\mathcal{N}$, finding optimal encoder-decoder pairs $\left(\mathcal{E},\mathcal{D}\right)$, such that
\begin{align}
    \mathcal{H}_L \stackrel{\mathcal{E}}{\longrightarrow} \mathcal{H}_P \stackrel{\mathcal{N}}{\longrightarrow} \mathcal{H}_P \stackrel{\mathcal{D}}{\longrightarrow} \mathcal{H}_L,
\end{align}
yields $\mathcal{D}\circ \mathcal{N} \circ \mathcal{E} \approx \operatorname{id}_L$. A well-established quantity in order to quantify the approximation is given by the \emph{channel fidelity}, which is defined as
\begin{align}\label{eq:def_channel_fidelity}
    F\left(\mathcal{N},d_L\right)\coloneqq \sup_{\left(\mathcal{E},\mathcal{D}\right)} \ \mathcal{F}\left(\Phi_{LL}, \mathcal{D}\circ \mathcal{N} \circ \mathcal{E} \left( \Phi_{LL}\right)\right).
\end{align}
For the channel fidelity a converging hierarchy has been established in \cite{Berta2021}. A general question since then has been whether the $O\left(1/n\right)$-convergence behavior is structurally included in the approximate quantum error correction problem (AQEC). In this section we aim to prove a practically relevant theorem regarding AQEC, solving two main issues in prior solutions:
\begin{enumerate}
    \item one-sided extensions such as shown in \cite{Berta2021} and \cite{kossmann2025_aqec} have a $O\left(1/\sqrt{n}\right)$-convergence behavior, which is strictly suboptimal from the perspective of de Finetti theorems (cf.~\cite{Christandl2007} with a $O\left(1/n\right)$-behavior). However, in \autoref{thm:double_extended_de_Finetti} and \autoref{thm:bose_de_Finetti} we show that in principle a $O\left(1/n\right)$ convergence behavior is possible. It would remain to show how to deal with the linear constraints. 
    \item in \cite{kossmann2025_aqec} a general framework for combined symmetry reduction has been established, as recalled in \autoref{eq:symmtry_theorem_problem}. For the error correction problem a natural ordering of subsystems is by physical and logical code spaces, denoted by $\mathcal{H}_P$ and $\mathcal{H}_L$. The general subsystem structure is then for one-sided extension given by
    \begin{equation}\label{eq:subsystem_structure}
\begin{aligned}
\begin{matrix}
P & \bar{P}^{\left(1\right)} & \bar{P}^{\left(2\right)} & \ldots & \bar{P}^{\left(n\right)}\\
L & \bar{L}^{\left(1\right)} & \bar{L}^{\left(2\right)} & \ldots & \bar{L}^{\left(n\right)}.
\end{matrix}
\end{aligned}
\end{equation}
    Solving the decomposition into irreducibles can be done separately for $\mathcal{H}_{P\bar{P}_1^n}$ and $\mathcal{H}_{L\bar{L}_1^n}$. Particularly on $\mathcal{H}_{L\bar{L}_1^n}$ a diagonal full unitary group $\mathcal{U}\left(\mathcal{H}_L\right)$ acts and on $\mathcal{H}_{P\bar{P}_1^n}$ some diagonal action of a closed subgroup $\mathcal{V}\subseteq \mathcal{U}\left(\mathcal{H}_P\right)$. The space $\mathcal{H}_{P\bar{P}_1^n}$ decomposes under $\mathcal{V}\times S_n$ into irreducible representations and similarly $\mathcal{H}_{L\bar{L}_1^n}$ into irreducibles from $\mathcal{U}\left(\mathcal{H}_L\right)\times S_n$. However, connecting them yields couplings between the $S_n$-representations
    \begin{equation}\label{eq:S_n_coupling}
    \begin{aligned}
        \mathcal{H}_{L\bar{L}_1^n} \otimes \mathcal{H}_{P\bar{P}_1^n} &\cong_{\mathcal{U}\left(\mathcal{H}_L\right) \times \mathcal{V}\times S_n} \left(\bigoplus_{i=1}^{k_\mathcal{V}} \bigoplus_{\lambda \vdash n} \mathbb{C}^{m_{i,\lambda}} \otimes \mathbb{C}^{d_i} \otimes \mathbb{C}^{d_\lambda} \right) \otimes \left(\bigoplus_{j=1}^{k_{\mathcal{U}\left(\mathcal{H}_L\right)}} \bigoplus_{\mu \vdash n} \mathbb{C}^{m_{j,\mu}} \otimes \mathbb{C}^{e_j} \otimes \mathbb{C}^{d_\mu} \right) \\
        &= \bigoplus_{i=1}^{k_\mathcal{V}} \bigoplus_{\lambda \vdash n} \bigoplus_{j=1}^{k_{\mathcal{U}\left(\mathcal{H}_L\right)}} \bigoplus_{\mu \vdash n} \mathbb{C}^{m_{i,\lambda}} \otimes \mathbb{C}^{m_{j,\mu}} \otimes \mathbb{C}^{e_j} \otimes \mathbb{C}^{d_i} \otimes \underbrace{\mathbb{C}^{d_\lambda} \otimes \mathbb{C}^{d_\mu}}_{\cong \bigoplus_{\nu \vdash n}\mathbb{C}^{g_{\lambda\mu\nu}} \otimes \mathbb{C}^{d_\nu}}. 
    \end{aligned}
    \end{equation}
    Here $g_{\lambda\mu\nu}$ denote the Kronecker coefficients of $S_n$, i.e.\ the multiplicities in the decomposition of the tensor product $V_\lambda\otimes V_\mu$ of Specht modules into irreducibles. In \cite[Sec.~5.4]{kossmann2025_aqec} it is reported that the coupling between the irreducibles $\mu,\lambda\vdash n$ from $S_n$ is a core issue in an implementation. Here, we show by employing \autoref{thm:bose_de_Finetti} that a similar Bose-symmetric hierarchy can be used where the calculation in \autoref{eq:S_n_coupling} can be drastically simplified as we are just considering $\lambda = \mu = (n)$, which is (trivially) known to be multiplicity free (for general Kronecker coefficients cf.~\cite{Bessenrodt2017}).
\end{enumerate}

In the following we give an example of an AQEC problem, solving both challenges raised in $\left(1\right)$ and $\left(2\right)$. Let $\mathcal{N}_{P\to P}$ be a depolarizing channel. We aim to combine the techniques from \cite{kossmann2025_aqec} with \autoref{thm:bose_de_Finetti}, building on \cite[Lem.~5.7]{zeiss2025approximatingfixedsizequantum}. For this purpose we purify the systems in \autoref{eq:subsystem_structure} and double-extend them as follows
\begin{equation}\label{eq:subsystem_structure_purified}
\begin{aligned}
\begin{matrix}
P^{\left(n\right)} & \ldots &P^{\left(2\right)}& P^{\left(1\right)} & \bar{P}^{\left(1\right)} & \bar{P}^{\left(2\right)} & \ldots & \bar{P}^{\left(n\right)}\\
P_{\operatorname{copy}}^{\left(n\right)} & \ldots &P_{\operatorname{copy}}^{\left(2\right)}& P_{\operatorname{copy}}^{\left(1\right)} & \bar{P}_{\operatorname{copy}}^{\left(1\right)} & \bar{P}_{\operatorname{copy}}^{\left(2\right)} & \ldots & \bar{P}_{\operatorname{copy}}^{\left(n\right)}\\
L^{\left(n\right)} & \ldots &L^{\left(2\right)}& L^{\left(1\right)} & \bar{L}^{\left(1\right)} & \bar{L}^{\left(2\right)} & \ldots & \bar{L}^{\left(n\right)}\\
L^{\left(n\right)}_{\operatorname{copy}} & \ldots &L^{\left(2\right)}_{\operatorname{copy}} & L^{\left(1\right)}_{\operatorname{copy}} & \bar{L}^{\left(1\right)}_{\operatorname{copy}} & \bar{L}^{\left(2\right)}_{\operatorname{copy}} & \ldots & \bar{L}^{\left(n\right)}_{\operatorname{copy}},
\end{matrix}
\end{aligned}
\end{equation}
whereby $\bar{P}^{\left(i\right)} \cong P^{\left(i\right)}$, for all $1\leq i \leq n$ and similarly for the $L$-systems. The $n$-level relaxation on these subsystems is then denoted by
\begin{equation}\label{eq:enlarge_symmetry_sdp}
        \begin{aligned}
            \operatorname{SDP}_n^\star = \ &\tr\left[C_{PP} \otimes \Phi_{L\bar{L}}\rho_{P\bar{P}L\bar{L}}\right]&  \\
            \operatorname{s.t.} \ &\tr\left[\rho_{\left(PP_{\operatorname{copy}}\right)^{\left(1\ldots n\right)}\left(\bar{P}\bar{P}_{\operatorname{copy}}\right)^{\left(1\ldots n\right)}\left(LL_{\operatorname{copy}}\right)^{\left(1\ldots n\right)}\left(\bar{L}\bar{L}_{\operatorname{copy}}\right)^{\left(1\ldots n\right)}}\right] = 1&\\
            &\tr_{L\left(P_{\operatorname{copy}}L_{\operatorname{copy}}\right)^{\left(n\right)}}\left[\rho_{\left(PP_{\operatorname{copy}}\right)^{\left(1\ldots n\right)}\left(\bar{P}\bar{P}_{\operatorname{copy}}\right)^{\left(1\ldots n\right)}\left(LL_{\operatorname{copy}}\right)^{\left(1\ldots n\right)}\left(\bar{L}\bar{L}_{\operatorname{copy}}\right)^{\left(1\ldots n\right)}}\right] =& \\
            &\hspace{2cm}\frac{1_{P^{\left(n\right)}}}{d_{P^{\left(n\right)}}} \otimes \rho_{\left(PP_{\operatorname{copy}}\right)^{\left(1\ldots n-1\right)}\left(\bar{P}\bar{P}_{\operatorname{copy}}\right)^{\left(1\ldots n\right)}\left(LL_{\operatorname{copy}}\right)^{\left(1\ldots n-1\right)}\left(\bar{L}\bar{L}_{\operatorname{copy}}\right)^{\left(1\ldots n\right)}}&\\
            &\tr_{\bar{P}^{\left(n\right)}\left(\bar{P}_{\operatorname{copy}}\bar{L}_{\operatorname{copy}}\right)^{\left(n\right)}}\left[\rho_{\left(PP_{\operatorname{copy}}\right)^{\left(1\ldots n\right)}\left(\bar{P}\bar{P}_{\operatorname{copy}}\right)^{\left(1\ldots n\right)}\left(LL_{\operatorname{copy}}\right)^{\left(1\ldots n\right)}\left(\bar{L}\bar{L}_{\operatorname{copy}}\right)^{\left(1\ldots n\right)}}\right] =& \\
            &\hspace{2cm}\rho_{\left(PP_{\operatorname{copy}}\right)^{\left(1\ldots n\right)}\left(\bar{P}\bar{P}_{\operatorname{copy}}\right)^{\left(1\ldots n-1\right)}\left(LL_{\operatorname{copy}}\right)^{\left(1\ldots n\right)}\left(\bar{L}\bar{L}_{\operatorname{copy}}\right)^{\left(1\ldots n-1\right)}}\otimes \frac{1_{\bar{L}^{\left(n\right)}}}{d_{\bar{L}^{\left(n\right)}}}& \\
            &\text{Bose-symmetric of} \ \left(PP_{\operatorname{copy}}LL_{\operatorname{copy}}\right)^{\left(1\ldots n\right)} \ \text{w.r.t.} \ \left(\bar{P}\bar{P}_{\operatorname{copy}}\bar{L}\bar{L}_{\operatorname{copy}}\right)^{\left(1\ldots n\right)}& \\
            &\text{Bose-symmetric of} \ \left(\bar{P}\bar{P}_{\operatorname{copy}}\bar{L}\bar{L}_{\operatorname{copy}}\right)^{\left(1\ldots n\right)} \ \text{w.r.t.} \ \left(PP_{\operatorname{copy}}LL_{\operatorname{copy}}\right)^{\left(1\ldots n\right)} & \\
            &\rho_{\left(PP_{\operatorname{copy}}\right)^{\left(1\ldots n\right)}\left(\bar{P}\bar{P}_{\operatorname{copy}}\right)^{\left(1\ldots n\right)}\left(LL_{\operatorname{copy}}\right)^{\left(1\ldots n\right)}\left(\bar{L}\bar{L}_{\operatorname{copy}}\right)^{\left(1\ldots n\right)}} \geq 0& 
        \end{aligned}
\end{equation}
Given the local action of the unitary group $\mathcal{U}\left(\mathcal{H}_P\right)$ on $\mathcal{H}_P$, we consider the action 
\begin{align}
    U \mapsto U\otimes \bar{U} \quad \text{on} \ \mathcal{H}_P \otimes \mathcal{H}_{P_{\operatorname{copy}}}.
\end{align}
Similarly we define its action on $\mathcal{H}_{\bar{P}}\otimes \mathcal{H}_{\bar{P}_{\operatorname{copy}}}$ and the action of $\mathcal{U}\left(\mathcal{H}_L\right)$ on $\mathcal{H}_L \otimes \mathcal{H}_{L_{\operatorname{copy}}}$ and correspondingly on $\mathcal{H}_{\bar{L}} \otimes \mathcal{H}_{\bar{L}_{\operatorname{copy}}}$. Given this action, we can consider the diagonal action on the systems 
\begin{align}
    \left(PP_{\operatorname{copy}}\right)^{\left(1\ldots n\right)}\left(\bar{P}\bar{P}_{\operatorname{copy}}\right)^{\left(1\ldots n\right)} \quad \text{and}\quad \left(LL_{\operatorname{copy}}\right)^{\left(1\ldots n\right)}\left(\bar{L}\bar{L}_{\operatorname{copy}}\right)^{\left(1\ldots n\right)}.
\end{align}
With this we are able to state the main theorem of this section. 
\begin{theorem}[Bose-symmetric hierarchy for AQEC]\label{thm:bose_hierarchy_aqec}
    Given a depolarizing channel $\mathcal{N}_{P\to P}$, the channel fidelity can be approximated with the semidefinite program \autoref{eq:enlarge_symmetry_sdp}, where the state can be chosen out of $\mathcal{A}^\prime_{S_{n_A}\times S_{n_B} \times \mathcal{U}\left(\mathcal{H}_L\right)\times\mathcal{U}\left(\mathcal{H}_P\right)}$ and we can restrict to the irreducibles $(n_A)\vdash n_A$ and $(n_B)\vdash n_B$, i.e.\ the symmetric subspaces\footnote{We write $n_A= n_B=n$ in order to distinguish the left and right copies of $S_n$.}. Furthermore we have the following estimate 
    \begin{align}
        \operatorname{SDP}_n^\star - F(\mathcal{N}_{P\to P},d_L) \leq 4 d_L^4 d_P^4 \frac{\sqrt{4 \ln 2 \left(d_L^2d_P^2 -1\right)\,\log\left(n+d_L^2 d_P^2 \right)}}{n}.
    \end{align}
\end{theorem}
\begin{proof}
    That any feasible point for the channel fidelity optimization problem \cite[Eq.\ (62)]{kossmann2025_aqec} has a double-sided $n$-extension such as in \autoref{eq:enlarge_symmetry_sdp} follows immediately from \cite[Lem.~5.7]{zeiss2025approximatingfixedsizequantum} by local purification. We introduce the diagonal action of the unitary group $\mathcal{U}\left(\mathcal{H}_P\right)$ on \autoref{eq:subsystem_structure_purified}
    \begin{align}\label{eq:diagonal_action_twirl_aqec}
         \int_{\mathcal{U}\left(\mathcal{H}_P\right)} d\mu_{\mathcal{U}\left(\mathcal{H}_P\right)}(U) \left(U \otimes \bar{U}\right)^{\otimes 2n} \, \cdot \, \left(\left(U \otimes \bar{U}\right)^{\otimes 2n}\right)^\dagger
    \end{align}
    and similarly for $\mathcal{U}\left(\mathcal{H}_L\right)$. It is easy to verify that this action commutes with the two symmetric groups as it is a diagonal action and furthermore the maximally entangled state is an isotropic state, thus invariant under the action yielding that the proof in \cite[Lem.~5.7]{zeiss2025approximatingfixedsizequantum} is compatible with the joint symmetry reduction framework \autoref{eq:strong_joint_symmetry_def} and \cite[Prop.~7.5.~(d)]{kossmann2025_aqec} such that we conclude to find a solution in $\mathcal{A}^\prime_{S_{n_A}\times S_{n_B} \times \mathcal{U}\left(\mathcal{H}_L\right)\times\mathcal{U}\left(\mathcal{H}_P\right)}$. By assumption the optimization problem \autoref{eq:enlarge_symmetry_sdp} is Bose-symmetric. For the bound we apply \autoref{thm:bose_de_Finetti} with the local dimensions $d_A = d_B = d_L^2 \, d_P^2$ of the purified blocks $\left(PP_{\operatorname{copy}}LL_{\operatorname{copy}}\right)^{\left(i\right)}$ and $\left(\bar{P}\bar{P}_{\operatorname{copy}}\bar{L}\bar{L}_{\operatorname{copy}}\right)^{\left(i\right)}$.
\end{proof}

\section{Conclusion and outlook}\label{sec:conclusion_outlook}

In this work we have developed a structural theory of de Finetti hierarchies under fixed point constraints. The guiding perspective is that the symmetries entering such hierarchies --- permutation invariance itself, invariance under local actions of compact groups, and, most generally, invariance under duals of conditional expectations --- all manifest as fixed point algebras of quantum channels, and that the Wedderburn block structure of these algebras is precisely the data governing the information-theoretic quantities on which de Finetti-type arguments rely. On the technical level, the entanglement-assisted capacity bound of \autoref{thm:upper_bound_entanglement_assisted_classical_capacity}, its specialization to group averages in \autoref{cor:entanglement_assisted_classcial_capacity_group_averages}, the blockwise distortion estimates for adapted informationally complete measurements in \autoref{prop:distortion_conditional_expectation} and \autoref{cor:distortion_correlation_postmeasurement}, and the exact type-based chain rule of \autoref{prop:chain_rule} combine into a modular toolbox. From it we derived the double-sided extension theorem \autoref{thm:double_extended_de_Finetti} with $O\left(\sqrt{\log n}/n\right)$ convergence, the interpolation theorem \autoref{thm:interpolation_deFinetti}, whose dimension dependence is governed solely by the block data of the involved fixed point algebras, the proper de Finetti theorem \autoref{thm:proper_deFinetti} for $k$ remaining systems, and the Bose-symmetric variant \autoref{thm:bose_de_Finetti}. Complementing these approximation statements, \autoref{sec:efficient_inner_sequence} established that the rounding scheme behind the inner sequence can be implemented in time polynomial in the hierarchy level for fixed local dimensions, and \autoref{sec:applications} translated the abstract results into convergence guarantees for symmetry-reduced bilinear optimization and into a Bose-symmetric hierarchy for approximate quantum error correction that avoids the Kronecker-coefficient couplings of previous approaches.

Several directions remain open. First, the efficiency results of \autoref{sec:efficient_inner_sequence} exploit Schur-Weyl duality for the one-sided action of the symmetric group; extending the algorithmic toolbox to mixed Schur-Weyl duality \cite{grinko2023gelfandtsetlinbasispartiallytransposed, grinko2025mixed} should render the double-sided hierarchies of \autoref{thm:double_extended_de_Finetti} computationally accessible in the same spirit. Second, in many applications the relevant operators have small rank $r \ll d$ compared to the local dimension, and we expect that the memory- and gate-efficient techniques of \cite{cerveromartin2024memorygateefficientalgorithm} can be combined with our rounding scheme to obtain further efficiency improvements. Third, while \autoref{thm:double_extended_de_Finetti} shows that the optimal $O\left(1/n\right)$-type behavior, up to logarithmic factors, survives double-sided extensions, it remains open whether the $O\left(1/\sqrt{N}\right)$ rate of \autoref{thm:proper_deFinetti} and the dependence of our bounds on the number of unmeasured systems are optimal; lower bounds for de Finetti theorems under fixed point constraints are, to the best of our knowledge, largely unexplored. Finally, it would be interesting to carry the fixed point perspective beyond quantum state spaces, e.g.\ to de Finetti theorems for general probabilistic theories and abstract cones \cite{zeiss2026finitefinetticonvexbodies, Plvala2025, Aubrun2024}, where duals of conditional expectations have to be replaced by suitable positive projections.

\section*{Acknowledgements}

JZ thanks Tobias Rippchen, Steven Kim and Nikolaos Louloudis for numerous discussions. JZ thanks Dmitry Grinko for discussions on \autoref{sec:efficient_inner_sequence}. GK and JZ acknowledge support from the Excellence Cluster - Matter and Light for Quantum Computing (ML4Q-2) and by the European Research Council (ERC Grant Agreement No. 948139). The authors acknowledge the use of \cite{claude}.

\bibliography{references_new}
\bibliographystyle{ultimate}
\appendix

\section{Recap on fixed points of quantum channels}

We recall a central theorem on fixed points for quantum channels, which can be found e.g.\ in \cite{Hayden_2004,Petz2008QITQS} or the recent work \cite{Carlen2025}. 
\begin{proposition}\label{thm:ergodic_projection_condexp}
    Given a finite dimensional Hilbert space $\mathcal{H}$ and a unital and  $2$-positive map 
    \begin{align}
        \Phi:\mathcal{B}\left(\mathcal{H}\right) \to \mathcal{B}\left(\mathcal{H}\right), 
    \end{align}
    let $\Phi^*$ denote the Hilbert-Schmidt dual of $\Phi$ (i.e.\ $\tr\left[\Phi^*\left(\rho\right)a\right]=\tr\left[\rho \Phi\left(a\right)\right]$).
    Then we have
    \begin{enumerate}
        \item There exists a state $\hat{\rho}\in\mathcal{S}\left(\mathcal{H}\right)$ such that $\Phi^*\left(\hat{\rho}\right)=\hat{\rho}$.
        For any such choice of $\hat{\rho}$, the ergodic means
        \begin{align}
            E_N \coloneqq \frac{1}{N} \sum_{t=0}^{N-1} \Phi^{\circ t}
        \end{align}
        converge, as $N\to\infty$, in operator norm on $\mathcal{H}_{\hat{\rho}}$ to the \emph{orthogonal projection}
        (with respect to $\langle \cdot,\cdot\rangle_{\hat{\rho}}$) onto the subspace of fixed points
        $\operatorname{Fix}\left(\Phi\right)\coloneqq \ker\left(\Phi-\operatorname{id}\right)$.

        \item If $\Phi^*$ has a full-rank fixed point $\hat{\rho}>0$, then $\operatorname{Fix}\left(\Phi\right)$ is a unital $*$-subalgebra of
        $\mathcal{B}\left(\mathcal{H}\right)$ and the limit map
        \begin{align}
             E:\mathcal{B}\left(\mathcal{H}\right) \to \operatorname{Fix}\left(\Phi\right)
        \end{align}
        is a conditional expectation.
    \end{enumerate}
\end{proposition}

\begin{proof}
Given the $2$-positive and unital map $\Phi$, we find its unique dual map
\begin{align}
    \tr\left[\Phi^*\left(\rho\right)a\right] = \tr\left[\rho \Phi\left(a\right)\right], \quad a \in \mathcal{B}\left(\mathcal{H}\right), \quad \rho \in \mathcal{B}\left(\mathcal{H}\right).
\end{align}
It is well-known that $\Phi^*$ is $2$-positive again and trace-preserving. By finite dimensionality this implies that
\begin{align}
    \mathcal{S}\left(\mathcal{H}\right) \coloneqq \{\rho \in \mathcal{B}\left(\mathcal{H}\right) \ \vert \ \tr\left[\rho\right] = 1, \ \rho\geq 0\}
\end{align}
is convex and compact and $\Phi^*\left(\mathcal{S}\left(\mathcal{H}\right)\right)\subseteq \mathcal{S}\left(\mathcal{H}\right)$. By Schauder's fixed point theorem \cite[Thm.~5.2.5]{Drbek2013} there exists a state $\hat{\rho} \in \mathcal{S}\left(\mathcal{H}\right)$ such that $\Phi^*\left(\hat{\rho}\right) = \hat{\rho}$. We consider the GNS representation with respect to $\hat{\rho} \in \mathcal{S}\left(\mathcal{H}\right)$, defined as
\begin{align}
    \langle a, b\rangle_{\hat{\rho}} \coloneqq \tr\left[\hat{\rho}  a^* b\right], \quad a,b\in \mathcal{B}\left(\mathcal{H}\right).
\end{align}
If $\hat{\rho} \in \mathcal{S}\left(\mathcal{H}\right)$ is not faithful we consider the quotient
\begin{align}
    \mathcal{H}_{\hat{\rho}} \coloneqq \mathcal{B}\left(\mathcal{H}\right)/N_{\hat{\rho}},
\end{align}
where $N_{\hat{\rho}} \coloneqq \{a \in \mathcal{B}\left(\mathcal{H}\right) \ \vert \ \lVert a\rVert_{\hat{\rho}} = 0\}$ and $\lVert a\rVert_{\hat{\rho}}^2 \coloneqq \tr\left[\hat{\rho}  a^* a\right]$.

In $\mathcal{H}_{\hat{\rho}}$ we have $\lVert \Phi\rVert_{\mathcal{H}_{\hat{\rho}}\to \mathcal{H}_{\hat{\rho}}} \leq 1$, because we apply the Schwarz inequality \cite[Thm.\ 5.3]{Wolf2012QuantumChannelsGuidedTour}
\begin{equation}
\begin{aligned}
    \lVert \Phi\left(a\right)\rVert^2_{\hat{\rho}} &= \tr\left[\hat{\rho}  \Phi\left(a\right)^* \Phi\left(a\right)\right] \\
    &\leq \tr\left[\hat{\rho}  \Phi\left(a^* a\right)\right] \\
    &= \tr\left[\Phi^*\left(\hat{\rho}\right)  a^* a\right]\\
    &= \tr\left[\hat{\rho}  a^* a\right] = \lVert a\rVert_{\hat{\rho}}^2.
\end{aligned}
\end{equation}
Thus $\Phi$ is a contraction on the Hilbert space $\mathcal{H}_{\hat{\rho}}$ and $\lVert \Phi^{\circ t} \rVert_{\mathcal{H}_{\hat{\rho}}\to \mathcal{H}_{\hat{\rho}}} \leq 1$. Now we consider the family
\begin{align}
    \mathcal{U} \coloneqq \{\Phi^{\circ t} \ \vert \ t \in \mathbb{N}_0\} \subseteq \mathcal{B}\left(\mathcal{H}_{\hat{\rho}}\right).
\end{align}
Then we have
\begin{enumerate}
    \item $\lVert U \rVert \leq 1$ for all $U \in \mathcal{U}$,
    \item $U_1 U_2 \in \mathcal{U}$ for all $U_1,U_2 \in \mathcal{U}$.
\end{enumerate}
Let $E$ be the orthogonal projection on the subspace of $\mathcal{H}_{\hat{\rho}}$ defined as
\begin{align}
    E\mathcal{H}_{\hat{\rho}} = \{\phi \in \mathcal{H}_{\hat{\rho}} \ \vert \ U\phi = \phi \quad \text{for all} \ U \in \mathcal{U}\}.
\end{align}
As $\Phi \in \mathcal{U}$, the fixed space above coincides with the $\Phi$-fixed space, i.e.,
\begin{align}
    E\mathcal{H}_{\hat{\rho}} = \ker\left(\Phi-\operatorname{id}\right) = \operatorname{Fix}\left(\Phi\right).
\end{align}

As in particular
\begin{align}
    E_N \coloneqq \frac{1}{N} \sum_{t=0}^{N-1} \Phi^{\circ t} \in \operatorname{co}\left(\mathcal{U}\right)
    \quad \text{and} \quad
    \lVert E_N\rVert \leq 1,
\end{align}
we can compute
\begin{align}\label{eq:fixed point_eq_proof_appendix}
    \Phi \circ E_N - E_N
    = \frac{1}{N}\sum_{t=0}^{N-1} \left(\Phi^{\circ \left(t+1\right)}-\Phi^{\circ t} \right)
    = \frac{1}{N}\left(\Phi^{\circ N}-\operatorname{id}\right).
\end{align}
Thus,
\begin{align}
    \lVert \Phi \circ E_N - E_N\rVert \leq \frac{1}{N}\lVert\Phi^{\circ N}-\operatorname{id}\rVert \to 0
    \quad \text{as} \quad N\to \infty.
\end{align}

Since $\mathcal{H}_{\hat{\rho}}$ is finite-dimensional, the unit ball in $\mathcal{B}\left(\mathcal{H}_{\hat{\rho}}\right)$ is compact in operator norm. Moreover, $\operatorname{co}\left(\mathcal{U}\right)$ is bounded (as $\lVert U\rVert \leq 1$ for all $U\in\mathcal{U}$), hence its norm-closure $\overline{\operatorname{co}\left(\mathcal{U}\right)}$ is a closed subset of a compact set and therefore compact. Consequently, there exists a subsequence $\left(E_{N_k}\right)$ converging in operator norm to some
\begin{align}
    F \in \overline{\operatorname{co}\left(\mathcal{U}\right)}.
\end{align}
By \autoref{eq:fixed point_eq_proof_appendix} and continuity of composition in operator norm, $F$ satisfies
\begin{align}
    \Phi \circ F = F.
\end{align}
Given this we have
\begin{enumerate}
    \item $\Phi\left(Fx\right) = Fx$, thus $\operatorname{ran}\left(F\right) \subseteq \operatorname{Fix}\left(\Phi\right)$.
    \item If $y \in \operatorname{Fix}\left(\Phi\right)$, then $E_Ny = y$ for all $N \in \mathbb{N}$, hence $F y = y$ by $E_{N_k}\to F$, so $F\vert_{\operatorname{Fix}\left(\Phi\right)} = \operatorname{id}$, implying $\operatorname{Fix}\left(\Phi\right)\subseteq \operatorname{ran}\left(F\right)$.
\end{enumerate}
Together we have $\operatorname{Fix}\left(\Phi\right) = \operatorname{ran}\left(F\right)$. Moreover, since $\lVert E_N\rVert \leq 1$ for all $N$, we obtain $\lVert F\rVert \leq 1$. In particular, from $F|_{\operatorname{Fix}\left(\Phi\right)}=\operatorname{id}$ it follows that $F^2=F$, so $F$ is a projection onto $\operatorname{Fix}\left(\Phi\right)$. By the standard lemma that any projection $P$ on a Hilbert space with $\lVert P\rVert \leq 1$ is the orthogonal projection onto its range, we conclude that $F$ is the orthogonal projection onto $\operatorname{Fix}\left(\Phi\right)$, hence $F=E$. This proves the first statement.

For the second statement assume that $\Phi^*$ has a \emph{full-rank} fixed point $\hat{\rho}>0$.
Then $N_{\hat{\rho}}=\{0\}$ and we can identify $\mathcal{H}_{\hat{\rho}}=\mathcal{B}\left(\mathcal{H}\right)$ as a Hilbert space
with inner product $\langle a,b\rangle_{\hat{\rho}}=\tr\left[\hat{\rho} a^* b\right]$.

By positivity of $\Phi$ we have $\Phi\left(a\right)^*=\Phi\left(a^*\right)$ for all $a\in\mathcal{B}\left(\mathcal{H}\right)$.
Thus, $\operatorname{Fix}\left(\Phi\right)$ is $*$-closed and contains the identity, because $\Phi$ is unital.

Let $a\in\operatorname{Fix}\left(\Phi\right)$. By the Schwarz inequality for unital $2$-positive maps,
\begin{align}
    \Phi\left(a\right)^* \Phi\left(a\right) \leq \Phi\left(a^* a\right).
\end{align}
Since $\Phi\left(a\right)=a$, we obtain $a^* a \leq \Phi\left(a^* a\right)$, hence $X\coloneqq \Phi\left(a^* a\right)-a^* a\geq0$.
Applying $\Phi^*\left(\hat{\rho}\right)=\hat{\rho}$ yields
\begin{align}
    \tr\left[\hat{\rho} X\right]
    = \tr\left[\hat{\rho} \Phi\left(a^* a\right)\right]-\tr\left[\hat{\rho} a^* a\right]
    = \tr\left[\Phi^*\left(\hat{\rho}\right) a^* a\right]-\tr\left[\hat{\rho} a^* a\right]
    = 0.
\end{align}
As $\hat{\rho}>0$ is full rank, the functional $X\mapsto \tr\left[\hat{\rho}X\right]$ is faithful on positive operators, so
$X\geq0$ and $\tr\left[\hat{\rho}X\right]=0$ implies $X=0$. Therefore,
\begin{align}\label{eq:md_eq1}
    \Phi\left(a^* a\right)=a^* a.
\end{align}
Applying the same argument to $a^*\in\operatorname{Fix}\left(\Phi\right)$ gives
\begin{align}\label{eq:md_eq2}
    \Phi\left(a a^*\right)=a a^*.
\end{align}
Equations \autoref{eq:md_eq1} -- \autoref{eq:md_eq2} mean that $a$ lies in the multiplicative domain of $\Phi$ (see e.g.\
\cite[Lem.~9.1]{Petz2008QITQS}), hence for all $x\in\mathcal{B}\left(\mathcal{H}\right)$,
\begin{align}\label{eq:md_property}
    \Phi\left(ax\right)=\Phi\left(a\right)\Phi\left(x\right)=a \Phi\left(x\right),
    \qquad
    \Phi\left(xa\right)=\Phi\left(x\right)\Phi\left(a\right)=\Phi\left(x\right) a.
\end{align}
In particular, if $a,b\in\operatorname{Fix}\left(\Phi\right)$ then $\Phi\left(ab\right)=a\Phi\left(b\right)=ab$, so $ab\in\operatorname{Fix}\left(\Phi\right)$.
Thus $\operatorname{Fix}\left(\Phi\right)$ is closed under multiplication, contains $1$, and is $*$-closed, i.e.\ it is a unital
$*$-subalgebra of $\mathcal{B}\left(\mathcal{H}\right)$.

Furthermore, $E$ is unital and positive as a norm-limit of the unital and positive maps $E_N$. 
Moreover, $E|_{\operatorname{Fix}\left(\Phi\right)}=\operatorname{id}$ and $\operatorname{ran}\left(E\right)=\operatorname{Fix}\left(\Phi\right)$ by the first part.

Let $b\in\operatorname{Fix}\left(\Phi\right)$. From \autoref{eq:md_property} we obtain by iteration that for all $t\in\mathbb{N}_0$ and all $x$,
\begin{align}
    \Phi^{\circ t}\left(bx\right)=b \Phi^{\circ t}\left(x\right),
    \qquad
    \Phi^{\circ t}\left(xb\right)=\Phi^{\circ t}\left(x\right) b.
\end{align}
Averaging these identities gives for all $N$,
\begin{align}
    E_N\left(bx\right)=b E_N\left(x\right),
    \qquad
    E_N\left(xb\right)=E_N\left(x\right) b.
\end{align}
Passing to the limit $N\to\infty$ yields
\begin{align}
    E\left(bx\right)=b E\left(x\right),
    \qquad
    E\left(xb\right)=E\left(x\right) b.
\end{align}
Therefore, $E$ satisfies the bimodule property $E\left(b x c\right) = b\, E\left(x\right)\, c$ for all $b,c\in\operatorname{Fix}\left(\Phi\right)$ and all $x\in\mathcal{B}\left(\mathcal{H}\right)$, i.e.\ $E$ is a conditional expectation onto the unital $*$-subalgebra $\operatorname{Fix}\left(\Phi\right)$. By Tomiyama's theorem, $E$ is then already completely positive (cf. e.g. \cite[Sec.~9.1]{Petz2008QITQS}).
\end{proof}

\begin{lemma}\label{lemma:conditioning_mutual_information}
Let $\rho_{AZ_1^n} \in \mathcal{S}\left(\mathcal{H}_A \otimes \mathcal{H}_Z^{\otimes n}\right)$ be a classical-quantum state with $Z_1^n$-systems classical. Then we have
\begin{align}\label{eq:cor_condtiional_mutual_information} 
    I\left(A:Z_{m+1}\vert Z_1^m\right)_{\rho_{AZ_1^n}} = \sum_{z_1^m} p\left(z_1^m\right) I\left(A:Z_{m+1}\right)_{\rho_{AZ_{m+1}\vert z_1^m}}.
\end{align}
\end{lemma}

\begin{proof}
Since the $Z$-systems are classical, we can write $\rho_{AZ_1^{m+1}} = \sum_{z_1^{m+1}} p\left(z_1^{m+1}\right) \rho_{A\vert z_1^{m+1}} \otimes \vert z_1^{m+1}\rangle\langle z_1^{m+1}\vert$. For such classical-quantum states, conditioning on the classical register $Z_1^m$ averages block-diagonally, i.e.\ for every choice of subsystems $S \in \lbrace A,\, Z_{m+1},\, AZ_{m+1}\rbrace$ we have
\begin{align}
    H\left(S\vert Z_1^m\right)_{\rho} = \sum_{z_1^m} p\left(z_1^m\right) H\left(S\right)_{\rho_{AZ_{m+1}\vert z_1^m}},
\end{align}
cf.\ e.g.\ \cite{Wilde2016}. Substituting this identity term by term into
\begin{align}
    I\left(A:Z_{m+1}\vert Z_1^m\right)_{\rho} = H\left(A\vert Z_1^m\right)_{\rho} + H\left(Z_{m+1}\vert Z_1^m\right)_{\rho} - H\left(AZ_{m+1}\vert Z_1^m\right)_{\rho}
\end{align}
yields
\begin{align}
\begin{split}
     I\left(A:Z_{m+1}\vert Z_1^m\right)_{\rho} &= \sum_{z_1^m} p\left(z_1^m\right) \left[H\left(A\right) + H\left(Z_{m+1}\right) - H\left(AZ_{m+1}\right)\right]_{\rho_{AZ_{m+1}\vert z_1^m}} \\
     &= \sum_{z_1^m} p\left(z_1^m\right) I\left(A:Z_{m+1}\right)_{\rho_{AZ_{m+1}\vert z_1^m}}.
\end{split}
\end{align}
\end{proof}

\subsection{Calculation of the adjoint of conditional expectations}\label{appendix:adjoints_conditional_expectations}

\begin{lemma}[\cite{Wolf2012QuantumChannelsGuidedTour}]\label{lem:calculation_adjoint_appendix}
      Let $\mathcal{B} \subseteq \mathcal{A} \cong M_n\left(\mathbb{C}\right)$ be a unital $*$-subalgebra, $\{V_i:\mathcal{H}\to \mathcal{H}_{i,1}\otimes\mathcal{H}_{i,2}\}_{i=1}^k$ a family of mutually orthogonal isometries and $\rho_i \in \mathcal{B}\left(\mathcal{H}_{i,2}\right)$ density operators. Then the Hilbert-Schmidt dual of 
    \begin{align}
        E\left(a\right) = \sum_{i=1}^k V_i^\dagger  \left( \tr_{i,2}  \left[V_i aV_i^\dagger \ 1_{d_i}\otimes \rho_i \right]\otimes 1_{m_i} \right)V_i, \quad a \in \mathcal{A},
    \end{align}
    is given by
    \begin{align}
        E^*\left(\omega\right) = \sum_{i=1}^k V_i^\dagger \left( \tr_{i,2}\left[V_i \omega V_i^\dagger\right]\otimes \rho_i \right)V_i, \quad \omega \in \mathcal{S}\left(\mathcal{H}\right).
    \end{align}
\end{lemma}
\begin{proof}
    We use that $\{V_i:\mathcal{H}\to \mathcal{H}_{i,1}\otimes\mathcal{H}_{i,2}\}_{i=1}^k$ are mutually orthogonal isometries, i.e. 
    \begin{align}
        V_i V_j^\dagger = \delta_{ij} 1_{\mathcal{H}_{i,1} \otimes \mathcal{H}_{i,2}}, \quad \sum_{i=1}^k V_i^\dagger V_i  = 1_{\mathcal{H}}. 
    \end{align}
    Fix $\omega \in \mathcal{S}\left(\mathcal{H}\right)$ and $a \in \mathcal{B}\left(\mathcal{H}\right)$. Then we have
    \begin{align}
        X_i \coloneqq V_i \omega V_i^\dagger \in \mathcal{S}\left(\mathcal{H}_{i,1} \otimes \mathcal{H}_{i,2}\right), \quad Y_i \coloneqq V_i a V_i^\dagger \in \mathcal{B}\left(\mathcal{H}_{i,1} \otimes \mathcal{H}_{i,2}\right).  
    \end{align}
    Then the dual $E^*: \mathcal{S}\left(\mathcal{H}\right) \to \mathcal{S}\left(\mathcal{H}\right)$ of $E$ is defined by
    \begin{align}
        \tr\left[E^*\left(\omega\right) a\right] = \tr\left[\omega E\left(a\right)\right], \quad \omega \in \mathcal{S}\left(\mathcal{H}\right), \quad a \in \mathcal{B}\left(\mathcal{H}\right).
    \end{align}
    Now we use \autoref{prop:structure_prop_cond_expectation_mmwolf} and the cyclicity of the trace
    \begin{equation}\label{eq:calculation_lem_adjoint1}
    \begin{aligned}
        \tr\left[\omega E\left(a\right)\right] &= \sum_{i=1}^k \tr \left[\omega V_i^\dagger  \left( \tr_{i,2}  \left[V_i aV_i^\dagger \ 1_{d_i}\otimes \rho_i \right]\otimes 1_{m_i} \right)V_i \right]\\
        &=\sum_{i=1}^k \tr \left[\omega V_i^\dagger  \left(\tr_{i,2}\left[Y_i \left(1_{d_i}\otimes \rho_i\right)\right]\otimes 1_{m_i} \right)V_i  \right] \\
        &=\sum_{i=1}^k \tr  \left[V_i \omega V_i^\dagger  \left(\tr_{i,2}\left[Y_i \left(1_{d_i}\otimes \rho_i\right)\right]\otimes 1_{m_i} \right) \right] \\
        &=\sum_{i=1}^k \tr \left[\tr_{i,2}\left[X_i\right] \ \tr_{i,2}\left[Y_i \left(1_{d_i}\otimes \rho_i\right)\right] \right].
    \end{aligned}
    \end{equation}
    Now we use that the adjoint of the map $F_i:Z \mapsto \tr_{i,2}\left[Z\left(1\otimes \rho_i\right)\right]$ is given by $T \mapsto T\otimes \rho_i$ , i.e.\ for all $T \in \mathcal{B}\left(\mathcal{H}_{i,1}\right)$ and all $Z \in \mathcal{B}\left(\mathcal{H}_{i,1} \otimes \mathcal{H}_{i,2}\right)$ we have
    \begin{align}\label{eq:calculation_lem_adjoint2}
        \tr \left[T\tr_{i,2}\left[Z\left(1_{d_i}\otimes \rho_i\right)\right] \right] = \tr\left[\left(T\otimes \rho_i\right) Z\right].
    \end{align}
    Substituting \autoref{eq:calculation_lem_adjoint2} into \autoref{eq:calculation_lem_adjoint1} yields
    \begin{align}
        \tr\left[\omega E\left(a\right)\right] = \sum_{i=1}^k \tr \left[\left(\tr_{i,2}\left[X_i\right] \otimes \rho_i\right)\ Y_i\right].
    \end{align}
    Substituting $X_i$ and $Y_i$ back and using again cyclicity yields
    \begin{align}
        E^*\left(\omega\right) = \sum_{i=1}^k V_i^\dagger \left( \tr_{i,2}\left[V_i \omega V_i^\dagger\right]\otimes \rho_i \right)V_i, \quad \omega \in \mathcal{S}\left(\mathcal{H}\right).
    \end{align}
\end{proof}

The following lemma is a standard result and can be deduced from e.g. \cite{Cover2006} in classical information theory or \cite{Wilde2016} for quantum information theory.
\begin{lemma}[Double chain rule for $I\left(Y_1^n:Z_1^n\right)$]\label{lem:double_chain_rule_MI}
Let $\rho_{A_1^n B_1^n}\in \mathcal{S}\left(\mathcal{H}_{A}^{\otimes n} \otimes \mathcal{H}_B^{\otimes n}\right)$ be a quantum state and let
$\mathcal{M}_A:\mathcal{S}\left(\mathcal{H}_A\right) \to \mathcal{M}_1\left(\mathcal{Y}\right)$ and $\mathcal{M}_B:\mathcal{S}\left(\mathcal{H}_B\right)\to \mathcal{M}_1\left(\mathcal{Z}\right)$ be measurement channels with classical output registers.
Define the output state
\begin{align}
\omega_{Y_1^n Z_1^n}
  \coloneqq   
 \left(\mathcal{M}_A^{\otimes n}\otimes \mathcal{M}_B^{\otimes n}\right)  \left(\rho_{A_1^nB_1^n}\right).
\end{align}
Then
\begin{align}
I\left(Y_1^n:Z_1^n\right)_\omega
=
\sum_{m_A=0}^{n-1}\sum_{m_B=0}^{n-1}
I  \left(Y_{m_A+1}:Z_{m_B+1}  | Y_1^{m_A}Z_1^{m_B}\right)_\omega,
\end{align}
with the convention $Y_1^0$ respectively $Z_1^0$ are the trivial registers.
\end{lemma}

\begin{proof}
For any state $\sigma_{XUVW}$, the conditional mutual information admits
the representation
\begin{align}
I\left(X:U\vert W\right)_\sigma = H\left(X\vert W\right)_\sigma - H\left(X\vert UW\right)_\sigma,
\end{align}
hence
\begin{align}
I\left(X:UV\vert W\right)_\sigma
&= H\left(X\vert W\right)_\sigma - H\left(X\vert UVW\right)_\sigma \notag\\
&= \left(H\left(X\vert W\right)_\sigma - H\left(X\vert UW\right)_\sigma\right)
  + \left(H\left(X\vert UW\right)_\sigma - H\left(X\vert UVW\right)_\sigma\right)\notag\\
&= I\left(X:U\vert W\right)_\sigma + I\left(X:V\vert UW\right)_\sigma. \label{eq:chain_two_blocks}
\end{align}
We will use \autoref{eq:chain_two_blocks} repeatedly.

Using $I\left(XY:Z\right)=I\left(X:Z\right)+I\left(Y:Z\vert X\right)$ iteratively, we obtain
\begin{align}
I\left(Y_1^n:Z_1^n\right)_\omega= \sum_{i=1}^n I  \left(Y_i:Z_1^n \vert Y_1^{i-1}\right)_\omega. \label{eq:chain_Y}
\end{align}

Fix $i\in\{1,\dots,n\}$ and apply \autoref{eq:chain_two_blocks} iteratively to the triple
$\left(X,U,V,W\right) = \left(Y_i,  Z_j,  Z_{j+1}^n,  Y_1^{i-1}Z_1^{j-1}\right)$.
This yields, for each $j$,
\begin{align}\label{eqn:help_mutual_info_relation}
I  \left(Y_i:Z_j^n \vert Y_1^{i-1}Z_1^{j-1}\right)_\omega=I  \left(Y_i:Z_j \vert Y_1^{i-1}Z_1^{j-1}\right)_\omega+I  \left(Y_i:Z_{j+1}^n \vert Y_1^{i-1}Z_1^{j}\right)_\omega.
\end{align}
Telescoping from $j=1$ to $n$ together with \autoref{eqn:help_mutual_info_relation} gives
\begin{align}
I  \left(Y_i:Z_1^n \vert Y_1^{i-1}\right)_\omega 
&= I  \left(Y_i:Z_1^n \vert Y_1^{i-1}\right)_\omega + \sum_{j=2}^n I  \left(Y_i:Z_j^n \vert Y_1^{i-1}Z_1^{j-1}\right)_\omega - \sum_{j=1}^{n-1} I  \left(Y_i:Z_{j+1}^n \vert Y_1^{i-1}Z_1^{j}\right)_\omega \\
&=\sum_{j=1}^n I  \left(Y_i:Z_j \vert Y_1^{i-1}Z_1^{j-1}\right)_\omega. \label{eq:chain_Z_given_Y}
\end{align}

Plugging \autoref{eq:chain_Z_given_Y} into \autoref{eq:chain_Y} yields
\begin{align}
I\left(Y_1^n:Z_1^n\right)_\omega=\sum_{i=1}^n\sum_{j=1}^n I  \left(Y_i:Z_j \vert Y_1^{i-1}Z_1^{j-1}\right)_\omega.
\end{align}
Finally, set $m_A=i-1$ and $m_B=j-1$ to obtain exactly the claimed form
\begin{align}
I\left(Y_1^n:Z_1^n\right)_\omega=\sum_{m_A=0}^{n-1}\sum_{m_B=0}^{n-1}
I \left(Y_{m_A+1}:Z_{m_B+1}\vert Y_1^{m_A}Z_1^{m_B}\right)_\omega. 
\end{align}
\end{proof}

\section{Young diagrams and tableaux}\label{sec:Young_diagrams_tableaux}

In this appendix we collect the combinatorial objects underlying the representation theory of the symmetric and unitary groups, as used in the main text and in \autoref{sec:weyls_dimension_formula} and \autoref{sec:subgroup_adapted_basis}. These notions are standard, and no originality is claimed; we refer to \cite{sagan2013symmetric,Fulton2004,james2006representation} for textbook treatments and to \cite{grinko2025mixed} for a recent review geared towards quantum information theory. Recall from \autoref{sec:notation_preliminaries} that $\lambda \vdash n$ denotes a partition of $n$ with length $\ell\left(\lambda\right)$, that $\lambda \vdash_d n$ additionally requires $\ell\left(\lambda\right)\leq d$, and that $\mu \vDash n$ denotes a weak composition.

\paragraph{\textbf{Young diagrams.}}
The Young diagram (also Ferrers diagram) of a partition $\lambda \vdash n$ is the left-justified array of $n$ cells (boxes, nodes) containing $\lambda_i$ cells in its $i$-th row. We use the English convention, in which rows are indexed from top to bottom and columns from left to right, such that a cell is a pair $x = \left(i,j\right)$ with row index $i \in \left[\ell\left(\lambda\right)\right]$ and column index $j \in \left[\lambda_i\right]$. We identify a partition with its diagram and write $\left(i,j\right)\in \lambda$; the size of $\lambda$ is $\lvert \lambda \rvert \coloneqq \sum_i \lambda_i = n$, and $\varnothing$ denotes the empty partition, i.e.\ the unique partition of $0$. A partition $\mu$ is a subpartition of $\lambda$, denoted $\mu \subseteq \lambda$, if $\mu_i \leq \lambda_i$ for all $i$, i.e.\ if the diagram of $\mu$ is contained in the diagram of $\lambda$. The conjugate (or transpose) partition $\lambda^\prime$ is obtained by reflecting the diagram along its main diagonal; equivalently, $\lambda^\prime_j \coloneqq \lvert \lbrace i \ \vert \ \lambda_i \geq j\rbrace\rvert$ counts the cells in the $j$-th column of $\lambda$. Conjugation is an involution, $\left(\lambda^\prime\right)^\prime = \lambda$, and $\ell\left(\lambda\right) = \lambda_1^\prime$.

\paragraph{\textbf{Contents, hooks and axial distances.}}
The content of a cell $x = \left(i,j\right)$ is
\begin{align}
    \operatorname{cont}\left(x\right) \coloneqq j - i,
\end{align}
which is constant along the diagonals of the diagram. The hook of a cell $\left(i,j\right)\in\lambda$ consists of the cell itself, all cells to its right in row $i$ (the arm, of length $a\left(i,j\right)\coloneqq \lambda_i - j$) and all cells below it in column $j$ (the leg, of length $l\left(i,j\right)\coloneqq \lambda_j^\prime - i$); the hook length of $\left(i,j\right)\in \lambda$ is
\begin{align}
    h\left(i,j\right)\coloneqq a\left(i,j\right) + l\left(i,j\right) + 1 = \lambda_i + \lambda_j^\prime - i - j + 1.
\end{align}
Hook lengths determine the dimensions of Specht modules via the hook length formula, \autoref{eq:hook_length_formula}. The axial distance from a cell $x = \left(i,j\right)$ to a cell $y = \left(i^\prime, j^\prime\right)$ is
\begin{align}
    \operatorname{ax}\left(x,y\right)\coloneqq \operatorname{cont}\left(y\right) - \operatorname{cont}\left(x\right) = \left(j^\prime - j\right) - \left(i^\prime - i\right),
\end{align}
i.e.\ the signed number of steps of any Manhattan path from $x$ to $y$, where every step to the right or upwards contributes $+1$ and every step to the left or downwards contributes $-1$. Axial distances govern the action of adjacent transpositions on Young's orthogonal basis of the Specht modules (cf.\ \autoref{sec:subgroup_adapted_basis}). As an example, consider $\lambda = \left(4,2,1\right)\vdash 7$ with conjugate $\lambda^\prime = \left(3,2,1,1\right)$: the diagrams of $\lambda$ and $\lambda^\prime$, as well as the hook lengths and the contents of the cells of $\lambda$, are
\begin{align}\label{eq:example_young_diagrams_appendix}
    {\ytableausetup{boxsize=1.5em}
    \ydiagram{4,2,1}\ ,
    \qquad
    \ydiagram{3,2,1,1}\ ,
    \qquad
    \begin{ytableau}
        6 & 4 & 2 & 1 \\
        3 & 1 \\
        1
    \end{ytableau}\ ,
    \qquad
    \begin{ytableau}
        0 & 1 & 2 & 3 \\
        -1 & 0 \\
        -2
    \end{ytableau}}\ ,
\end{align}
respectively.

\paragraph{\textbf{Addable and removable cells.}}
A cell $\left(i,j\right)\in \lambda$ is called removable (or an inner corner) if deleting it from the diagram leaves the diagram of a partition of $n-1$, i.e.\ if $j = \lambda_i$ and $\lambda_{i+1}<\lambda_i$; a cell $\left(i,j\right)\notin \lambda$ is called addable (or an outer corner) if adjoining it to the diagram yields the diagram of a partition of $n+1$. We denote by $\lambda^-$ (respectively $\lambda^+$) the set of partitions obtained from $\lambda$ by removing a removable (adding an addable) cell, such that
\begin{align}
    \mu \in \lambda^- \quad \Longleftrightarrow \quad \lambda \in \mu^+.
\end{align}
The number of removable cells of $\lambda$ equals the number of distinct part sizes of $\lambda$, and the number of addable cells exceeds it by one; in particular, $\lvert \lambda^-\rvert \leq \ell\left(\lambda\right)$ and, since addable and removable cells lie in pairwise different rows and columns and hence on pairwise different diagonals, their contents are pairwise distinct. If $\ell\left(\lambda\right)\leq d$, then at most $d$ elements of $\lambda^+$ have length at most $d$. For the example in \autoref{eq:example_young_diagrams_appendix}, the removable cells of $\lambda = \left(4,2,1\right)$ are $\left(1,4\right), \left(2,2\right), \left(3,1\right)$ and the addable cells are $\left(1,5\right), \left(2,3\right), \left(3,2\right), \left(4,1\right)$, such that
\begin{align}
    \lambda^- = \lbrace \left(3,2,1\right), \left(4,1,1\right), \left(4,2\right)\rbrace, \qquad \lambda^+ = \lbrace \left(5,2,1\right), \left(4,3,1\right), \left(4,2,2\right), \left(4,2,1,1\right)\rbrace.
\end{align}
These single-cell modifications encode the multiplicity-free branching rules for Specht and Weyl modules (Young's branching rule and Pieri's rule) used in the main text in \autoref{sec:efficient_inner_sequence} and in \autoref{sec:subgroup_adapted_basis}.

\paragraph{\textbf{Young tableaux.}}
A Young tableau of shape $\lambda$ is a filling $T$ of the cells of $\lambda$ with positive integers; we write $T\left(x\right)$ for the entry of the cell $x \in \lambda$ and $\operatorname{sh}\left(T\right) = \lambda$ for the shape. A tableau is called semistandard if its entries weakly increase along each row (from left to right) and strictly increase down each column, and we denote by $\operatorname{SSYT}\left(\lambda, d\right)$ the set of semistandard Young tableaux of shape $\lambda$ with entries in $\left[d\right]$. Since the entries of the first column must increase strictly, $\operatorname{SSYT}\left(\lambda,d\right)\neq \emptyset$ if and only if $\ell\left(\lambda\right)\leq d$, consistent with the absence of partitions with more than $d$ rows in the Schur--Weyl decomposition (cf.\ \autoref{sec:weyls_dimension_formula}). The weight of $T \in \operatorname{SSYT}\left(\lambda,d\right)$ is the vector $w\left(T\right)\in \mathbb{N}_0^d$ with entries
\begin{align}
    w\left(T\right)_i \coloneqq \lvert \lbrace x \in \lambda \ \vert \ T\left(x\right) = i \rbrace \rvert,
\end{align}
i.e.\ a weak composition $w\left(T\right)\vDash n$ recording how often each value occurs (cf.\ the type maps in \autoref{sec:notation_preliminaries}). A standard Young tableau of shape $\lambda \vdash n$ is a semistandard tableau in which each of the entries $1,\ldots,n$ occurs exactly once, i.e.\ a tableau of weight $\left(1,\ldots,1\right)$ whose rows and columns increase strictly, and the set of these is denoted by $\operatorname{SYT}\left(\lambda\right)$. The cardinalities of both sets carry representation-theoretic meaning: for $\lambda \vdash_d n$,
\begin{align}\label{eq:tableaux_count_dimensions}
    \lvert \operatorname{SYT}\left(\lambda\right)\rvert = d_\lambda = \dim V_\lambda, \qquad \lvert \operatorname{SSYT}\left(\lambda, d\right)\rvert = m_\lambda = \dim U_\lambda,
\end{align}
given explicitly by the hook length formula \autoref{eq:hook_length_formula} and Weyl's dimension formula \autoref{eq:weyl_dimension_formula}, respectively (cf.\ \cite[Thm.~2.6.5]{sagan2013symmetric} and \cite[Thm.~6.3]{Fulton2004}). For instance,
\begin{align}
    {\ytableausetup{boxsize=1.5em}
    \begin{ytableau}
        1 & 2 & 5 & 7\\
        3 & 4 \\
        6
    \end{ytableau} \in \operatorname{SYT}\left(\left(4,2,1\right)\right),
    \qquad
    \begin{ytableau}
        1 & 1 & 2 & 3\\
        2 & 3 \\
        3
    \end{ytableau} \in \operatorname{SSYT}\left(\left(4,2,1\right), 3\right)},
\end{align}
where the semistandard tableau has weight $\left(2,2,3\right)\vDash 7$.

\paragraph{\textbf{Growth sequences and contents of standard tableaux.}}
For $T \in \operatorname{SYT}\left(\lambda\right)$ with $\lambda \vdash n$ and $k \in \lbrace 0,1,\ldots,n\rbrace$, let $T^{\left(k\right)}$ denote the diagram formed by the cells of $T$ with entries at most $k$. Then
\begin{align}\label{eq:SYT_growth_sequence}
    \varnothing = T^{\left(0\right)} \subseteq T^{\left(1\right)}\subseteq \cdots \subseteq T^{\left(n\right)} = \lambda, \qquad T^{\left(k\right)} \in \left(T^{\left(k-1\right)}\right)^+ \quad \text{for all} \ k \in \left[n\right],
\end{align}
i.e.\ every standard tableau describes a growth sequence of diagrams in which exactly one addable cell is adjoined at each step; conversely, every such chain of partitions determines a unique standard tableau, its removal history. This identifies $\operatorname{SYT}\left(\lambda\right)$ bijectively with the set of paths from $\varnothing$ to $\lambda$ in the Young lattice, i.e.\ in the Bratteli diagram of the tower $\mathbb{C}\left[S_0\right]\hookrightarrow \mathbb{C}\left[S_1\right]\hookrightarrow \cdots \hookrightarrow \mathbb{C}\left[S_n\right]$ (cf.\ \autoref{fig:Bratelli_S_n} and \autoref{sec:subgroup_adapted_basis}). For $k \in \left[n\right]$ we write $\operatorname{cont}_k\left(T\right)$ for the content of the cell of $T$ containing the entry $k$, i.e.\ of the unique cell of $T^{\left(k\right)}\setminus T^{\left(k-1\right)}$, and call
\begin{align}
    c_T \coloneqq \left(\operatorname{cont}_1\left(T\right), \operatorname{cont}_2\left(T\right),\ldots, \operatorname{cont}_n\left(T\right)\right)
\end{align}
the content vector of $T$. Since the addable cells of any diagram have pairwise distinct contents, the content vector determines $T$ uniquely. For the standard tableau displayed above, the growth sequence is $\varnothing \subseteq \left(1\right) \subseteq \left(2\right)\subseteq \left(2,1\right)\subseteq \left(2,2\right)\subseteq \left(3,2\right)\subseteq \left(3,2,1\right)\subseteq \left(4,2,1\right)$ with content vector $c_T = \left(0,1,-1,0,2,-2,3\right)$. Content vectors are precisely the joint spectra of the Jucys--Murphy elements on the Gelfand--Tsetlin basis, cf.\ \autoref{sec:subgroup_adapted_basis}.

\section{Dimension formulas and polynomial bounds for Schur--Weyl blocks}\label{sec:weyls_dimension_formula}

In this appendix we show bounds on the size of the multiplicity spaces and the number of irreducible representations occurring in a decomposition of the action 
\begin{align}
     S_n \to \operatorname{GL}\left( \mathcal{H}^{\otimes n}\right), \quad \sigma \mapsto \left[v_1\otimes v_2\otimes \cdots \otimes v_n \mapsto v_{\sigma^{-1}\left(1\right)}\otimes v_{\sigma^{-1}\left(2\right)}\otimes \cdots \otimes v_{\sigma^{-1}\left(n\right)}\right]. 
\end{align}
The irreducible representations of $S_n$ can be identified with partitions of $n$, denoted as $\lambda \vdash n$. Furthermore let $d \coloneqq \dim \mathcal{H}$. Then \cite[Thm.\ 6.3.\ (1)]{Fulton2004} states that the decomposition of $\mathcal{H}^{\otimes n}$ does not admit any contribution of a partition $\lambda_{1}\geq \lambda_2 \geq \ldots \geq \lambda_{d}\geq \lambda_{d+1} \geq 0$, if $\lambda_{d+1}\neq 0$. Thus we have
\begin{align}
    \mathcal{H}^{\otimes n} \cong  \bigoplus_{\lambda \vdash n, \ \lambda_{d+1}=0} \mathbb{C}^{m_{\lambda}} \otimes \mathbb{C}^{d_\lambda},
\end{align}
where $\mathbb{C}^{m_{\lambda}}$ corresponds to the multiplicity space of the irreducible representation $\lambda\vdash n$, whereby $\mathbb{C}^{d_\lambda}$ corresponds to the space isomorphic to the irreducible representation $\lambda \vdash n$, i.e.\ the action of $S_n$ on $\mathbb{C}^{d_\lambda}$ has just the trivial commutator. The dimension $d_{\lambda}$ of Specht modules index by $\lambda$ can be computed via the Hook-length formula (cf.\ \cite[Thm.\ 3.10.2]{Fulton2004})
\begin{align}\label{eq:hook_length_formula}
    d_{\lambda}=\frac{n!}{\prod_{(i,j)\in \lambda} h(i,j)}.
\end{align}
Here $h\left(i,j\right)$ denotes the hook length of the cell $\left(i,j\right)\in \lambda$, as introduced in \autoref{sec:Young_diagrams_tableaux}. The dimension of the multiplicity spaces is characterized by Weyl's dimension formula given by (Hook content formula, cf.\ e.g.\ \cite[Thm.\ 6.3\ (2)]{Fulton2004})
\begin{align}\label{eq:weyl_dimension_formula}
    m_\lambda   =  \prod_{1\leq i<j\leq d} \frac{\lambda_i - \lambda_j +j-i}{j-i}.
\end{align}
\begin{lemma}\label{lem:bounds_symmetric_group}
    Given $d< \infty $ then we have
    \begin{enumerate}
        \item $\#\{\lambda \vdash n \ \vert \ \lambda_{d+1}= 0\} \leq \left(n+d\right)^{d-1}$
        \item $m_\lambda \leq \left(n+d\right)^{d\left(d-1\right)/2}$ for all $\lambda \vdash n$ such that $\lambda_{d+1}=0$
        \item $m_{\left(n\right)} \leq \left(n+d\right)^{d-1}$ for the symmetric subspace corresponding to the partition $\left(n\right) \vdash n$.
    \end{enumerate}
\end{lemma}
\begin{proof}
    \begin{enumerate}
        \item Dropping the assumption that there needs to be an order in $\left(\lambda_1,\lambda_2,\ldots,\lambda_d\right)$, yields that the number of $d$-tuples of non-negative integers summing up to $n$ is given by $\binom{n+d-1}{d-1}$, which can be bounded with $\binom{n}{k}\leq n^k$ to 
        \begin{align}
            \#\{\lambda \vdash n \ \vert \ \lambda_{d+1}= 0\} \leq \binom{n+d-1}{d-1} \leq \left(n+d-1\right)^{d-1} \leq \left(n+d\right)^{d-1}.
        \end{align}
        \item Since $\lambda_i \leq n$ and $\lambda_j\geq 0$, we have $\lambda_i-\lambda_j\leq n$ and $1\leq j-i\leq d$, such that we conclude with \autoref{eq:weyl_dimension_formula}
        \begin{align}
            m_{\lambda} \leq \prod_{1\leq i < j \leq d} \left(n+d-1\right),
        \end{align}
        but there are $\binom{d}{2}$ factors in the product, such that we have
        \begin{align}
            m_{\lambda} \leq \left(n+d\right)^{d\left(d-1\right)/2}.
        \end{align}
        \item This is a well-known formula for the symmetric subspace, see e.g.\ \cite{harrow2013churchsymmetricsubspace}, such that we have
        \begin{align}
            m_{\left(n\right)} \leq \binom{n+d-1}{d-1}\leq \left(n+d\right)^{d-1}.
        \end{align}
    \end{enumerate}
\end{proof}

\section{Gelfand-Tsetlin bases and Schur-Weyl duality}\label{sec:subgroup_adapted_basis}

In this appendix we assemble the representation-theoretic toolbox underlying the algorithmic results of \autoref{sec:efficient_inner_sequence}: the Gelfand-Tsetlin (or Gel'fand-Zetlin) construction \cite{gelfand1950matrix} of canonical orthonormal bases for multiplicity-free families of algebras (\autoref{subsec:gelfand_tsetlin_basis}), and its specialization to Schur-Weyl duality (\autoref{subsec:schur_weyl_orthonormal_basis}), including the Schur basis, the branching rules for Specht and Weyl modules, the Clebsch-Gordan series and transform, and the computational complexity of all objects involved. See e.g.\ \cite{bacon2005quantumschurtransformi, Bacon2006, grinko2023gelfandtsetlinbasispartiallytransposed, Okounkov1996, Vershik2005, Doty2019} for further information and \cite{harrow2005applicationscoherentclassicalcommunication, grinko2025mixed} for detailed overviews with reference to quantum information theory. We assume familiarity with the basic concepts of representation theory and claim no novelty for the contents of this appendix.

Before turning to the technical aspects, we briefly address the generalization of representation theory from groups to associative algebras, as both perspectives are prevalent in the literature and relevant to the considerations of this work; see e.g.\ \cite{cox2008representation, dlab1994finite} for details. By Maschke's theorem\footnote{$\operatorname{char}\left(\mathbb{C}\right)=0$ does not divide $\lvert G\rvert$.}, every finite-dimensional representation of a finite group $G$ over $\mathbb{C}$ is completely reducible. For general associative algebras $\Al$ over $\mathbb{C}$ this no longer holds: there can exist finite-dimensional $\Al$-modules\footnote{In the setting of associative algebras, one usually adopts the terminology of modules rather than that of representations.} that have submodules but no invariant complementary submodule,\footnote{Nilpotent elements in the Jacobson radical of the algebra yield reducible but non-split modules.} i.e.\ modules over associative algebras may be reducible yet indecomposable. The Krull--Schmidt theorem still ensures that every finite-dimensional module decomposes uniquely, up to isomorphism and reordering, into a direct sum of indecomposable modules, but complete reducibility requires additional structure: an algebra $\Al$ is semisimple, which can be characterized by the vanishing of its Jacobson radical, if and only if every finite-dimensional $\Al$-module is completely reducible. With the applications from finite-dimensional quantum information theory in mind, we restrict throughout to semisimple algebras; this covers all cases relevant to this work, since the group algebra $\mathbb{C}\left[G\right]$ of a finite group is semisimple and so is every finite-dimensional matrix $*$-algebra. By the Wedderburn--Artin theorem, a finite-dimensional semisimple algebra over $\mathbb{C}$ is isomorphic to a direct sum $\bigoplus_{i=1}^k \operatorname{End}\left(\mathbb{C}^{n_i}\right)$ of full matrix algebras for some integers $k, n_1,\ldots,n_k \geq 1$ \cite{cox2008representation}. Throughout, we use both terminologies, referring to representations in the group setting and to modules in the context of associative algebras, and we assume all groups to be either finite or compact Lie groups.

Concerning terminology, we reserve the term \emph{Gelfand-Tsetlin basis} for the canonical basis attached to a multiplicity-free family of algebras as constructed in \autoref{subsec:gelfand_tsetlin_basis}. The orthonormal basis of $\mathcal{H}^{\otimes n}$ obtained by applying this construction to both sides of Schur-Weyl duality is called the \emph{Schur basis}, cf.\ \autoref{subsubsec:CG_series_transform}.

\subsection{Gelfand-Tsetlin basis}\label{subsec:gelfand_tsetlin_basis}

The goal of this subsection is a canonical construction of orthonormal bases for all simple modules of a semisimple algebra, involving as few arbitrary choices as possible. Recall that an element $\epsilon$ of an algebra $\Al$ is an idempotent if $\epsilon^2 = \epsilon$, that two idempotents $a,b \in \Al$ are orthogonal if $ab = ba = 0$, that an idempotent is central if it commutes with every element of $\Al$, and that an idempotent (respectively, a central idempotent) is primitive if it cannot be written as a sum of two non-zero orthogonal idempotents (central idempotents). For a finite-dimensional semisimple algebra $\Al$ over $\mathbb{C}$ we denote by $\hat{\Al}$ the set of isomorphism classes of simple $\Al$-modules and by $V_\lambda$ the simple module corresponding to $\lambda \in \hat{\Al}$. The primitive central idempotents of $\Al$ are in one-to-one correspondence with $\hat{\Al}$, and the Wedderburn--Artin isomorphism refines to
\begin{align}\label{eq:wedderburn_idempotent_decomposition}
    \Al = \bigoplus_{\lambda \in \hat{\Al}} \epsilon_\lambda \Al \cong \bigoplus_{\lambda \in \hat{\Al}} \operatorname{End}\left(V_\lambda\right), \qquad \sum_{\lambda \in \hat{\Al}} \epsilon_\lambda = 1,
\end{align}
where the left-hand side is understood as a decomposition of the left-regular representation of $\Al$ and the right-hand side expresses the resolution of the unit element by the primitive central idempotents $\epsilon_\lambda$ (cf.\ \cite{cox2008representation} and \cite[Sec.~2.7.2]{grinko2025mixed}). The idempotent $\epsilon_\lambda$ projects onto the isotypic component of $V_\lambda$; our aim is to canonically refine each $\epsilon_\lambda$ into a sum of pairwise orthogonal primitive idempotents which are of rank one in the block $\operatorname{End}\left(V_\lambda\right)$ and thereby single out a distinguished basis of $V_\lambda$. This is possible whenever $\Al = \Al_n$ is the top of a tower of algebras with multiplicity-free restrictions: the approach originates in the Okounkov--Vershik treatment of the symmetric group algebras \cite{Okounkov1996, Vershik2005} and was formulated for general semisimple algebras by Doty, Lauve and Seelinger \cite{Doty2019}; our presentation follows \cite[Sec.~2.7]{grinko2025mixed}.

We first fix the terminology for towers. For algebras $\Al$ and $\Bl$, a map $\psi : \Al \to \Bl$ is an algebra embedding if $\psi$ is an injective algebra homomorphism, abbreviated $\Al \hookrightarrow \Bl$, and it is unity-preserving if $\psi\left(1_{\Al}\right) = 1_{\Bl}$; in this case we identify $\Al$ with the subalgebra $\psi\left(\Al\right)\subseteq \Bl$. For a $\Bl$-module $M$ and a subalgebra $\Al \subseteq \Bl$ we denote by $\operatorname{Res}_{\Al}^{\Bl} M$ the restriction of $M$ to the action of $\Al$. If $G_0 \subseteq G_1 \subseteq \cdots \subseteq G_n$ is a tower of finite groups, the group algebras $\mathbb{C}\left[G_0\right]\hookrightarrow \mathbb{C}\left[G_1\right] \hookrightarrow \cdots \hookrightarrow \mathbb{C}\left[G_n\right]$ provide such a chain of unity-preserving embeddings, and the construction below then reproduces the subgroup-adapted bases of \cite[Chap.~7]{harrow2005applicationscoherentclassicalcommunication}; the algebraic formulation covers, in addition, algebras which are not group algebras, such as the permutation matrix algebras $\Al_n^d$. The recursive structure along the tower is precisely what enables efficient implementations of the Clebsch--Gordan and Schur transforms \cite{bacon2005quantumschurtransformi, Bacon2006, harrow2005applicationscoherentclassicalcommunication}, cf.\ the next subsection.

\begin{definition}[{Multiplicity-free family, \cite[Def.~1.1]{Doty2019}; see also \cite[Def.~2.7.1]{grinko2025mixed}}]\label{def:multiplicity_free_family}
    A family $\Al_0, \Al_1, \ldots, \Al_n$ of finite-dimensional semisimple algebras over $\mathbb{C}$ is called a multiplicity-free family if the following axioms hold:
    \begin{enumerate}[label=(\alph*)]
        \item $\Al_0 \cong \mathbb{C}$,
        \item for each $k \in \left[n\right]$ there is a unity-preserving algebra embedding $\Al_{k-1}\hookrightarrow \Al_k$,
        \item for each $k \in \left[n\right]$, the restriction of any simple $\Al_k$-module to $\Al_{k-1}$ is isomorphic to a direct sum of pairwise non-isomorphic simple $\Al_{k-1}$-modules; the restriction (branching) from $\Al_k$ to $\Al_{k-1}$ is then called multiplicity-free.
    \end{enumerate}
\end{definition}

\begin{example}\label{ex:multiplicity_free_families}
    The canonical example is the family of symmetric group algebras $\mathbb{C}\left[S_0\right]\hookrightarrow \mathbb{C}\left[S_1\right]\hookrightarrow\cdots \hookrightarrow \mathbb{C}\left[S_n\right]$, where $S_{k-1}\subseteq S_k$ acts on $\left[k\right]$ as the permutations fixing $k$; the branching is multiplicity-free by Young's branching rule \cite[Chap.~2.8]{sagan2013symmetric} (cf.\ also \cite{Okounkov1996, Vershik2005}). A second example, central to this work, is the family of permutation matrix algebras $\Al_0^d \hookrightarrow \Al_1^d \hookrightarrow \cdots \hookrightarrow \Al_n^d$ with $\Al_k^d = \psi_k^d\left(\mathbb{C}\left[S_k\right]\right)\subseteq \mathcal{B}\left(\mathcal{H}^{\otimes k}\right)$ and $d = \dim \mathcal{H}$, embedded via the restriction of the unity-preserving map $X \mapsto X \otimes 1_{\mathcal{H}}$, which sends $\psi_k^d\left(\sigma\right)$ to $\psi_{k+1}^d\left(\sigma\right)$ for every $\sigma \in S_k$. Since the simple $\Al_k^d$-modules are exactly the Specht modules $V_\lambda$ with $\lambda \vdash_d k$, the branching is Young's rule truncated to partitions with at most $d$ rows and hence again multiplicity-free.
\end{example}

\subsubsection{Bratteli diagrams}\label{subsubsec:bratteli_diagrams}

To a multiplicity-free family one associates a graded graph which records all branchings along the tower \cite{Bratteli1972}.

\begin{definition}[{Bratteli diagram, cf.~\cite{Bratteli1972} and \cite[Def.~2.7.2]{grinko2025mixed}}]\label{def:bratteli_diagram}
    Let $\Al_0,\ldots,\Al_n$ be a multiplicity-free family of algebras. Its Bratteli diagram $\mathscr{A}$ is the directed acyclic graph whose vertex set is the disjoint union $\bigsqcup_{k=0}^n \hat{\Al}_k$ of the isomorphism classes of simple $\Al_k$-modules, with an edge $\lambda \rightarrow \mu$ from $\lambda \in \hat{\Al}_k$ to $\mu \in \hat{\Al}_{k+1}$ if and only if $V_\lambda$ is isomorphic to a direct summand of $\operatorname{Res}_{\Al_k}^{\Al_{k+1}} V_\mu$. The set $\hat{\Al}_k$ is called the $k$-th level of $\mathscr{A}$, the unique vertex $\varnothing \in \hat{\Al}_0$ is called the root, and the vertices in $\hat{\Al}_n$ are called leaves.
\end{definition}

By multiplicity-freeness the Bratteli diagram is a simple graph; for general families of semisimple algebras one keeps track of the branching multiplicities by assigning them to the edges, resulting in a graded multigraph.\footnote{The permutation matrix algebras of \autoref{ex:multiplicity_free_families} illustrate that a Bratteli diagram may terminate: for $d < n$, some vertices of level $k \leq d$ have no continuation to level $n$ once $\ell\left(\lambda\right)\leq d$ is violated; such vertices simply do not occur.} A (directed) path in $\mathscr{A}$ traverses edges from lower to higher levels, and for vertices $\lambda \in \hat{\Al}_i$ and $\mu \in \hat{\Al}_j$ with $i<j$ we denote by $\operatorname{Paths}_{i,j}\left(\lambda,\mu,\mathscr{A}\right)$ the set of all paths from $\lambda$ to $\mu$. Paths starting at the root are abbreviated by $\operatorname{Paths}_j\left(\lambda, \mathscr{A}\right)\coloneqq \operatorname{Paths}_{0,j}\left(\varnothing, \lambda, \mathscr{A}\right)$, and if $\lambda$ is a leaf we write $\operatorname{Paths}\left(\lambda,\mathscr{A}\right)\coloneqq \operatorname{Paths}_n\left(\lambda,\mathscr{A}\right)$; finally, $\operatorname{Paths}\left(\mathscr{A}\right)$ denotes the set of all root-to-leaf paths in $\mathscr{A}$. Whenever the Bratteli diagram is clear from the context, we drop it from the notation and write $\operatorname{Paths}_{i,j}\left(\lambda,\mu\right)$, $\operatorname{Paths}_j\left(\lambda\right)$ and $\operatorname{Paths}\left(\lambda\right)$. A path is written as its sequence of vertices, $T = T^0 \rightarrow T^1 \rightarrow \cdots \rightarrow T^n \in \operatorname{Paths}\left(\mathscr{A}\right)$, and as a summation index, $\lambda : \lambda \rightarrow \mu$ ranges over all vertices $\lambda$ of the previous level connected to $\mu$, while $\tilde{\mu} : \lambda \rightarrow \tilde{\mu}$ ranges over all vertices $\tilde{\mu}$ of the next level connected to $\lambda$ (cf.\ \cite[Sec.~2.7.1]{grinko2025mixed}). Iterating axiom (c) of \autoref{def:multiplicity_free_family} along the tower shows that dimensions are counted by paths: since $\dim V_\mu = \sum_{\lambda : \lambda \rightarrow \mu} \dim V_\lambda$ and the root module is one-dimensional, an induction over the levels yields
\begin{align}\label{eq:dimension_counts_paths}
    \dim V_\lambda = \lvert \operatorname{Paths}_k\left(\lambda, \mathscr{A}\right)\rvert \qquad \text{for all} \ k \in \lbrace 0,1,\ldots,n\rbrace \ \text{and} \ \lambda \in \hat{\Al}_k.
\end{align}

\begin{example}[Young lattice]\label{ex:young_lattice}
    For the family of symmetric group algebras, $\hat{\mathbb{C}\left[S_k\right]}$ is identified with the set of partitions $\lambda \vdash k$ labeling the Specht modules $V_\lambda$, and by Young's branching rule there is an edge $\mu \rightarrow \lambda$ between the levels $k-1$ and $k$ if and only if $\mu \in \lambda^-$, i.e.\ if $\mu$ is obtained from $\lambda$ by removing a removable cell (cf.\ \autoref{sec:Young_diagrams_tableaux}). The resulting Bratteli diagram is the Young lattice depicted in \autoref{fig:Bratelli_S_n}, and the root-to-$\lambda$ paths are exactly the growth sequences of standard Young tableaux, such that \autoref{eq:dimension_counts_paths} recovers $\dim V_\lambda = \lvert \operatorname{SYT}\left(\lambda\right)\rvert = d_\lambda$ from \autoref{eq:tableaux_count_dimensions}. The Bratteli diagram of the permutation matrix algebras $\left(\Al_k^d\right)_{k}$ is the truncation of the Young lattice to partitions with at most $d$ rows.
\end{example}

\begin{figure}[ht]
\centering
{\color{black}%
\begin{tikzpicture}[
  x=1cm,y=1cm,
  >={Latex[length=2mm]},
  every node/.style={inner sep=1pt},
  lab/.style={font=\large},
  diag/.style={inner sep=0pt}
]

\node[lab] at (0.0, 4.2) {$\mathbb{C}S_0$};
\node[lab] at (2.2, 4.2) {$\mathbb{C}S_1$};
\node[lab] at (4.6, 4.2) {$\mathbb{C}S_2$};
\node[lab] at (7.0, 4.2) {$\mathbb{C}S_3$};
\node[lab] at (9.4, 4.2) {$\mathbb{C}S_4$};

\node (e) at (0.0,0.0) {$\varnothing$};

\node[diag] (a) at (2.0,0.0) {$\ydiagram{1}$};

\node[diag] (b1) at (4.2,1.0) {$\ydiagram{2}$};
\node[diag] (b2) at (4.2,-1.0) {$\ydiagram{1,1}$};

\node[diag] (c1) at (6.6,2.0) {$\ydiagram{3}$};
\node[diag] (c2) at (6.6,0.0) {$\ydiagram{2,1}$};
\node[diag] (c3) at (6.6,-2.0) {$\ydiagram{1,1,1}$};

\node[diag] (d1) at (9.0,3) {$\ydiagram{4}$};
\node[diag] (d2) at (9.0,1.25) {$\ydiagram{3,1}$};
\node[diag] (d3) at (9.0,-0.5) {$\ydiagram{2,2}$};
\node[diag] (d4) at (9.0,-2.55) {$\ydiagram{2,1,1}$};
\node[diag] (d5) at (9.0,-5.05) {$\ydiagram{1,1,1,1}$};

\draw[->] (e) -- (a);

\draw[->] (a) -- (b1);
\draw[->] (a) -- (b2);

\draw[->] (b1) -- (c1);
\draw[->] (b1) -- (c2);

\draw[->] (b2) -- (c2);
\draw[->] (b2) -- (c3);

\draw[->] (c1) -- (d1);
\draw[->] (c1) -- (d2);

\draw[->] (c2) -- (d2);
\draw[->] (c2) -- (d3);
\draw[->] (c2) -- (d4);

\draw[->] (c3) -- (d4);
\draw[->] (c3) -- (d5);

\end{tikzpicture}%
}
\caption{The Bratteli diagram of the family of symmetric group algebras $\mathbb{C}\left[S_0\right]\hookrightarrow \mathbb{C}\left[S_1\right]\hookrightarrow \cdots \hookrightarrow \mathbb{C}\left[S_4\right]$, also known as the Young lattice (cf.\ \autoref{ex:young_lattice}). The Bratteli diagram of the permutation matrix algebras $\left(\Al_k^d\right)_{k}$ is obtained by removing all vertices violating $\ell\left(\lambda\right)\leq d$.}\label{fig:Bratelli_S_n} 
\end{figure}
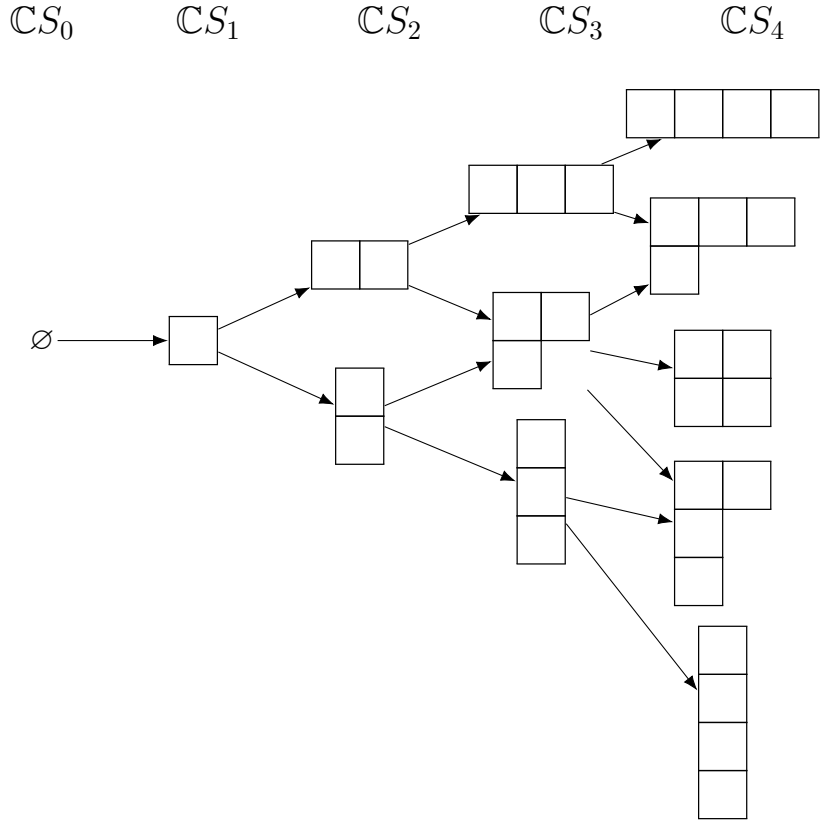

\subsubsection{Gelfand-Tsetlin subalgebra and Jucys-Murphy elements}\label{subsubsec:GT_subalgebra_JM}

Throughout, let $\Al_0,\ldots,\Al_n$ denote a multiplicity-free family with Bratteli diagram $\mathscr{A}$. For a leaf $\lambda \in \hat{\Al}_n$, the restriction of $V_\lambda$ to $\Al_{n-1}$ decomposes multiplicity-free by axiom (c), and iterating the restriction along the tower towards the root, whose algebra $\Al_0\cong \mathbb{C}$ only admits one-dimensional simple modules, splits $V_\lambda$ into a direct sum of distinguished lines: one for each path $T \in \operatorname{Paths}\left(\lambda\right)$, in accordance with \autoref{eq:dimension_counts_paths}. Choosing a unit vector $\ket{T}$ in each line, which is unique up to a phase, yields the Gelfand-Tsetlin basis
\begin{align}\label{eq:def_GT_basis}
    \big\lbrace \ket{T} \ \big\vert \ T \in \operatorname{Paths}\left(\mathscr{A}\right)\big\rbrace \qquad \text{of} \qquad \bigoplus_{\lambda \in \hat{\Al}_n} V_\lambda
\end{align}
(cf.\ \cite[Sec.~2.7.3]{grinko2025mixed}), whose elements are thus labeled by the root-to-leaf paths of the Bratteli diagram. The projectors onto the Gelfand-Tsetlin basis vectors can be constructed explicitly inside $\Al_n$; the relevant subalgebra is generated by the centers $\mathcal{Z}\left(\Al_i\right)$ along the tower.

\begin{definition}[{Gelfand-Tsetlin subalgebra, \cite[Def.~2.7.5]{grinko2025mixed}}]\label{def:GT_subalgebra}
    For each $k \in \left[n\right]$, the Gelfand-Tsetlin subalgebra of $\Al_k$ is
    \begin{align}\label{eq:def_GT_subalgebra}
        \Xl_k \coloneqq \left\langle \mathcal{Z}\left(\Al_1\right),\ldots, \mathcal{Z}\left(\Al_k\right)\right\rangle \subseteq \Al_k,
    \end{align}
    i.e.\ the subalgebra generated by the centers of $\Al_1,\ldots,\Al_k$.
\end{definition}

By construction $\Xl_1 \subseteq \Xl_2 \subseteq \cdots \subseteq \Xl_n$, and each $\Xl_k$ is a maximal commutative subalgebra of $\Al_k$ \cite{Okounkov1996, Vershik2005}. To each path $T = T^0 \rightarrow T^1 \rightarrow \cdots \rightarrow T^n \in \operatorname{Paths}\left(\mathscr{A}\right)$ we associate the element
\begin{align}\label{eq:def_path_idempotent}
    \epsilon_T \coloneqq \epsilon_{T^1}\epsilon_{T^2}\cdots \epsilon_{T^n},
\end{align}
where $\epsilon_{T^i}$ denotes the primitive central idempotent of $\Al_i$ corresponding to $T^i \in \hat{\Al}_i$ as in \autoref{eq:wedderburn_idempotent_decomposition}; since $\epsilon_{T^i}\in \mathcal{Z}\left(\Al_i\right)$ for each $i$, the factors commute and $\epsilon_T \in \Xl_n$.

\begin{proposition}[{\cite[Prop.~1.6 \& Cor.~1.7]{Doty2019}; see also \cite[Prop.~2.7.6]{grinko2025mixed}}]\label{prop:canonical_primitive_idempotents}
    The collection $\lbrace \epsilon_T \ \vert \ T \in \operatorname{Paths}\left(\mathscr{A}\right)\rbrace$ is a family of pairwise orthogonal primitive idempotents in $\Al_n$ that sums to the identity and is a basis of the Gelfand-Tsetlin subalgebra $\Xl_n$. Moreover, the primitive central idempotents of $\Al_n$ are given by
    \begin{align}\label{eq:central_idempotents_from_paths}
        \epsilon_\lambda = \sum_{T \in \operatorname{Paths}\left(\lambda\right)} \epsilon_T.
    \end{align}
\end{proposition}

\begin{remark}\label{rem:idempotents_not_central}
    Under the isomorphism in \autoref{eq:wedderburn_idempotent_decomposition}, the idempotent $\epsilon_T$ with $T \in \operatorname{Paths}\left(\lambda\right)$ corresponds to the rank-one projector $\ketbra{T}{T}$ inside the block $\operatorname{End}\left(V_\lambda\right)$ and vanishes on all other blocks. In particular, while $\epsilon_\lambda$ is central and projects onto the full isotypic component, the $\epsilon_T$ are primitive but in general not central. For $\mathbb{C}\left[S_n\right]$ this distinction is familiar from the Young symmetrizers, which are primitive idempotents adapted to Young's natural (non-orthogonal) basis and central only for the trivial and the sign representation.
\end{remark}

The Gelfand-Tsetlin subalgebras, and with them the idempotents $\epsilon_T$, admit a convenient set of generators.

\begin{definition}[{Jucys-Murphy sequence, \cite[Def.~3.1]{Doty2019}; see also \cite[Def.~2.7.7]{grinko2025mixed}}]\label{def:jucys_murphy_sequence}
    Let $\Al_0,\ldots,\Al_n$ be a multiplicity-free family of algebras with Gelfand-Tsetlin subalgebras $\Xl_1,\ldots,\Xl_n$. A sequence of elements $J_1,\ldots,J_n \in \Al_n$ with $J_k \in \Xl_k$ for each $k\in\left[n\right]$ is called
    \begin{enumerate}[label=(\alph*)]
        \item additively central if $J_1 + \cdots + J_k \in \mathcal{Z}\left(\Al_k\right)$ for all $k \in \left[n\right]$,
        \item separating if $\Xl_k = \left\langle J_1,\ldots,J_k\right\rangle$ for all $k \in \left[n\right]$,
    \end{enumerate}
    and a Jucys-Murphy sequence if it is both additively central and separating.
\end{definition}

Since $\lbrace \epsilon_T \ \vert \ T \in \operatorname{Paths}\left(\mathscr{A}\right)\rbrace$ is a basis of $\Xl_n$ by \autoref{prop:canonical_primitive_idempotents}, every $J_k$ expands as
\begin{align}\label{eq:JM_content_expansion}
    J_k = \sum_{T \in \operatorname{Paths}\left(\mathscr{A}\right)} c_T\left(k\right)\, \epsilon_T, \qquad c_T\left(k\right)\in \mathbb{C},
\end{align}
and we call $c_T \coloneqq \left(c_T\left(1\right),\ldots,c_T\left(n\right)\right)$ the content (or weight) vector of the path $T$; equivalently, $c_T\left(k\right)$ is the eigenvalue of $J_k$ on the Gelfand-Tsetlin basis vector $\ket{T}$. By \cite[Lem.~3.9]{Doty2019}, a consequence of additive centrality, the value $c_T\left(k\right)$ does not depend on the whole path $T$ but only on the vertices $T^{k-1}$ and $T^k$, such that it can be assigned to the edge $T^{k-1}\rightarrow T^k$ of the Bratteli diagram,
\begin{align}\label{eq:edge_contents}
    c_{T^{k-1}\rightarrow T^k} \coloneqq c_T\left(k\right);
\end{align}
the separating property furthermore ensures that the values attached to the different outgoing edges of any fixed vertex are pairwise distinct (cf.\ \cite{Doty2019}). Following \cite{Doty2019}, we assign to each edge $\lambda \rightarrow \mu$ between the levels $k-1$ and $k$ the interpolation polynomial
\begin{align}\label{eq:def_interpolation_polynomial}
    P_{\lambda \rightarrow \mu}\left(x\right)\coloneqq \prod_{\tilde{\mu}\,:\, \lambda \rightarrow \tilde{\mu}\neq \mu} \frac{x - c_{\lambda \rightarrow \tilde{\mu}}}{c_{\lambda\rightarrow\mu} - c_{\lambda \rightarrow \tilde{\mu}}},
\end{align}
where the product runs over all edges outgoing from $\lambda$ other than $\lambda \rightarrow \mu$; it is well defined by the pairwise distinctness of the edge values. This yields an explicit algorithm for computing all idempotents from the spectral data of a Jucys-Murphy sequence.

\begin{proposition}[{\cite[Thm.~3.8 \& Thm.~3.11]{Doty2019}; see also \cite[Sec.~2.7.5]{grinko2025mixed}}]\label{prop:idempotent_algorithm}
    Let $J_1,\ldots,J_n$ be a Jucys-Murphy sequence for the multiplicity-free family $\Al_0,\ldots,\Al_n$. With the base case $\epsilon_{\varnothing} = 1$, the primitive central idempotents of $\Al_k$ can be computed recursively for all $k \in \left[n\right]$ and $\mu \in \hat{\Al}_k$ via
    \begin{align}\label{eq:recursion_central_idempotents}
        \epsilon_\mu = \sum_{\lambda\,:\, \lambda\rightarrow \mu} P_{\lambda \rightarrow \mu}\left(J_k\right)\, \epsilon_\lambda,
    \end{align}
    where the sum runs over all edges incoming into $\mu$. Moreover, the canonical primitive idempotents of \autoref{prop:canonical_primitive_idempotents} are given by
    \begin{align}\label{eq:path_idempotent_product_formula}
        \epsilon_T = \prod_{k=1}^n P_{T^{k-1}\rightarrow T^k}\left(J_k\right) = \prod_{k=1}^n \prod_{\mu\,:\, T^{k-1}\rightarrow \mu \neq T^k} \frac{J_k - c_{T^{k-1}\rightarrow \mu}}{c_{T^{k-1}\rightarrow T^k} - c_{T^{k-1}\rightarrow \mu}} \qquad \text{for all} \ T \in \operatorname{Paths}\left(\mathscr{A}\right).
    \end{align}
\end{proposition}

Evaluating these formulas requires exactly the following data about the family $\Al_0,\ldots,\Al_n$: the Bratteli diagram $\mathscr{A}$, a Jucys-Murphy sequence $J_1,\ldots,J_n$, and the edge values $c_{\lambda\rightarrow\mu}$ of \autoref{eq:edge_contents}. For the families relevant to this work, all three are explicit.

\begin{example}[Symmetric group algebras]\label{ex:JM_symmetric_group}
    For the family $\mathbb{C}\left[S_0\right]\hookrightarrow \cdots \hookrightarrow \mathbb{C}\left[S_n\right]$ of \autoref{ex:multiplicity_free_families}, the elements
    \begin{align}\label{eq:JM_elements_symmetric_group}
        J_1 \coloneqq 0, \qquad J_k \coloneqq \sum_{i=1}^{k-1}\left(i,k\right) \quad \text{for} \ 2 \leq k \leq n,
    \end{align}
    with $\left(i,k\right)\in S_k$ the transposition exchanging $i$ and $k$, form a Jucys-Murphy sequence \cite{Jucys1974, Murphy1981}: the partial sum $J_1 + \cdots + J_k$ is the sum of all transpositions in $S_k$ and hence central, and the separating property is established in \cite{Okounkov1996, Vershik2005}. The eigenvalues are given by the cell contents of standard Young tableaux: identifying paths in the Young lattice with standard Young tableaux as in \autoref{ex:young_lattice}, one has $J_k \ket{T} = \operatorname{cont}_k\left(T\right)\ket{T}$, i.e.\ $c_T\left(k\right) = \operatorname{cont}_k\left(T\right)$ and $c_{\mu\rightarrow\lambda} = \operatorname{cont}\left(\lambda \setminus \mu\right)$, such that the content vector of the path $T$ coincides with the content vector of the corresponding tableau introduced in \autoref{sec:Young_diagrams_tableaux} \cite{Okounkov1996, Vershik2005}. Since a surjective homomorphism of semisimple algebras maps the center onto the center, the images $\psi_n^d\left(J_k\right)$ form a Jucys-Murphy sequence for the family $\left(\Al_k^d\right)_k$ of permutation matrix algebras, whose Bratteli diagram is the truncated Young lattice of \autoref{ex:young_lattice} with the same edge values; all data entering \autoref{prop:idempotent_algorithm} are therefore explicit in this case as well (cf.\ \cite{grinko2025mixed}).
\end{example}

Finally, the primitive idempotents $\epsilon_T$ extend to a complete set of matrix units realizing the Wedderburn--Artin decomposition of $\Al_n$ explicitly: for every leaf $\lambda \in \hat{\Al}_n$ and all $T,S \in \operatorname{Paths}\left(\lambda\right)$ there exist elements $E_{T,S}\in \Al_n$ which, under the isomorphism in \autoref{eq:wedderburn_idempotent_decomposition}, act as $\ketbra{T}{S}$ on the block $\operatorname{End}\left(V_\lambda\right)$ and vanish on all other blocks, such that
\begin{align}\label{eq:matrix_units_relations}
\begin{split}
    \Al_n &= \operatorname{span}_{\mathbb{C}}\big\lbrace E_{T,S} \ \big\vert \ T,S \in \operatorname{Paths}\left(\lambda\right), \ \lambda \in \hat{\Al}_n\big\rbrace, \\
    E_{T,S}\, E_{T^\prime\!,S^\prime} &= \delta_{S,T^\prime}\, E_{T,S^\prime}, \qquad \sum_{T \in \operatorname{Paths}\left(\mathscr{A}\right)} E_{T,T} = 1,
\end{split}
\end{align}
with $E_{T,T} = \epsilon_T$ (cf.\ \cite[Sec.~2.7.6]{grinko2025mixed}). In the following subsection this framework is specialized to the algebras appearing in Schur--Weyl duality.

\subsection{Orthonormal basis for Schur-Weyl duality}\label{subsec:schur_weyl_orthonormal_basis}

In this subsection we specialize the framework of \autoref{subsec:gelfand_tsetlin_basis} to the two mutually commuting actions on the tensor space $\mathcal{H}^{\otimes n}$ with $d = \dim \mathcal{H}$: the tensor representation $\psi_n^d$ of $S_n$ from \autoref{eq:def_permutation_action_Sn_H_otimes_n} and the diagonal action of the unitary group,
\begin{align}\label{eq:def_diagonal_unitary_action}
    \phi_n^d : \mathcal{U}\left(d\right)\to \mathcal{B}\left(\mathcal{H}^{\otimes n}\right), \quad U \mapsto U^{\otimes n}.
\end{align}
Besides the permutation matrix algebra $\Al_n^d = \psi_n^d\left(\mathbb{C}\left[S_n\right]\right)$ we consider the algebra
\begin{align}\label{eq:def_diagonal_unitary_algebra}
    \mathcal{U}_n^d \coloneqq \operatorname{span}_{\mathbb{C}}\lbrace U^{\otimes n} \ \vert \ U \in \mathcal{U}\left(d\right)\rbrace,
\end{align}
which remains unchanged if the span is taken over $\operatorname{GL}\left(d\right)$, $\operatorname{SL}\left(d\right)$ or $\operatorname{SU}\left(d\right)$ instead of $\mathcal{U}\left(d\right)$ (cf.\ \cite[Lem.~2.10.2]{grinko2025mixed}); in particular, statements about $\mathcal{U}\left(d\right)$ and $\operatorname{GL}\left(d\right)$ are interchangeable in this context. The two algebras are full mutual commutants.

\begin{theorem}[{Schur-Weyl duality, cf.~\cite[Chap.~9]{Fulton2004}, \cite[Thm.~2.10.3]{grinko2025mixed} and \cite{harrow2005applicationscoherentclassicalcommunication}}]\label{thm:schur_weyl_duality}
    The algebras $\Al_n^d$ and $\mathcal{U}_n^d$ are each other's commutants in $\mathcal{B}\left(\mathcal{H}^{\otimes n}\right)$,
    \begin{align}\label{eq:sw_commutants}
        \mathcal{U}_n^d = \operatorname{End}_{\Al_n^d}\left(\mathcal{H}^{\otimes n}\right), \qquad \Al_n^d = \operatorname{End}_{\mathcal{U}_n^d}\left(\mathcal{H}^{\otimes n}\right),
    \end{align}
    and $\psi_n^d$ is faithful, i.e.\ $\Al_n^d \cong \mathbb{C}\left[S_n\right]$, if and only if $d \geq n$. As an $S_n\times \mathcal{U}\left(d\right)$-bimodule, the tensor space decomposes as
    \begin{align}\label{eq:schur_weyl_decomposition_appendix}
        \mathcal{H}^{\otimes n} \cong \bigoplus_{\lambda \vdash_d n} U_\lambda \otimes V_\lambda,
    \end{align}
    where $V_\lambda$ and $U_\lambda$ are the Specht and Weyl modules introduced in \autoref{sec:notation_preliminaries}, carrying the irreducible actions $\psi_\lambda : S_n \to \mathcal{B}\left(V_\lambda\right)$ and $\phi_\lambda : \mathcal{U}\left(d\right)\to \mathcal{B}\left(U_\lambda\right)$. Equivalently, there exists a unitary $U_{\operatorname{Sch}}$ on $\mathcal{H}^{\otimes n}$, the Schur transform, such that for all $\sigma \in S_n$ and $U \in \mathcal{U}\left(d\right)$
    \begin{align}\label{eq:schur_transform_intertwining}
        U_{\operatorname{Sch}}\, \psi_n^d\left(\sigma\right) U_{\operatorname{Sch}}^\dagger = \bigoplus_{\lambda \vdash_d n} 1_{U_\lambda}\otimes \psi_\lambda\left(\sigma\right), \qquad U_{\operatorname{Sch}}\, \phi_n^d\left(U\right) U_{\operatorname{Sch}}^\dagger = \bigoplus_{\lambda \vdash_d n} \phi_\lambda\left(U\right)\otimes 1_{V_\lambda}.
    \end{align}
\end{theorem}

Two consequences are used repeatedly in \autoref{sec:efficient_inner_sequence}. First, on the level of the algebras, \autoref{eq:schur_weyl_decomposition_appendix} realizes the Wedderburn--Artin decompositions
\begin{align}\label{eq:sw_isotypic_algebras}
    \Al_n^d \cong \bigoplus_{\lambda \vdash_d n} \operatorname{End}\left(V_\lambda\right), \qquad \mathcal{U}_n^d \cong \bigoplus_{\lambda \vdash_d n} \operatorname{End}\left(U_\lambda\right),
\end{align}
such that the commutant of the permutation action decomposes into blocks
\begin{align}\label{eq:sw_commutant_blocks}
    \operatorname{End}_{\mathbb{C}\left[S_n\right]}\left(\mathcal{H}^{\otimes n}\right) = \operatorname{End}_{\Al_n^d}\left(\mathcal{H}^{\otimes n}\right) \cong \bigoplus_{\lambda \vdash_d n} \operatorname{End}\left(U_\lambda\right)\otimes 1_{V_\lambda}.
\end{align}
Second, comparing dimensions in \autoref{eq:schur_weyl_decomposition_appendix} yields the identity $d^{\,n} = \sum_{\lambda \vdash_d n} m_\lambda\, d_\lambda$; the number of summands and the dimensions $m_\lambda$ are polynomial in $n$ for fixed $d$, whereas the $d_\lambda$ are in general exponential (cf.\ \autoref{lem:bounds_symmetric_group}). The remainder of this subsection equips both tensor factors of every block with their Gelfand-Tsetlin bases and describes the resulting Schur basis and the transforms between it and the computational basis.

\subsubsection{Representation theory of $S_n$}\label{subsubsec:rep_theory_Sn}

The irreducible representations of $S_n$, equivalently the simple $\mathbb{C}\left[S_n\right]$-modules, are exactly the Specht modules $V_\lambda$ with $\lambda \vdash n$ and $d_\lambda = \dim V_\lambda = \lvert \operatorname{SYT}\left(\lambda\right)\rvert$ given by the hook length formula (cf.\ \autoref{eq:tableaux_count_dimensions} and \autoref{eq:hook_length_formula}); this classification follows, for instance, from the Okounkov--Vershik approach outlined in \autoref{subsec:gelfand_tsetlin_basis} \cite{Okounkov1996, Vershik2005, sagan2013symmetric}. For the quotient algebra $\Al_n^d$ the simple modules are precisely the Specht modules surviving the truncation,
\begin{align}\label{eq:simple_modules_permutation_matrix_algebra}
    \hat{\Al}_n^d = \lbrace \lambda \vdash n \ \vert \ \ell\left(\lambda\right)\leq d\rbrace,
\end{align}
reflecting the non-trivial kernel of $\psi_n^d$ for $d < n$: all isotypic components of partitions with more than $d$ rows are annihilated, cf.\ \autoref{ex:multiplicity_free_families} and \cite[Sec.~2.8]{grinko2025mixed}.

By \autoref{ex:JM_symmetric_group}, the Gelfand-Tsetlin basis $\lbrace \ket{T}\rbrace$ of a Specht module $V_\lambda$ is labeled by the paths $T \in \operatorname{Paths}\left(\lambda\right)$ in the Young lattice, equivalently by the standard Young tableaux $\operatorname{SYT}\left(\lambda\right)$ via the growth sequences of \autoref{sec:Young_diagrams_tableaux}. A third equivalent encoding is the Yamanouchi word of $T$: the tuple $\left(w_1,\ldots,w_n\right)$ where $w_k$ is the row index of the cell added at the $k$-th step of the growth sequence. For example, the path
\begin{align}\label{eq:example_yamanouchi_word}
    T = \Big(\varnothing, \ \ydiagram{1}\,, \ \ydiagram{2}\,, \ \ydiagram{3}\,, \ \ydiagram{3,1}\, \Big)
\end{align}
has Yamanouchi word $\left(1,1,1,2\right)$ (cf.\ \cite[p.~21]{grinko2025mixed}). The action of $S_n$ on this basis is given by Young's orthogonal form, which is fully determined by the axial distances of \autoref{sec:Young_diagrams_tableaux}: writing $s_i \coloneqq \left(i, i+1\right)\in S_n$ for the adjacent transpositions, which generate $S_n$, one has the following classical result.

\begin{theorem}[{Young--Yamanouchi basis, \cite[Thm.~2.8.3]{grinko2025mixed}; cf.~\cite{james2006representation, Okounkov1996}}]\label{thm:young_yamanouchi_basis}
    Let $\lambda \vdash n$ and $T \in \operatorname{SYT}\left(\lambda\right)$, let $u_i, u_{i+1}$ denote the cells of $T$ containing $i$ and $i+1$, and set $r_i\left(T\right)\coloneqq \operatorname{ax}\left(u_i, u_{i+1}\right)$. Then for every $i \in \left[n-1\right]$
    \begin{align}\label{eq:youngs_orthogonal_form}
        \psi_\lambda\left(s_i\right)\ket{T} = \frac{1}{r_i\left(T\right)}\, \ket{T} + \sqrt{1 - \frac{1}{r_i\left(T\right)^2}}\ \ket{s_i T},
    \end{align}
    where $s_i T$ denotes the filling obtained from $T$ by exchanging the entries $i$ and $i+1$. If $i$ and $i+1$ lie in the same row or the same column of $T$, then $r_i\left(T\right) = 1$ or $r_i\left(T\right) = -1$, respectively, the second term vanishes, and $\ket{T}$ is an eigenvector with eigenvalue $\pm 1$; otherwise $\lvert r_i\left(T\right)\rvert \geq 2$ and $s_i T \in \operatorname{SYT}\left(\lambda\right)$.
\end{theorem}

\begin{remark}[Relation to Young symmetrizers]\label{rem:young_symmetrizers}
    The Young--Yamanouchi basis is the unitarization of the classical Young symmetrizer
    construction. For $T \in \operatorname{SYT}\left(\lambda\right)$ with row group $R_T$ and column
    group $C_T$, the normalized Young symmetrizer
    $e_T \coloneqq \frac{d_\lambda}{n!}\sum_{q \in C_T}\sum_{p \in R_T}\operatorname{sgn}\left(q\right) q\, p$
    is a primitive idempotent realizing the Specht module as
    $\mathbb{C}\left[S_n\right]e_T \cong V_\lambda$, adapted to Young's natural, non-orthogonal basis
    \cite{Fulton2004, james2006representation}. For distinct standard tableaux of the same shape
    these idempotents are only triangularly orthogonal --- products vanish in one order of the
    last-letter order, but not in both --- and performing Gram--Schmidt along this order turns the
    family $\lbrace e_T \rbrace$ into Young's seminormal units \cite{james2006representation},
    i.e.\ precisely the pairwise orthogonal canonical idempotents
    $\lbrace \epsilon_T\rbrace$ of \autoref{prop:canonical_primitive_idempotents} under the
    identification of paths with standard tableaux; accordingly, $\ket{T}$ spans the image
    $\epsilon_T V_\lambda$. The two presentations differ sharply in computational cost: $e_T$ is a
    sum over $\lvert R_T\rvert \cdot \lvert C_T \rvert$ group elements and thus has, in general,
    exponentially many terms, whereas the interpolation formula of \autoref{prop:idempotent_algorithm}
    assembles $\epsilon_T$ from the Jucys-Murphy spectra in time polynomial in $n$. This is the
    $S_n$-counterpart of the comparison between the semistandard and the Gelfand-Tsetlin bases of
    the Weyl modules in the discussion preceding \autoref{thm:bose_de_Finetti}.
\end{remark}

In particular, all representing matrices in this basis are real orthogonal with entries that are square roots of rational numbers. Together with the Jucys-Murphy elements and the idempotent algorithm of \autoref{subsubsec:GT_subalgebra_JM}, whose spectral data for $S_n$ are the cell contents, \autoref{thm:young_yamanouchi_basis} makes the representation theory of $S_n$ and of $\Al_n^d$ completely explicit. The branching along the tower is Young's rule:

\begin{proposition}[{Young's branching rule, \cite[Thm.~2.8.3]{sagan2013symmetric}}]\label{prop:young_branching_rule}
    For every $\lambda \vdash k$ the restriction of the Specht module $V_\lambda$ to $S_{k-1}$ decomposes multiplicity-free as
    \begin{align}\label{eq:young_branching_rule}
        \operatorname{Res}^{S_k}_{S_{k-1}} V_\lambda \cong \bigoplus_{\mu \in \lambda^-} V_\mu, \qquad \text{in particular} \qquad \dim V_\lambda = \sum_{\mu \in \lambda^-} \dim V_\mu.
    \end{align}
    Consequently, for every $\mu \in \lambda^-$ there is an isometric inclusion $V_\mu \hookrightarrow V_\lambda$ of $S_{k-1}$-modules, unique up to a phase, and the inclusions for different $\mu$ have pairwise orthogonal ranges.
\end{proposition}

Since the number of removable and addable cells is controlled by \autoref{sec:Young_diagrams_tableaux}, we have $\lvert \lambda^-\rvert \leq \ell\left(\lambda\right)\leq d$ and $\lvert \lbrace \nu \in \lambda^+ \ \vert \ \ell\left(\nu\right)\leq d\rbrace\rvert \leq d$ for every $\lambda \vdash_d k$, i.e.\ within Schur-Weyl duality every vertex of the branching graph has at most $d$ incoming and at most $d$ outgoing edges.

\begin{remark}[Branching and truncation]\label{rem:branching_truncation}
    The truncation to $\lambda \vdash_d k$ is compatible with the branching in both directions. Restricting downwards, removing a cell can never increase the number of rows, such that $\lambda \vdash_d k$ implies $\mu \vdash_d \left(k-1\right)$ for every $\mu \in \lambda^-$; \autoref{eq:young_branching_rule} therefore holds verbatim for the simple modules of the permutation matrix algebras. Inducing upwards, i.e.\ tensoring with an additional factor $\mathcal{H}$, partitions $\nu \in \mu^+$ with $\ell\left(\nu\right) = d+1$ would appear on the level of $\mathbb{C}\left[S_k\right]$, but the corresponding Weyl modules vanish on $\mathbb{C}^d$, $U_\nu = 0$ (cf.\ \autoref{sec:Young_diagrams_tableaux}), such that these summands are absent from $\mathcal{H}^{\otimes k}$; this is precisely the truncation of the Young lattice in \autoref{fig:Bratelli_S_n}.
\end{remark}

\subsubsection{Representation theory of $\mathcal{U}(d)$}\label{subsubsec:rep_theory_Ud}

The irreducible polynomial representations of $\operatorname{GL}\left(d\right)$ are classified by partitions $\lambda \vdash_d n$, $n \in \mathbb{N}_0$, realized by the Weyl modules $U_\lambda$ with highest weight $\lambda$ and $\dim U_\lambda = m_\lambda$ given by Weyl's dimension formula (cf.\ \autoref{eq:weyl_dimension_formula} and \cite[Chap.~8]{Fulton2004}). Since $\mathcal{U}\left(d\right)$ is a maximal compact subgroup of $\operatorname{GL}\left(d\right)$, every finite-dimensional continuous representation of $\mathcal{U}\left(d\right)$ arises as the restriction of a unique rational representation of $\operatorname{GL}\left(d\right)$, and on the tensor spaces $\left(\mathbb{C}^d\right)^{\otimes n}$ only the polynomial representations occur; the representation theories of $\mathcal{U}\left(d\right)$ and $\operatorname{GL}\left(d\right)$ are therefore interchangeable for our purposes \cite[Sec.~2.9]{grinko2025mixed}. We denote the irreducible action on $U_\lambda$ by $\phi_\lambda$, as in \autoref{thm:schur_weyl_duality}.

The Gelfand-Tsetlin construction applies to the unitary side as well: the chain $\operatorname{GL}\left(1\right)\hookrightarrow \operatorname{GL}\left(2\right)\hookrightarrow \cdots \hookrightarrow \operatorname{GL}\left(d\right)$, embedded as upper-left blocks, is multiplicity-free \cite{Vilenkin1992}. Concretely, the restriction of $U_\lambda$ to $\operatorname{GL}\left(d-1\right)$ decomposes multiplicity-free into the Weyl modules $U_\mu$ of all highest weights $\mu = \left(\mu_1,\ldots,\mu_{d-1}\right)$ interlacing $\lambda$,
\begin{align}\label{eq:GL_branching_interlacing}
    \lambda_1 \geq \mu_1 \geq \lambda_2 \geq \mu_2 \geq \cdots \geq \mu_{d-1}\geq \lambda_d, \qquad \text{abbreviated} \quad \mu \sqsubseteq \lambda.
\end{align}
Iterating along the chain, the resulting Gelfand-Tsetlin basis of $U_\lambda$ is labeled by the tuples $\left(m_1, m_2,\ldots,m_d\right)$ of interlacing highest weights with $m_d = \lambda$, arranged into triangular arrays:
\begin{align}\label{eq:def_GT_pattern}
    M = \begin{pNiceArray}{ccccccc}
        m_{1,d} & & m_{2,d} & & \ldots & & m_{d,d} \\
        & m_{1,d-1} & & \ldots & & m_{d-1,d-1} & \\
        & & & \ddots & & & \\
        & & m_{1,2} & & m_{2,2} & & \\
        & & & m_{1,1} & & &
    \end{pNiceArray},
\end{align}
where $m_j \coloneqq \left(m_{1,j},\ldots,m_{j,j}\right)\in \mathbb{N}_0^j$ denotes the $j$-th row. A Gelfand-Tsetlin pattern of shape $\lambda \vdash_d n$ and length $d$ is such an array with top row $m_d = \left(\lambda_1,\ldots,\lambda_d\right)$, where $\lambda$ is padded with zeros to length $d$, whose entries satisfy the interlacing (in-betweenness) conditions
\begin{align}\label{eq:def_interlacing_condition}
    m_{i,j}\geq m_{i,j-1}\geq m_{i+1,j} \qquad \text{for all} \ 1 \leq i < j \leq d,
\end{align}
i.e.\ $m_1 \sqsubseteq m_2 \sqsubseteq \cdots \sqsubseteq m_d = \lambda$; the set of these is denoted by $\operatorname{GT}\left(\lambda, d\right)$, or $\operatorname{GT}\left(\lambda\right)$ if $d$ is clear from the context. Gelfand-Tsetlin patterns are in bijection with semistandard Young tableaux: given $T \in \operatorname{SSYT}\left(\lambda, d\right)$, the pattern entries
\begin{align}\label{eq:GT_SSYT_bijection}
\begin{split}
    m_{i,j} &= \lvert \lbrace \text{entries} \leq j \ \text{in row} \ i \ \text{of} \ T\rbrace\rvert, \qquad \text{equivalently}\\
    m_{i,j}-m_{i,j-1} &= \lvert\lbrace \text{entries} = j \ \text{in row} \ i \ \text{of} \ T\rbrace\rvert,
\end{split}
\end{align}
define the corresponding $M \in \operatorname{GT}\left(\lambda,d\right)$, such that $m_j$ is the shape of the subtableau of entries at most $j$; in particular,
\begin{align}\label{eq:GT_pattern_count}
    \lvert \operatorname{GT}\left(\lambda, d\right)\rvert = \lvert \operatorname{SSYT}\left(\lambda,d\right)\rvert = m_\lambda,
\end{align}
consistent with \autoref{eq:tableaux_count_dimensions}. The weight of a pattern $M$ is the vector of consecutive row-sum differences,
\begin{align}\label{eq:def_GT_weight}
    w\left(M\right)_j \coloneqq \lvert m_j \rvert - \lvert m_{j-1}\rvert \quad \text{with} \quad \lvert m_j\rvert \coloneqq \sum_{i=1}^j m_{i,j}, \qquad j \in \left[d\right],
\end{align}
and coincides with the weight of the corresponding semistandard tableau from \autoref{sec:Young_diagrams_tableaux}. The Gelfand-Tsetlin basis vectors $\ket{M}$, $M \in \operatorname{GT}\left(\lambda,d\right)$, diagonalize the maximal torus,
\begin{align}\label{eq:GT_torus_weight_action}
    \phi_\lambda\left(\operatorname{diag}\left(u_1,\ldots,u_d\right)\right)\ket{M} = \Big(\prod_{j=1}^d u_j^{\,w\left(M\right)_j}\Big)\ket{M}, \qquad u_1,\ldots,u_d \in \mathcal{U}\left(1\right).
\end{align}
Explicit formulas for the action of the remaining generators of the Lie algebra $\mathfrak{gl}_d$ in this basis go back to Gelfand and Tsetlin \cite{gelfand1950matrix} and can be found in \cite[p.\ 363]{Vilenkin1992} and \cite[Thm.\ 2.9.5]{grinko2025mixed}; we do not need them here. A concrete realization of the basis, which also fixes the phases, is obtained recursively in the next two sections: starting from the computational basis of $ U_{(1)}\simeq\mathbb{C}^d$, the isometries $C_{\mu \leftarrow \lambda}$ with the explicit real coefficients of \autoref{subsubsec:CG_coefficients_GT} embed every Gelfand--Tsetlin vector into $\HS^{\otimes k}$, cf.\ \autoref{lem:schur_transform_efficiency}.

\begin{remark}[Gelfand-Tsetlin versus semistandard bases]\label{rem:GT_vs_semistandard}
    The bijection \autoref{eq:GT_SSYT_bijection} identifies the index sets of two natural bases
    of the Weyl module $U_\lambda$, which should not be confused. The classical construction
    realizes $U_\lambda$ inside $\mathcal{H}^{\otimes n}$ as the image of a Young symmetrizer
    associated with a fixed standard tableau of shape $\lambda$, and the vectors obtained by
    applying the symmetrizer to the computational basis vectors labeled by the semistandard
    tableaux $S \in \operatorname{SSYT}\left(\lambda,d\right)$ form the \emph{semistandard basis}
    \cite[Chap.~8]{Fulton2004}. Both bases consist of weight vectors, with weights matched by
    \autoref{eq:GT_SSYT_bijection} and \autoref{eq:def_GT_weight}, so the transition matrix is
    block diagonal along weight spaces and, with respect to a suitable order on
    $\operatorname{SSYT}\left(\lambda,d\right)$, triangular. The difference lies in the metric
    properties: the semistandard basis has integral coordinates in the computational basis, but
    it is neither orthogonal nor canonical --- the Young symmetrizer is not self-adjoint, its
    image depends on the chosen standard tableau, and inner products are encoded in a nontrivial
    Gramian (cf.\ \cite[Sec.~2.10]{sagan2013symmetric} for the analogous discussion for Specht
    modules). Consequently, coordinates with respect to the semistandard basis define an algebra
    isomorphism but not a $*$-isomorphism: adjoints, positivity, POVM elements and isometries are
    not preserved coordinatewise. The Gelfand-Tsetlin basis is the orthonormalization singled out
    by the multiplicity-free chain \autoref{eq:GL_branching_interlacing} --- canonical up to
    phases, with no auxiliary choice of tableau --- and it is precisely this property that the
    algorithms of \autoref{sec:efficient_inner_sequence} rely on, e.g.\ for the positivity of the
    measurement blocks in \autoref{prop:coarse_grained_blocks} and for the isometry property of
    the Pieri inclusions $C_{\mu\leftarrow\lambda}$ of \autoref{subsubsec:CG_series_transform};
    working semistandardly instead would reintroduce exactly the Gramian bookkeeping alluded to
    in the discussion preceding \autoref{thm:bose_de_Finetti}.
\end{remark}

\subsubsection{The Clebsch-Gordan series and transform}\label{subsubsec:CG_series_transform}

Combining the two preceding subsubsections, every block $U_\lambda\otimes V_\lambda$ of \autoref{eq:schur_weyl_decomposition_appendix} carries the orthonormal basis $\lbrace \ket{M}\otimes \ket{T}\rbrace$ with $M \in \operatorname{GT}\left(\lambda,d\right)$ and $T \in \operatorname{SYT}\left(\lambda\right)$, and we call
\begin{align}\label{eq:def_schur_basis}
    \big\lbrace \ket{\lambda, M, T} \ \big\vert \ \lambda \vdash_d n, \ M \in \operatorname{GT}\left(\lambda, d\right), \ T \in \operatorname{SYT}\left(\lambda\right)\big\rbrace
\end{align}
the Schur basis of $\mathcal{H}^{\otimes n}$; the Schur transform $U_{\operatorname{Sch}}$ of \autoref{thm:schur_weyl_duality} is the basis change from the computational basis to \autoref{eq:def_schur_basis}. Its standard construction is recursive and rests on the decomposition of tensor products of Weyl modules. For highest weights $\lambda, \mu$ of polynomial representations of $\operatorname{GL}\left(d\right)$, the Littlewood--Richardson rule gives
\begin{align}\label{eq:littlewood_richardson}
    U_\lambda \otimes U_\mu \cong \bigoplus_{\nu \vdash_d \lvert \lambda\rvert + \lvert \mu\rvert} U_\nu \otimes \mathbb{C}^{c_{\lambda\mu}^{\nu}}
\end{align}
with multiplicities $c_{\lambda\mu}^{\nu}\in \mathbb{N}_0$, the Littlewood--Richardson coefficients \cite[Chap.~5]{Fulton2004}; the right-hand side is referred to as the Clebsch-Gordan series. A unitary $\operatorname{CG}_{\lambda,\mu}$ achieving this decomposition explicitly, i.e.\ satisfying
\begin{align}\label{eq:def_CG_transform}
    \operatorname{CG}_{\lambda,\mu}\left(\phi_\lambda\left(U\right)\otimes \phi_\mu\left(U\right)\right)\operatorname{CG}_{\lambda,\mu}^\dagger = \bigoplus_{\nu \vdash_d \lvert\lambda\rvert + \lvert \mu\rvert} \phi_\nu\left(U\right)\otimes 1_{c_{\lambda\mu}^{\nu}} \qquad \text{for all} \ U \in \mathcal{U}\left(d\right),
\end{align}
is called a Clebsch-Gordan transform, and its matrix entries with respect to Gelfand-Tsetlin bases are the Clebsch-Gordan (or Wigner) coefficients \cite{Biedenharn1968, Vilenkin1992}. For the recursive construction only the special case of Pieri's rule is needed: if $\mu = \left(m\right)$ is a single row, then $c_{\lambda\mu}^\nu = 1$ if $\nu$ is obtained from $\lambda$ by adding $m$ cells with no two in the same column, and $c_{\lambda\mu}^\nu = 0$ otherwise (cf.\ \cite[Thm.~2.9.6]{grinko2025mixed}); dually for a single column and rows. In particular, for the defining representation $\mu = \left(1\right)$, i.e.\ $U_{\left(1\right)}\cong_{\mathcal{U}\left(d\right)} \mathbb{C}^d$, the tensor product decomposes multiplicity-free according to the addable cells,
\begin{align}\label{eq:pieri_single_box}
    U_\mu \otimes \mathbb{C}^d \cong \bigoplus_{\substack{\lambda \in \mu^+ \\ \ell\left(\lambda\right)\leq d}} U_\lambda,
\end{align}
which is exactly the branching used in \autoref{sec:efficient_inner_sequence}: for every $\mu \vdash_d \left(k-1\right)$ and $\lambda \in \mu^+$ with $\ell\left(\lambda\right)\leq d$ there is an isometric inclusion $C_{\mu \leftarrow \lambda}: U_\lambda \hookrightarrow U_\mu \otimes \mathbb{C}^d$, unique up to a phase, and the inclusions for different $\lambda$ have pairwise orthogonal ranges.

The Schur transform is now built by coupling one tensor factor at a time \cite{bacon2005quantumschurtransformi, Bacon2006, harrow2005applicationscoherentclassicalcommunication}: identifying $\mathcal{H}^{\otimes k} \cong \mathcal{H}^{\otimes\left(k-1\right)}\otimes \mathcal{H}$ and applying \autoref{eq:pieri_single_box} blockwise defines, for $k = 2,\ldots,n$, the unitary Clebsch-Gordan step
\begin{align}\label{eq:CG_step_spaces}
    \operatorname{CG}^{\left(k\right)}: \Big(\bigoplus_{\mu \in \hat{\Al}_{k-1}^d} U_\mu \otimes V_\mu\Big)\otimes \mathbb{C}^d \longrightarrow \bigoplus_{\lambda \in \hat{\Al}_{k}^d} U_\lambda \otimes V_\lambda,
\end{align}
acting on the Schur basis of level $k-1$ and a computational basis vector $\ket{x}$, $x \in \left[d\right]$, as
\begin{align}\label{eq:CG_transform_action}
    \operatorname{CG}^{\left(k\right)}\left(\ket{\mu, M, T}\otimes \ket{x}\right) = \sum_{\substack{\lambda\,:\, \mu \rightarrow \lambda}} \ \sum_{\substack{N \in \operatorname{GT}\left(\lambda, d\right)\\ w\left(N\right) = w\left(M\right) + e_x}} c^{x}_{N,M}\, \ket{\lambda, N, T \rightarrow \lambda},
\end{align}
where $\mu \rightarrow \lambda$ ranges over the edges of the truncated Young lattice, $T \rightarrow \lambda$ denotes the path $T$ extended by the vertex $\lambda$, $e_x \in \mathbb{N}_0^d$ is the $x$-th standard basis vector, and the $c^x_{N,M}\in \mathbb{R}$ are the Clebsch-Gordan coefficients of \autoref{subsubsec:CG_coefficients_GT}; this is the specialization of \cite[Eq.~4.13]{grinko2025mixed} to the non-mixed case. The full transform is the cascade
\begin{align}\label{eq:schur_transform_cascade}
    U_{\operatorname{Sch}} = \operatorname{CG}^{\left(n\right)}\left(\operatorname{CG}^{\left(n-1\right)}\otimes 1\right)\cdots \left(\operatorname{CG}^{\left(2\right)}\otimes 1^{\otimes\left(n-2\right)}\right),
\end{align}
cf.\ \cite[Eq.~4.14]{grinko2025mixed}. Each step leaves the Specht registers untouched apart from the path relabeling: for every $\sigma \in S_{k-1}$,
\begin{align}\label{eq:CG_intertwining_Sk}
    \operatorname{CG}^{\left(k\right)}\bigg(\Big(\bigoplus_{\mu \in \hat{\Al}_{k-1}^d} 1_{U_\mu}\otimes \psi_\mu\left(\sigma\right)\Big)\otimes 1_{\mathbb{C}^d}\bigg)\operatorname{CG}^{\left(k\right)\dagger} = \bigoplus_{\lambda \in \hat{\Al}_k^d} 1_{U_\lambda}\otimes \Big(\bigoplus_{\substack{\mu\,:\, \mu\rightarrow \lambda}}\psi_\mu\left(\sigma\right)\Big),
\end{align}
where the inner direct sum is the decomposition of $\operatorname{Res}^{S_k}_{S_{k-1}}V_\lambda$ from \autoref{prop:young_branching_rule}; this is the specialization of \cite[Eq.~4.20]{grinko2025mixed}. Consequently, the cascade implements the Gelfand-Tsetlin construction of \autoref{subsec:gelfand_tsetlin_basis} on both registers simultaneously, and the basis realized on the Specht registers coincides, including phases, with the Young--Yamanouchi basis of \autoref{thm:young_yamanouchi_basis} \cite[Sec.~4.2.2]{grinko2025mixed, harrow2005applicationscoherentclassicalcommunication}. Reading \autoref{eq:CG_transform_action} backwards also identifies the reverse construction used in \autoref{sec:efficient_inner_sequence}: the matrix elements of the isometries $C_{\mu\leftarrow \lambda}$ of \autoref{eq:pieri_single_box} with respect to Gelfand-Tsetlin bases are
\begin{align}\label{eq:def_C_isometry_matrix_elements}
    \left(\bra{M}\otimes \bra{x}\right) C_{\mu \leftarrow \lambda} \ket{N} = c^x_{N,M}, \qquad M \in \operatorname{GT}\left(\mu, d\right), \ x \in \left[d\right], \ N \in \operatorname{GT}\left(\lambda,d\right),
\end{align}
i.e.\ the reverse (dual) Clebsch-Gordan transform requires no data beyond the coefficients $c^x_{N,M}$ themselves.

The following efficiency statements summarize why all of the above is computationally accessible.

\begin{lemma}[{Efficiency of the Schur transform, \cite{bacon2005quantumschurtransformi, Bacon2006, harrow2005applicationscoherentclassicalcommunication} and \cite[Lem.~4.2.1]{grinko2025mixed}}]\label{lem:schur_transform_efficiency}
    The Schur transform $U_{\operatorname{Sch}}$ on $\mathcal{H}^{\otimes n}$ with $d = \dim \mathcal{H}$ can be implemented by a quantum circuit of size polynomial in $n$, $d$ and $\log\left(1/\epsilon\right)$, where $\epsilon$ is the accuracy in operator norm. Moreover, every individual matrix entry $\bra{\lambda, M, T} U_{\operatorname{Sch}}\ket{x_1,\ldots,x_n}$ can be computed classically in time $n^{O\left(d^2\right)}$, i.e.\ in time polynomial in $n$ for fixed $d$.
\end{lemma}
 
 Alternatively, the block structure \autoref{eq:sw_commutant_blocks} of the commutant can be accessed directly: for fixed $d$, the algorithms of \cite{Litjens2016, polak2020new} block-diagonalize $\operatorname{End}_{\mathbb{C}\left[S_n\right]}\left(\mathcal{H}^{\otimes n}\right)\cong \bigoplus_{\lambda\vdash_d n} \mathbb{C}^{m_\lambda \times m_\lambda}$ in time polynomial in $n$, albeit with respect to non-orthonormal representative sets constructed from Young symmetrizers and semistandard tableaux (cf.\ \autoref{rem:GT_vs_semistandard}). The compressions in \autoref{sec:efficient_inner_sequence} instead compute all blocks directly in the Gelfand-Tsetlin bases, recursively along the cascade from the data of \autoref{prop:cascade_data_poly}; cf.\ the proof of \autoref{prop:coarse_grained_blocks}.

\subsubsection{Clebsch-Gordan coefficients via Gelfand-Tsetlin patterns}\label{subsubsec:CG_coefficients_GT}

It remains to discuss the computation of the Clebsch-Gordan coefficients $c^x_{N,M}$ entering \autoref{eq:CG_transform_action} and \autoref{eq:def_C_isometry_matrix_elements}. Explicit closed-form expressions, going back to \cite[Chap.~18]{Vilenkin1992}, are collected in \cite[Sec.~4.A.1]{grinko2025mixed}, \cite[App.~A]{grinko2023gelfandtsetlinbasispartiallytransposed} and \cite[App.~C]{Bastin_2025}; we summarize the structural facts needed in \autoref{sec:efficient_inner_sequence}.

Fix $\mu \vdash_d \left(k-1\right)$, $M \in \operatorname{GT}\left(\mu,d\right)$ and $x \in \left[d\right]$. The coefficient $c^x_{N,M}$ can only be non-zero if the shape $\lambda$ of $N$ satisfies $\lambda \in \mu^+$ with $\ell\left(\lambda\right)\leq d$ and if the weight selection rule $w\left(N\right) = w\left(M\right) + e_x$ of \autoref{eq:CG_transform_action} holds. Combinatorially, every such $N$ is obtained from $M$ by choosing integers $i_x, i_{x+1},\ldots,i_d$ with $1 \leq i_j \leq j$ and incrementing the entry $m_{i_j, j}$ by one in each of the rows $j = x,\ldots,d$, i.e.\ the difference $N - M$ is a triangular $\lbrace 0,1\rbrace$-shift pattern supported in the top $d-x+1$ rows; the interpretation on the level of semistandard tableaux is the insertion of a cell with entry $x$ followed by consecutive bumping \cite[Sec.~4.A.1]{grinko2025mixed}, cf.\ also \cite[App.~C]{Bastin_2025}. In particular, for fixed $M$ and $x$ there are at most
\begin{align}\label{eq:CG_support_bound}
    \prod_{j=x}^{d} j \ \leq \ d!
\end{align}
patterns $N$ with $c^x_{N,M}\neq 0$, a constant for fixed $d$. Each non-zero coefficient factorizes into a product of at most $d - x + 1$ reduced Wigner coefficients, one per row of the shift pattern, and each factor is an explicit square root of a rational function of the shifted pattern entries $\ell_{i,j}\coloneqq m_{i,j} - i$ \cite[Chap.~18]{Vilenkin1992}; in particular all $c^x_{N,M}$ are real, a single coefficient can be evaluated exactly using $O\left(d^2\right)$ arithmetic operations, and unitarity of the Clebsch-Gordan transform yields the orthogonality relation
\begin{align}\label{eq:CG_orthogonality}
    \sum_{x \in \left[d\right]}\ \sum_{\mu \in \lambda^- }\ \sum_{M \in \operatorname{GT}\left(\mu, d\right)} c^{x}_{N,M}\, c^{x}_{N^\prime\!,M} = \delta_{\lambda,\lambda^\prime}\,\delta_{N,N^\prime}, \qquad N \in \operatorname{GT}\left(\lambda,d\right), \ N^\prime \in \operatorname{GT}\left(\lambda^\prime, d\right).
\end{align}
We combine these observations into the complexity statement used in the main text.

\begin{proposition}[Polynomial-time data for the Schur-Weyl cascade]\label{prop:cascade_data_poly}
    Fix $d \in \mathbb{N}$. Given $n \in \mathbb{N}$, the following data can be computed in time polynomial in $n$:
    \begin{enumerate}[label=(\alph*)]
        \item the truncated Young lattice up to level $n$, i.e.\ all vertices $\lambda \vdash_d k$, $k \leq n$, together with their edges, the dimensions $d_\lambda$ and $m_\lambda$, and the Jucys-Murphy edge contents of \autoref{ex:JM_symmetric_group};
        \item all Clebsch-Gordan coefficients $\lbrace c^x_{N,M}\rbrace$ that appear in the steps $\operatorname{CG}^{\left(2\right)},\ldots,\operatorname{CG}^{\left(n\right)}$ of the cascade \autoref{eq:schur_transform_cascade};
        \item explicit matrix representations, of size $m_\mu d \times m_\lambda$ each, of all isometries $C_{\mu\leftarrow \lambda}$ with $\lambda \vdash_d k$, $\mu \in \lambda^-$ and $k \leq n$.
    \end{enumerate}
\end{proposition}

\begin{proof}
    The number of vertices at level $k$ is at most $\left(k+d\right)^{d-1}$ and each has at most $d$ incoming and outgoing edges (\autoref{lem:bounds_symmetric_group} and \autoref{prop:young_branching_rule}), the dimensions are given by the explicit formulas \autoref{eq:hook_length_formula} and \autoref{eq:weyl_dimension_formula}, and the edge contents are cell contents, which proves (a). For (b), the number of patterns per shape is $m_\lambda \leq \left(n+d\right)^{d\left(d-1\right)/2}$ by \autoref{lem:bounds_symmetric_group}, for each pattern and each $x \in \left[d\right]$ there are at most $d!$ non-vanishing coefficients by \autoref{eq:CG_support_bound}, and each is evaluated with $O\left(d^2\right)$ operations from the closed-form product of reduced Wigner coefficients \cite[Chap.~18]{Vilenkin1992}, \cite[App.~C]{Bastin_2025}. Claim (c) follows from (b) via \autoref{eq:def_C_isometry_matrix_elements}.
\end{proof}

\end{document}